\documentclass[letterpaper,11pt]{article}
\usepackage{amsmath}
\usepackage{amssymb}
\usepackage{amsfonts}
\usepackage{amsthm}
\usepackage{setspace, verbatim}
\usepackage{bbm}
\usepackage{color}
\usepackage{float}
\usepackage{subcaption}
\usepackage[final]{graphicx}
\usepackage{enumerate}
\usepackage{mathrsfs}
\usepackage[longnamesfirst]{natbib}
\usepackage{graphicx}
\usepackage[usenames,dvipsnames,svgnames,table]{xcolor}
\usepackage{hyperref}
\usepackage{enumitem}
\usepackage{multirow}
\usepackage{pdflscape}
\usepackage{afterpage}
\usepackage{rotating}
\usepackage{capt-of}

\usepackage{tikz}
\usepackage[margin=1.00in]{geometry}
\hypersetup{
     colorlinks   = true,
     citecolor    = Blue,
     urlcolor     = Black,
     linkcolor    = Blue
}

\newtheorem{theorem}{Theorem}[section]
\newtheorem{lemma}{Lemma}[section]
\newtheorem{proposition}{Proposition}[section]
\newtheorem{assumption}{Assumption}[section]

\newtheorem*{defn*}{Definition}
\newtheorem{remark}{Remark}
\newtheoremstyle{theoremd}
  {\topsep}
  {\topsep}
  {\itshape}
  {0pt}
  {\bfseries}
  {.\!\!$^{\,\boldsymbol\prime}$}
  { }
  {\thmname{#1}\thmnumber{ #2}\thmnote{ (#3)}}

\theoremstyle{theoremd}
\newtheorem{aprime}{Assumption}[section]

\newcommand{\ind}{1\!\mathrm{l}}

\newcommand{\mcr}[1]{\mathscr{#1}}
\newcommand{\mb}[1]{\mathbb{#1}}
\newcommand{\wh}[1]{\widehat{#1}}
\newcommand{\wt}[1]{\widetilde{#1}}
\newcommand{\mf}[1]{\mathbf{#1}}
\newcommand{\mr}[1]{\mathrm{#1}}
\newcommand{\mc}[1]{\mathcal{#1}}

\newcommand{\p}{\mb P}
\newcommand{\ul}[1]{\underline{#1}}
\newcommand{\ol}[1]{\overline{#1}}

\makeatletter
\renewcommand\paragraph{\@startsection{paragraph}{4}{\z@}%
                                    {0pt \@plus1ex \@minus.2ex}%
                                    {-1em}%
                                    {\normalfont\normalsize\bfseries}}
\makeatother

\begin{document}

\title{Monte Carlo Confidence Sets for Identified Sets\thanks{%
We are grateful to K. Hirano and three reviewers, S. Bonhomme, and L. P. Hansen for insightful suggestions. We also thank T. Cogley, K. Evdokimov, H. Hong, B. Honor\'e, M. Keane, K. Menzel, M. Koles\'ar, U. M\"uller, J. Montiel Olea, T. Sargent, F. Schorfheide, C. Sims and participants at the 2015 ESWC meetings in Montreal, the September 2015 ``Big Data Big Methods'' Cambridge-INET conference, and workshops at many institutions for useful comments. We thank Keith O'Hara for his excellent RA work on the MCMC simulations and empirical results reported in the first version.}}
\author{Xiaohong Chen\thanks{%
Cowles Foundation for Research in Economics, Yale
University. E-mail address: \texttt{xiaohong.chen@yale.edu}}
\quad Timothy M. Christensen%
\thanks{%
Department of Economics, New York University. E-mail address: \texttt{timothy.christensen@nyu.edu}}
\quad Elie Tamer%
\thanks{%
Department of Economics, Harvard University. E-mail address: \texttt{elietamer@fas.harvard.edu}}}
\date{First draft: August 2015; Revised September 2017}
\maketitle

\begin{abstract}
\singlespacing
\noindent In complicated/nonlinear parametric models, it is generally hard to know whether the model parameters are point identified. We provide computationally attractive procedures to construct confidence sets (CSs) for identified sets of full parameters and of subvectors in models defined through a likelihood or a vector of moment equalities or inequalities.
These CSs are based on level sets of optimal sample criterion functions (such as likelihood or optimally-weighted or continuously-updated GMM criterions). The level sets are constructed using cutoffs that are computed via Monte Carlo (MC) simulations directly from the quasi-posterior distributions of the criterions. We establish new Bernstein-von Mises (or Bayesian Wilks) type theorems for the quasi-posterior distributions of the quasi-likelihood ratio (QLR) and profile QLR  in partially-identified regular models and some non-regular models. These results imply that our MC CSs have \emph{exact} asymptotic frequentist coverage for identified sets of full parameters and of subvectors in partially-identified regular models, and have valid but potentially conservative coverage in models with reduced-form parameters on the boundary. Our MC CSs for identified sets of subvectors are shown to have \emph{exact} asymptotic coverage in models with singularities. We also provide results on uniform validity of our CSs over classes of DGPs that include point and partially identified models. We demonstrate good finite-sample coverage properties of our procedures in two simulation experiments. Finally, our procedures are applied to two non-trivial empirical examples: an airline entry game and a model of trade flows.
\end{abstract}

\newpage

\section{Introduction}

It is often difficult to verify whether parameters in complicated nonlinear structural models are globally point identified. This is especially the case when conducting a sensitivity analysis to examine the impact of various model assumptions on the estimates of parameters of interest, where relaxing some suspect assumptions may lead to loss of point identification. This difficulty of verifying point identification naturally calls for  inference procedures that are valid whether or not the parameters of interest are point identified. Our goal is to contribute to this sensitivity literature by proposing relatively simple inference procedures that allow for partial identification in models defined through a likelihood or a vector of moment equalities or inequalities.

To that extent, we provide computationally attractive and asymptotically valid confidence set (CS) constructions for the identified set $\Theta_I$ of the full vector of parameters $\theta \equiv (\mu,\eta)\in \Theta$,\footnote{Following the literature, the identified set $\Theta_I$ is the argmax of a population criterion over the whole parameter space $\Theta$. A model is point identified if $\Theta_I$ is a singleton, say $\{\theta_0\}$, and partially identified if $\{\theta_0\} \subsetneq \Theta_I \subsetneq \Theta$.} and for the identified sets $M_I$ of subvectors $\mu$. As a sensitivity check in an empirical study, a researcher could report conventional CSs based on inverting a $t$ or Wald statistic, which are valid under point identification only, alongside our new CSs that are asymptotically optimal under point identification and robust to failure of point identification.

Our CS constructions are criterion-function based, as in \cite{CHT} (CHT) and the subsequent literature on CSs for identified sets. That is, contour sets of the sample criterion function
 are used as CSs for $\Theta_I$ and contour sets of the sample profile criterion are used as CSs for $M_I$. However, our CSs are constructed using critical values that are calculated \emph{differently} from those in the existing literature. In two of our proposed CS constructions, we estimate critical values using quantiles of the sample criterion function (or profile criterion) that are simulated from a quasi-posterior distribution, which is formed by combining the sample criterion function with a prior over the model parameter space $\Theta$.\footnote{In correctly-specified likelihood models the quasi-posterior is a true posterior distribution over $\Theta$. We refer to the distribution as a quasi-posterior because we accommodate non-likelihood based models, such as moment-based models with GMM criterions.}
 
We propose three procedures for constructing various CSs. To construct a CS for the identified set $\Theta_I$, our Procedure 1 draws a sample $\{\theta^1,...,\theta^B\}$ from the quasi-posterior, computes the $\alpha$-\emph{quantile of the sample criterion evaluated at the draws}, and then defines our CS $\wh\Theta_\alpha$ for $\Theta_I$ as the contour set at said $\alpha$-quantile. The computational complexity here is simply as hard as the problem of taking draws from the quasi-posterior, a well-researched and understood area in the literature on Monte Carlo (MC) algorithms in Bayesian computation (see, e.g., \cite{liu}, \cite{CRobert}). Many MC samplers (including the popular Markov Chain Monte Carlo (MCMC) algorithms) could, in principle, be used for this purpose. In our simulations and empirical applications, we use an adaptive sequential Monte Carlo (SMC) algorithm that is well-suited to drawing from irregular, multi-modal (quasi-)posteriors and is also easily parallelizable for fast computation (see, e.g., \cite{HS2014}, \cite{DDJ2012}, \cite{DG2014}). Our Procedure 2 produces a CS $\wh M_\alpha$ for $M_I$ of a general subvector using the same draws from the quasi-posterior as in Procedure 1. Here an added computation step is needed to obtain critical values that guarantee the exact asymptotic coverage for $M_I$. Finally, our Procedure 3 CS for $M_I$ of a scalar subvector is simply the contour set of the profiled quasi-likelihood ratio (QLR) with its critical value being the $\alpha$ quantile of a chi-square distribution with one degree of freedom. Our Procedure 3 CS is  simple to compute but is valid only for scalar subvectors.

Our CS constructions are valid for ``optimal'' criterions, which include (but are not limited to) correctly-specified likelihood models, GMM models with optimally-weighted or continuously-updated or GEL criterions,\footnote{Moment inequality-based models are special cases of moment equality-based models as one can add nuisance parameters to transform moment inequalities into moment equalities. Although moment inequality models are allowed, our criterion differs from the popular GMS criterion for moment inequalities in \cite{AndrewsSoares} and others; see Subsections \ref{s:missing}, \ref{s:miq0} and \ref{s:miq}.} or sandwich quasi-likelihoods. For point- or partially-identified regular models, optimal criterions correspond to criterions that satisfy a generalized information equality. But our optimal criterions also allow for some correctly-specified non-regular (or non-standard) models such as models with parameter-dependent support, an important feature of set identified models (see Appendix \ref{ax:pds}). Our Procedure 1 and 2 CSs, $\wh\Theta_\alpha$ and $\wh M_\alpha$, are shown to have \emph{exact} asymptotic coverage for $\Theta_I$ and $M_I$ in potentially partially identified regular models, and are valid but possibly conservative in potentially partially identified models with reduced-form parameters on the boundary (in which the local tangent space is a convex cone). Our Procedure 1 and 2 CSs are also shown to be uniformly valid over DGPs that include both point- and partially identified models (see Appendix \ref{s:uniformity}).
Moreover, our Procedure 2 CS is shown to have \emph{exact} asymptotic coverage for $M_I$ in models with singularities, which are particularly relevant in applications when parameters are close to point-identified or point-identified.
 Our Procedure 3 CS has \emph{exact} asymptotic coverage in regular models that are point-identified\footnote{In fact, all three of our procedures are efficient in point-identified regular models.}. Although theoretically slightly conservative in partially identified models, our Procedure 3 CS performs  well in our simulations and empirical examples.

Our Procedure 1 and 2 CSs are Monte Carlo (MC) based. To establish their theoretical validity, we derive new Bernstein-von Mises (or Bayesian Wilks) type theorems for the (quasi-)posterior distributions of the QLR and profile QLR in partially identified models, allowing for regular models and some important non-regular cases (e.g. models in which the local tangent space is a convex cone, models with singularities, and models with parameter-dependent support). These theorems establish that the (quasi-)posterior distributions of the QLR and profile QLR converge to their frequentist counterparts in regular models; see Section \ref{sec-property} and Appendix \ref{ax:pds} for similar results in some non-regular cases. As an illustration we briefly mention some results for Procedure 1 here:
Section \ref{sec-property} presents conditions under which the sample QLR statistic and the (quasi-)posterior distribution of the QLR both converge to a chi-square distribution with unknown degree of freedom in regular models.\footnote{In point-identified models, Wilks-type asymptotics imply the degree of freedom is equal to the dimension of $\theta$ for QLR statistics. In partially identified models, the degree of freedom is some $d^*$, typically less than or equal to $\dim(\theta)$. The correct $d^*$ may not be easy to infer from the context, which is why we refer to it as ``unknown''.} Appendix \ref{ax:pds} shows that the QLR and the (quasi-)posterior of the QLR both converge to a gamma distribution with scale parameter of 2 and unknown shape parameter in more general partially-identified models. These results ensure that the quantiles of the QLR evaluated at the MC draws from its quasi-posterior consistently estimate the correct critical values needed for Procedure 1 CS to have exact asymptotic coverage for $\Theta_I$. See Section \ref{sec-property} for similar results for the profile QLR and Procedure 2 CSs for $M_I$ for subvectors.

We demonstrate the computational feasibility and good finite-sample coverage of our proposed methods in two simulation experiments: a missing data example and a complete information entry game with correlated payoff shocks. We use the missing data example to illustrate the conceptual difficulties in a transparent way, studying both numerically and theoretically the behaviors of our CSs when this model is partially-identified, close to point-identified, and point-identified. Although the length of a confidence interval for the identified set $M_I$ of a scalar $\mu$ is by definition no shorter than that for $\mu$ itself, our simulations demonstrate that the differences in length between our Procedures 2 and 3 CSs for $M_I$ and the GMS CSs of \cite{AndrewsSoares} for $\mu$ are negligible. Finally, our CS constructions are applied to two real data examples: an airline entry game with correlated payoff shocks and an empirical trade flow model. The airline entry game example has 17 partially-identified structural parameters.  Our empirical findings using Procedures 2 and 3 CSs show that the data are informative about some equilibrium selection probabilities. The trade example has 46 structural parameters. Here, point-identification may be difficult to verify, especially when conducting a sensitivity analysis of restrictive model assumptions.

\paragraph{Literature Review.} Several papers have recently proposed Bayesian (or pseudo Bayesian) methods for constructing CSs for $\Theta_I$ that have correct frequentist coverage properties. See section 3.3 in 2009 NBER working paper version of \cite{MoonSchorfheide}, \cite{Kitagawa}, \cite{NoretsTang}, \cite{KlineTamer}, \cite{LiaoSimoni} and the references therein.
All these papers consider  \emph{separable} regular models and use various renderings of a similar intuition. First, there exists a finite-dimensional reduced-form parameter, say $\phi$, that is (globally) point-identified and $\sqrt n$-consistently and asymptotically normal estimable from the data, and is linked to the model structural parameter $\theta$ via a \emph{known} global mapping. Second, a prior is placed on the reduced-form parameter $\phi$, and third, a classical Bernstein-von Mises theorem stating the asymptotic normality of the posterior distribution for $\phi$ is assumed to hold. Finally, the known global mapping between the reduced-form and the structural parameters is inverted, which, by step 3, guarantees correct coverage for $\Theta_I$ in large samples. In addition to this literature's focus on separable models, it is not clear whether the results there remain valid in various non-regular models we study.

 Our approach is valid regardless of whether the model is separable or not. We show that for general separable or non-separable partially identified likelihood or moment-based models, a {\it local reduced-form reparameterization} exists (see Section \ref{s:suff}). We use this local reparameterization as a proof device to show that the (quasi-)posterior distributions of the QLR and the profile QLR statistics have a frequentist interpretation in large samples. Importantly, since our Procedures 1 and 2 impose priors on the model parameter $\theta$ only, there is no need for obtaining a global reduced-form reparameterization or deriving its dimension to implement our procedures. This is in contrast with the above-mentioned existing Bayesian methods for partially identified separable models, for which researchers need to impose priors on global reduced-form parameters $\phi$ to ensure that its posterior lies on $\{\phi(\theta) : \theta \in \Theta\}$ (i.e. the set of reduced-form parameters consistent with the structural model), which could be difficult even in some empirically relevant separable models; see the airline entry game application in Section \ref{sec-empirical}. Moreover, our new Bernstein-von Mises type theorems for the (quasi-)posterior distributions of the QLR and profile QLR allow for several important non-regular cases in which the local reduced-form parameter is typically not $\sqrt n$-consistent and asymptotically normally estimable.

When specialized to point- or partially-identified likelihood models, our Procedure 1 CS for $\Theta_I$ is equivalent to Bayesian credible set for $\theta$ based on inverting a LR statistic. With flat priors, these CSs are also the highest posterior density (HPD) credible sets. Our general theoretical results imply that HPD credible sets give correct frequentist coverage in partially identified regular models and conservative coverage in some non-standard circumstances. These findings complement those of \cite{MoonSchorfheide} who showed that HPD credible sets can under-cover (in a frequentist sense) in separable partially identified regular models under their conditions.\footnote{Note that this is not a contradiction since our priors are imposed on structural parameters $\theta$ only, violating  Assumption 2 in \cite{MoonSchorfheide}.} In point-identified regular models satisfying a generalized information equality with $\sqrt n$-consistent and asymptotically normally estimable parameters $\theta$, \cite{CH} (CH hereafter) propose constructing CSs for scalar subvectors $\mu$  by taking the upper and lower quantiles of the MCMC draws $\{\mu^1,\ldots,\mu^B\}$ where $(\mu^b,\eta^b) \equiv \theta^b$. Our CS constructions for scalar subvectors are asymptotically equivalent to CH's CSs in such models, but they differ otherwise. Our CS constructions, which are based on quantiles of the {\it criterion} evaluated at the MC draws $\{\theta^1,\ldots,\theta^B\}$ rather than of the raw parameter draws themselves, are valid irrespective of whether the model is point- or partially-identified.
Intuitively, this is because the population criterion is always point-identified irrespective of whether  $\theta$ is point- or partially-identified.

There are several published works on frequentist CS constructions for $\Theta_I$: see, e.g., CHT and  \cite{RomanoShaikh} where subsampling based methods are used for general partially identified models, \cite{Bugni} and \cite{Armstrong14} where bootstrap methods are used for moment inequality models, and \cite{BM} where random set methods are used when $\Theta_I$ is strictly convex. For inference on identified sets of subvectors, both the subsampling-based papers of CHT and  \cite{RomanoShaikh} deliver valid tests with a judicious choice of the subsample size for a profiled criterion function. The subsampling-based CS construction allows for general criterion functions,
but is computationally demanding and sensitive to choice of subsample size in realistic empirical structural models.\footnote{There is a large literature on frequentist approach for {\it inference on the true parameter} $\theta \in \Theta_I$ or $\mu \in M_I$ (e.g., \cite{IM}, \cite{Rosen}, \cite{AndrewsGuggenberger}, \cite{Stoye}, \cite{AndrewsSoares}, \cite{andrews2012inference}, \cite{canay}, \cite{RSW}, \cite{bugni/canay/shi:16} and \cite{KMS} among many others), which generally uses discontinuous-in-parameters asymptotic (repeated sampling) approximations to test statistics.  These existing frequentist methods
are difficult to implement in realistic empirical models.} Our methods are computationally attractive and typically have asymptotically correct coverage, but require ``optimal'' criterion functions.

The rest of the paper is organized as follows. Section \ref{sec-procedure} describes our new procedures for CSs for identified sets $\Theta_I$ and $M_I$. Section \ref{sec-empirical0} presents simulations and real data applications. Section \ref{sec-property} first establishes new BvM (or Bayesian Wilks) results for the QLR and profile QLR in partially identified models. It then derives the frequentist validity of our CSs. Section \ref{s:suff} provides some sufficient conditions to the key regularity conditions for the general theory in Section \ref{sec-property}. Section \ref{sec-conclusion} briefly concludes. Appendix \ref{a:mc} describes the implementation details for the simulations and real data applications in Section  \ref{sec-empirical0}. Appendix \ref{s:uniformity} shows that our CSs for $\Theta_I$ and $M_I$ are valid uniformly over a class of DGPs. 
Appendix \ref{s:uniform-ex} verifies the main regularity conditions for uniform validity in the missing data and a moment inequality examples. Appendix \ref{s:lp} presents results on local power. Appendix \ref{ax:pds} establishes a new BvM (or Bayesian Wilks) result which shows that the limiting (quasi-)posterior distribution of the QLR in a partially identified model is a gamma distribution with unknown shape parameter and scale parameter of 2. There, results on models with parameter-dependent support are given. Appendix \ref{a:proofs} contains all the proofs and additional lemmas.

\section{Description of our Procedures}\label{sec-procedure}

In this section we first describe our method for constructing CSs for $\Theta_I$. We then describe methods for constructing CSs for $M_I$ of any subvector. We finally present an extremely simple method for constructing CSs for $M_I$ of a scalar subvector in certain situations.

Let $\mf X_n = (X_1,\ldots,X_n)$ denote a sample of i.i.d. or strictly stationary and ergodic data of size $n$. Consider a population objective function  $L: \Theta \to \mb R$, such as a log-likelihood function for correctly specified likelihood models, an optimally-weighted or continuously-updated GMM objective function, or a sandwich quasi-likelihood function. The function $L$ is assumed to be an upper semicontinuous function of $\theta$ with $\sup_{\theta \in \Theta} L(\theta) < \infty$.
The population objective $L$ may not be maximized uniquely over $\Theta$, but rather its maximizers, the {\it identified set}, may be a nontrivial set of parameters:
\begin{equation} \label{Theta-I}
 \Theta_I := \left\{ \theta \in \Theta : L(\theta) = \textstyle \sup_{\vartheta \in \Theta} L(\vartheta)\right\}\,.
\end{equation}
The set $ \Theta_I$ is our first object of interest. In many applications, it may be of interest to provide a CS for a {\it subvector} of interest. Write $\theta \equiv (\mu,\eta)$ where $\mu$ is the subvector of interest and $\eta$ is a nuisance parameter. Our second object of interest is the identified set for the subvector $\mu$:
\begin{equation} \label{M-I}
 M_I := \{ \mu : (\mu,\eta) \in \Theta_I \mbox{ for some } \eta \}\,.
\end{equation}
Given the data $\mf X_n$, we seek to construct computationally attractive CSs that cover $\Theta_I$ or $M_I$ with a pre-specified probability (in repeated samples) as sample size $n$ gets large.

To describe our approach, let $L_n$ denote an (upper semicontinuous) sample criterion function that is a jointly measurable function of the data $\mf X_n$ and $\theta$. This objective function $L_n$  can be a natural sample analogue of $L$. We give a few examples of objective functions that we consider.

{\bf Parametric likelihood:}   Given a parametric model: $\{ P_\theta : \theta \in \Theta\},$  with a corresponding density $p_\theta(.)$ (with respect to some dominating measure), the identified set is $\Theta_I = \{ \theta \in \Theta: P_0= P_\theta\}$ where $P_0$ is the true data distribution.  We take $L_n$ to be the average log-likelihood function:
\begin{equation} \label{e:ll}
 L_n(\theta) = \frac{1}{n} \sum_{i=1}^n \log p_\theta(X_i)\,.
\end{equation}

{\bf GMM models:}  Consider a set of {\it moment equalities} $E[\rho_\theta(X_i)] = 0$  such that the solution to this vector of equalities may not be unique. The identified set is  $\Theta_I = \{ \theta \in \Theta: \, E[\rho_\theta(X_i)] = 0\}$. The sample objective function $L_n$ can be the continuously-updated GMM objective function:
\begin{equation} \label{e:cue}
 L_n(\theta) = -\frac{1}{2} \rho_n(\theta)' W_n(\theta)  \rho_n(\theta)
\end{equation}
where $\rho_n(\theta) = \frac{1}{n} \sum_{i=1}^n \rho_\theta(X_i)$ and $W_n(\theta) = \left( \frac{1}{n} \sum_{i=1}^n \rho_\theta(X_i)\rho_\theta(X_i)' - \rho_n(\theta) \rho_n(\theta)' \right)^{-}$ (the superscript $^-$ denotes generalized inverse) for iid data or other suitable choices. Given an optimal weighting matrix $\wh W_n$, we could also use an optimally-weighted GMM objective function:
\begin{equation} \label{e:ow}
 L_n(\theta) = -\frac{1}{2} \rho_n(\theta)' \wh W_n \rho_n(\theta) \,.
\end{equation}
Generalized empirical likelihood objective functions could also be used with our procedures.

Our main CS constructions (Procedures 1 and 2 below) are based on Monte Carlo (MC) simulation methods from a quasi-posterior. Given $L_n$ and a prior $\Pi$ over $\Theta$, the quasi-posterior distribution $\Pi_n$ for $\theta$ given $\mf X_n$ is defined as
\begin{equation} \label{e:posterior}
 \mr d \Pi_n(\theta|\mf X_n) = \frac{e^{nL_n(\theta)} \mr d \Pi(\theta)}{\int_\Theta e^{nL_n(\theta)} \mr d\Pi(\theta)}\,.
\end{equation}
Our procedures 1 and 2 require drawing a sample $\{\theta^1,\ldots,\theta^B\}$ from the quasi-posterior $\Pi_n$. In practice we use an adaptive sequential Monte Carlo (SMC) algorithm which is known to be well suited to drawing from irregular, multi-modal distributions, but any MC sampler could, in principle, be used. The SMC algorithm  is described in detail in Appendix \ref{s:smc}.

\subsection{Confidence sets for the identified set $\Theta_I$}

Here we seek a 100$\alpha$\% CS $\wh \Theta_{\alpha}$ for $\Theta_I$ using $L_n(\theta)$ that has asymptotically exact coverage, i.e.:
\[
 \lim_{n \to \infty} \p(\Theta_I \subseteq \wh \Theta_{\alpha}) = \alpha\,.
\]

 \centerline{\it \large \sc [Procedure 1: Confidence sets for the identified set]}

\begin{enumerate}
\item Draw a sample $\{\theta^1,\ldots,\theta^B \}$ from the quasi-posterior distribution $\Pi_n$ in (\ref{e:posterior}).
\item Calculate the $(1-\alpha)$ quantile of $\{L_n(\theta^1),\ldots,L_n(\theta^B)\}$; call it $\zeta_{n,\alpha}^{mc}$.
\item Our 100$\alpha$\% confidence set for $\Theta_I$ is then:
\begin{equation} \label{e:cs:full}
\wh\Theta_\alpha = \{ \theta \in \Theta : L_n(\theta) \geq \zeta_{n,\alpha}^{mc}\}\,.
\end{equation}
\end{enumerate}
Notice that no optimization of $L_n$ itself is required in order to construct $\wh \Theta_\alpha$. Further, an exhaustive grid search over the full parameter space $\Theta$ is not required as the MC draws $\{\theta^1,\ldots,\theta^B\}$ will concentrate around $\Theta_I$ and thereby indicate the  regions in $\Theta$ over which to search.

CHT  considered inference on the set of minimizers of a {\it nonnegative population criterion function} $Q: \, \Theta \to \mb R_+$ using a sample analogue $Q_n$ of $Q$. Let $\xi_{n,\alpha}$ denote a consistent estimator of the $\alpha$ quantile of $\sup_{\theta \in \Theta_I} Q_n(\theta)$. The 100$\alpha$\% CS for $\Theta_I$ at level $\alpha \in (0,1)$ proposed  is $\wh \Theta_\alpha^{CHT} = \{ \theta \in \Theta : Q_n(\theta) \leq \xi_{n,\alpha}\}$. In the existing literature, subsampling or bootstrap based methods have been used to compute $\xi_{n, \alpha}$ which can be tedious to implement. Instead, our procedure replaces $\xi_{n,\alpha}$ with a cut off based on Monte Carlo simulations.
 The next remark provides an equivalent approach to Procedure 1 but that is constructed in terms of $Q_n$, which is the quasi likelihood ratio statistic associated with $L_n$.

\begin{remark}\label{rmk:full}
Let $\hat \theta \in \Theta$ denote an approximate maximizer of $L_n$, i.e.:
\[
L_n(\hat \theta) = \sup_{\theta \in \Theta}L_n(\theta) + o_\p(n^{-1})\,.
\]
and define the quasi-likelihood ratio (QLR) (at a point $\theta \in \Theta$) as:
\begin{equation} \label{e:qlr}
 Q_n(\theta) = 2n [L_n(\hat \theta) - L_n (\theta)]\,.
\end{equation}
Let $\xi_{n,\alpha}^{mc}$ denote the $\alpha$ quantile of $\{Q_n(\theta^1),\ldots,Q_n(\theta^B)\}$. The confidence set:
\[
 \wh\Theta_\alpha' = \{ \theta \in \Theta : Q_n(\theta) \leq \xi_{n,\alpha}^{mc}\}
\]
is equivalent to $\wh \Theta_\alpha$ defined in (\ref{e:cs:full}) because $L_n(\theta) \geq \zeta_{n,\alpha}^{mc}$ if and only if $Q_n(\theta) \leq \xi_{n,\alpha}^{mc}$.
\end{remark}

In Procedure 1 and Remark \ref{rmk:full} above, the posterior-like quantity involves the use of a prior distribution $\Pi$ over $\Theta$. This prior is user chosen and typically would be the uniform prior but other choices are possible. In our simulations, various choices of prior did not matter much, unless they assigned extremely small mass near the true parameter (which is avoided by using a uniform prior whenever $\Theta$ is compact).

The next lemma presents high-level conditions under which \emph{any} 100$\alpha$\% criterion-based CS for $\Theta_I$  has asymptotically correct (frequentist) coverage. Similar statements appear in CHT. Let $F_W (c):= \Pr(W \leq c )$ denote the (probability) distribution function of a random variable $W$ and $w_{\alpha}:=\inf \{c \in \mb R: F_W (c) \geq \alpha \}$ be the $\alpha$ quantile of $F_W$.

\begin{lemma}\label{l:basic}
Let (i) $\sup_{\theta \in \Theta_I} Q_n(\theta) \rightsquigarrow W$ where $W$ is a random variable for which $F_W$ is continuous at $w_\alpha$, and (ii) $(w_{n,\alpha})_{n \in \mb N}$ be a sequence of random variables such that $w_{n,\alpha} \geq w_\alpha + o_\p(1)$. Define:
\[ \wh \Theta_\alpha = \{ \theta \in \Theta : Q_n(\theta) \leq w_{n,\alpha}\}\,.
\]
Then: $\liminf_{n \to \infty} \mb P(\Theta_I \subseteq \wh \Theta_\alpha) \geq \alpha$. Moreover, if condition (ii) is replaced by the condition $w_{n,\alpha} = w_\alpha + o_\p(1)$, then: $\lim_{n \to \infty} \mb P(\Theta_I \subseteq \wh \Theta_\alpha) = \alpha$.
\end{lemma}

Our MC CSs for $\Theta_I$ are shown to be valid by verifying parts (i) and (ii) with $w_{n,\alpha} = \xi_{n,\alpha}^{mc}$. To verify part (ii), we shall establish a new Bernstein-von Mises (BvM) (or a new Bayesian Wilks) type result for the quasi-posterior distribution of the QLR under loss of identifiability.

\subsection{Confidence sets for the identified set $M_I$ of subvectors}

We seek a CS $\widehat M_{\alpha}$ for $M_I$ such that:
\[
 \lim_{n \to \infty} \p(M_I \subseteq \wh M_{\alpha}) = \alpha\,.
\]
A well-known method to construct a CS for $M_I$ is based on projection, which maps a CS $\wh \Theta_\alpha $ for $\Theta_I$ into one for $M_I$. The projection CS:
\begin{equation}\label{projection}
  \wh {M}_{\alpha}^{proj} = \{\mu : (\mu,\eta) \in \wh \Theta_\alpha \mbox{ for some } \eta \}
\end{equation}
is a valid $100\alpha$\% CS for $M_I$ whenever $\wh \Theta_\alpha$ is a valid $100\alpha$\%  CS for $\Theta_I$. As is well documented, $\wh M_\alpha^{proj}$ is typically conservative, and especially so when the dimension of $\mu$ is small relative to the dimension of $\theta$. Indeed, our simulations below indicate that $\wh M_\alpha^{proj}$ is very conservative even in reasonably low-dimensional parametric models.

We propose CSs for $M_I$ based on a profile criterion for $M_I$. Let $M = \{\mu: (\mu,\eta) \in \Theta \mbox{ for some } \eta\}$ and $H_\mu = \{ \eta : (\mu,\eta) \in \Theta\}$. The profile criterion for a point $\mu \in M$ is $\sup_{\eta \in H_\mu} L_n(\mu,\eta)$,
and the profile criterion for $M_I$ is
\begin{equation}\label{PL-set}
PL_n(M_I) \equiv \inf_{\mu \in M_I} \sup_{\eta \in H_\mu} L_n(\mu,\eta).
\end{equation}
Let $\Delta(\theta^b)$ be an equivalence set for $\theta^b$. In likelihood models we define $\Delta(\theta^b) = \{\theta \in \Theta : p_\theta = p_{\theta^b}\}$ and in moment-based models we define $\Delta(\theta^b) = \{ \theta \in \Theta : E[\rho(X_i,\theta)] = E[\rho(X_i,\theta^b)] \}$. Let $M(\theta^b) = \{ \mu : (\mu,\eta) \in \Delta(\theta^b) \mbox{ for some } \eta\}$, and the profile criterion for $M(\theta^b)$ is
\begin{equation}\label{PL-set-b}
PL_n(M(\theta^b)) \equiv \inf_{\mu \in M(\theta^b)} \sup_{\eta \in H_\mu} L_n(\mu,\eta)\,.
\end{equation}

\bigskip

 \centerline{\it \large \sc [Procedure 2: Confidence sets for subvectors]}

\begin{enumerate}
\item Draw a sample $\{\theta^1,\ldots,\theta^B\}$ from the quasi-posterior distribution $\Pi_n$ in (\ref{e:posterior}).
\item Calculate the $(1-\alpha)$ quantile of $\big\{ PL_n(M(\theta^b)) : b = 1,\ldots,B \big\}$; call it $\zeta_{n,\alpha}^{mc,p}$.
\item Our 100$\alpha$\% confidence set for $M_I$ is then:
\begin{equation} \label{e:cs:subvec}
 \wh M_\alpha = \Big\{  \mu \in M : \sup_{\eta \in H_\mu} L_n(\mu,\eta) \geq \zeta_{n,\alpha}^{mc,p} \Big\} \,.
\end{equation}

\end{enumerate}

By forming  $\wh M_\alpha$ in terms of the profile criterion we avoid having to do an exhaustive grid search over $\Theta$. An additional computational advantage is that the subvectors of the draws, say $\{\mu^1,\ldots,\mu^B\}$, concentrate around $M_I$, thereby indicating the region in $M$ over which to search.

\bigskip

\begin{remark}\label{rmk:subvec}
Recall the definition of the QLR $Q_n$ in (\ref{e:qlr}), we define the profile QLR for the set $M(\theta^b)$ analogously as
\begin{equation}\label{PQLR-set-b}
PQ_n(M(\theta^b)) \equiv 2n [L_n(\hat \theta) - PL_n(M(\theta^b))] \;=\;  \sup_{\mu \in M(\theta^b)} \inf_{\eta \in H_\mu} Q_n (\mu, \eta)\,.
\end{equation}
Let $\xi_{n,\alpha}^{mc,p}$ denote the $\alpha$ quantile of the profile QLR draws $\big\{PQ_n(M(\theta^b)) : b=1,\ldots,B\big\}$. The confidence set:
\[
 \wh M_\alpha' = \Big\{ \mu \in M : \inf_{\eta \in H_\mu} Q_n(\mu,\eta) \leq \xi_{n,\alpha}^{mc,p}\Big\}
\]
is equivalent to $\wh M_\alpha$ because $\sup_{\eta \in H_\mu} L_n(\mu,\eta) \geq \zeta_{n,\alpha}^{mc,p}$ if and only if $\inf_{\eta \in H_\mu} Q_n(\mu,\eta) \leq \xi_{n,\alpha}^{mc,p}$.
\end{remark}

Our Procedure 2 and Remark \ref{rmk:subvec} above are \emph{different} from taking quantiles of the MC parameter draws. A popular percentile CS (denoted as $\wh{M}_{\alpha}^{perc}$) for a scalar subvector $\mu$ is computed by taking the upper and lower $100(1-\alpha)/2$ percentiles of $\{\mu^1,\ldots,\mu^B\}$. For point-identified regular models with $\sqrt n$-consistent and asymptotically normal parameters $\theta$, this approach is known to be valid for correctly-specified likelihood models in the standard Bayesian literature and its validity for criterion-based models satisfying a generalized information equality has been established by \cite{CH}. However, in partially identified models this approach is no longer valid and under-covers, as evidenced in the simulation results below. 

The following result presents high-level conditions under which any 100$\alpha$\% criterion-based CS for $M_I$ is asymptotically valid. A similar statement appears in \cite{RomanoShaikh}.

\begin{lemma}\label{l:basic:profile}
Let (i) $\sup_{\mu \in M_I} \inf_{\eta \in H_\mu} Q_n(\mu,\eta) \rightsquigarrow W$ where $W$ is a random variable for which $F_W$ is continuous at $w_\alpha$, and (ii) $(w_{n,\alpha})_{n \in \mb N}$ be a sequence of random variables such that $w_{n,\alpha} \geq w_\alpha + o_\p(1)$. Define:
\[
 \wh M_\alpha =\Big\{ \mu \in M : \inf_{\eta \in H_\mu} Q_n(\mu,\eta)  \leq  w_{n,\alpha}\Big\}\,.
\]
Then: $\liminf_{n \to \infty} \mb P(M_I \subseteq \wh M_\alpha) \geq \alpha$. Moreover, if condition (ii) is replaced by the condition $w_{n,\alpha} = w_\alpha + o_\p(1)$, then: $\lim_{n \to \infty} \mb P(M_I \subseteq \wh M_\alpha) = \alpha$.
\end{lemma}

Our MC CSs for $M_I$ are shown to be valid by verifying parts (i) and (ii) with $w_{n,\alpha} = \xi_{n,\alpha}^{mc,p}$. To verify part (ii), we shall derive a new BvM type result for the quasi-posterior of the profile QLR under loss of identifiability.

\subsection{A simple but slightly conservative CS for $M_I$ of scalar subvectors}

For a class of partially identified models with one-dimensional subvectors of interest, we now propose another CS $\wh M_\alpha^\chi$ which is extremely simple to construct. This new CS for $M_I$ is slightly conservative (whereas $\wh M_\alpha$ could be asymptotically exact), but its coverage is much less conservative than that of the projection-based CS $\wh {M}_{\alpha}^{proj}$.

\bigskip
\bigskip

 \centerline{\it \large \sc [Procedure 3: Simple conservative CSs for scalar subvectors]}
\begin{enumerate}
\item Calculate a maximizer $\hat \theta$ for which $L_n(\hat \theta) \geq \sup_{\theta \in \Theta} L_n(\theta) + o_\p(n^{-1})$.
\item Our 100$\alpha$\% confidence set for $M_I \subset \mb R$ is then:
\begin{equation} \label{e:mchi}
 \wh M_\alpha^\chi = \Big\{ \mu \in M : \inf_{\eta \in H_\mu} Q_n(\mu,\eta) \leq \chi^2_{1,\alpha} \Big\}
\end{equation}
where $Q_n$ is the QLR in (\ref{e:qlr}) and $\chi^2_{1,\alpha}$ denotes the $\alpha$ quantile of the $\chi^2_1$ distribution.
\end{enumerate}

Procedure 3 above is justified when the limit distribution of the profile QLR for $M_I$ is stochastically dominated by the $\chi^2_1$ distribution (i.e., $F_W (z)\geq F_{\chi^2_1} (z)$ for all $z \geq 0 $ in Lemma \ref{l:basic:profile}). This allows for computationally simple construction using repeated evaluations on a scalar grid.
Unlike $\wh M_\alpha$, the CS $\wh M_\alpha^\chi$ for $M_I$ is typically asymptotically conservative and is only valid for scalar functions of $\Theta_I$ (see Section \ref{s:mchi}). Nevertheless, the CS $\wh M_\alpha^\chi$ is asymptotically exact when $M_I$ happens to be a singleton belonging to the interior of $M$, and, for confidence levels of $\alpha \geq 0.85$, its degree of conservativeness for the set $M_I$ is negligible (see Section \ref{s:mchi}). It is extremely simple to implement and performs very favorably in simulations. As a sensitivity check in empirical estimation of a complicated structural model, one could report the conventional CS based on a $t$-statistic (that is valid under point identification only) as well as our CS $\wh M_\alpha^\chi$ (that remains valid under partial identification); see Section \ref{sec-empirical}.

\section{Simulation Evidence and Empirical Applications}\label{sec-empirical0}

This section presents simulation evidence and empirical applications to demonstrate the good performances of our new procedures for general possibly partially identified models. See Appendix \ref{a:mc} for implementation details.

\subsection{Simulation evidence}

In this subsection we investigate the finite-sample behavior of our proposed CSs in two leading examples of partially identified models: missing data and entry game with correlated payoff shocks. Both have been studied in the existing literature as leading examples of partially-identified moment \emph{inequality} models; we instead use them as examples of likelihood and moment {\it equality} models.

 We use samples of size $n = 100$, $250$, $500$, and $1000$. For each sample, we calculate the posterior quantile of the QLR or profile QLR statistic using $B=10000$ draws from an adaptive SMC algorithm (see Appendix \ref{s:smc} for a description of the algorithm).

\subsubsection{Example 1: missing data}\label{s:missing}

We first consider the simple but insightful missing data example. Suppose we observe a random sample $\{(D_i,Y_iD_i)\}_{i=1}^n$ where both the outcome variable $Y_i$ and the selection variable $D_i$ take values in $\{0,1\}$. The parameter of interest is the true mean $\mu_0 = \mb E[Y_i]$. Without further assumptions, $\mu_0$ is not point identified when $\Pr(D_i = 0) > 0$ as we only observe $Y_i$ when $D_i=1$.

Denote the true probabilities of observing $(D_i,Y_iD_i) = (1,1)$, $(0,0)$ and $(1,0)$ by $\tilde \gamma_{11}$, $\tilde \gamma_{00}$, and $\tilde \gamma_{10} = 1-\tilde \gamma_{11} - \tilde \gamma_{00}$ respectively. We view $\tilde \gamma_{00}$ and $\tilde \gamma_{11}$ as true {\it reduced-form parameters} that are consistently estimable. The reduced-form parameters are functions of the structural parameter $\theta = (\mu, \eta_1, \eta_2)$ where $\mu = \mb E[Y_i]$,  $\eta_1 = \Pr(Y_i = 1|D_i =0)$, and $\eta_2 = \Pr(D_i = 1)$.
Under this model parameterization, $\theta$ is related to the reduced form parameters via $\tilde \gamma_{00} (\theta ) = 1-\eta_2$ and $\tilde \gamma_{11} (\theta ) = \mu - \eta_1 (1- \eta_2)$.
The parameter space $\Theta$ for $\theta$ is defined as:
\begin{equation}\label{e:theta:md}
 \Theta = \{ (\mu , \eta_1 , \eta_2) \in \mb [0,1]^3 : 0 \leq \mu - \eta_1(1-\eta_2) \leq \eta_2  \}\,.
\end{equation}
The identified set for $\theta$ is:
\begin{equation}\label{e:thetaI:md}
 \Theta_I = \{ (\mu,\eta_1,\eta_2) \in \Theta : \tilde \gamma_{00} = 1-\eta_2, \tilde \gamma_{11} = \mu - \eta_1(1-\eta_2) \}.
\end{equation}
Here, $\eta_2$ is point-identified but only an affine combination of $\mu$ and $\eta_1$ are identified.
The identified set for $\mu = E[Y_i]$ is:
\[
M_I = [\tilde \gamma_{11},\tilde \gamma_{11}+\tilde \gamma_{00}]
\]
and the identified set for the nuisance parameter $\eta_1$ is $[0,1]$.

We set the true values of the parameters to be $\mu = 0.5$, $\eta_1 = 0.5$, and take $\eta_2 = 1-c/\sqrt n$ for $c = 0,1,2$ to cover both partially-identified but ``drifting-to-point-identification'' ($c = 1,2$) and point-identified $(c = 0)$ cases. We first implement the procedures using a likelihood criterion and a flat prior on $\Theta$. The likelihood function of $(D_i,Y_iD_i)=(d,yd)$ is
\begin{align*}
 p_\theta(d,yd) & = [\tilde \gamma_{11}(\theta)]^{yd} [1-\tilde \gamma_{11}(\theta)-\tilde \gamma_{00}(\theta)]^{d - y d} [\tilde \gamma_{00}(\theta)]^{1-d}\,.
\end{align*}
In Appendix \ref{a:mc} we present and discuss additional results for a likelihood criterion with a curved prior and a continuously-updated GMM criterion based on the moments $E [ \ind\{D_i = 0 \} - \tilde \gamma_{00} (\theta)] = 0$ and $E [ \ind\{(D_i,Y_iD_i) = (1,1) \} - \tilde \gamma_{11} (\theta) ]  = 0$ with a flat prior (this GMM case may be interpreted as a moment inequality model with $\eta_1(1-\eta_2)$ playing the role of a slackness parameter).

We implement the SMC algorithm as described in Appendix \ref{ax:smc:ex1}. To illustrate sampling via the SMC algorithm and the resulting posterior of the QLR, Figure \ref{f:ex1-plots} displays histograms of the draws for $\mu$, $\eta_1$ and $\eta_2$ for one run of the adaptive SMC algorithm for a sample of size $1000$ with $\eta_2 = 0.8$. Here $\mu$ is partially identified with $M_I = [0.4,0.6]$. The histograms in Figure \ref{f:ex1-plots} show that the draws for $\mu$ and $\eta_1$ are both approximately flat across their identified sets. In contrast, the draws for $\eta_2$, which is point identified, are approximately normally distributed and centered at the MLE. The Q-Q plot in Figure \ref{f:ex1-plots} shows that the quantiles of $Q_n(\theta)$ computed from the draws are very close to the quantiles of a $\chi^2_2$ distribution, as predicted by our theoretical results below.

\begin{figure}
\begin{subfigure}{.5\textwidth}
  \centering
  \includegraphics[width=.8\linewidth]{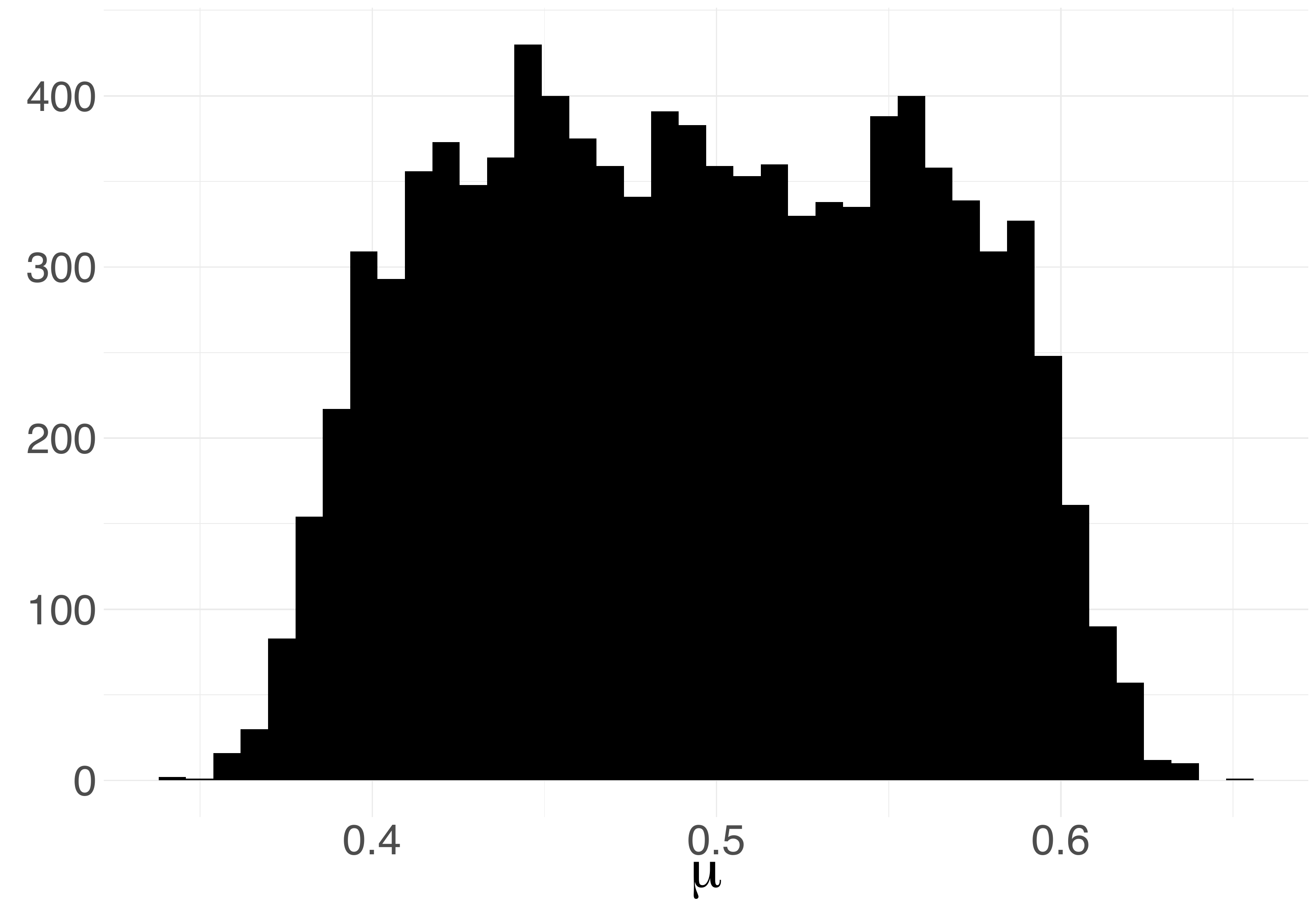}
  \end{subfigure}%
\begin{subfigure}{.5\textwidth}
  \centering
  \includegraphics[width=.8\linewidth]{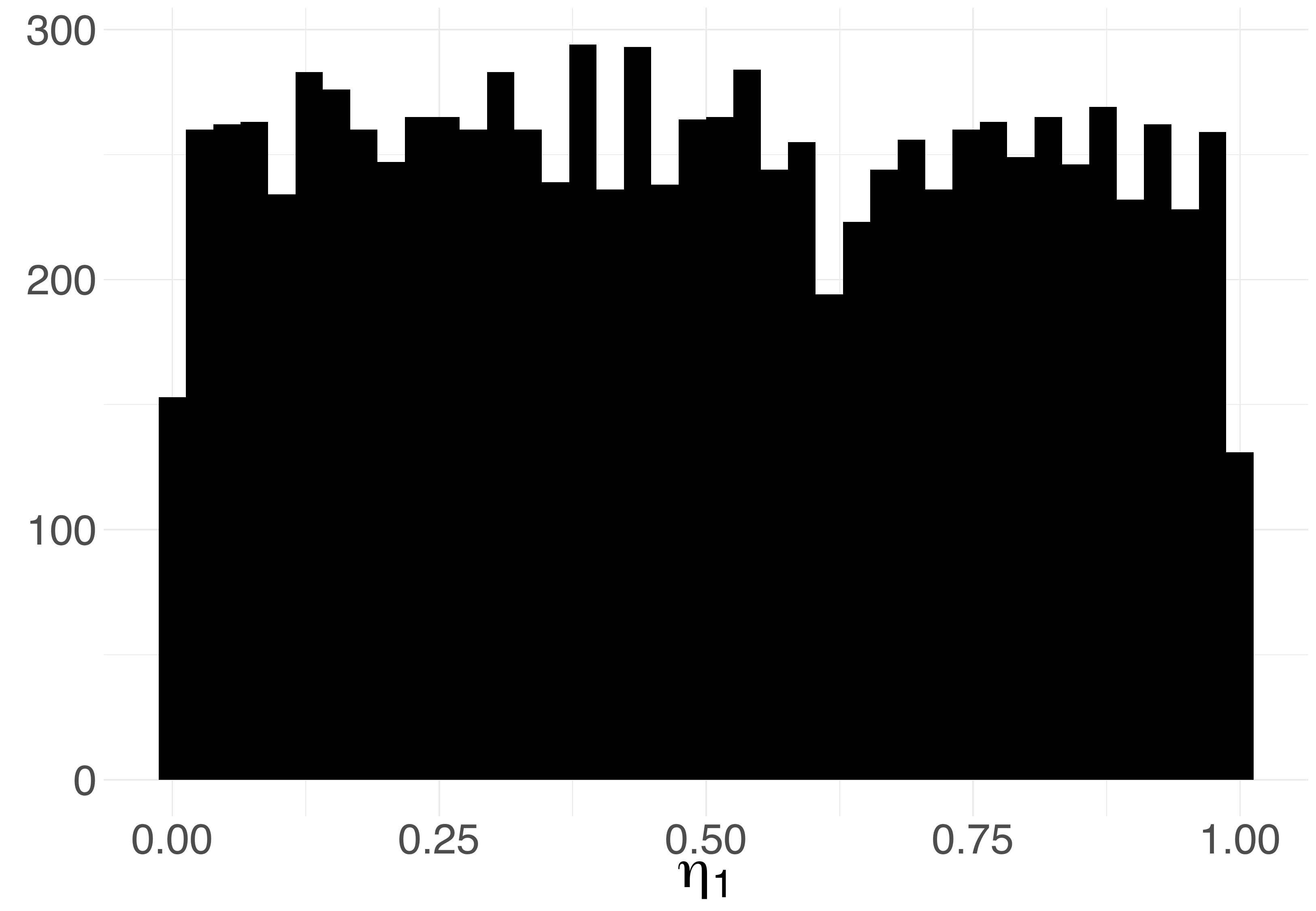}
\end{subfigure}
\begin{subfigure}{.5\textwidth}
  \centering
  \includegraphics[width=.8\linewidth]{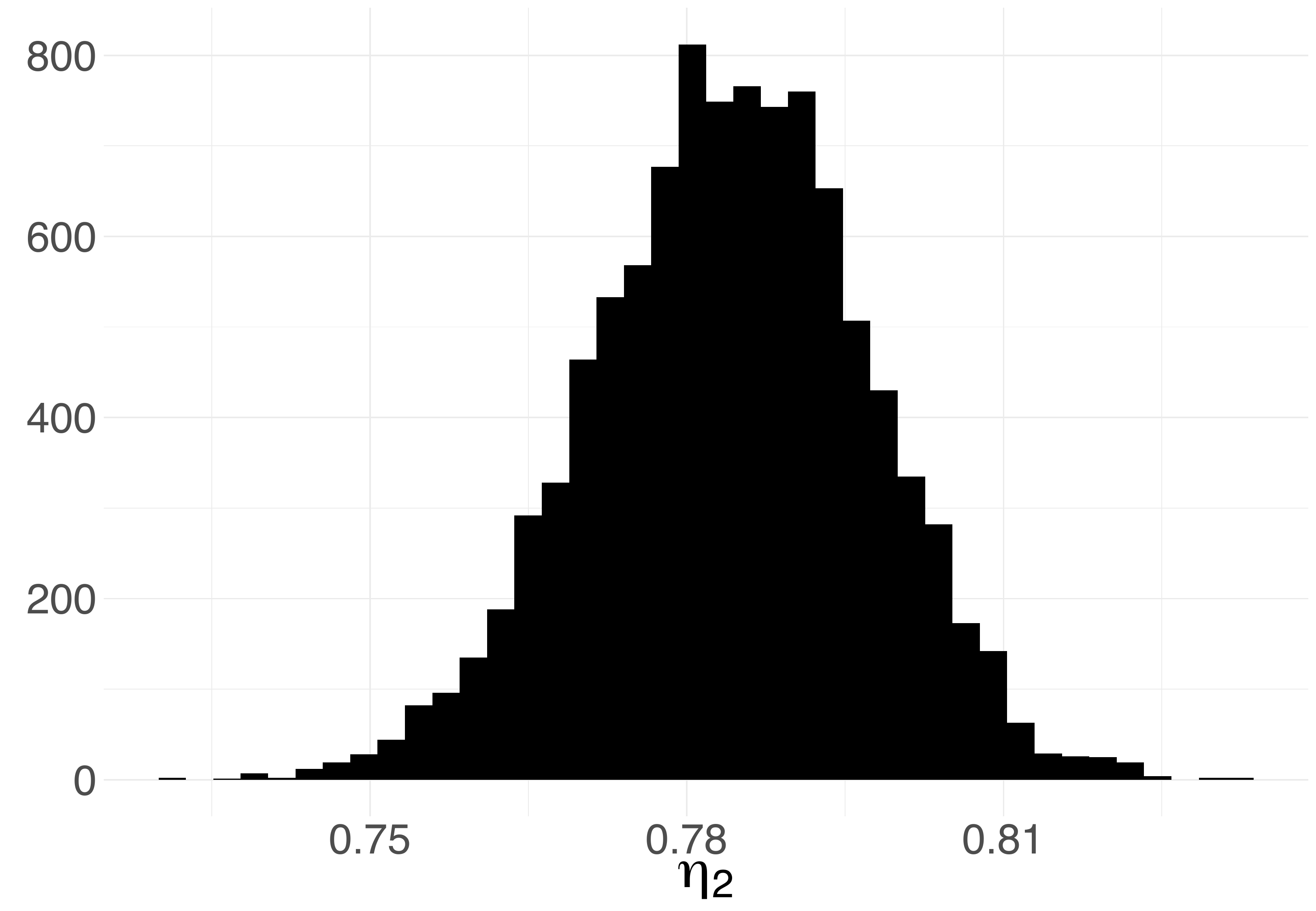}
  \end{subfigure}%
\begin{subfigure}{.5\textwidth}
  \centering
  \includegraphics[width=.8\linewidth]{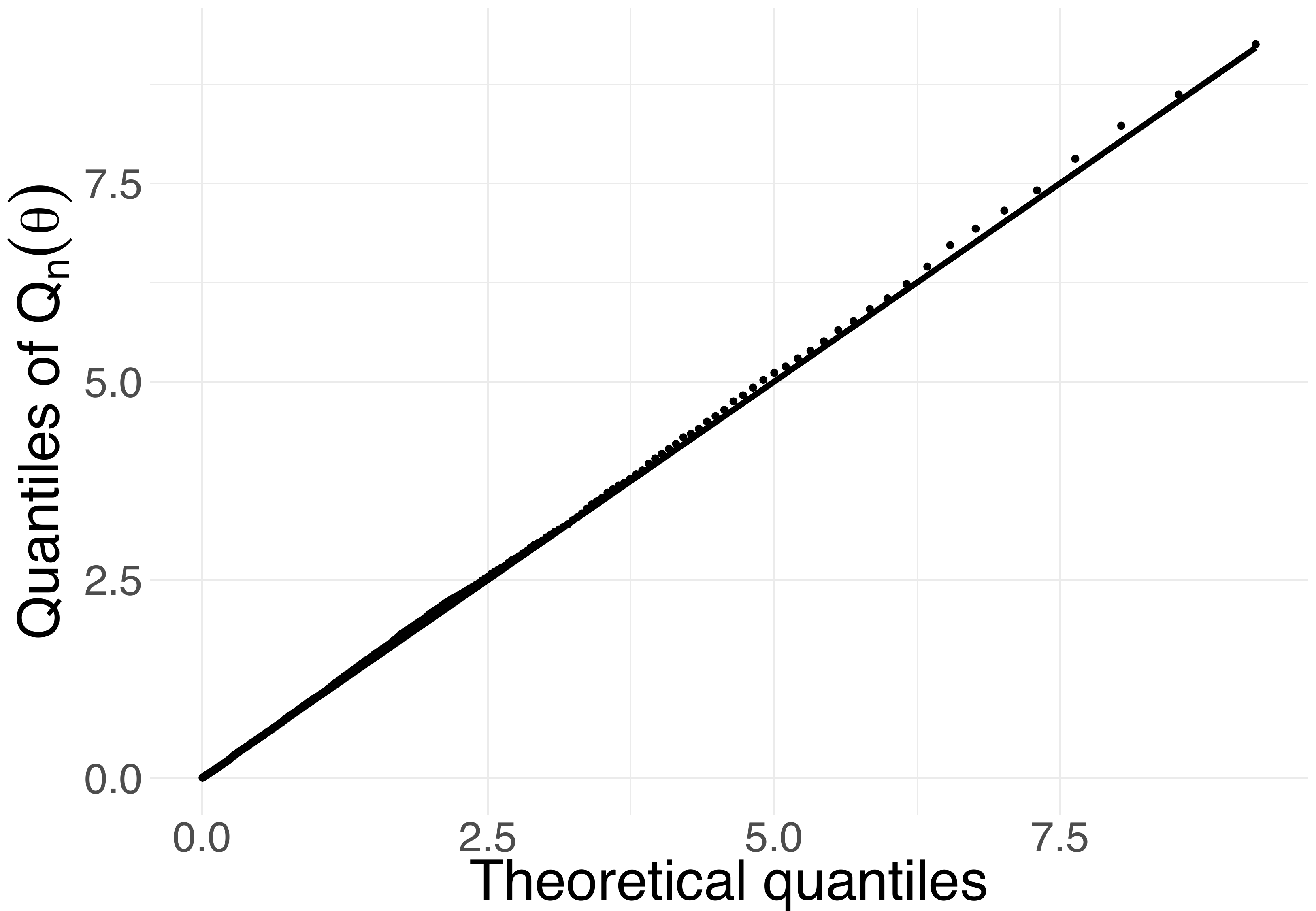}
\end{subfigure}
\centering

\parbox{12cm}{\caption{\small\label{f:ex1-plots} Missing data example: histograms of the SMC draws for $\mu$ (top left), $\eta_1$ (top right), and $\eta_2$ (bottom left) and Q-Q plot of $Q_n(\theta)$ computed from the draws against $\chi^2_2$ quantiles (bottom right) for a sample of size $n = 1000$ with $\eta_2 = 0.8$. The identified sets for $\mu$ and $\eta_1$ are $[0.4,0.6]$ and $[0,1]$, respectively.} }
\end{figure}

\begin{sidewaystable}[p]{\begin{center} \footnotesize
\begin{tabular}{|c|cccccc|cccccc|cccccc|} \hline
	& \multicolumn{6}{c|}{$\eta_2 = 1-\frac{2}{\sqrt n}$} & \multicolumn{6}{c|}{$\eta_2 = 1-\frac{1}{\sqrt n}$}& \multicolumn{6}{c|}{$\eta_2 = 1$ (Point ID)}   \\
	& \multicolumn{2}{c}{0.90} & \multicolumn{2}{c}{0.95} & \multicolumn{2}{c|}{0.99} & \multicolumn{2}{c}{0.90} & \multicolumn{2}{c}{0.95} & \multicolumn{2}{c|}{0.99}  & \multicolumn{2}{c}{0.90} & \multicolumn{2}{c}{0.95} & \multicolumn{2}{c|}{0.99}  \\ \hline
	& & \multicolumn{16}{c}{$\widehat{\Theta}_{\alpha}$ (Procedure 1)} & \\
100     &  .910  & ---   &  .957  & ---   &  .994  & ---   &  .903  & ---   &  .953  & ---   &  .993  & ---   &  .989  & ---   &  .997  & ---   &  1.000  & ---  \\
250     &  .901  & ---   &  .947  & ---   &  .991  & ---   &  .912  & ---   &  .955  & ---   &  .992  & ---   &  .992  & ---   &  .997  & ---   &  1.000  & ---  \\
500     &  .913  & ---   &  .956  & ---   &  .991  & ---   &  .908  & ---   &  .957  & ---   &  .991  & ---   &  .995  & ---   &  .997  & ---   &  \phantom{0}.999  & ---  \\
1000    &  .910  & ---   &  .958  & ---   &  .992  & ---   &  .911  & ---   &  .958  & ---   &  .994  & ---   &  .997  & ---   &  .999  & ---   &  1.000  & ---  \\
	& & \multicolumn{16}{c}{$\widehat{M}_{\alpha}$ (Procedure 2)} & \\
100     &  .920  & [$ .32 ,\! .68 $]  &  .969  & [$ .30 ,\! .70 $]  &  .997  & [$ .27 ,\! .73 $]  &  .918  & [$ .37 ,\! .63 $]  &  .964  & [$ .35 ,\! .65 $]  &  .994  & [$ .32 ,\! .68 $]  &  .911  & [$ .42 ,\! .59 $]  &  .958  & [$ .40 ,\! .60 $]  &  \phantom{0}.990  & [$ .37 ,\! .63 $] \\
250     &  .917  & [$ .39 ,\! .61 $]  &  .961  & [$ .38 ,\! .62 $]  &  .992  & [$ .36 ,\! .64 $]  &  .920  & [$ .42 ,\! .58 $]  &  .963  & [$ .41 ,\! .59 $]  &  .991  & [$ .39 ,\! .61 $]  &  .915  & [$ .45 ,\! .55 $]  &  .959  & [$ .44 ,\! .56 $]  &  \phantom{0}.991  & [$ .42 ,\! .58 $] \\
500     &  .914  & [$ .42 ,\! .58 $]  &  .961  & [$ .41 ,\! .59 $]  &  .993  & [$ .40 ,\! .60 $]  &  .914  & [$ .44 ,\! .56 $]  &  .958  & [$ .43 ,\! .57 $]  &  .992  & [$ .42 ,\! .58 $]  &  .916  & [$ .46 ,\! .54 $]  &  .959  & [$ .46 ,\! .54 $]  &  \phantom{0}.990  & [$ .44 ,\! .56 $] \\
1000    &  .917  & [$ .44 ,\! .56 $]  &  .956  & [$ .44 ,\! .56 $]  &  .993  & [$ .43 ,\! .57 $]  &  .914  & [$ .46 ,\! .54 $]  &  .955  & [$ .45 ,\! .55 $]  &  .993  & [$ .44 ,\! .56 $]  &  .916  & [$ .47 ,\! .53 $]  &  .959  & [$ .47 ,\! .53 $]  &  \phantom{0}.992  & [$ .46 ,\! .54 $] \\
	& & \multicolumn{16}{c}{$\widehat{M}^{\chi}_{\alpha}$ (Procedure 3)} & \\
100     &  .920  & [$ .32 ,\! .68 $]  &  .952  & [$ .31 ,\! .69 $]  &  .990  & [$ .28 ,\! .72 $]  &  .916  & [$ .37 ,\! .63 $]  &  .946  & [$ .36 ,\! .64 $]  &  .989  & [$ .33 ,\! .67 $]  &  .902  & [$ .42 ,\! .58 $]  &  .937  & [$ .41 ,\! .60 $]  &  \phantom{0}.986  & [$ .38 ,\! .63 $] \\
250     &  .915  & [$ .39 ,\! .61 $]  &  .952  & [$ .38 ,\! .62 $]  &  .990  & [$ .36 ,\! .64 $]  &  .914  & [$ .42 ,\! .58 $]  &  .954  & [$ .41 ,\! .59 $]  &  .990  & [$ .39 ,\! .61 $]  &  .883  & [$ .45 ,\! .55 $]  &  .949  & [$ .44 ,\! .56 $]  &  \phantom{0}.991  & [$ .42 ,\! .58 $] \\
500     &  .894  & [$ .42 ,\! .58 $]  &  .954  & [$ .41 ,\! .59 $]  &  .989  & [$ .40 ,\! .60 $]  &  .906  & [$ .44 ,\! .56 $]  &  .949  & [$ .44 ,\! .56 $]  &  .990  & [$ .42 ,\! .58 $]  &  .899  & [$ .46 ,\! .54 $]  &  .945  & [$ .46 ,\! .54 $]  &  \phantom{0}.988  & [$ .44 ,\! .56 $] \\
1000    &  .909  & [$ .44 ,\! .56 $]  &  .950  & [$ .44 ,\! .56 $]  &  .993  & [$ .43 ,\! .57 $]  &  .904  & [$ .46 ,\! .54 $]  &  .954  & [$ .45 ,\! .55 $]  &  .989  & [$ .45 ,\! .55 $]  &  .906  & [$ .48 ,\! .52 $]  &  .946  & [$ .47 ,\! .53 $]  &  \phantom{0}.991  & [$ .46 ,\! .54 $] \\
	& & \multicolumn{16}{c}{$\widehat{M}^{proj}_{\alpha}$ (Projection)} & \\
100     &  .972  & [$ .30 ,\! .70 $]  &  .990  & [$ .28 ,\! .71 $]  &  .999  & [$ .25 ,\! .75 $]  &  .969  & [$ .35 ,\! .65 $]  &  .989  & [$ .33 ,\! .67 $]  &  .998  & [$ .30 ,\! .70 $]  &  .989  & [$ .37 ,\! .63 $]  &  .997  & [$ .36 ,\! .64 $]  &  1.000  & [$ .33 ,\! .67 $] \\
250     &  .971  & [$ .37 ,\! .63 $]  &  .986  & [$ .36 ,\! .64 $]  &  .998  & [$ .34 ,\! .66 $]  &  .976  & [$ .40 ,\! .60 $]  &  .988  & [$ .39 ,\! .61 $]  &  .998  & [$ .37 ,\! .63 $]  &  .992  & [$ .42 ,\! .58 $]  &  .997  & [$ .41 ,\! .59 $]  &  1.000  & [$ .39 ,\! .61 $] \\
500     &  .972  & [$ .41 ,\! .59 $]  &  .985  & [$ .40 ,\! .60 $]  &  .999  & [$ .39 ,\! .61 $]  &  .972  & [$ .43 ,\! .57 $]  &  .989  & [$ .42 ,\! .58 $]  &  .999  & [$ .41 ,\! .59 $]  &  .995  & [$ .44 ,\! .56 $]  &  .997  & [$ .43 ,\! .57 $]  &  \phantom{0}.999  & [$ .42 ,\! .58 $] \\
1000    &  .973  & [$ .44 ,\! .56 $]  &  .990  & [$ .43 ,\! .57 $]  &  .999  & [$ .42 ,\! .58 $]  &  .973  & [$ .45 ,\! .55 $]  &  .988  & [$ .45 ,\! .55 $]  &  .999  & [$ .44 ,\! .56 $]  &  .997  & [$ .45 ,\! .55 $]  &  .999  & [$ .45 ,\! .55 $]  &  1.000  & [$ .44 ,\! .56 $] \\
	& & \multicolumn{16}{c}{$\widehat{M}^{perc}_{\alpha}$ (Percentile)} & \\
100     &  .416  & [$ .38 ,\! .62 $]  &  .676  & [$ .36 ,\! .64 $]  &  .945  & [$ .32 ,\! .68 $]  &  .661  & [$ .40 ,\! .59 $]  &  .822  & [$ .39 ,\! .61 $]  &  .963  & [$ .35 ,\! .65 $]  &  .896  & [$ .42 ,\! .58 $]  &  .946  & [$ .40 ,\! .60 $]  &  \phantom{0}.989  & [$ .38 ,\! .63 $] \\
250     &  .402  & [$ .42 ,\! .58 $]  &  .669  & [$ .41 ,\! .59 $]  &  .917  & [$ .38 ,\! .62 $]  &  .662  & [$ .44 ,\! .56 $]  &  .822  & [$ .43 ,\! .57 $]  &  .960  & [$ .41 ,\! .59 $]  &  .899  & [$ .45 ,\! .55 $]  &  .950  & [$ .44 ,\! .56 $]  &  \phantom{0}.990  & [$ .42 ,\! .58 $] \\
500     &  .400  & [$ .44 ,\! .56 $]  &  .652  & [$ .43 ,\! .57 $]  &  .914  & [$ .42 ,\! .58 $]  &  .652  & [$ .46 ,\! .54 $]  &  .812  & [$ .45 ,\! .55 $]  &  .955  & [$ .43 ,\! .57 $]  &  .903  & [$ .46 ,\! .54 $]  &  .953  & [$ .46 ,\! .54 $]  &  \phantom{0}.988  & [$ .44 ,\! .56 $] \\
1000    &  .405  & [$ .46 ,\! .54 $]  &  .671  & [$ .45 ,\! .55 $]  &  .917  & [$ .44 ,\! .56 $]  &  .662  & [$ .47 ,\! .53 $]  &  .819  & [$ .46 ,\! .54 $]  &  .953  & [$ .45 ,\! .55 $]  &  .905  & [$ .47 ,\! .53 $]  &  .953  & [$ .47 ,\! .53 $]  &  \phantom{0}.990  & [$ .46 ,\! .54 $] \\
	& & \multicolumn{16}{c}{Comparison with GMS CSs for $\mu$ via moment inequalities} & \\
100     &  .815  & [$ .34 ,\! .66 $]  &  .908  & [$ .32 ,\! .68 $]  &  .981  & [$ .29 ,\! .71 $]  &  .803  & [$ .39 ,\! .61 $]  &  .904  & [$ .37 ,\! .63 $]  &  .980  & [$ .34 ,\! .66 $]  &  .889  & [$ .42 ,\! .58 $]  &  .938  & [$ .40 ,\! .60 $]  &  \phantom{0}.973  & [$ .39 ,\! .62 $] \\
250     &  .798  & [$ .40 ,\! .60 $]  &  .899  & [$ .39 ,\! .61 $]  &  .979  & [$ .36 ,\! .63 $]  &  .811  & [$ .43 ,\! .57 $]  &  .897  & [$ .42 ,\! .58 $]  &  .980  & [$ .40 ,\! .60 $]  &  .896  & [$ .45 ,\! .55 $]  &  .944  & [$ .44 ,\! .56 $]  &  \phantom{0}.981  & [$ .42 ,\! .57 $] \\
500     &  .794  & [$ .43 ,\! .57 $]  &  .898  & [$ .42 ,\! .58 $]  &  .976  & [$ .40 ,\! .60 $]  &  .789  & [$ .45 ,\! .55 $]  &  .892  & [$ .44 ,\! .56 $]  &  .975  & [$ .43 ,\! .57 $]  &  .897  & [$ .46 ,\! .54 $]  &  .948  & [$ .46 ,\! .54 $]  &  \phantom{0}.986  & [$ .45 ,\! .55 $] \\
1000    &  .802  & [$ .45 ,\! .55 $]  &  .900  & [$ .44 ,\! .56 $]  &  .978  & [$ .43 ,\! .57 $]  &  .812  & [$ .46 ,\! .54 $]  &  .900  & [$ .46 ,\! .54 $]  &  .978  & [$ .45 ,\! .55 $]  &  .898  & [$ .47 ,\! .53 $]  &  .949  & [$ .47 ,\! .53 $]  &  \phantom{0}.990  & [$ .46 ,\! .54 $] \\ \hline
			\end{tabular}
			\parbox{14cm}{\caption{\small\label{t:ex1} Missing data example: average coverage probabilities for $\Theta_I$ and $M_I$ and average lower and upper bounds of CSs for $M_I$ across 5000 MC replications. Procedures 1--3, Projection and Percentile are implemented using a likelihood criterion and flat prior.}}
	\end{center}	}	
\end{sidewaystable}

\paragraph{Confidence sets for $\Theta_I$:}

The top panel of Table \ref{t:ex1} displays MC coverage probabilities of $\wh \Theta_\alpha$ for 5000 replications. The MC coverage probability should be equal to its nominal value in large samples when $\eta_2 < 1$ (see Theorem \ref{t:main}). It is perhaps surprising that the nominal and MC coverage probabilities are close even in samples as small as $n = 100$. When $\eta_2 = 1$ the CSs for $\Theta_I$ are conservative, as predicted by our theoretical results (see Theorem \ref{t:main:prime}).

\paragraph{Confidence sets for $M_I$:}

We now consider various CSs for the identified set $M_I$ for $\mu$. We first compute the projection CS $\wh M_\alpha^{proj}$, as defined in (\ref{projection}), for $M_I$. As we can see from Table \ref{t:ex1}, this results in conservative CSs for $M_I$. For example, when $\alpha = 0.90$ the projection CSs cover $M_I$ in around 97\% of repeated samples. As the models with $c = 1,2$ are close to point-identified, one might be tempted to report simple percentile CSs $\wh{M}_{\alpha}^{perc}$ for $M_I$ using \cite{CH} procedure, which is valid under point identification, and taking the upper and lower $100(1-\alpha)/2$ quantiles from of the draws for $\mu$.\footnote{Note that we use exactly the same draws for implementing the percentile CS and procedures 1 and 2. As the SMC algorithm uses a particle approximation to the posterior, in practice we compute posterior quantiles for $\mu$ using the particle weights in a manner similar to (\ref{e:quantile-weighted}).} The results in Table \ref{t:ex1} show that $\wh M_\alpha^{perc}$ has correct coverage when $\mu$ is point identified (i.e. $\eta_2 =1$) but it under-covers when $\mu$ is not point identified. For instance, the coverage probabilities of 90\% CSs for $M_I$ are about 66\% with $c = 1$.

In contrast, our criterion-based procedures 2 and 3 remain valid under partial identification. We show below (see Theorem \ref{t:main:profile}) that the coverage probabilities of our procedure 2 CS $\wh M_\alpha$ (for $M_I$) should be equal to their nominal values $\alpha$ when $n$ is large irrespective of whether the model is partially identified with (i.e. $\eta_2 < 1$) or point identified (i.e. $\eta_2 = 1$). The results in Table \ref{t:ex1} show that this is indeed the case, and that the coverage probabilities for procedure 2 are close to their nominal level even for small values of $n$, irrespective of whether the model is point- or partially-identified. In Section \ref{s:md}, we show that the asymptotic distribution of the profile QLR for $M_I$ is stochastically dominated by the $\chi^2_1$ distribution. Table \ref{t:ex1} also presents results for procedure 3 using $\wh M_\alpha^\chi$ as in (\ref{e:mchi}). As we can see from these tables, the coverage results look remarkably close to their nominal values even for small sample sizes and for all values of $\eta_2$.

Finally, we compare our the length of CSs for $M_I$ using procedures 2 and 3 with the length of CSs for the parameter $\mu$ constructed using the generalized moment selection (GMS) procedure of \cite{AndrewsSoares}. We implement their procedure using the inequalities
\begin{align} \label{e:mom-ineq}
 E[\mu - Y_iD_i] & \geq 0 &
 E[ Y_i D_i + (1-D_i) - \mu] & \geq 0
\end{align}
with their smoothing parameter $\kappa_n = (\log n)^{1/2}$, their GMS function $\varphi_j^{(1)}$, and with critical values computed via a multiplier bootstrap. Of course, GMS CSs are for the parameter $\mu$ rather than the set $M_I$, which is why the coverage for $M_I$ reported in Table \ref{t:ex1} appears lower than nominal under partial identification (GMS CSs are known to be asymptotically valid CSs for $\mu$). Importantly, the average lower and upper bounds of our CSs for $M_I$ constructed using Procedures 2 and 3 are very close to those using GMS, whereas projection-based CSs are, in turn, larger. On the other hand, CSs computed using percentiles of the draws for $\mu$ are narrower.

\subsubsection{Example 2: entry game with correlated payoff shocks}\label{s:game}

We now consider the complete information entry game example described in Table \ref{table:simplegame}. We assume that $(\epsilon_1, \epsilon_2)$, observed by the players, are jointly normally distributed with variance 1 and correlation $\rho$, an important parameter of interest. We also assume that $\Delta_1$ and $\Delta_2$ are both negative and that players play a pure strategy Nash equilibrium.  When $-\beta_j \leq \epsilon_j \leq -\beta_j -\Delta_j$, $j=1,2$, the game has two equilibria: for given values of the epsilons in this region, the model predicts $(1,0)$ {\it and} $(0,1)$. Let $D_{a_1 a_2}$ denote a binary random variable taking the value $1$ if and only if player 1 takes action $a_1$ and player 2 takes action $a_2$. We observe a random sample of $\{(D_{00,i},D_{10,i},D_{01,i},D_{11,i})\}_{i=1}^n$. So the data provides information of four choice probabilities $(P(0,0), P(1,0), P(0,1), P(1,1))$, but there are six parameters that need to be estimated: $\theta = (\beta_1, \beta_2, \Delta_1, \Delta_1, \rho, s)$ where $s \in [0,1]$ is the equilibrium selection probability. The model parameter is partially identified as we have 3 non-redundant choice probabilities from which we need to learn about 6 parameters.

\begin{table}[t]
\begin{center}
  \setlength{\extrarowheight}{2pt}
  \begin{tabular}{*{4}{c|}}
    \multicolumn{2}{c}{} & \multicolumn{2}{c}{Player $2$}\\\cline{3-4}
    \multicolumn{1}{c}{} &  & $0$  & $1$ \\\cline{2-4}
    \multirow{2}*{Player $1$}  & $0$ & $\phantom{\beta_1 + \epsilon_1}(0,0)\phantom{0}$ & $\phantom{\beta_2  + \epsilon_2}(0,\beta_2  + \epsilon_2)\phantom{0}$ \\\cline{2-4}
    & $1$ & $(\beta_1 + \epsilon_1 ,0)$ & $(\beta_1 + \Delta_1 + \epsilon_1,\beta_2 + \Delta_2 + \epsilon_2)$ \\\cline{2-4}
  \end{tabular}
  \parbox{12cm}{\caption{\small\label{table:simplegame}Payoff matrix for the binary entry game. The first entry in each cell is the payoff to player 1 and the second entry is the payoff to player 2. }}
\end{center}
\end{table}

We can link the choice probabilities (reduced-form parameters) to $\theta$ via:
\begin{align*}
 \tilde \gamma_{00}(\theta) :=& Q_\rho(\epsilon_1 \leq -\beta_1; \, \epsilon_2 \leq -\beta_2)\\
 \tilde \gamma_{11}(\theta) :=& Q_\rho(\epsilon_1 \geq -\beta_1 - \Delta_1; \, \epsilon_2 \geq -\beta_2 - \Delta_2 )\\
 \tilde \gamma_{10}(\theta) :=& s\times Q_\rho(-\beta_1 \leq \epsilon_1 \leq -\beta_1 - \Delta_1; \, -\beta_2 \leq \epsilon_2 \leq -\beta_2 - \Delta_2)
 \\ & \quad + Q_\rho( \epsilon_1 \geq - \beta_1; \epsilon_2 \leq - \beta_2) + Q_\rho( \epsilon_1 \geq -\beta_1 - \Delta_1 ; -\beta_2 \leq \epsilon_2 \leq -\beta_2 - \Delta_2)
\end{align*}
and $\tilde \gamma_{01}(\theta) = 1-\tilde \gamma_{00}(\theta) - \tilde \gamma_{11}(\theta) - \tilde \gamma_{10}(\theta)$, where $Q_\rho$ denotes the joint probability distribution of $(\epsilon_1,\epsilon_2)$ indexed by the correlation parameter $\rho$. Let $(\tilde \gamma_{00}, \tilde \gamma_{10}, \tilde \gamma_{01}, \tilde \gamma_{11})$ denote the true choice probabilities $(P(0,0), P(1,0), P(0,1), P(1,1))$.
This naturally suggests a likelihood approach, where the likelihood of $(D_{00,i},D_{10,i},D_{11,i},D_{01,i})=(d_{00},d_{10},d_{11},1-d_{00}-d_{10}-d_{11})$ is:
\[
 p_\theta(d_{00},d_{10},d_{11})  = [\tilde \gamma_{00}(\theta)]^{d_{00}}[\tilde \gamma_{10}(\theta)]^{d_{10}}[\tilde \gamma_{11}(\theta)]^{d_{11}}[1-\tilde \gamma_{00}(\theta)-\tilde \gamma_{10}(\theta)-\tilde \gamma_{11}(\theta)]^{1-d_{00}-d_{10}-d_{11}} \,.
\]
In the simulations, we use a likelihood criterion with parameter space:
\[
\Theta = \{ (\beta_1, \beta_2, \Delta_1, \Delta_2, \rho,s) \in \mb R^6 :  (\beta_1,\beta_2) \in [-1,2]^2, \; (\Delta_1,\Delta_2) \in [-2,0]^2, \; (\rho,s) \in [0,1]^2 \}\,.
\]
We simulate the data using $\beta_1 = \beta_2 = 0.2$, $\Delta_1 = \Delta_2 = -0.5$, $\rho = 0.5$ and $s = 0.5$.

\begin{figure}[t]\centering
\begin{subfigure}{.5\textwidth}
  \centering
  \includegraphics[width=.8\linewidth]{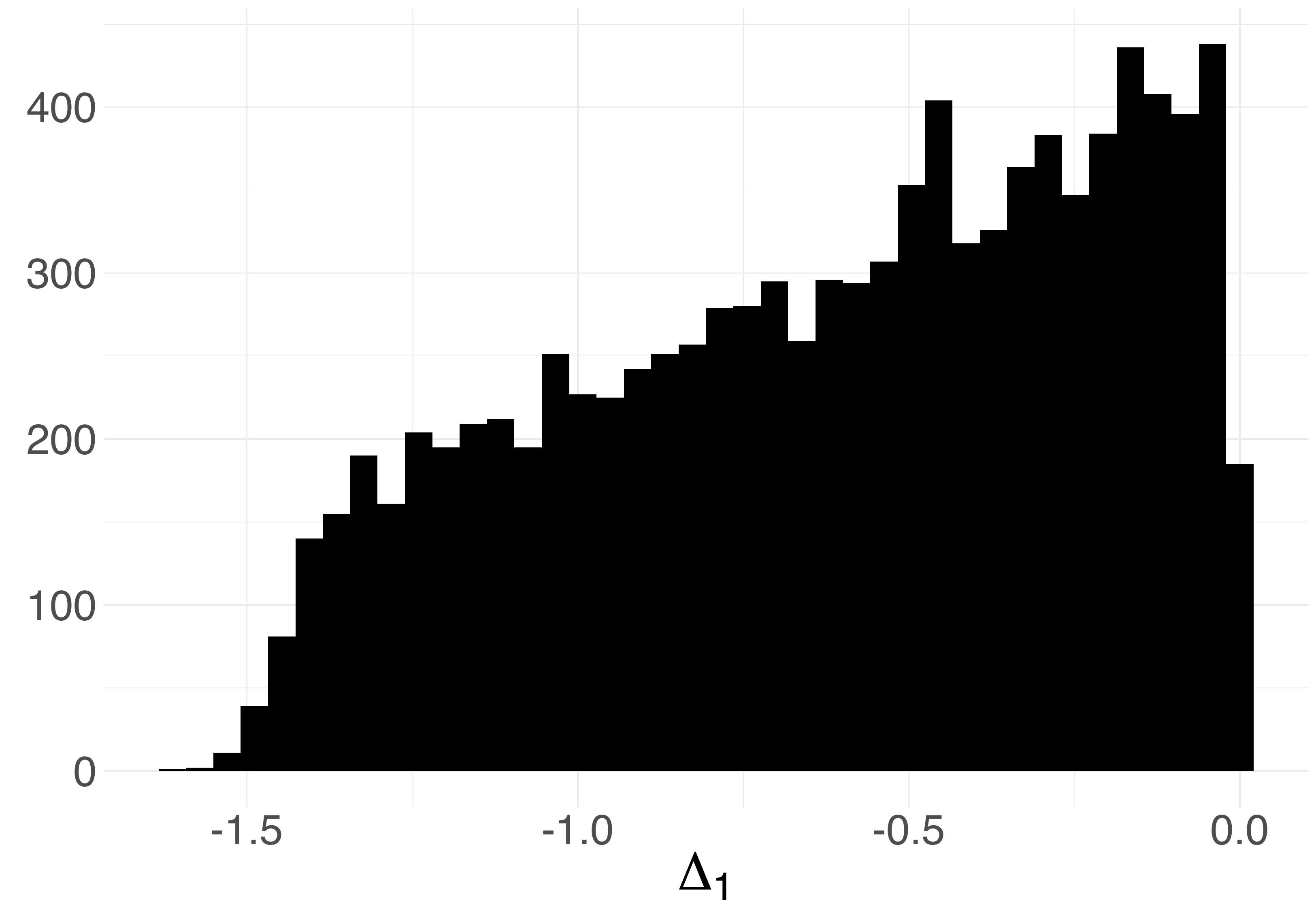}
  \end{subfigure}%
\begin{subfigure}{.5\textwidth}
  \centering
  \includegraphics[width=.8\linewidth]{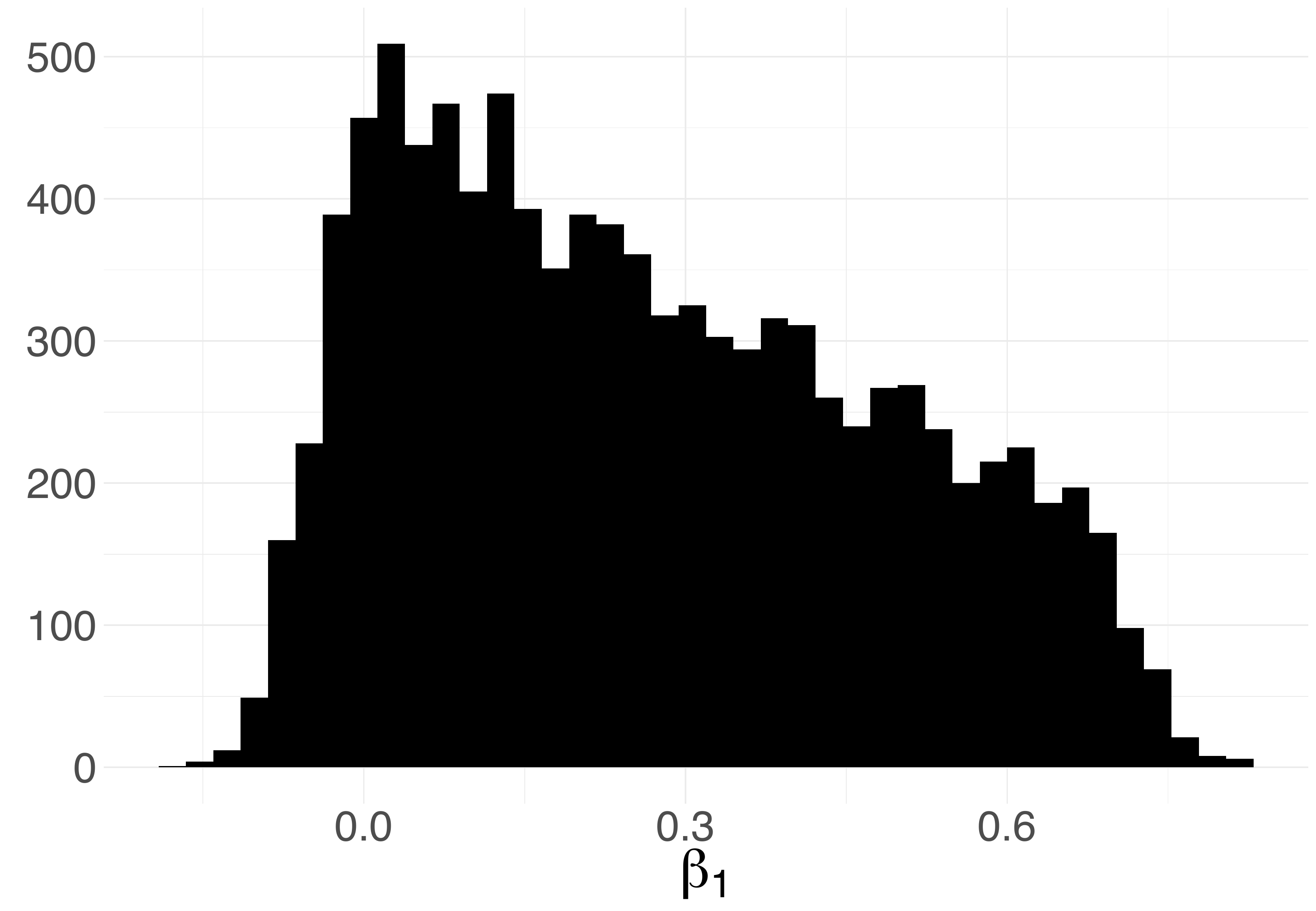}
\end{subfigure}
\begin{subfigure}{.5\textwidth}
  \centering
  \includegraphics[width=.8\linewidth]{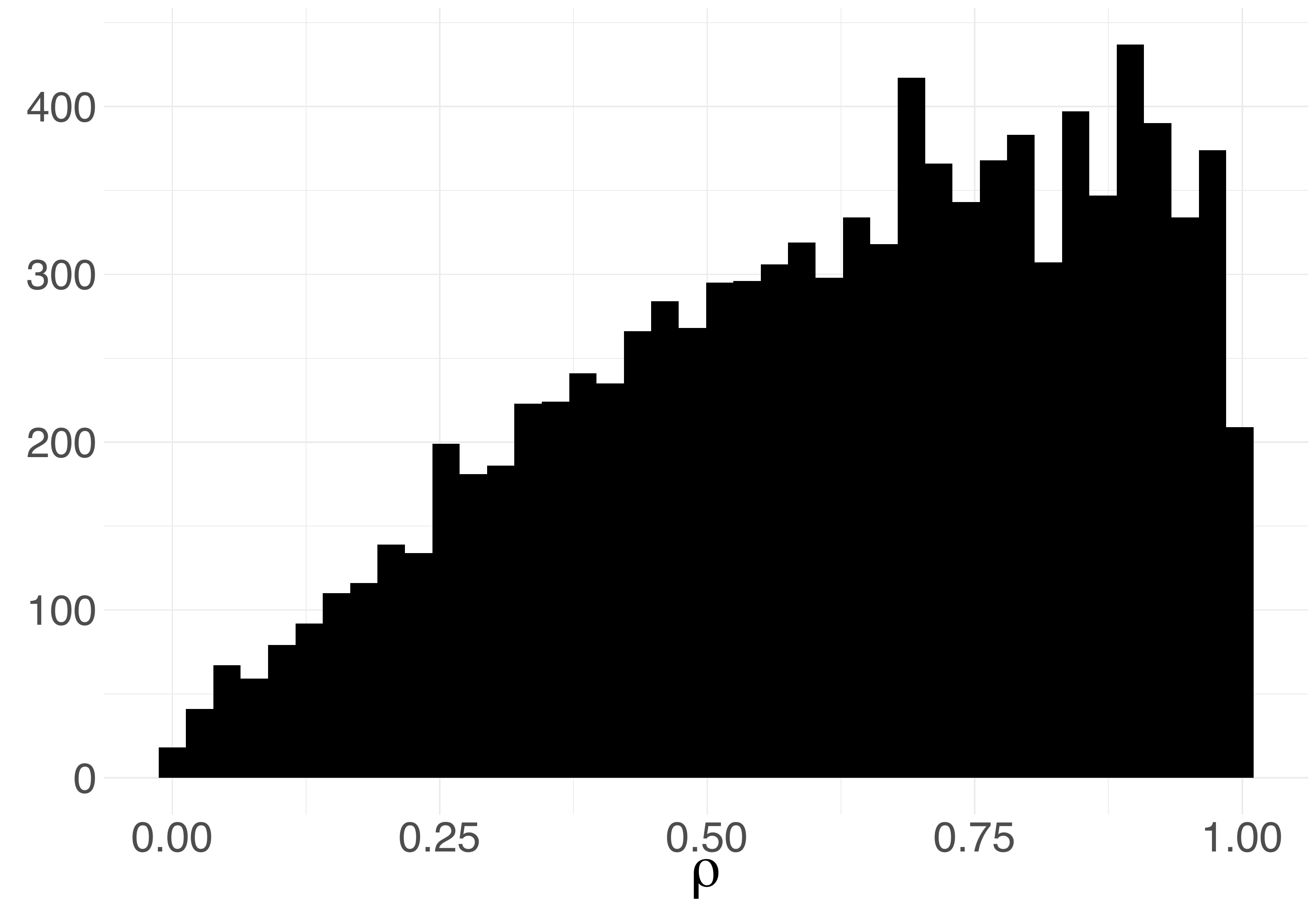}
  \end{subfigure}%
\begin{subfigure}{.5\textwidth}
  \centering
  \includegraphics[width=.8\linewidth]{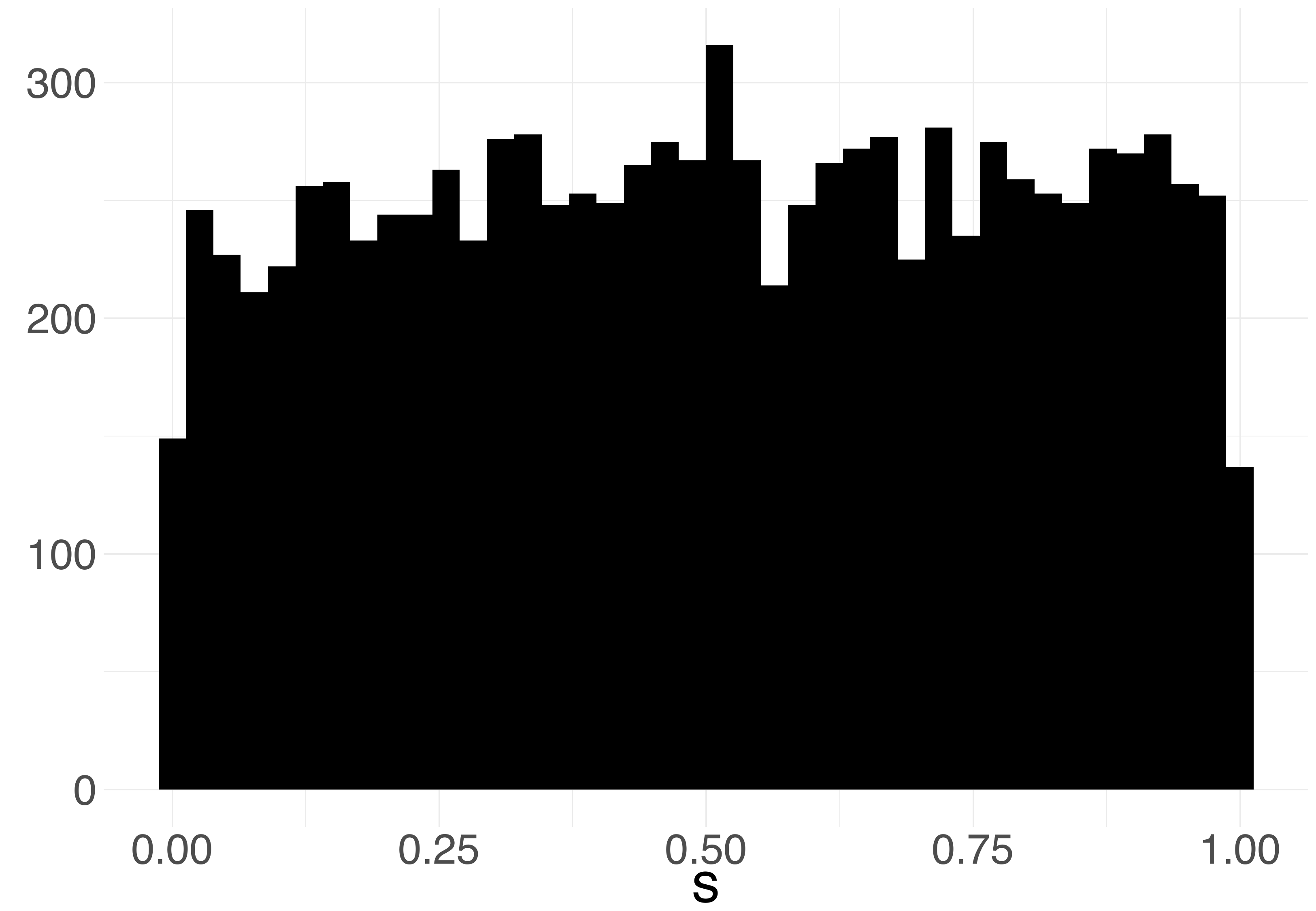}
\end{subfigure}
\begin{subfigure}{.5\textwidth}
  \centering
  \includegraphics[width=.8\linewidth]{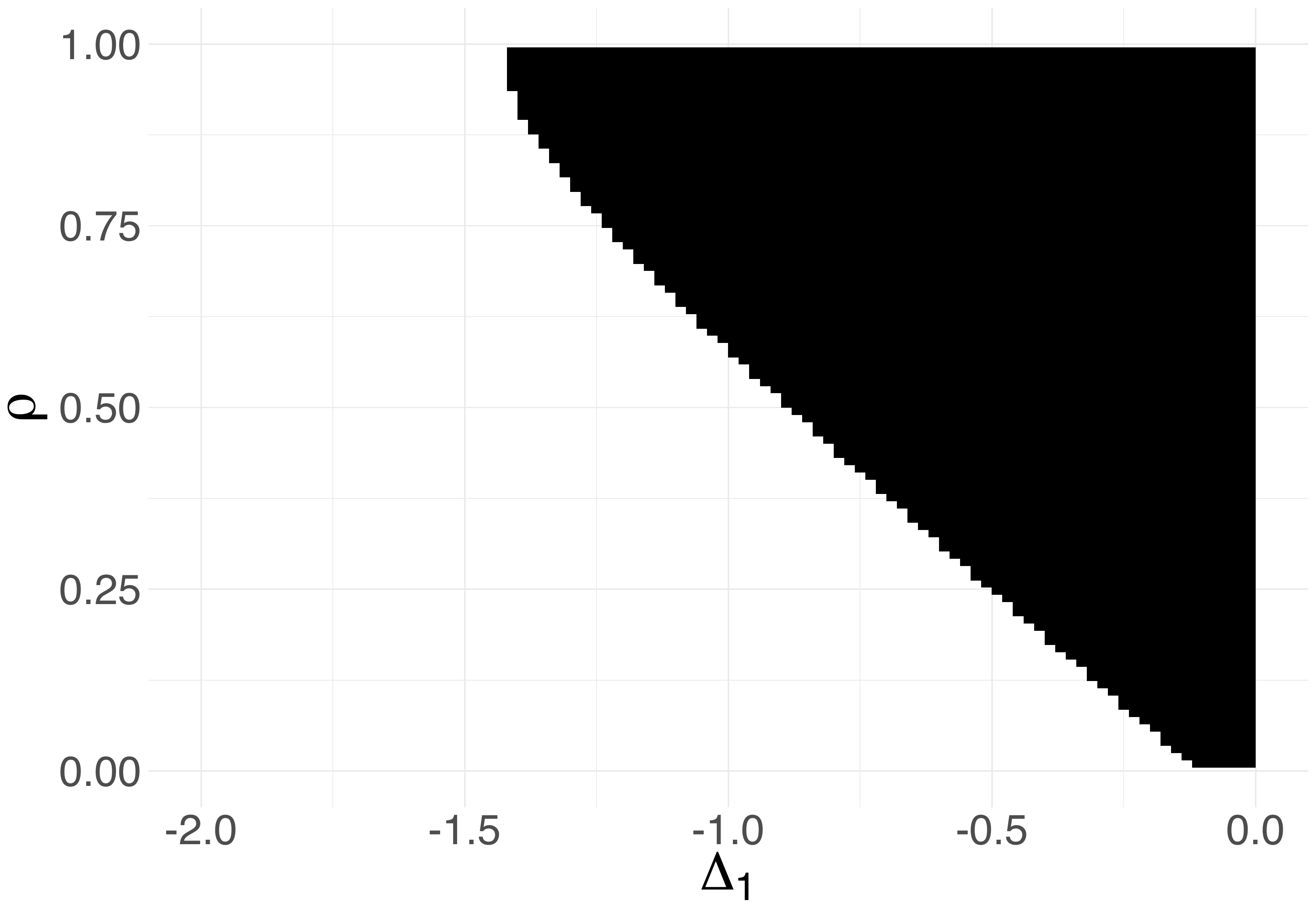}
\end{subfigure}%
\begin{subfigure}{.5\textwidth}
  \centering
  \includegraphics[width=.8\linewidth]{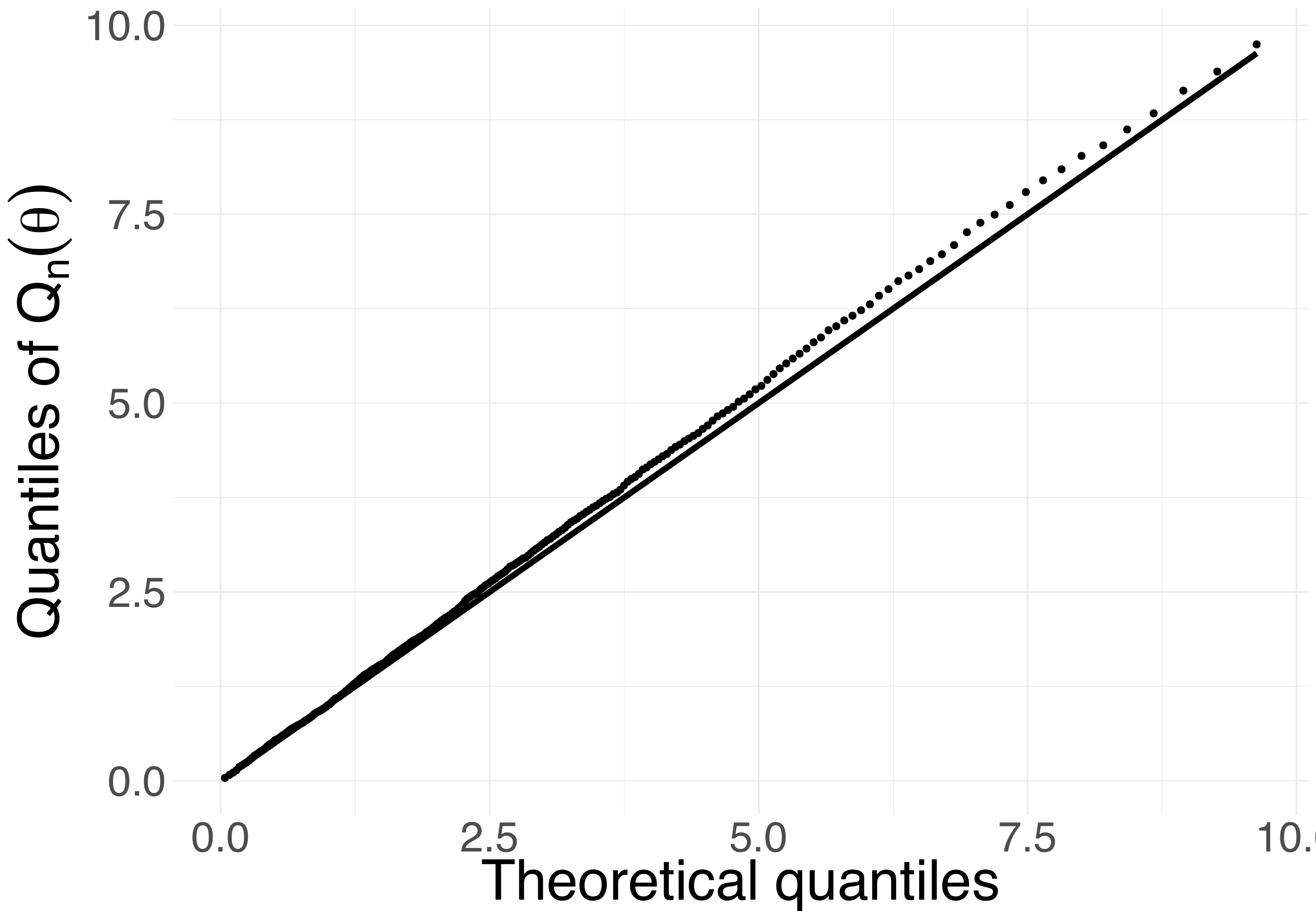}
\end{subfigure}
\centering
\parbox{12cm}{\caption{\small\label{f:ex2-plots} Entry game example: histograms of the SMC draws for $\Delta_1$ (top left), $\beta_1$ (top right) $\rho$ (mid left), and $s$ (mid right), and Q-Q plot of $Q_n(\theta)$ computed from the draws against $\chi^2_3$ quantiles (bottom right) for a sample of size $n = 1000$. The identified set for $(\Delta_1,\rho)$ is also shown (bottom left).} }
\end{figure}

We put a flat prior on $\Theta$ and implement the SMC algorithm as described in Appendix \ref{ax:smc:ex2}. Figure \ref{f:ex2-plots} displays histograms of the marginal draws for $\Delta_1$, $\beta_1$, $s$ and $\rho$ for one run of the SMC algorithm with a sample of size $n = 1000$. The plots for $\Delta_2$ and $\beta_2$ are very similar to those for $\Delta_1$ and $\beta_1$ (which is to be expected as the parameters are symmetric) and are therefore omitted. The draws for $\Delta_1$ and $\beta_1$ are supported on and around their respective identified sets, which are approximately $[-1.42,0]$ and $[-0.05,0.66]$ (the identified sets for $\rho$ and $s$ are $[0,1]$). Note that here the draws for $\Delta_1$, $\beta_1$ and $\rho$ are not flat over their identified sets, in contrast with the draws for $\mu$ and $\eta_1$ in Figure \ref{f:ex1-plots}. To see why, consider the marginal identified set for $(\Delta_1,\rho)$, plotted as the shaded region in Figure \ref{f:ex2-plots}. This plot shows that when $\Delta_1$ is close to the upper bound of its identified set, $(\Delta_1,\rho)$ is in the shaded region for any $\rho \in [0,1]$. However, when $\Delta_1$ is close to the lower bound of its identified set, $(\Delta_1,\rho)$ is only in the shaded region for very large values of $\rho$. This structure of $\Theta_I$, together with the flat prior on $\Theta$, means that the marginal posterior for $\Delta_1$ assigns relatively more mass towards the upper limit of the identified set for $\Delta_1$. Similar logic applies for $\beta_1$ and $\rho$. Figure \ref{f:ex2-plots} also shows that the quantiles of $Q_n(\theta)$ computed from the draws are very close to the $\chi^2_3$ quantiles, as predicted by our theoretical results below.

Table \ref{t:game:MC} reports average coverage probabilities and CS limits for the various procedures across 1000 replications. We form CSs for $\Theta_I$ using procedure 1, as well as CSs for the identified sets of scalar subvectors $\Delta_1$ and $\beta_1$ using procedures 2 and 3.\footnote{As the parameterization is symmetric, the identified sets for $\Delta_2$ and $\beta_2$ are the same as for $\Delta_1$ and $\beta_1$ so we omit them. We also omit CSs for $\rho$ and $s$, whose identified sets are both $[0,1]$.} We also compare our CS for identified sets for $\Delta_1$ and $\beta_1$ with projection-based and percentile-based CSs. Appendix \ref{ax:smc:ex2} provides additional details on computation of $M(\theta)$ for implementation of procedure 2. We do not use the reduced-form reparameterization in terms of choice probabilities to compute $M(\theta)$.  Coverage of $\wh \Theta_\alpha$ for $\Theta_I$ is extremely good, even with the small sample size $n = 100$. Coverage of procedures 2 and 3 for the identified sets for $\Delta_1$ and $\beta_1$ is slightly conservative for the small sample size $n$, but close to nominal for $n =1000$. As expected, projection CSs are valid but very conservative (the coverage probabilities of 90\% CSs are all at least 98\%) whereas percentile-based CSs undercover.

\begin{table}[p]{  \footnotesize
\begin{center}
\begin{tabular}{|c|cccccc|} \hline
	& \multicolumn{2}{c}{0.90} & \multicolumn{2}{c}{0.95} & \multicolumn{2}{c|}{0.99}  \\ \hline
 & & \multicolumn{4}{c}{CSs for the identified set $\Theta_I$} & \\
 & & \multicolumn{4}{c}{$\widehat{\Theta}_{\alpha}$ (Procedure 1)} & \\
100     &  0.924  & ---   &  0.965  & ---   &  0.993  & ---  \\
250     &  0.901  & ---   &  0.952  & ---   &  0.996  & ---  \\
500     &  0.913  & ---   &  0.958  & ---   &  0.991  & ---  \\
1000    &  0.913  & ---   &  0.964  & ---   &  0.992  & ---  \\ \hline
& & \multicolumn{4}{c}{CSs for the identified set for $\Delta_1$} & \\
 & & \multicolumn{4}{c}{$\widehat{M}_{\alpha}$ (Procedure 2)} & \\
100     &  0.958  & [$ -1.70 ,\! 0.00 $]  &  0.986  & [$ -1.77 ,\! 0.00 $]  &  0.997  & [$ -1.87 ,\! 0.00 $] \\
250     &  0.930  & [$ -1.58 ,\! 0.00 $]  &  0.960  & [$ -1.62 ,\! 0.00 $]  &  0.997  & [$ -1.70 ,\! 0.00 $] \\
500     &  0.923  & [$ -1.52 ,\! 0.00 $]  &  0.961  & [$ -1.55 ,\! 0.00 $]  &  0.996  & [$ -1.60 ,\! 0.00 $] \\
1000    &  0.886  & [$ -1.48 ,\! 0.00 $]  &  0.952  & [$ -1.50 ,\! 0.00 $]  &  0.989  & [$ -1.54 ,\! 0.00 $] \\
& & \multicolumn{4}{c}{$\widehat{M}_{\alpha}^\chi$ (Procedure 3)} & \\
100     &  0.944  & [$ -1.70 ,\! 0.00 $]  &  0.973  & [$ -1.75 ,\! 0.00 $]  &  0.993  & [$ -1.84 ,\! 0.00 $] \\
250     &  0.939  & [$ -1.59 ,\! 0.00 $]  &  0.957  & [$ -1.62 ,\! 0.00 $]  &  0.997  & [$ -1.69 ,\! 0.00 $] \\
500     &  0.937  & [$ -1.53 ,\! 0.00 $]  &  0.971  & [$ -1.55 ,\! 0.00 $]  &  0.996  & [$ -1.60 ,\! 0.00 $] \\
1000    &  0.924  & [$ -1.49 ,\! 0.00 $]  &  0.966  & [$ -1.51 ,\! 0.00 $]  &  0.992  & [$ -1.54 ,\! 0.00 $] \\
& & \multicolumn{4}{c}{$\widehat{M}_{\alpha}^{proj}$ (Projection)} & \\
100     &  0.993  & [$ -1.84 ,\! 0.00 $]  &  0.997  & [$ -1.88 ,\! 0.00 $]  &  1.000  & [$ -1.94 ,\! 0.00 $] \\
250     &  0.996  & [$ -1.69 ,\! 0.00 $]  &  0.999  & [$ -1.72 ,\! 0.00 $]  &  1.000  & [$ -1.79 ,\! 0.00 $] \\
500     &  0.996  & [$ -1.60 ,\! 0.00 $]  &  0.999  & [$ -1.62 ,\! 0.00 $]  &  1.000  & [$ -1.67 ,\! 0.00 $] \\
1000    &  0.989  & [$ -1.54 ,\! 0.00 $]  &  0.996  & [$ -1.56 ,\! 0.00 $]  &  0.999  & [$ -1.59 ,\! 0.00 $] \\
& & \multicolumn{4}{c}{$\widehat{M}_{\alpha}^{perc}$ (Percentiles)} & \\
100     &  0.000  & [$ -1.43 ,\! -0.06 $]  &  0.000  & [$ -1.54 ,\! -0.03 $]  &  0.000  & [$ -1.72 ,\! -0.01 $] \\
250     &  0.000  & [$ -1.37 ,\! -0.06 $]  &  0.000  & [$ -1.45 ,\! -0.03 $]  &  0.000  & [$ -1.57 ,\! -0.01 $] \\
500     &  0.000  & [$ -1.34 ,\! -0.05 $]  &  0.000  & [$ -1.41 ,\! -0.03 $]  &  0.000  & [$ -1.50 ,\! -0.01 $] \\
1000    &  0.000  & [$ -1.33 ,\! -0.05 $]  &  0.000  & [$ -1.39 ,\! -0.03 $]  &  0.000  & [$ -1.46 ,\! -0.01 $] \\ \hline
& & \multicolumn{4}{c}{CSs for the identified set for $\beta_1$} & \\
 & & \multicolumn{4}{c}{$\widehat{M}_{\alpha}$ (Procedure 2)} & \\
100     &  0.960  & [$ -0.28 ,\! 0.89 $]  &  0.974  & [$ -0.32 ,\! 0.94 $]  &  0.994  & [$ -0.40 ,\! 1.03 $] \\
250     &  0.935  & [$ -0.18 ,\! 0.81 $]  &  0.958  & [$ -0.20 ,\! 0.84 $]  &  0.995  & [$ -0.26 ,\! 0.89 $] \\
500     &  0.925  & [$ -0.14 ,\! 0.76 $]  &  0.958  & [$ -0.16 ,\! 0.78 $]  &  0.995  & [$ -0.19 ,\! 0.82 $] \\
1000    &  0.926  & [$ -0.11 ,\! 0.72 $]  &  0.970  & [$ -0.12 ,\! 0.74 $]  &  0.994  & [$ -0.15 ,\! 0.76 $] \\
& & \multicolumn{4}{c}{$\widehat{M}_{\alpha}^\chi$ (Procedure 3)} & \\
100     &  0.918  & [$ -0.26 ,\! 0.87 $]  &  0.963  & [$ -0.30 ,\! 0.92 $]  &  0.992  & [$ -0.38 ,\! 1.01 $] \\
250     &  0.914  & [$ -0.17 ,\! 0.80 $]  &  0.953  & [$ -0.20 ,\! 0.83 $]  &  0.991  & [$ -0.25 ,\! 0.88 $] \\
500     &  0.912  & [$ -0.13 ,\! 0.75 $]  &  0.957  & [$ -0.15 ,\! 0.77 $]  &  0.990  & [$ -0.19 ,\! 0.81 $] \\
1000    &  0.917  & [$ -0.11 ,\! 0.72 $]  &  0.962  & [$ -0.12 ,\! 0.73 $]  &  0.993  & [$ -0.14 ,\! 0.76 $] \\
& & \multicolumn{4}{c}{$\widehat{M}_{\alpha}^{proj}$ (Projection)} & \\
100     &  0.990  & [$ -0.38 ,\! 1.00 $]  &  0.997  & [$ -0.41 ,\! 1.05 $]  &  1.000  & [$ -0.49 ,\! 1.13 $] \\
250     &  0.989  & [$ -0.24 ,\! 0.88 $]  &  0.997  & [$ -0.27 ,\! 0.90 $]  &  1.000  & [$ -0.32 ,\! 0.96 $] \\
500     &  0.989  & [$ -0.19 ,\! 0.81 $]  &  0.996  & [$ -0.20 ,\! 0.82 $]  &  1.000  & [$ -0.24 ,\! 0.86 $] \\
1000    &  0.990  & [$ -0.14 ,\! 0.76 $]  &  0.998  & [$ -0.15 ,\! 0.77 $]  &  1.000  & [$ -0.18 ,\! 0.80 $] \\
& & \multicolumn{4}{c}{$\widehat{M}_{\alpha}^{perc}$ (Percentiles)} & \\
100     &  0.395  & [$ -0.11 ,\! 0.71 $]  &  0.654  & [$ -0.16 ,\! 0.78 $]  &  0.937  & [$ -0.26 ,\! 0.90 $] \\
250     &  0.169  & [$ -0.05 ,\! 0.66 $]  &  0.478  & [$ -0.09 ,\! 0.71 $]  &  0.883  & [$ -0.16 ,\! 0.80 $] \\
500     &  0.085  & [$ -0.04 ,\! 0.63 $]  &  0.399  & [$ -0.07 ,\! 0.68 $]  &  0.840  & [$ -0.12 ,\! 0.74 $] \\
1000    &  0.031  & [$ -0.03 ,\! 0.62 $]  &  0.242  & [$ -0.05 ,\! 0.65 $]  &  0.803  & [$ -0.09 ,\! 0.70 $] \\ \hline
	\end{tabular}
\parbox{12cm}{\caption{\small\label{t:game:MC} Entry game example: average coverage probabilities for $\Theta_I$ and identified sets for $\Delta_1$ and $\beta_1$ across MC replications and average lower and upper bounds of CSs for identified sets for $\Delta_1$ and $\beta_1$ across MC replications using a likelihood criterion function and flat prior. The identified sets for $\Delta_1$ and $\beta_1$ are approximately $[-1.42,0]$ and $[-0.05,0.66]$.}}
\end{center}
}	
\end{table}

\subsection{Empirical applications}\label{sec-empirical}

This subsection implements our procedures in two non-trivial empirical applications. The first application estimates an entry game with correlated payoff shocks using data from the US airline industry. Here there are 17 model parameters to be estimated. The second application estimates a model of trade flows initially examined in \cite{HMR} (HMR henceforth). We use a version of the empirical model in HMR with 46 parameters to be estimated.

Although the entry game model is separable, we do not make use of separability in implementing our procedures. In fact, the existing Bayesian approaches that impose priors on the globally-identified reduced-form parameters will be problematic in this example. This separable model has 24 non-redundant choice probabilities (global reduced-form parameters, i.e., $\dim (\phi ) =24$) and 17 model structural parameters (i.e., $\dim (\theta ) =17$), and there is no explicit closed form expression for the identified set. Both \cite{MoonSchorfheide} and \cite{KlineTamer} would sample from the posterior for the reduced-form parameter $\phi$. But, unless the posterior for $\phi$ is constrained to lie on $\{\phi(\theta) : \theta \in \Theta\}$ (i.e. the set of reduced-form probabilities consistent with the model, rather than the full 24-dimensional space), certain values of $\phi$ drawn from their posteriors for $\phi$ will not be consistent with the model.

The empirical trade example is a {\it nonseparable} likelihood model that cannot be handled by either (a) existing Bayesian approaches that rely on a point-identified, $\sqrt n$-estimable and asymptotically normal reduced-form parameter, or (b) inference procedures based on moment inequalities.

In both applications, our approach only puts a prior on the model structural parameter $\theta$ so it does not matter whether the model is separable or not. Both applications illustrate how our procedures may be used to examine the robustness of estimates to various ad hoc modeling assumptions in a theoretically valid and computationally feasible way.

\subsubsection{Bivariate Entry Game with US Airline Data}

This section estimates a version of the entry game that we study in Subsection \ref{s:game} above.
We use data from the second quarter of $2010$'s
Airline Origin and Destination Survey (DB1B) to estimate a  binary game where the payoff for firm $i$ from entering market $m$ is
\[
 \beta_i + \beta_i^xx_{im} + \Delta_iy_{3-i} + \epsilon_{im} \quad i=1,2
\]
where the $\Delta_i$ are assumed to be negative (as usually the case in entry models). The data contain 7882 markets which are formally defined as trips between two airports irrespective of stopping and we examine the entry behavior of two kinds of firms: LC (low cost) firms,\footnote{The low cost carriers are:  JetBLue, Frontier, Air Tran,  Allegiant Air, Spirit, Sun Country, USA3000, Virgin America, Midwest Air, and Southwest.} and OA (other airlines) which includes all the other firms. The unconditional choice probabilities are $(.16, .61, .07, .15)$ which are respectively the probabilities that OA and LC serve a market, that OA and not LC serve a market, that LC and not OA serve a market, and finally whether no airline serve the market.

The regressors are {\it market presence} and {\it market size}. Market presence is a market- and airline-specific variable defined as follows: from a given airport, we compute the ratio of markets a given carrier (we take the maximum within the category OA or LC, as appropriate) serves divided by the total number of markets served from that given airport. The market presence variable $MP$ is the average of the ratios from the two endpoints and it provides a proxy for an airline's presence in a given airport (See \cite{berry} for more on this variable).  This variable acts as an excluded regressor: the market presence for OA only enters OA's payoffs, so $MP$ is both market- and airline-specific.  The second regressor we use is {\it market size} $MS$ which is defined as the  population at the endpoints, so this variable is market-specific. We discretize both $MP$ and $MS$ into binary variables that take the value of one if the variable is higher than its median (in the data) value and zero otherwise. The choice probabilities are $P(y_{OA}, y_{LC}|MS, \, MP_{OA},\, MP_{LC})$ are conditional on the three-dimensional vector $(MS, MP_{OA}, MP_{LC})$. We therefore have 4 choice probabilities for every value of the conditioning variables (and there are 8 values for these).\footnote{With binary values, the conditioning set takes the following eight values: (1,1,1), (1,1,0), (1,0,1), (1,0,0), (0,1,1), (0,1,0), (0,0,1), (0,0,0).} To use notation similar to that in Subsection \ref{s:game}, let OA be player $1$ and firm LC be player $2$. Denote $\beta_1 (x_{mOA}):=\beta^0_{OA} + \beta_{OA}'x_{mOA}$ and $\beta_2 (x_{mLC}):= \beta^0_{LC} + \beta_{LC}'x_{mLC}$ with $x_{mOA} = (MS_m, MP_{mOA})'$ and  $x_{mLC} = (MS_m, MP_{mLC})'$. The likelihood for market $m$ depends on the choice probabilities:

\vskip -20pt

{
\footnotesize
\begin{align*}
 \tilde\gamma_{11}(\theta; x_m) :=& P(\epsilon_{1m} \geq -\beta_1 (x_{mOA}) - \Delta_{OA}; \, \epsilon_{2m} \geq -\beta_2 (x_{mLC}) - \Delta_{LC} )\\
 \tilde\gamma_{00}(\theta; x_m) :=& P(\epsilon_{1m} \leq -\beta_1 (x_{mOA}) ; \epsilon_{2m} \leq -\beta_2 (x_{mLC}) )\\
 \tilde\gamma_{10}(\theta; x_m) := & s(x_m)\times P(-\beta_1 (x_{mOA}) \leq \epsilon_{1m} \leq -\beta_1 (x_{mOA}) - \Delta_{OA}; -\beta_2 (x_{mLC})  \leq \epsilon_{2m} \leq -\beta_2 (x_{mLC}) - \Delta_{LC})
 \\ & \quad + P( \epsilon_{1m} \geq - \beta_{1}(x_{mOA}); \epsilon_{2m} \leq -\beta_2 (x_{mLC}) )\\ &\quad + P( \epsilon_{1m} \geq -\beta_{1}(x_{mOA}) - \Delta_{OA} ;-\beta_2 (x_{mLC}) \leq \epsilon_{2m} \leq -\beta_2 (x_{mLC}) - \Delta_{LC}) \,.
\end{align*}
}

\vskip -20pt

Here $s(x_m)$ is a nuisance parameter which corresponds to the various {\it aggregate} equilibrium selection probabilities. Here $s(\cdot)$ is defined on the support of $x_m$, so in the model this function takes $2^3=8$ values each belonging to  $[0,1]$. In the {\it full model} we make no assumptions on the equilibrium selection mechanism. Therefore, the full model has 17 parameters: 4 parameters per profit function (namely $\Delta_i$, $\beta_i^0$, $\beta_i^{MS}$, and $\beta_i^{MP}$), the correlation $\rho$ between $\epsilon_{i1}$ and $\epsilon_{i2}$, and the 8 parameters in the aggregate equilibrium choice probabilities $s(\cdot)$. We also estimate a restricted version of the model called {\it fixed $s$} in which we restrict the aggregate selection probabilities to be the same across markets, for a total of 10 parameters. Note that these are just one version of the econometric model for a game; a less parsimonious version would allow, for example, for the parameters to change with regressor values, or allow for the regressors' support to be richer (rather than binary). We analyze this case precisely to highlight the fact that our CSs provide coverage guarantees regardless of whether the parameter vector is point identified.

\begin{sidewaystable}[p]
\begin{center}{\footnotesize
\begin{tabular}{|c|cccc|cccc|} \hline
 & \multicolumn{4}{c|}{Full model} & \multicolumn{4}{c|}{Fixed-$s$ model}\\
\hline
 &  Procedure 2 & Procedure 3 & Projection & Percentile &  Procedure 2 & Procedure 3 & Projection & Percentile \\
$\Delta_{OA}$        & [$ -1.599 ,\! -1.178 $]  & [$ -1.539 ,\! -1.303 $]  & [$ -1.707 ,\! -0.701 $]  & [$ -1.515 ,\! -1.117 $]  & [$ -1.563 ,\! -1.335 $]  & [$ -1.543 ,\! -1.363 $]  & [$ -1.655 ,\! -1.194 $]  & [$ -1.536 ,\! -1.326 $] \\
$\Delta_{LC}$        & [$ -1.527 ,\! -1.218 $]  & [$ -1.503 ,\! -1.246 $]  & [$ -1.719 ,\! -1.018 $]  & [$ -1.489 ,\! -1.225 $]  & [$ -1.567 ,\! -1.343 $]  & [$ -1.547 ,\! -1.367 $]  & [$ -1.671 ,\! -1.222 $]  & [$ -1.548 ,\! -1.339 $] \\
$\beta_{OA}^0$       & [$ 0.443 ,\! 0.581 $]  & [$ 0.455 ,\! 0.575 $]  & [$ 0.341 ,\! 0.695 $]  & [$ 0.447 ,\! 0.578 $]  & [$ 0.431 ,\! 0.551 $]  & [$ 0.437 ,\! 0.539 $]  & [$ 0.365 ,\! 0.611 $]  & [$ 0.427 ,\! 0.540 $] \\
$\beta_{OA}^{MS}$    & [$ 0.365 ,\! 0.539 $]  & [$ 0.383 ,\! 0.521 $]  & [$ 0.238 ,\! 0.665 $]  & [$ 0.389 ,\! 0.544 $]  & [$ 0.347 ,\! 0.479 $]  & [$ 0.353 ,\! 0.467 $]  & [$ 0.275 ,\! 0.551 $]  & [$ 0.348 ,\! 0.477 $] \\
$\beta_{OA}^{MP}$    & [$ 0.413 ,\! 0.581 $]  & [$ 0.425 ,\! 0.569 $]  & [$ 0.275 ,\! 0.713 $]  & [$ 0.424 ,\! 0.579 $]  & [$ 0.479 ,\! 0.641 $]  & [$ 0.497 ,\! 0.623 $]  & [$ 0.389 ,\! 0.719 $]  & [$ 0.504 ,\! 0.648 $] \\
$\beta_{LC}^0$       & [$ -1.000 ,\! -0.729 $]  & [$ -1.000 ,\! -0.723 $]  & [$ -1.000 ,\! -0.453 $]  & [$ -0.993 ,\! -0.751 $]  & [$ -0.910 ,\! -0.627 $]  & [$ -0.874 ,\! -0.657 $]  & [$ -1.000 ,\! -0.507 $]  & [$ -0.917 ,\! -0.655 $] \\
$\beta_{LC}^{MS}$    & [$ 0.226 ,\! 0.431 $]  & [$ 0.238 ,\! 0.419 $]  & [$ 0.064 ,\! 0.599 $]  & [$ 0.220 ,\! 0.405 $]  & [$ 0.299 ,\! 0.443 $]  & [$ 0.305 ,\! 0.431 $]  & [$ 0.220 ,\! 0.527 $]  & [$ 0.303 ,\! 0.442 $] \\
$\beta_{LC}^{MP}$    & [$ 1.591 ,\! 1.868 $]  & [$ 1.633 ,\! 1.832 $]  & [$ 1.423 ,\! 1.988 $]  & [$ 1.615 ,\! 1.821 $]  & [$ 1.573 ,\! 1.790 $]  & [$ 1.597 ,\! 1.760 $]  & [$ 1.489 ,\! 1.880 $]  & [$ 1.590 ,\! 1.776 $] \\
$\rho$               & [$ 0.874 ,\! 0.986 $]  & [$ 0.910 ,\! 0.978 $]  & [$ 0.713 ,\! 0.998 $]  & [$ 0.867 ,\! 0.977 $]  & [$ 0.938 ,\! 0.990 $]  & [$ 0.948 ,\! 0.986 $]  & [$ 0.886 ,\! 0.998 $]  & [$ 0.935 ,\! 0.986 $] \\
$s$                  & --- & --- & --- & ---    & [$ 0.926 ,\! 0.980 $]  & [$ 0.932 ,\! 0.976 $]  & [$ 0.888 ,\! 0.992 $]  & [$ 0.927 ,\! 0.977 $] \\
$s_{000}$             & [$ 0.587 ,\! 0.964 $]  & [$ 0.679 ,\! 0.950 $]  & [$ 0.000 ,\! 1.000 $]  & [$ 0.572 ,\! 0.934 $]  & --- & --- & --- & ---  \\
$s_{001}$             & [$ 0.812 ,\! 1.000 $]  & [$ 0.854 ,\! 1.000 $]  & [$ 0.439 ,\! 1.000 $]  & [$ 0.797 ,\! 0.995 $]  & --- & --- & --- & ---  \\
$s_{010}$             & [$ 0.000 ,\! 1.000 $]  & [$ 0.000 ,\! 0.906 $]  & [$ 0.000 ,\! 1.000 $]  & [$ 0.018 ,\! 0.828 $]  & --- & --- & --- & ---  \\
$s_{100}$             & [$ 0.637 ,\! 0.998 $]  & [$ 0.794 ,\! 0.998 $]  & [$ 0.000 ,\! 1.000 $]  & [$ 0.612 ,\! 0.990 $]  & --- & --- & --- & ---  \\
$s_{011}$             & [$ 0.916 ,\! 1.000 $]  & [$ 0.930 ,\! 1.000 $]  & [$ 0.804 ,\! 1.000 $]  & [$ 0.915 ,\! 0.999 $]  & --- & --- & --- & ---  \\
$s_{101}$             & [$ 0.491 ,\! 0.920 $]  & [$ 0.607 ,\! 0.842 $]  & [$ 0.000 ,\! 1.000 $]  & [$ 0.449 ,\! 0.799 $]  & --- & --- & --- & ---  \\
$s_{110}$             & [$ 0.000 ,\! 1.000 $]  & [$ 0.000 ,\! 1.000 $]  & [$ 0.000 ,\! 1.000 $]  & [$ 0.042 ,\! 0.986 $]  & --- & --- & --- & ---  \\
$s_{111}$             & [$ 0.942 ,\! 1.000 $]  & [$ 0.966 ,\! 1.000 $]  & [$ 0.856 ,\! 1.000 $]  & [$ 0.941 ,\! 0.999 $]  & --- & --- & --- & ---  \\
\hline
\end{tabular}}
\vskip 4pt
\parbox{14cm}{\caption{\label{g:app} \small Entry game application: 95\% CSs for structural parameters computed via our Procedures 2 and 3 as well as via Projection and Percentile methods. The full model contains a general specification for equilibrium selection while the fixed-$s$ model restricts the equilibrium selection probability to be the same across markets with different regressor values.}}

\end{center}
\end{sidewaystable}

We again take a flat prior on $\Theta$ and implement the procedures using a likelihood criterion. We restrict the support of $\Delta_i$ to $[-2,0]$, $\beta_i$ to $[-1,2]^3$, $\rho$ to $[0,1]$ and the selection probabilities to $[0,1]$. We implement the procedure using the adaptive SMC algorithm as described in Appendix \ref{ax:smc:app1} with $B = 10000$ draws. Histograms of the SMC draws for the selection probabilities are presented in Figure \ref{f:game_1:s}; histograms of draws for the profit function parameters and $\rho$ are presented in Figures \ref{f:game_0} and \ref{f:game_1} in Appendix \ref{ax:smc:app1}. We construct CSs for each of the parameters using procedure 2 and procedure 3 and compare these to projection-based CSs (projecting $\wh \Theta_\alpha$ using our procedure 1) and percentile CSs. The empirical findings are presented in Table \ref{g:app} below. Appendix \ref{ax:smc:app1} contains further details on computation of $M(\theta)$ for implementation of procedure 2. As with the game simulation, we do not explicitly use the reduced-form reparameterization $\theta \mapsto \tilde \gamma(\theta)$ when computing $M(\theta)$.

The results in Table \ref{g:app} show that CSs computed via procedures 2 and 3 are generally similar (though there are some differences, with CSs via procedure 2, which is valid under weaker conditions than procedure 3, appearing wider for some of the selection probabilities in the full model). On the other hand, projection CSs are very wide, especially in the full model. For instance, the projection CS for $s_{101}$ is $[0,1]$ whereas CSs via procedures 2 and 3 are $[0.49,0.92]$ and $[0.61,0.84]$ respectively. As expected, percentile CSs are narrower than procedure 2 and 3 CSs, reflecting the fact that percentile CSs under-cover in partially identified models.

Starting with the {full model} results, we see that the estimates are meaningful economically and are inline with recent estimates obtained in the literature. For example, fixed costs (the intercepts) are positive and significant for the large airlines (OA) but are negative for the LC carriers. Typically, the presence of higher fixed costs can signal various barriers to entry prevent LCs from entering: the higher these fixed costs the less likely it is for LCs to enter. On the other hand, higher fixed costs of large airlines are associated with a bigger presence (such as a hub) and so OAs are more likely to enter. As expected, both market presence and market size are associated with a positive probability of entry for both OA and LC. Results for the fixed-$s$ model are in agreement with the corresponding ones for the full model and tell a consistent story. Note also the very high correlation in the errors, which could indicate missing profitability variables whereby firms enter a particularly profitable markets regardless of competition.

One interesting observation are the CSs for the selection probabilities (also see Figure \ref{f:game_1:s}). Consider $s_{010}$ and $s_{110}$: these are the aggregate selection probabilities which, according to the results, are not identified. This is likely due to the rather small number of markets with small size, large presence for OA but small presence for LC (for $s_{010}$) and the small number of markets with large market size, large presence for OA but small presence for LC (for $s_{110}$). The strength of our approach is its {\it adaptivity} to lack of identification in a particular data set: for example, 95\% CSs for the identified set for $s_{010}$ are $[0,1]$ (via procedure 2), indicating that the model (and data) has no information about this parameter, while the corresponding CS for the identified set for $s_{111}$ is  the narrow and informative interval $[0.94,1.00]$.

\begin{figure}[p]
\begin{subfigure}{.45\textwidth}
  \centering
  \includegraphics[width=.9\linewidth]{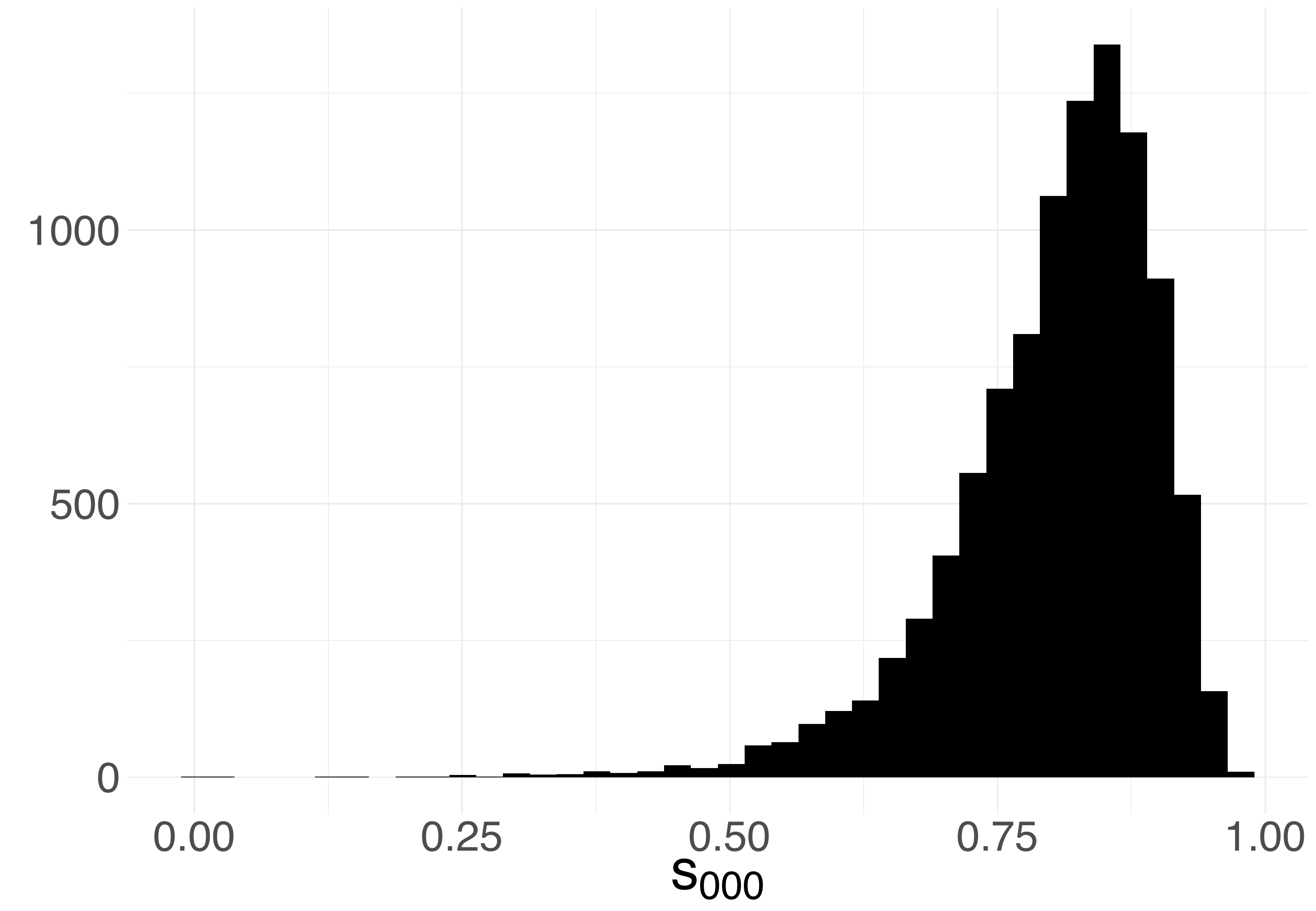}
  \end{subfigure}%
\begin{subfigure}{.45\textwidth}
  \centering
  \includegraphics[width=.9\linewidth]{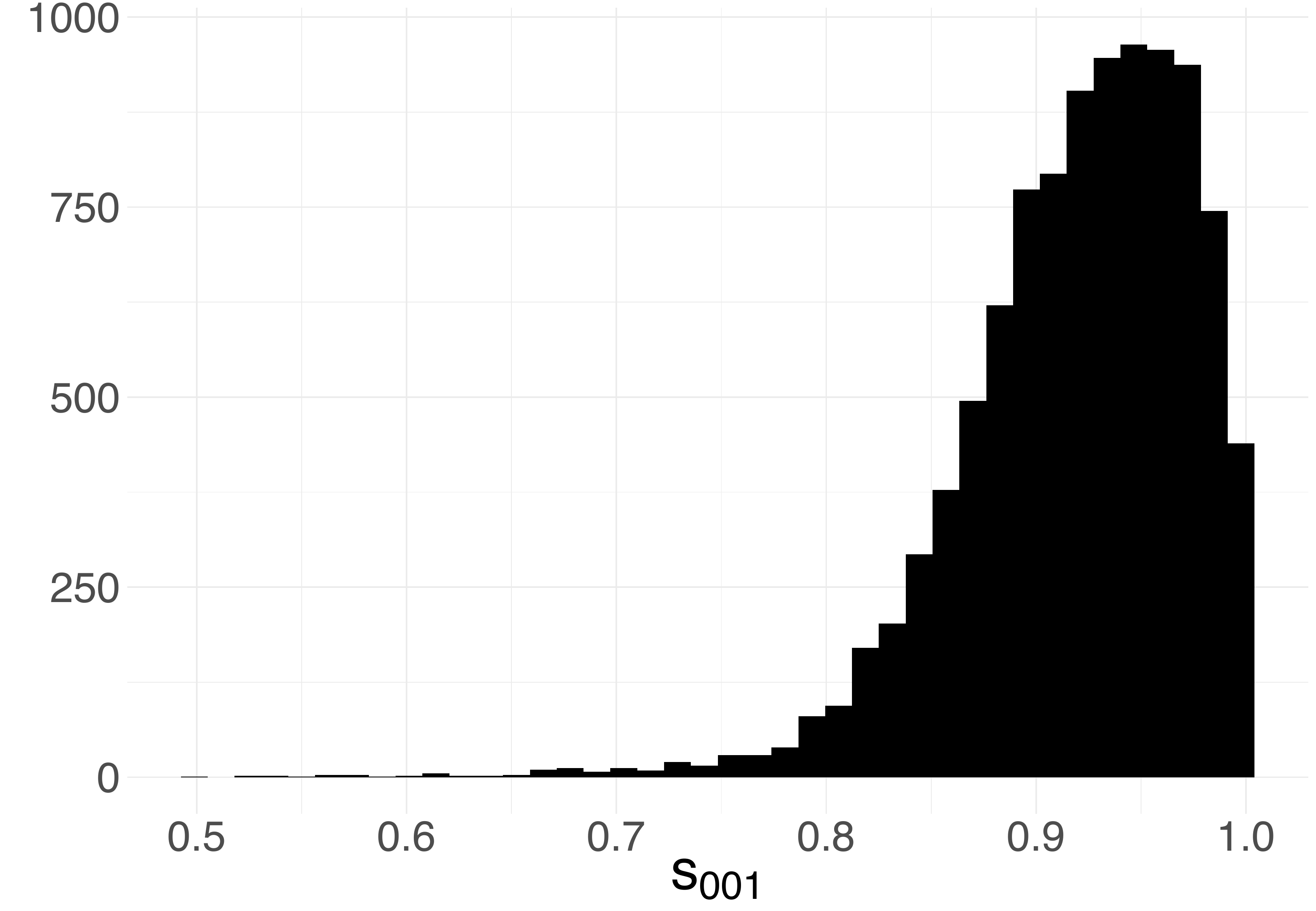}
\end{subfigure}
\begin{subfigure}{.45\textwidth}
  \centering
  \includegraphics[width=.9\linewidth]{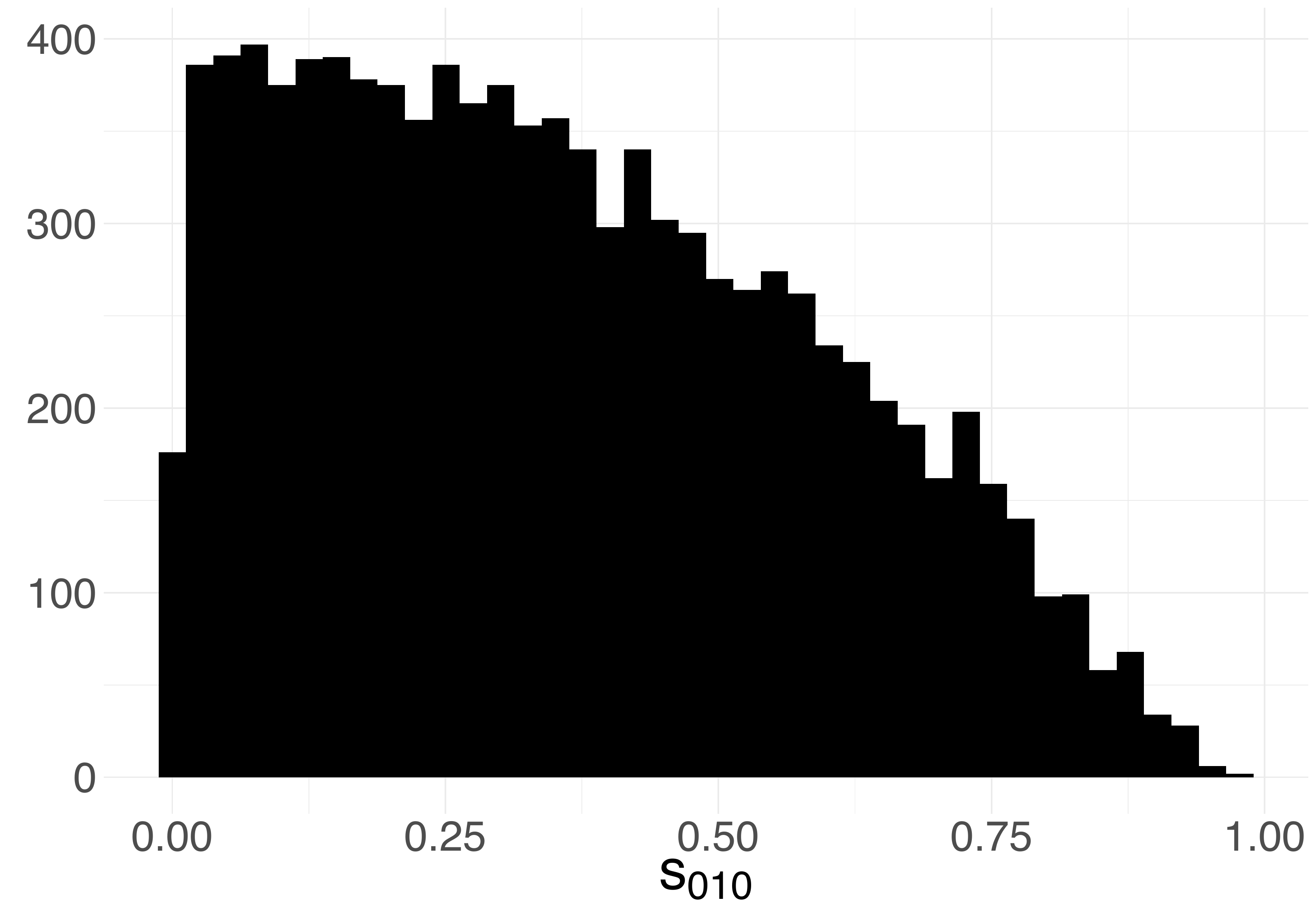}
  \end{subfigure}%
\begin{subfigure}{.45\textwidth}
  \centering
  \includegraphics[width=.9\linewidth]{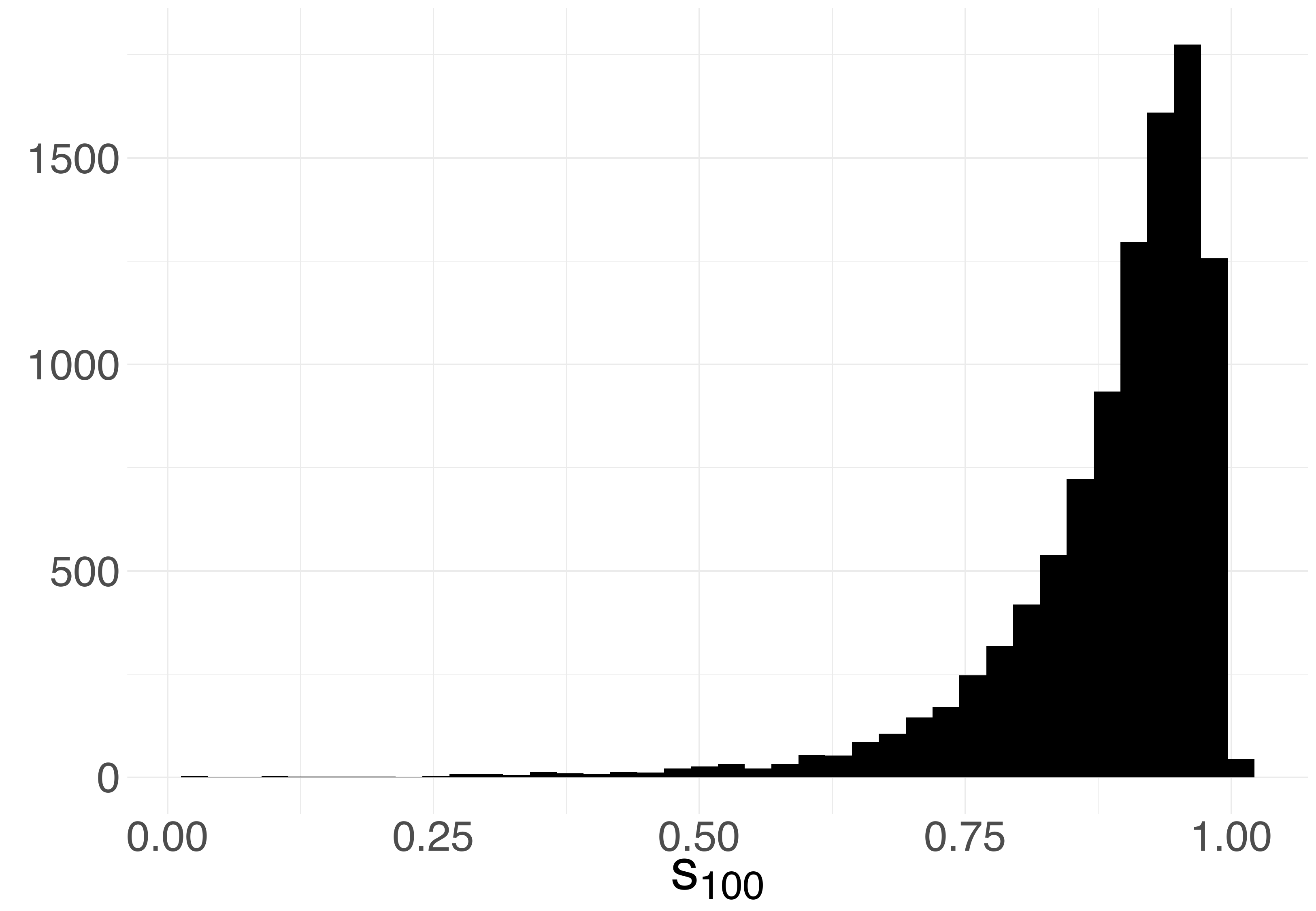}
\end{subfigure}
\begin{subfigure}{.45\textwidth}
  \centering
  \includegraphics[width=.9\linewidth]{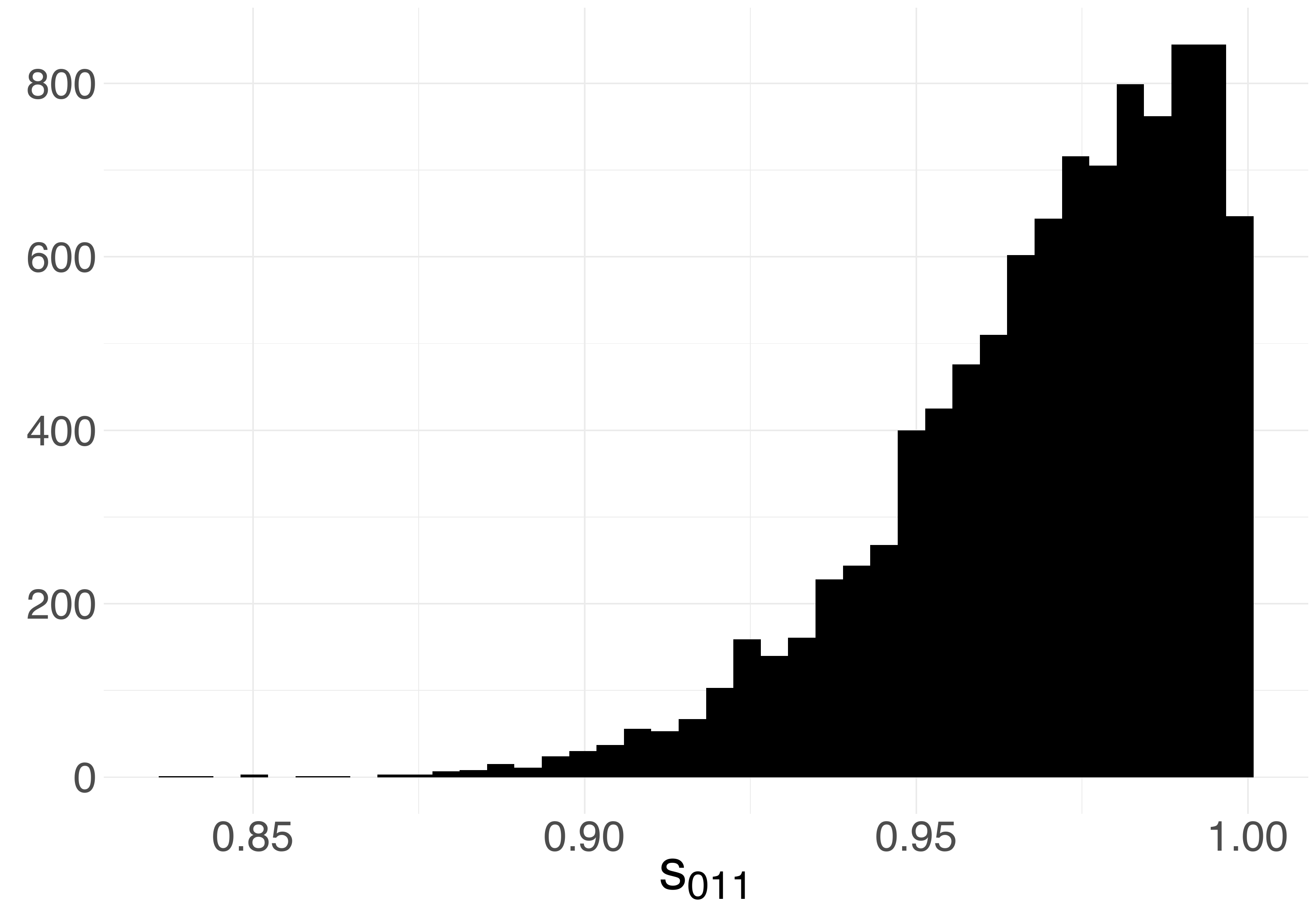}
  \end{subfigure}%
\begin{subfigure}{.45\textwidth}
  \centering
  \includegraphics[width=.9\linewidth]{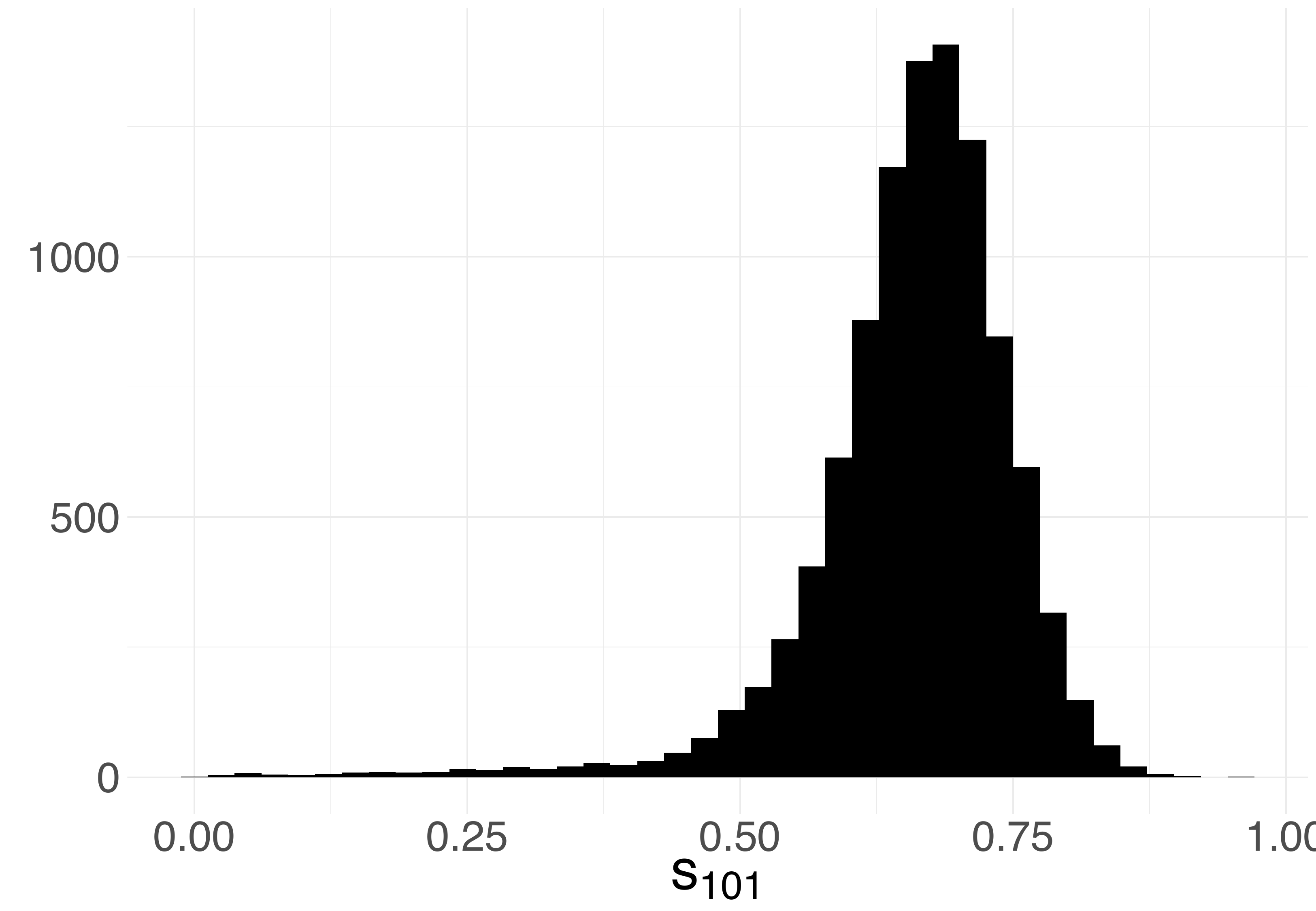}
\end{subfigure}
\begin{subfigure}{.45\textwidth}
  \centering
  \includegraphics[width=.9\linewidth]{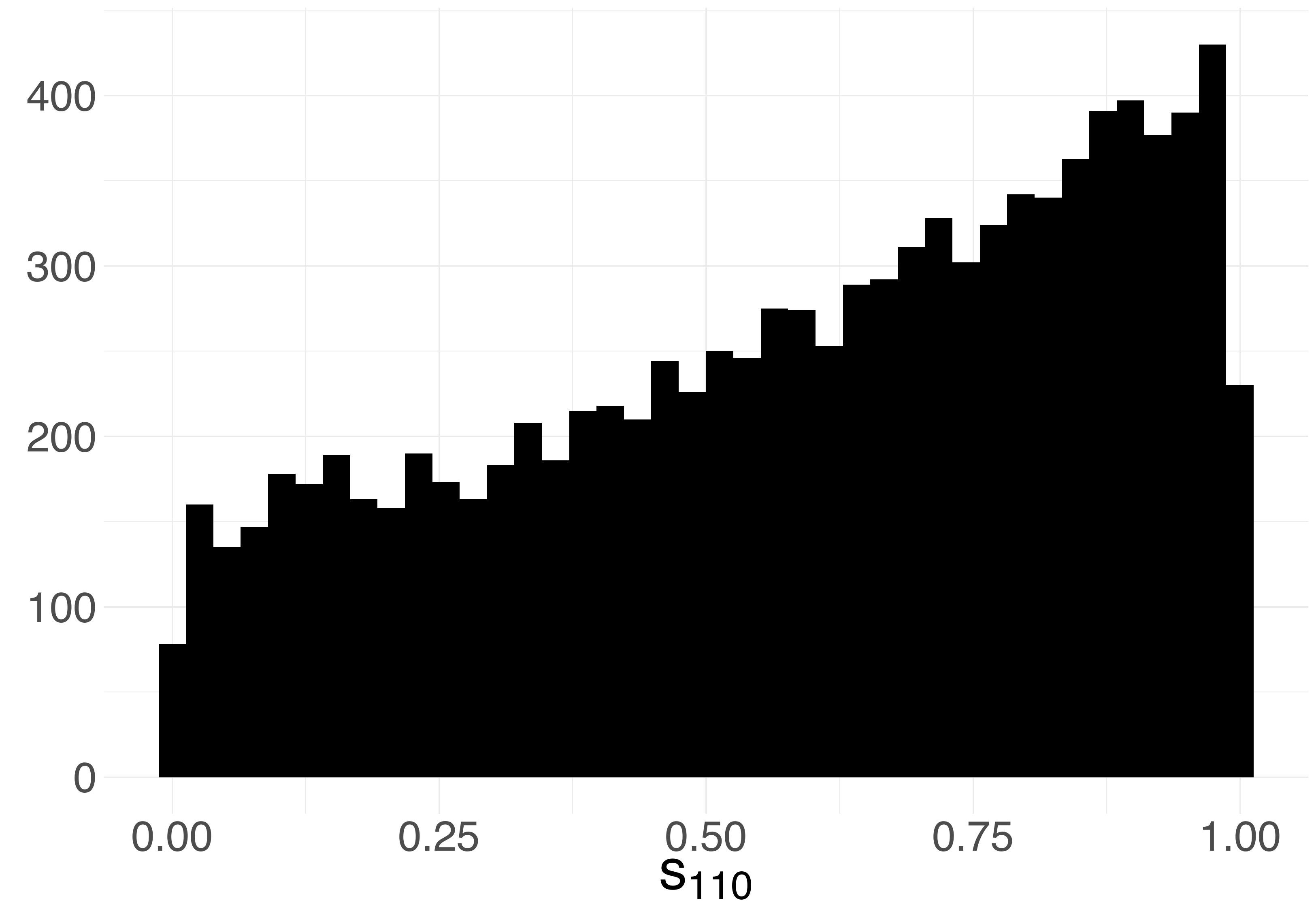}
  \end{subfigure}%
\begin{subfigure}{.45\textwidth}
  \centering
  \includegraphics[width=.9\linewidth]{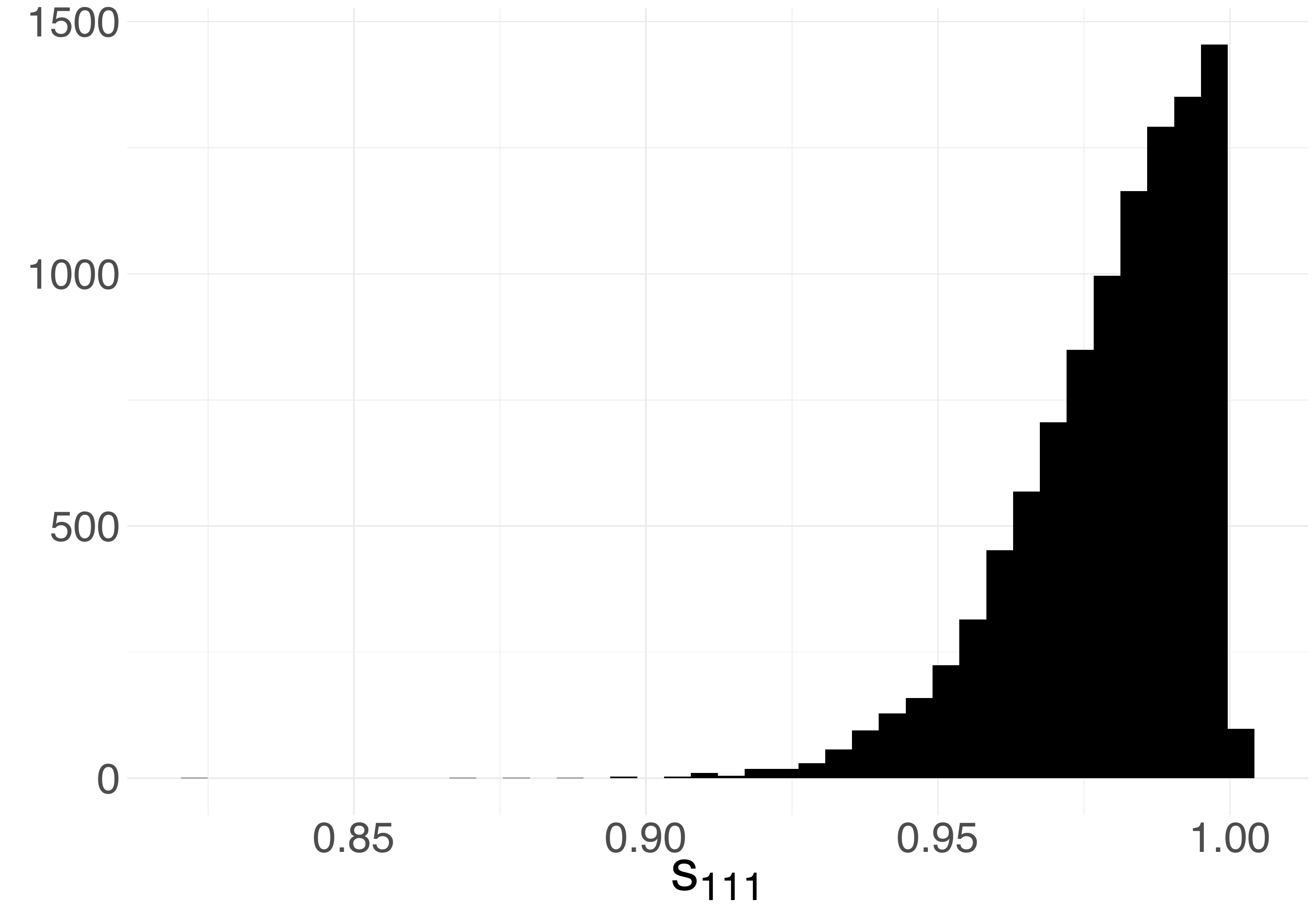}
\end{subfigure}
\centering
\parbox{12cm}{\caption{\small\label{f:game_1:s} Entry game application: histograms of the SMC draws for selection probabilities $s_{000}$, $s_{001}$, $s_{010}$, $s_{100}$, $s_{011}$, $s_{101}$, $s_{110}$, and $s_{111}$ for the \textbf{full model}.} }
\end{figure}

\subsubsection{ An empirical model of trade flows}

In an influential paper, \cite{HMR} examine the extensive margin of trade using a structural model estimated with current trade data. The following is a brief description of their empirical framework.  Let $M_{ij}$ denote the value of country $i$'s imports from country $j$. This is only observed if country $j$ exports to country $i$. If a random draw for productivity from country $j$ to $i$ is sufficiently high then $j$ will export to $i$. To model this, \cite{HMR} introduce a latent variable $z_{ij}^*$ which measures trade volume between $i$ and $j$. Here $z_{ij}^*$ takes the value zero if $j$ does not export to $i$ and is strictly positive otherwise. We adapt slightly their empirical model to obtain a selection model of the form:
\begin{align*}
 \log M_{ij} & = \left\{ \begin{array}{ll}
  \beta_0 + \lambda_j + \chi_i - \nu' f_{ij} + \delta z_{ij}^* + u_{ij} & \mbox{if $z_{ij}^* > 0$} \\
  \mbox{not observed} & \mbox{if $z_{ij}^* \leq 0$} \end{array} \right. \\
  z_{ij}^* & = \beta_0^* + \lambda^*_j + \chi^*_i - \nu^{*\prime} f_{ij} + \eta_{ij}^*
\end{align*}
in which $\lambda_j$, $\chi_i$, $\lambda_j^*$ and $\chi_i^*$  are exporting and importing continent fixed effects, $f_{ij}$ is a vector of observable trade frictions between $i$ and $j$, and $u_{ij}$ and $\eta_{ij}^*$ are error terms described below.
Exclusion restrictions can be imposed by setting at least one of the elements of $\nu$ equal to zero.

There are three differences between our empirical model and that of \cite{HMR}. First, we let $z_{ij}^*$ enter the outcome equation linearly instead of nonlinearly.\footnote{Their nonlinear specification is known to be problematic (see, e.g., \cite{HMRcritical}).} Second, we use continent fixed effect instead of country fixed effects. This reduces the number of parameters from over 400 to 46. Third, we allow for heteroskedasticity in the selection equation, which is known to be a problem in trade data. This illustrates the robustness approach we advocate which relaxes parametric assumptions on part of the model that is suspect (homoskedasticity) without worrying about loss of point identification.

\begin{sidewaystable}[p]
{\small
\begin{center}
\begin{tabular}{|c|cc|ccccc|}
 \multicolumn{1}{c}{} & \multicolumn{2}{c}{Homoskedastic} & \multicolumn{5}{c}{Heteroskedastic} \\ \hline
 Variable & MLE & $t$-stat CI &  MLE & $t$-stat CI & Procedure 2 & Procedure 3 & Percentile \\ \hline
Distance  &   2.352 &   [1.154,3.549] &   0.314 &   [0.273,0.355] &   [0.216,0.749] &   [0.242,0.509] &   [0.207,0.397]  \\
Border  &   -5.191 &   [-7.077,-3.304] &   -2.265 &   [-2.452,-2.077] &   [-2.651,-1.859] &   [-2.611,-1.898] &   [-2.618,-1.816]  \\
Island  &   -1.302 &   [-1.913,-0.691] &   -0.728 &   [-0.868,-0.589] &   [-1.060,-0.308] &   [-1.060,-0.308] &   [-0.983,-0.397]  \\
Landlock  &   -7.275 &   [-10.769,-3.780] &   -1.369 &   [-1.602,-1.137] &   [-2.194,-0.914] &   [-2.194,-0.890] &   [-1.801,-0.954]  \\
Legal  &   0.358 &   [0.002,0.715] &   -0.122 &   [-0.183,-0.061] &   [-0.254,0.004] &   [-0.242,-0.009] &   [-0.248,0.011]  \\
Language  &   -4.098 &   [-6.430,-1.766] &   -0.095 &   [-0.168,-0.021] &   [-0.868,0.049] &   [-0.868,0.026] &   [-0.237,0.067]  \\
Colonial  &   -17.378 &   [-26.002,-8.755] &   -2.822 &   [-3.029,-2.615] &   [-4.980,-2.373] &   [-4.980,-2.461] &   [-3.231,-2.298]  \\
Currency  &   -1.550 &   [-2.780,-0.320] &   -0.631 &   [-0.946,-0.315] &   [-1.315,0.020] &   [-1.282,-0.013] &   [-1.274,0.062]  \\
FTA  &   -19.540 &   [-29.783,-9.298] &   -2.151 &   [-2.410,-1.892] &   [-2.686,-1.589] &   [-2.631,-1.616] &   [-2.680,-1.577]  \\  \hline
\end{tabular}
\vskip 4pt
\parbox{14cm}{\caption{\label{t:app} \small Maximum likelihood estimates of the coefficients $\nu$ of the trade friction variables in the outcome equation (MLE) together with their 95\% confidence sets based on inverting $t$-statistics ($t$-stat CI). Also shown are 95\% CSs computed by our Procedures 2 and 3 as well as via Percentile methods.}}
\end{center}
}
\end{sidewaystable}

To allow for heteroskedasticity, we suppose that the distribution of $(u_{ij},\eta_{ij}^*)$ conditional on observables is Normal with mean zero and covariance:
\[
 \Sigma(X_{ij}) = \left[ \begin{array}{cc}
 \sigma_m^2  & \rho \sigma_m \sigma_z (X_{ij}) \\
 \rho \sigma_m \sigma_z (X_{ij}) & \sigma_z^2 (X_{ij})
 \end{array} \right]
\]
where $X_{ij}$ denotes $f_{ij}$, the exporter's continent, and the importer's continent and where
\[
 \sigma_z(X_{ij}) = \exp(\varpi_1 \log(\mr{distance}_{ij}) + \varpi_2 \log(\mr{distance}_{ij})^2) \,.
\]

We estimate the model from data on 24,649 country pairs in the selection equation and 11,146 country pairs in the outcome equation using the same data from 1986 as in \cite{HMR}. We also impose the exclusion restriction that the coefficient in $\nu$ corresponding to religion is equal to zero, else there is an exact linear relationship between the coefficients in the outcome and selection equation. This leaves a total of 46 parameters to be estimated. We only report estimates for the trade friction coefficients $\nu$ in the outcome equation as these are the most important. We estimate the model first by maximum likelihood under homoskedasticity and report conventional ML estimates for $\nu$ together with 95\% CSs based on inverting $t$-statistics. We then re-estimate the model under heteroskedasticity and report conventional ML estimates together with confidence sets based on inverting $t$-statistics, percentile CSs (i.e. the \cite{CH} procedure under point identification), and our procedures 2 and 3. To implement our Procedure 2 and percentile CSs, we use the adaptive SMC algorithm as described in Appendix \ref{ax:smc:app2} with $B = 10000$ draws.

The results are presented in Table \ref{t:app}.\footnote{Note that the friction variables enter negatively in the outcome equation. The coefficient of distance is positive meaning that distance negatively affects trade flows; the remaining variables are dummy variables, so a negative coefficient of border means that sharing a border positively affects trade flows, and so forth.} Overall, though the model is sensitive to the presence of heteroskedasticity, the results for the heteroskedastic specification show that the confidence sets seem reasonably insensitive to the type of procedure used, which suggests that partial identification may not be an issue even allowing for heteroskedasticity. We also notice some difference in results relative to \cite{HMR}.
For instance, they document strong positive effects of common legal systems and currency unions on trade flows, whereas we find much weaker evidence for this. We also find a positive effect of landlocked status on trade flows whereas they document a negative effect.
Under heteroskedasticity, the magnitudes of coefficients of the trade friction variables are generally smaller than under homoskedasticity but of the same sign. The exception is the legal variable, whose coefficient is positive under homoskedasticity but negative under heteroskedasticity. A remaining question is whether the estimates are also sensitive to the normality assumption on the errors. This question can be examined within the context of our results by, for example, using a flexible form for the joint distribution of the errors.

\section{Large Sample Properties}\label{sec-property}

 This section provides regularity conditions under which $\wh \Theta_\alpha$ (Procedure 1),  $\wh M_\alpha$ (Procedure 2) and $\wh M_\alpha^\chi$ (Procedure 3) are asymptotically valid confidence sets for $\Theta_I$ and $M_I$. The main new theoretical contributions are the derivations of the large-sample (quasi)-posterior distributions of the QLR for $\Theta_I$ and of the profile QLR for $M_I$ under loss of identifiability.

\subsection{Coverage properties of $\wh \Theta_\alpha$ for $\Theta_I$}

We first state some high-level regularity conditions. A discussion of these assumptions follows.

\begin{assumption}\label{a:rate} (Posterior contraction) \\
	(i) $L_n(\hat \theta) = \sup_{\theta \in \Theta_{osn}} L_n(\theta) + o_\p(n^{-1})$, with $(\Theta_{osn})_{n \in \mb N}$ a sequence of local neighborhoods of $\Theta_I$; \\
    (ii) $\Pi_n(\Theta_{osn}^c |\,\mf X_n) = o_\p(1)$, where $\Theta_{osn}^c = \Theta \!\setminus \!\Theta_{osn}$.
\end{assumption}

We presume the existence of a fixed neighborhood $\Theta_I^N$ of $\Theta_I^{\phantom *}$ (with $\Theta_{osn} \subset \Theta_I^N$ for all $n$ sufficiently large) upon which there exists a \emph{local} reduced-form reparameterization $\theta \mapsto \gamma(\theta)$ from $\Theta_I^N$ into $\Gamma \subseteq \mb R^{d^*}$ for a possibly unknown dimension $d^* \in [1,\infty)$, with $\gamma(\theta) =\gamma_0 \equiv  0$ if and only if $\theta \in \Theta_I$. Here $\gamma (\cdot) $ is merely a proof device and is only required to exist for $\theta$ in a fixed neighborhood of $\Theta_I$. To accommodate situations in which the true reduced-form parameter value $\gamma_0 =0$ may be ``on the boundary'' of $\Gamma$, we assume that the sets $T_{osn} \equiv \{ \sqrt n \gamma(\theta) : \theta \in \Theta_{osn}\}$ \emph{cover}\footnote{We say that a sequence of sets $A_n \subseteq \mb R^{d^*}$ \emph{covers} a set $A \subseteq \mb R^{d^*}$ if (i) $\sup_{b : \|b\| \leq M} |\inf_{a \in A_n} \| a- b\|^2 - \inf_{a \in A} \| a - b\|^2 | = o_\p(1)$ for each $M$,  and (ii) there is a sequence of closed balls $B_{k_n}$ of radius $k_n \to \infty$ centered at the origin with each $C_n := A_n \cap B_{k_n}$ convex, $C_n \subseteq C_{n'}$ for each $n' \geq n$, and $A = \ol{ \cup_{n \geq 1} C_n }$ (almost surely).} \emph{a closed convex cone} $T\subseteq \mb R^{d^*}$. We note that this is trivially satisfied with $T = \mb R^{d^*}$ whenever each $T_{osn}$ contains a ball of radius $k_n \to \infty$ centered at the origin. A similar approach is taken for point-identified models by \cite{Chernoff1954}, \cite{Geyer1994}, and \cite{Andrews1999}. Let $\|\gamma\|^2:=\gamma'\gamma$ and for any $v \in \mb R^{d^*}$, let $\mf T v= \mr{arg}\min_{t \in T} \|v -t\|^2$ denote the orthogonal (or metric) projection of $v$ onto $T$.

\begin{assumption}\label{a:quad} (Local quadratic approximation) \\
	There exist sequences of random variables $\ell_n$ and $\mb R^{d^*}$-valued random vectors $\hat \gamma_n$ (both measurable in $\mf X_n$) such that as $n \to \infty$:
	\begin{equation} \label{e:quad}
	\sup_{\theta \in \Theta_{osn}} \left| n L_n(\theta) - \left(\ell_n +\frac{1}{2} \| \sqrt n \hat \gamma_n \|^2 - \frac{1}{2} \|\sqrt n (\hat \gamma_n - \gamma(\theta))\|^2 \right) \right| = o_\p(1)
	\end{equation}
	with $\sup_{\theta \in \Theta_{osn}} \|\gamma(\theta)\| \to 0$ and $\sqrt n \hat \gamma_n = \mf T \mb V_n$ where $\mb V_n \rightsquigarrow N(0,\Sigma)$.
\end{assumption}

Let $\Pi_\Gamma$ denote the image measure (under the map $\theta \mapsto \gamma(\theta)$) of the prior $\Pi$  on $\Theta_I^N$, namely $\Pi_\Gamma (A) = \Pi(\{\theta \in \Theta_I^N : \gamma(\theta) \in A\})$. Let $B_\delta \subset \mb R^{d^*}$ be a ball of radius $\delta$ centered at the origin.

\begin{assumption}\label{a:prior} (Prior) \\
	(i) $\int_\Theta e^{nL_n(\theta)} \,\mr d \Pi(\theta) < \infty $ almost surely; \\
	(ii) $\Pi_\Gamma$ has a continuous, strictly positive density $\pi_\Gamma$ on $B_\delta \cap \Gamma$ for some $\delta > 0$.
\end{assumption}

\paragraph{Discussion of Assumptions:}
Assumption \ref{a:rate}(i) is a standard condition on any approximate extremum estimator, and Assumption \ref{a:rate}(ii) is a standard posterior contraction condition. The choice of $\Theta_{osn}$ is deliberately general and will depend on the particular model under consideration. See Section \ref{s:suff} for verification of Assumption \ref{a:rate}.
Assumption \ref{a:quad} is a standard local quadratic expansion condition imposed on the local reduced form parameter around $\gamma =0$. It is readily verified for likelihood and GMM models (see Section \ref{s:suff}) with $\hat \gamma_n =\gamma (\hat \theta )$ and $\mb V_n$ typically a normalized score function of the data. For these models with i.i.d. data the vector $\mb V_n$ is typically of the form: $\mb V_n = n^{-1/2} \sum_{i=1}^n v(X_i)+ o_\p(1)$ with $\mb E[v(X_i) ] = 0$ and $\mr{Var}[v(X_i)]=\Sigma$. In fact, Appendix \ref{s:uniformity-qe}  shows that this quadratic expansion assumption is satisfied uniformly over a large class of DGPs in models with discrete data.
Assumption \ref{a:prior}(i) requires the quasi-posterior to be proper. Assumption \ref{a:prior}(ii) is a standard prior mass and smoothness condition used to establish BvM theorems for identified parametric models (see, e.g., Section 10.2 of \cite{vdV}) but applied to $\Pi_\Gamma$. Under a flat prior on $\Theta$ and a continuous local mapping $\gamma :\Theta_I^N \mapsto \Gamma$, this assumption is easily satisfied (see its verification in examples of Section \ref{s:suff}).

Assumptions \ref{a:rate}(i) and \ref{a:quad} imply that the QLR statistic for $\Theta_I$ satisfies
\begin{equation} \label{Fqlr}
\sup_{\theta \in \Theta_I} Q_n(\theta) = \|\mf T \mb V_n\|^2 + o_\p(1)
\end{equation}
(see Lemma \ref{l:quad}). Therefore, under the {\it generalized information equality} $\Sigma = I_{d^*}$, which holds for a correctly-specified likelihood, an optimally-weighted or continuously-updated GMM, or various (generalized) empirical-likelihood criterions, the asymptotic distribution of $\sup_{\theta \in \Theta_I} Q_n(\theta)$ becomes $F_T$, which is defined as
\begin{equation} \label{e:F_T}
 F_T (z) := \p_Z (\| \mf T Z \|^2 \leq z)
\end{equation}
where $\p_Z$ denotes the distribution of a $N(0,I_{d^*})$ random vector $Z$. This recovers the known asymptotic distribution result for QLR statistics under point identification. If $T = \mb R^{d^*}$ then $F_T$ reduces to $ F_{\chi^2_{d^*}}$, the cdf of $\chi^2_{d^*}$ (a chi-square random variable with $d^*$ degree of freedom). If $T$ is polyhedral then $F_T$ is the distribution of a chi-bar-squared random variable (i.e. a mixture of chi-squared distributions with different degrees of freedom where the mixture weights depend on $T$).

Let $\mb P_{Z|\mf X_n}$ denote the distribution of a $N(0,I_{d^*})$ random vector $Z$ (conditional on the data), and $T -v$ denote the convex cone $T$ translated to have vertex at $-v$. The next lemma establishes the large sample behavior of the posterior distribution of the QLR statistic.

\begin{lemma}\label{l:post}
	Let Assumptions \ref{a:rate}, \ref{a:quad} and \ref{a:prior} hold. Then:
\begin{equation} \label{e:c:post:1}
	\sup_z \left| \Pi_n \big(\{ \theta : Q_n(\theta) \leq z \}\big|\,\mf X_n\big) - \p_{Z | \mf X_n} \Big( \|Z\|^2 \leq z  \Big| Z \in T - \sqrt n \hat \gamma_n \Big) \right| = o_\p(1)\,.
\end{equation}
And hence we have:\\
 (i) If $T \subsetneq \mb R^{d^*}$ then: $\Pi_n \big(\{ \theta : Q_n(\theta) \leq z \}\big|\,\mf X_n\big) \leq F_T (z)$ for all $z \geq 0$.\\
 (ii) If $T = \mb R^{d^*}$ then: $\sup_z \left| \Pi_n \big(\{ \theta : Q_n(\theta) \leq z\} \,\big|\,\mf X_n\big) - F_{\chi^2_{d^*}}\!\!( z ) \right| = o_\p(1)$.
\end{lemma}

This result shows that the posterior distribution of the QLR statistic is asymptotically $\chi^2_{d^*}$ when $T = \mb R^{d^*}$, which may be viewed as a Bayesian Wilks theorem for partially identified models, and asymptotically (first-order) stochastically dominates $F_T$ when $T$ is a closed convex cone. Note that Lemma \ref{l:post} does not require the generalized information equality $\Sigma=I_{d^*}$ to hold. This lemma extends known BvM results for possibly misspecified likelihood models with point-identified $\sqrt n$-consistent and asymptotically normally estimable parameters (see \cite{KleijnvdV} and the references therein) to allow for other models with failure of $\Sigma=I_{d^*}$, with partially-identified parameters and/or parameters on a boundary.

Let $\xi_{n,\alpha}^{post}$ denote the $\alpha$ quantile of $Q_n(\theta)$ under the posterior distribution $\Pi_n$, and let $\xi_{n,\alpha}^{mc}$ be as stated in Remark \ref{rmk:full}.

\begin{assumption} \label{a:mcmc}  (MC convergence)\\
	$\xi_{n,\alpha}^{mc} = \xi_{n,\alpha}^{post} + o_\p(1)$.
\end{assumption}

Lemma \ref{l:post} and Assumption \ref{a:mcmc} together imply
that our Procedure 1 CS $\wh \Theta_\alpha$ is always a well-defined (quasi-)Bayesian credible set (BCS) regardless of whether $\Sigma=I_{d^*}$ holds or not. Further, together with Equation (\ref{Fqlr}), they imply the following result.

\begin{theorem}\label{t:main}
Let Assumptions \ref{a:rate}, \ref{a:quad}, \ref{a:prior}, and \ref{a:mcmc} hold with $\Sigma = I_{d^*}$. Then for any $\alpha$ such that $F_T (\cdot )$ is continuous at its $\alpha$ quantile, we have:
\newline
(i) $\liminf_{n \to \infty} \p(\Theta_I \subseteq \wh \Theta_{\alpha}) \geq \alpha $;
\newline
(ii) If $T = \mb R^{d^*}$ then: $\lim_{n \to \infty} \p(\Theta_I \subseteq \wh \Theta_{\alpha}) = \alpha $.
\end{theorem}

Theorem \ref{t:main} shows that we need the generalized information equality $\Sigma=I_{d^*}$ to hold so that our Procedure 1 CS $\wh \Theta_\alpha$ has valid frequentist coverage for $\Theta_I$ in large samples.\footnote{This is consistent with the fact that percentile CSs also need $\Sigma=I_{d^*}$ in order to have a correct coverage for a point-identified scalar parameter (see, e.g., \cite{CH} and \cite{CRobert}).} This is because the asymptotic distribution of $\sup_{\theta \in \Theta_I} Q_n(\theta)$ is $F_T$ only under $\Sigma=I_{d^*}$. It follows that, with a criterion satisfying  $\Sigma = I_{d^*}$, our CS $\wh \Theta_\alpha$ will be asymptotically exact (for $\Theta_I$) when $T = \mb R^{d^*}$, and asymptotically valid but possibly conservative when $T$ is a convex cone.

\medskip

\begin{remark}\label{s:opt}
Theorem \ref{t:main} is still applicable to misspecified, separable partially-identified likelihood models. We can write the density in such models as $p_\theta(\cdot) = q_{\tilde \gamma(\theta)}(\cdot)$ where $\tilde \gamma(\theta)$ is an identifiable reduced-form parameter (see Section \ref{s:rfr} below). Under misspecification the identified set is $\Theta_I = \{ \theta : \tilde \gamma(\theta) = \tilde \gamma^*\}$ where $\tilde \gamma^*$ is the unique maximizer of $E[\log q_{\tilde \gamma}(X_i)]$ over $\wt \Gamma = \{\tilde \gamma(\theta) : \theta \in \Theta\}$. Following the insight of \cite{Mueller}, we could base our inference on the sandwich log-likelihood function:
\[
 L_n(\theta) = -\frac{1}{2}  (\check \gamma - \tilde \gamma(\theta)) ' (\wh \Sigma_S)^{-1} (\check \gamma - \tilde \gamma(\theta))
\]
where $\check \gamma$ approximately maximizes $\frac{1}{n} \sum_{i=1}^n \log q_\gamma(X_i)$ over $\wt \Gamma$ and $\wh \Sigma_S$ is the sandwich covariance matrix estimator for $\check \gamma$. If $\sqrt n (\check \gamma - \tilde \gamma^*) \rightsquigarrow N(0,\Sigma_S)$ and $\hat \Sigma_S \to_p \Sigma_S$ with $\Sigma_S$ positive definite, then Assumption \ref{a:quad} will hold with $\hat \gamma_n = \Sigma^{-1/2}_S(\check \gamma - \tilde \gamma^*)$ where $\sqrt n \hat \gamma_n \to_d N(0,I_{d^*})$ and $\gamma(\theta) = \Sigma^{-1/2}_S(\tilde \gamma(\theta) - \gamma^*)$.

\end{remark}

\medskip

\begin{remark}
In correctly specified likelihood models with flat priors, one may interpret $\wh \Theta_\alpha$ as a  HPD 100$\alpha$\% BCS for $\theta$. \cite{MoonSchorfheide} (MS hereafter) show that BCSs for a partially identified parameter $\theta$ (or subvectors) can under-cover asymptotically. As CSs for $\Theta_I$ should be larger than CSs for $\theta$, MS's result might appear to suggest that our Procedure 1 CS $\wh \Theta_\alpha$ would under-cover for $\Theta_I$. The ``apparent contradiction'' is because a key regularity condition in MS's under-coverage result (their Assumption 2 on p. 767) is violated in our setting.  For partially identified separable models, MS put a prior on the globally identified reduced-form parameter $\gamma$, say $\Pi(\gamma)$, and then a conditional prior, say $\Pi(\theta|\gamma)$, on the structural parameter $\theta$ given $\gamma$. The conditional prior $\Pi(\cdot|\gamma)$ needs to be supported on what would be the identified set for $\theta$ if $\gamma$ were the true reduced form parameter. Their Assumption 2 requires that $\Pi(\cdot|\gamma)$ is (locally) Lipschitz in $\gamma$, which is violated in our setting. We only put a prior on $\theta$. This prior on $\theta$ induces a prior on $\gamma = \gamma(\theta)$ and a conditional prior $\Pi(\theta|\gamma)$ that is supported on $\{\theta \in \Theta : \gamma(\theta) = \gamma\}$. Since $\gamma(\theta) = 0$ if and only if $\theta \in \Theta_I$, for any $\bar \gamma \neq 0$ our induced conditional prior satisfies
\[
 \Pi(\{\theta \in \Theta: \gamma(\theta) = 0\} | \gamma = 0) - \Pi(\{\theta \in \Theta : \gamma(\theta) = 0\} | \gamma = \bar \gamma ) = 1-0 = 1,
\]
thereby violating MS's Lipschitz condition (their Assumption 2).
See Remark 3 in MS for additional discussion of violation of their Lipschitz condition.
\end{remark}

\subsubsection{Models with singularities}

In this subsection we consider (possibly) partially identified models with singularities.\footnote{Such models are also referred to as non-regular models or models with non-regular parameters.}
In identifiable parametric models $\{ P_\theta : \theta \in \Theta\}$, the standard notion of differentiability in quadratic mean requires that the mass of the part of $P_\theta$ that is singular with respect to the true distribution $P_0 = P_{\theta_0}$ vanishes faster than $\|\theta - \theta_0\|^2$ as $\theta \to \theta_0$ \cite[section 6.2]{LeCamYang}. If this condition fails then the log-likelihood will not be locally quadratic at $\theta_0$. By analogy with the identifiable case, we say a non-identifiable model has a singularity if it does not admit a local quadratic approximation (in the reduced-form reparameterization) like that in Assumption \ref{a:quad}. One example is the missing data model under identification (see Subsection \ref{s:md} below).

To allow for partially identified models with singularities, we first  generalize the notion of the local reduced-form reparameterization to be of the form $\theta \mapsto (\gamma(\theta),\gamma_\bot(\theta))$ from $\Theta_I^N$ into $\Gamma \times \Gamma_\bot$ where $\Gamma \subseteq \mb R^{d^*}$ and $\Gamma_\bot \subseteq \mb R^{\dim(\gamma_\bot)}$ with $(\gamma(\theta),\gamma_\bot(\theta)) = 0$ if and only if $\theta \in \Theta_I$. The following regularity conditions generalize Assumptions \ref{a:quad} and \ref{a:prior} to allow for singularity.

\setcounter{aprime}{1}

\begin{aprime}\label{a:quad:prime} (Local quadratic approximation with singularity) \\
	(i) There exist sequences of random variables $\ell_n$ and $\mb R^{d^*}$-valued random vectors $\hat \gamma_n$ (both measurable in $\mf X_n$), and a sequence of functions $f_{n,\bot}: \Gamma_\bot \to \mb R_{+}$ (measurable in $\mf X_n$) with $f_{n,\bot}(0)=0$ (almost surely), such that as $n \to \infty$:
	\begin{equation} \label{e:quad:prime}
	\sup_{\theta \in \Theta_{osn}} \left| n L_n(\theta) - \left( \ell_n + \frac{1}{2} \|\sqrt n \hat \gamma_n\|^2 - \frac{1}{2} \| \sqrt n(\hat \gamma_n - \gamma(\theta))\|^2 -  f_{n,\perp}(\gamma_\perp(\theta)) \right)  \right| = o_\p(1)
	\end{equation}
	with $\sup_{\theta \in \Theta_{osn}} \|(\gamma(\theta),\gamma_\bot(\theta))\| \to 0$ and $\sqrt n \hat \gamma_n = \mf T \mb V_n$ where $\mb V_n \rightsquigarrow N(0,\Sigma)$; \\
	(ii) $\{(\gamma(\theta),\gamma_{\bot}(\theta)) : \theta \in \Theta_{osn}\} = \{\gamma(\theta) : \theta \in \Theta_{osn}\} \times \{\gamma_{\bot}(\theta) : \theta \in \Theta_{osn}\}$.
\end{aprime}

Let $\Pi_{\Gamma^*}$ denote the image of the measure $\Pi$ under the map $\Theta_I^N \ni \theta \mapsto (\gamma(\theta),\gamma_\bot(\theta))$. Let $B_r^* \subset \mb R^{d^*+\dim(\gamma_\bot)}$ denote a ball of radius $r$ centered at the origin.

\begin{aprime}\label{a:prior:prime} (Prior with singularity) \\
	(i) $\int_\Theta e^{nL_n(\theta)} \,\mr d \Pi(\theta) < \infty $ almost surely\\
	(ii) $\Pi_{\Gamma^*}$ has a continuous, strictly positive density $\pi_{\Gamma^*}$ on $B_\delta^* \cap (\Gamma \times \Gamma_\bot)$ for some $\delta > 0$.
\end{aprime}

\paragraph{Discussion of Assumptions:}
Assumption \ref{a:quad:prime}' is generalizes of Assumption \ref{a:quad} to the singular case. Assumption \ref{a:quad:prime}' implies that the peak of the likelihood does not concentrate on sets of the form $\{ \theta : f_{n,\perp}(\gamma_\perp(\theta)) > \epsilon >0\}$. Recently, \cite{BochkinaGreen} established a BvM result for \emph{identifiable} parametric likelihood models with singularities. They assume the likelihood is locally quadratic in some parameters and locally linear in others (similar to Assumption \ref{a:quad:prime}'(i)) and  that the local parameter space satisfies conditions similar to our Assumption \ref{a:quad:prime}'(ii). Assumption \ref{a:prior:prime}' generalizes Assumption \ref{a:prior} to the singular case. We impose no further restrictions on the set $\{\gamma_\bot(\theta) : \theta \in \Theta_I^N \}$.

The next lemma shows that the posterior distribution of the QLR asymptotically (first-order) stochastically dominates $F_T$ in partially identified models with singularity.

\begin{lemma}\label{l:post:prime}
	Let Assumptions \ref{a:rate}, \ref{a:quad:prime}' and \ref{a:prior:prime}' hold. Then:
\begin{equation} \label{e:post:qlr:prime}
	\sup_z  \left( \Pi_n \big(\{ \theta : Q_n(\theta) \leq z \}\big|\,\mf X_n\big) - \p_{Z | \mf X_n} \Big( \|Z\|^2 \leq z  \Big| Z \in T - \sqrt n \hat \gamma_n \Big)  \right) \leq o_\p(1) \,.
	\end{equation}
Hence: $\sup_z  \left( \Pi_n \big(\{ \theta : Q_n(\theta) \leq z \}\big|\,\mf X_n\big) - F_T(z) \right) \leq o_\p(1)$.
\end{lemma}

Lemma \ref{l:post:prime} immediately implies the following result.

\begin{theorem}\label{t:main:prime}
	Let Assumptions \ref{a:rate}, \ref{a:quad:prime}', \ref{a:prior:prime}', and \ref{a:mcmc} hold with $\Sigma=I_{d^*}$. Then for any $\alpha$ such that $F_T (\cdot )$ is continuous at its $\alpha$ quantile, we have: $\liminf_{n \to \infty} \p(\Theta_I \subseteq \wh \Theta_{\alpha}) \geq \alpha$.
\end{theorem}

For non-singular models, Theorem \ref{t:main} establishes that $\wh \Theta_\alpha$ is asymptotically valid for $\Theta_I$, with asymptotically exact coverage when $T$ is linear and can be conservative when $T$ is a closed convex cone. For singular models, Theorem \ref{t:main:prime} shows that $\wh \Theta_\alpha$ is still asymptotically valid for $\Theta_I$ but can be conservative even when $T$ is linear.\footnote{It might be possible to establish asymptotically exact coverage of $\wh \Theta_\alpha$ for $\Theta_I$ in singular models where the singular part $f_{n,\perp}(\gamma_\perp(\theta))$ in Assumption \ref{a:quad:prime}' possesses some extra structure.} When applied to the missing data example, Theorems \ref{t:main} and \ref{t:main:prime} imply that $\wh \Theta_\alpha$ for $\Theta_I$ is asymptotically exact under partial identification but conservative under point identification. This is consistent with simulation results reported in Table \ref{t:ex1}; see Section \ref{s:md} below for details.

\subsection{Coverage properties of $\wh M_\alpha$ for $M_I$}\label{s:subvec}

Here we present conditions under which $\wh M_\alpha$ has correct coverage for the identified set $M_I$ of subvectors $\mu$. Recall the definition of $M(\theta) \equiv  \{ \mu : (\mu,\eta) \in \Delta(\theta) \mbox{ for some } \eta\}$ from Section \ref{sec-procedure}. The profile criterion $PL_n(M(\theta))$ for $M(\theta)$ and the profile QLR $PQ_n(M(\theta))$ for $M(\theta)$ are defined the same way as those in (\ref{PL-set-b}) and (\ref{PQLR-set-b}) respectively:
\[
PL_n(M(\theta)) \equiv \inf_{\mu \in M(\theta)} \sup_{\eta \in H_\mu} L_n(\mu,\eta) \quad \mbox{and} \quad PQ_n(M(\theta))\equiv 2n [L_n(\hat \theta) - PL_n(M(\theta))].
\]

\begin{assumption}\label{a:qlr:profile} (Profile QL) \\
	There exists a measurable $f : \mb R^{d^*} \to \mb R_{+}$ such that:
	\begin{align*}
	& \sup_{\theta \in \Theta_{osn}} \left|  nP L_n(M(\theta)) - \left( \ell_n + \frac{1}{2} \|\sqrt n \hat \gamma_n \|^2 - \frac{1}{2} f \left( \sqrt n (\hat \gamma_n - \gamma(\theta))  \right)  \right) \right| = o_\p(1)
	\end{align*}
	with $\hat \gamma_n$ and $\gamma(\cdot)$ from Assumption \ref{a:quad} or \ref{a:quad:prime}'.
\end{assumption}

\paragraph{Discussion of Assumption \ref{a:qlr:profile}:} By definition of $M_I$ (in display (\ref{M-I})) we have: $M_I = \{\mu : \gamma(\mu,\eta)=0 \mbox{ for some } \eta \in H_\mu\}$ and also
$M_I = M(\theta)$ for any $\theta \in \Theta_I$. Thus
\[
PQ_n(M_I) = \sup_{\mu \in M_I} \inf_{\eta \in H_\mu} Q_n (\mu, \eta)
=PQ_n(M(\theta))\quad  \mbox{for all $\theta \in \Theta_I$}~.
\]
Assumption \ref{a:qlr:profile} imposes some structure on the profile QLR statistic for $M_I$ over the local neighborhood $\Theta_{osn}$. It implies that the profile QLR for $M_I$ is of the form:
\begin{equation} \label{e:qlr:subvec}
 PQ_n(M_I) = f (\mf T \mb V_n ) + o_\p(1) \,.
\end{equation}
When $\Sigma = I_{d^*}$, the asymptotic distribution of $\sup_{\theta \in \Theta_I} PQ_n(M(\theta))=PQ_n(M_I)$ becomes $G_T$:
\[
 G_T (z) := \p_Z (f( \mf T Z ) \leq z) \,~~~\mbox{where}~~Z \sim N(0,I_{d^*})~.
\]
The functional form of $f$ depends on the local reparameterization $\gamma$ and the geometry of $M_I$. When $M_I$ is a singleton and $T = \mb R^{d^*}$ then equation (\ref{e:qlr:subvec}) is typically satisfied with $f(v) = \inf_{t \in T_1} \|v - t\|^2$ where $T_1= \mb R^{d_1^*}$ with $d_1^* < d^*$ and the QLR statistic is asymptotically $\chi^2_{d^*-d_1^*}$. For a non-singleton set $M_I$, $f$ will typically be more complex. In the missing data example, we show in Section \ref{s:suff} that $f(v) = \max_{\mu \in \{ \ul \mu,\ol \mu\}} \inf_{t \in T_\mu} \|v - t\|^2$ where $T_{\ul \mu}$ and $T_{\ol \mu}$ are halfspaces and $T = \mb R^{d^*}$. Here the resulting profile QLR statistic for $M_I$ is asymptotically  the maximum of two mixtures of $\chi^2$ random variables. Note that the existence of $f$ is merely a proof device, and one does not need to know its precise expression in the implementation of our Procedure 2 CS $\wh M_{\alpha}$ for $M_I$.

The next lemma is a new BvM-type result for the posterior distribution of the profile QLR for $M_I$. Note that this result also allows for singular models.

\begin{lemma}\label{l:post:profile}
	Let Assumptions \ref{a:rate}, \ref{a:quad}, \ref{a:prior}, and \ref{a:qlr:profile} or \ref{a:rate}, \ref{a:quad:prime}', \ref{a:prior:prime}', and \ref{a:qlr:profile} hold. Then for any interval $I$ such that $\p_{Z} ( f(Z) \leq z )$ is continuous on $I$, we have:
	\begin{align*}
	\sup_{z \in I} \left| { \Pi_n \big( \{\theta:PQ_n( M(\theta)) \leq  z \} \,\big|\, \mf X_n \big) } - \p_{Z | \mf X_n} \Big( f(Z) \leq z  \Big| Z \in \sqrt n \hat \gamma_n - T \Big)  \right| = o_\p(1)\,.
	\end{align*}
And hence we have:\\
 (i) If $T \subsetneq \mb R^{d^*}$ and $f$ is subconvex,\footnote{We say that $f : \mb R^{d^*} \to \mb R_+$ is quasiconvex if $f^{-1}(z) := \{ v : f(v) \leq z\}$ is convex for each $z \geq 0$ and subconvex if, in addition, $f(v) = f(-v)$ for all $v \in \mb R^{d^*}$. The conclusion of Lemma \ref{l:post:profile}(i) remains valid under the weaker condition that (i) $f$ is quasiconvex and (ii) $\p_Z( Z \in ( f^{-1}(\xi_\alpha)- T^o)) \leq \p_Z( f(\mf T Z) \leq \xi_\alpha) $ holds, where $\xi_\alpha$ is the $\alpha$ quantile of $G_T$ and $T^o$ is the polar cone of $T$.} then: $\Pi_n \big(\{ \theta : Q_n(\theta) \leq z \}\big|\,\mf X_n\big) \leq G_T (z)$ for all $z \geq 0$.\\
 (ii) If $T= \mb R^{d^*}$ then: $\sup_z \left| \Pi_n \big(\{ \theta : Q_n(\theta) \leq z\} \,\big|\,\mf X_n\big) - \mb P_{Z} \big( f ( Z )  \leq z\big) \right| = o_\p(1)$.
\end{lemma}

 Let $\xi_{n,\alpha}^{post,p}$ denote the $\alpha$ quantile of the profile QLR $PQ_n(M(\theta))$ under the posterior distribution $\Pi_n$, and $\xi_{n,\alpha}^{mc,p}$ be given in Remark \ref{rmk:subvec}.

\begin{assumption}\label{a:mcmc:profile} (MC convergence) \\
	$\xi_{n,\alpha}^{mc,p} = \xi_{n,\alpha}^{post,p} + o_\p(1)$.
\end{assumption}

The next theorem is an important consequence of Lemma \ref{l:post:profile}.

\begin{theorem} \label{t:main:profile}
	Let Assumptions \ref{a:rate}, \ref{a:quad}, \ref{a:prior}, \ref{a:qlr:profile}, and \ref{a:mcmc:profile} or \ref{a:rate}, \ref{a:quad:prime}', \ref{a:prior:prime}', \ref{a:qlr:profile}, and \ref{a:mcmc:profile} hold with $\Sigma=I_{d^*}$ and suppose that $G_T (\cdot )$ is continuous at its $\alpha$ quantile.
\newline
(i) If $T \subsetneq \mb R^{d^*}$ and $f$ is subconvex,\footnote{The conclusion of Theorem \ref{t:main:profile}(i) remains valid under the weaker condition stated in footnote for Lemma \ref{l:post:profile}(i).} 
then: $\liminf_{n \to \infty} \p(M_I \subseteq \wh M_{\alpha}) \geq \alpha $ ;
\newline
(ii) If $T = \mb R^{d^*}$ then: $\lim_{n \to \infty} \p(M_I \subseteq \wh M_{\alpha}) = \alpha $.
\end{theorem}

Theorem \ref{t:main:profile}(ii) shows that our Procedure 2 CSs $\wh M_{\alpha}$ for $M_I$ can have asymptotically exact coverage if $T = \mb R^{d^*}$ even if the model is singular. In the missing data example, Theorem \ref{t:main:profile}(ii) implies that $\wh M_\alpha$ for $M_I$ is asymptotically exact irrespective of whether the model is point-identified or not (see Subsection \ref{s:md} below). Theorem \ref{t:main:profile}(i) shows that the CSs $\wh M_{\alpha}$ for $M_I$ can have conservative coverage when $T$ is a convex cone (see Appendix \ref{s:drift} for a moment inequality example).

\subsection{Coverage properties of $\wh M_\alpha^\chi$ for $M_I$ of scalar subvectors} \label{s:mchi}

This section presents one sufficient condition for validity of our Procedure 3 CS $\wh M_\alpha^\chi$ for $M_I \subset \mb R$. We say a half-space is \emph{regular} if it is of the form $\{v \in \mb R^{d^*} : a'v \leq 0\}$ for some $a \in \mb R^{d^*}$.

\begin{assumption}\label{a:qlr:chi}(Profile QLR, $\chi^2$ bound) \\
	$PQ_n(M(\theta)) \rightsquigarrow W \leq \max_{i \in \{1,2\}} \inf_{t \in T_i} \| Z - t\|^2$ for all $\theta \in \Theta_I$, where $Z \sim N(0,I_{d^*})$ for some $d^* \geq 1$ and $T_1$ and $T_2$ are regular half-spaces in $\mb R^{d^*}$.
\end{assumption}

\begin{theorem}\label{t:profile:chi}
	Let Assumption \ref{a:qlr:chi} hold and let the distribution of $W$ be continuous at its $\alpha$ quantile. Then: $\liminf_{n \to \infty} \p( M_I \subseteq \wh M_\alpha^\chi ) \geq \alpha$.
\end{theorem}

The following proposition presents a set of sufficient conditions for Assumption \ref{a:qlr:chi}.

\begin{proposition}\label{p:qlr:chi}
	Let the following hold:\\
	(i) Assumptions \ref{a:rate}(i), \ref{a:quad} or \ref{a:quad:prime}' hold with $\Sigma = I_{d^*}$ and $T = \mb R^{d^*}$; \\
    (ii) $\inf_{\mu \in M_I} \sup_{\eta \in H_\mu} L_n(\mu,\eta) = \min_{\mu \in \{ \ul \mu, \ol \mu\}} \sup_{\eta \in H_\mu} L_n(\mu,\eta) + o_\p(n^{-1})$; \\
	(iii) for each $\mu \in \{ \ul \mu, \ol \mu\}$ there exists a sequence of sets $(\Gamma_{\mu,osn})_{n \in \mb N}$ with $\Gamma_{\mu,osn} \subseteq \Gamma$ for each $n$ and a halfspace $T_\mu$ in $\mb R^{d^*}$ such that:
	\[
	\sup_{\eta \in H_\mu} n L_n(\mu,\eta) = \sup_{\gamma \in \Gamma_{\mu,osn}} \left( \ell_n + \frac{1}{2} \|\mb V_n\|^2 - \frac{1}{2} \|  \sqrt n \gamma - \mb V_n \|^2 \right) + o_\p(1)
	\]
	and $\inf_{\gamma \in \Gamma_{\mu,osn}} \| \sqrt n \gamma - \mb V_n \|^2 = \inf_{t \in T_\mu} \|t - \mb V_n\|^2 + o_\p(1)$. \\
	Then: Assumption \ref{a:qlr:chi} holds with $W=\max_{i \in \{ \ul \mu, \ol \mu\}} \inf_{t \in T_i} \| Z - t\|^2$.
\end{proposition}

Suppose $M_I = [\ul \mu, \ol \mu]\subsetneq \mb R$
(which is true when $\Theta_I$ is connected and bounded). If $\sup_{\eta \in H_\mu} L_n(\mu,\eta)$ is strictly concave in $\mu$ then condition (ii) of Proposition \ref{p:qlr:chi} holds. The other conditions of Proposition \ref{p:qlr:chi} are easy to verify as in the missing data example (see Subsection \ref{s:md} below).

The exact distribution of $\max_{i \in \{1,2\}} \inf_{t \in T_i} \| Z - t\|^2$ depends on the geometry of $T_1$ and $T_2$. For the missing data example, the polar cones of $T_1$ and $T_2$ are at least $90^o$ apart. The worst-case coverage (i.e., the case in which asymptotic coverage of $\wh M_\alpha^\chi$ will be most conservative) will occur when the polar cones of $T_1$ and $T_2$ are orthogonal, in which case $\max_{i \in \{1,2\}} \inf_{t \in T_i} \| Z - t\|^2$ has the mixture distribution $W^* :=\frac{1}{4} \delta_0 + \frac{1}{2}\chi^2_1 + \frac{1}{4} (\chi^2_1 \cdot \chi^2_1) $ where $\delta_0$ is a point mass at zero and $\chi^2_1 \cdot \chi^2_1$ is the distribution of the product of two independent $\chi^2_1$ random variables. The quantiles of the distribution of $\max_{i \in \{1,2\}} \inf_{t \in T_i} \| Z - t\|^2$ are continuous in $\alpha$ for all $\alpha > \frac{1}{4}$. For all configurations of $T_1$ and $T_2$ in this example, the distribution of $\max_{i \in \{1,2\}} \inf_{t \in T_i} \| Z - t\|^2$ (first-order) stochastically dominates $F_{W^*}$ and is (first-order) stochastically dominated by $F_{\chi^2_1}$ (i.e., $F_{W^*} (w) \geq F_{W} (w) \geq  F_{\chi^2_1} (w)$). Notice that this is different from the usual chi-bar-squared case encountered when testing whether a parameter belongs to the identified set on the basis of finitely many moment inequalities \citep{Rosen}.

To get an idea of the degree of conservativeness of  $\wh M_\alpha^\chi$, consider the class of models satisfying conditions for Proposition \ref{p:qlr:chi}. Figure \ref{f:cvg:chi} plots the asymptotic coverage of  $\wh M_\alpha$ and $\wh M_\alpha^\chi$ against nominal coverage for models in this class where $\wh M_\alpha^\chi$ is most conservative for the missing data example (i.e., the worst-case coverage). For each model in this class, the asymptotic coverage of $\wh M_\alpha$ and $\wh M_\alpha^\chi$ is between the nominal coverage and the worst-case coverage. As can be seen, the coverage of $\wh M_\alpha$ is exact at all levels $\alpha \in (0,1)$ for which the distribution of the profile QLR is continuous at its $\alpha$ quantile, as shown in Theorem \ref{t:main:profile}(ii). On the other hand, $\wh M_\alpha^\chi$ is asymptotically conservative, but the level of conservativeness decreases as $\alpha$ increases towards one. Indeed, for levels of $\alpha$ in excess of $0.85$ the level of conservativeness is negligible.

\begin{figure}[t]
\begin{center}
\includegraphics[trim = 0cm 0cm 0cm 0cm, clip, width = 0.6\textwidth]{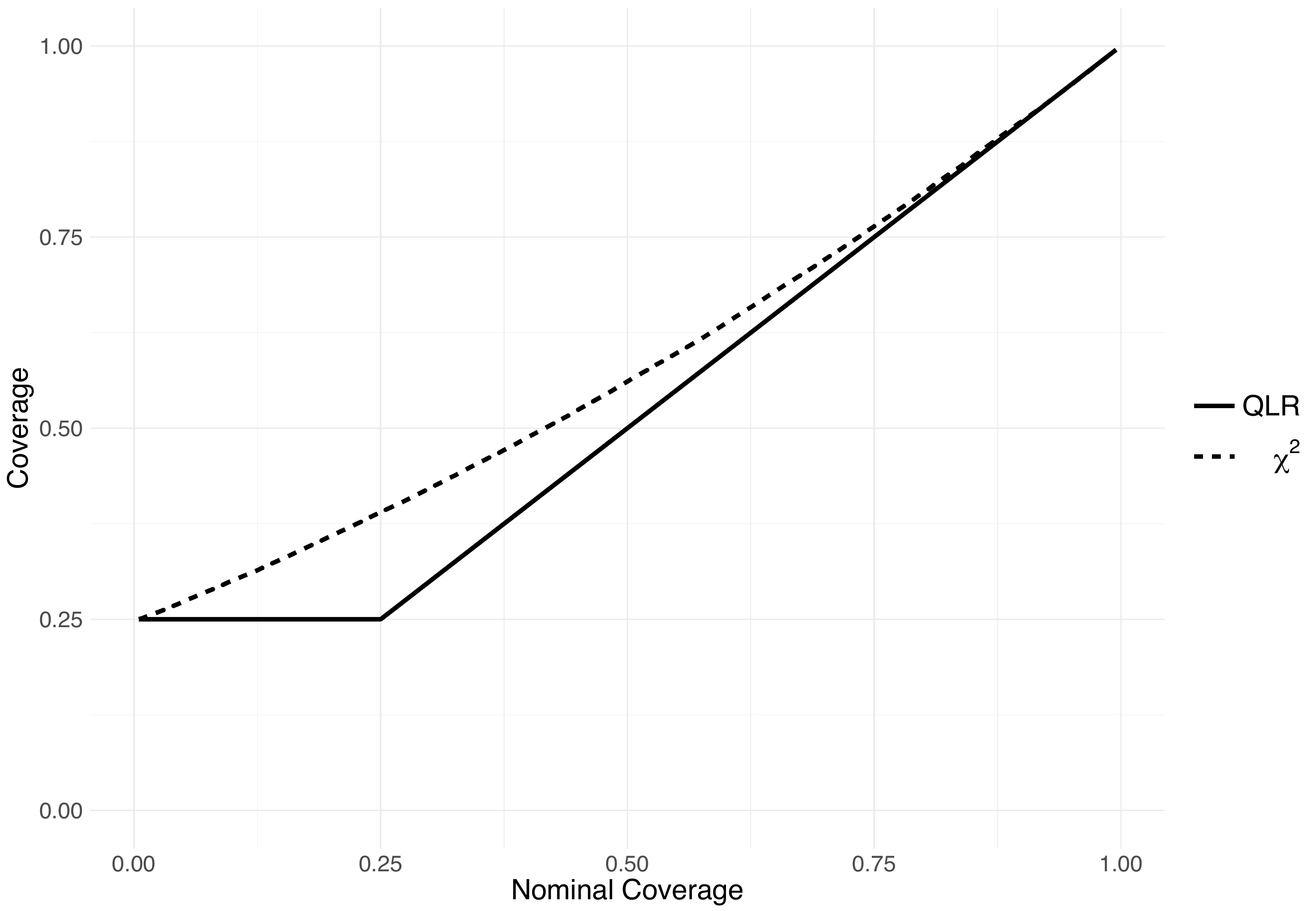}
\parbox{12cm}{\caption{\small\label{f:cvg:chi}  Missing data example: comparison of asymptotic coverage of $\wh M_\alpha$  (QLR -- solid kinked line) and $\wh M_\alpha^\chi$ ($\chi^2$ -- dashed curved line) with their nominal coverage for models where $\wh M_\alpha^\chi$ is valid for $M_I$ but most conservative.}}
\end{center}
\end{figure}

Since empirical papers typically report CSs for scalar parameters, Theorem \ref{t:profile:chi} will be very useful in applied work. One could generalize $\wh M_\alpha^\chi$ to deal with vector-valued subvectors by allowing $\chi_d^2$ quantiles with higher degrees of freedom $d \in (1,\dim(\theta))$, but it might be difficult to provide sufficient conditions as those in Proposition \ref{p:qlr:chi} to establish results like Theorem \ref{t:profile:chi}.

\section{Sufficient Conditions and Examples}\label{s:suff}

 This section provides sufficient conditions for the key regularity condition, Assumption \ref{a:quad}, in possibly partially identified likelihood and moment-based models with i.i.d. data. See Appendix \ref{s:uniformity-qe} for low-level conditions to ensure that Assumption \ref{a:quad} holds uniformly over a large class of DGPs in discrete models. We also verify Assumptions \ref{a:rate}, \ref{a:quad} (or \ref{a:quad:prime}'), \ref{a:prior} and \ref{a:qlr:profile}) in examples.

We use standard empirical process notation: $P_0 g$ denotes the expectation of $g(X_i)$ under the true probability measure $P_0$,  $\p_n g = n^{-1} \sum_{i=1}^n g(X_i)$ denotes expectation of $g(X_i)$ under the empirical measure, and  $\mb G_n g = \sqrt n (\p_n - P_0)g$ denotes the empirical process.

\subsection{Partially identified likelihood models}

Consider a parametric likelihood model $\mc P = \{p_\theta : \theta \in \Theta\}$ where each $p_\theta(\cdot)$ is a probability density with respect to a common $\sigma$-finite dominating measure $\lambda$. Let $p_0 \in \mc P$ be the true density under the data-generating probability measure, $D_{KL}(p \| q)$ denote the Kullback-Leibler divergence, and $h(p,q)^2 = \int (\sqrt p - \sqrt q)^2 \, \mr d \lambda$ denote the squared Hellinger distance between densities $p$ and $q$. The identified set is $\Theta_I = \{\theta \in \Theta : D_{KL}(p_0 \| p_\theta) = 0\}= \{\theta \in \Theta : h(p_0, p_\theta) = 0\}$.

\subsubsection{Separable likelihood models}\label{s:rfr}

For a large class of partially identified parametric likelihood models $\mc P= \{p_\theta : \theta \in \Theta\}$, there exists a function $\tilde \gamma : \Theta \to \wt \Gamma \subset \mb R^{d^*}$ for some possibly unknown $d^* \in [1, + \infty)$, such that $p_\theta(\cdot) = q_{\tilde \gamma(\theta)}(\cdot)$ for each $\theta \in \Theta$ and some densities $\{q_{\tilde \gamma(\theta)}(\cdot) : \tilde \gamma \in \wt \Gamma\}$. In this case we say that the model $\mc P$ is separable and admits a (global) reduced-form reparameterization. The reparameterization is assumed to be identifiable, i.e. $D_{KL}(q_{\tilde \gamma_0} \| q_{\tilde \gamma}) > 0$ for any $\tilde \gamma \neq \tilde \gamma_0$.  The identified set is $\Theta_I = \{\theta \in \Theta : \tilde \gamma(\theta) = \tilde \gamma_0\}$ where $\tilde \gamma_0$ is the true parameter, i.e. $p_0 = q_{\tilde \gamma_0}$. Models with discrete choice probabilities (such as the missing data and entry game designs we used in simulations) fall into this framework, where the vector $\tilde \gamma$ maps the structural parameters $\theta$ to the model-implied probabilities of discrete outcomes and the true probabilities $\tilde \gamma_0 \in \wt \Gamma$ of discrete outcomes are point-identified.

The following result presents one set of sufficient conditions for Assumptions \ref{a:rate}(ii) and \ref{a:quad} under conventional smoothness assumptions.

Let $\ell_{\tilde \gamma}(\cdot) := \log q_{\tilde \gamma}(\cdot)$, let $\dot \ell_{\tilde \gamma}$ and $\ddot \ell_{\tilde \gamma}$ denote the score and Hessian, let $\mb I_{0} := -P_0( \ddot \ell_{\tilde \gamma_0}^{\phantom\prime} )$ and let $\gamma(\theta) = \mb I_{0}^{1/2} (\tilde \gamma(\theta) - \tilde \gamma_0)$ and  $\Gamma = \{\mb I_{0}^{1/2} (\tilde \gamma\ - \tilde \gamma_0) : \tilde \gamma \in \wt \Gamma\}$.

\medskip

\begin{proposition} \label{p:rfr}
	Suppose that $\{q_{\tilde \gamma} : \tilde \gamma \in \wt \Gamma\}$ satisfies the following regularity conditions:\\
	(a) $X_1,\ldots,X_n$ is an i.i.d. sample from $q_{\tilde \gamma_0}$  with $\tilde \gamma_0$ identifiable and on the interior of $\wt \Gamma$; \\
	(b) $\tilde \gamma \mapsto P_0 \ell_{\tilde \gamma}$ is continuous and there is a neighborhood $U$ of $\tilde \gamma_0$ on which $\ell_{\tilde \gamma}(x)$ is twice continuously differentiable for each $x$, with $\dot \ell_{\tilde \gamma_0} \in L^2(P_0)$ and $\sup_{\tilde \gamma \in U}  \|\ddot \ell_{\tilde \gamma}(x) \| \leq \bar \ell(x)$ for some $\bar \ell \in L^2(P_0)$; \\
	(c) $P_0 \dot \ell_{\tilde \gamma} = 0$, $\mb I_0$ is non-singular, and $\mb I_0 = P_0( \dot \ell_{\tilde \gamma_0}^{\phantom \prime}\dot \ell_{\tilde \gamma_0}^{\prime})$; \\
	(d) $\wt \Gamma$ is compact and $\pi_\Gamma$ is strictly positive and continuous on $U$. \\
	Then: there exists a sequence $(r_n)_{n \in \mb N}$ with $r_n \to \infty$ and $r_n = o(n^{1/4})$ such that Assumptions \ref{a:rate}(ii) and \ref{a:quad} hold for the average log-likelihood (\ref{e:ll}) over $\Theta_{osn}:= \{ \theta \in \Theta : \|\gamma(\theta)\| \leq r_n/\sqrt n\}$ with $\ell_n = n\p_n \log p_0$, $\sqrt n \hat \gamma_n = \mb V_n = \mb I_{0}^{-1/2} \mb G_n (\dot \ell_{\tilde \gamma_0})$, $\Sigma = I_{d^*}$ and $T = \mb R^{d^*}$.
\end{proposition}

\subsubsection{General non-identifiable likelihood models}\label{s:genscore}

It is possible to define a local reduced-form reparameterization for non-identifiable likelihood models, even when $\mc P= \{p_\theta : \theta \in \Theta\}$ does not admit an explicit (global) reduced-form reparameterization. Let $\mc D \subset L^2(P_0)$ denote the set of all limit points of:
\[
\mc D_\epsilon := \left\{ \frac{\sqrt{p/p_0}-1}{h(p,p_0)} : p \in \mc P, 0 < h(p,p_0) \leq \epsilon\right\}
\]
as $\epsilon \to 0$ and let $\ol{\mc D}_{\epsilon} = \mc D_\epsilon \cup \mc D$. The set $\mc D$ is the set of generalized Hellinger scores,\footnote{It is possible to define sets of generalized scores via other measures of distance between densities. See \cite{LiuShao} and \cite{AGM}. Our results can easily be adapted to these other cases.} which consists of functions of $X_i$ with mean zero and unit variance. The cone $\mc T = \{ \tau d : \tau \geq 0, d \in \mc D\}$ is the tangent cone of the model $\mc P$ at $p_0$. We say that $\mc P$ is differentiable in quadratic mean (DQM) if each $p \in \mc P$ is absolutely continuous with respect to $p_0$ and for each $p \in \mc P$ there are elements $g_p \in \mc T$ and remainders $R_p \in L^2(\lambda)$ such that:
\[
\sqrt{p_{\phantom{.}}} - \sqrt{p_0} = g_p  \sqrt{p_0} + h(p,p_0) R_p
\]
with $\sup\{ \| R_p \|_{L^2(\lambda)} : h(p,p_0) \leq \varepsilon \} \to 0$ as $\varepsilon \to 0$. If the linear hull $\mathrm{Span}(\mc T)$ of $\mc T$ has finite dimension $d^* \geq 1$, then we can write each $g \in \mc T$ as $g = c(g)' \psi$ where $c(g) \in \mb R^{d^*}$ and the elements of $\psi = (\psi_1,\ldots,\psi_{d^*})$ form an orthonormal basis for $\mathrm{Span}(\mc T)$ in $L^2(P_0)$. Let $\mb T$ denote the orthogonal projection\footnote{If $\mc T \subseteq L^2(P_0)$ is a closed convex cone, the projection $\mb T f$ of any $f \in L^2(P_0)$ is defined as the unique element of $\mc T$ such that $\|f - \mb T f\|_{L^2(P_0)} = \inf_{t \in \mc T} \|f - t\|_{L^2(P_0)}$.} onto $\mc T$ and let $\gamma(\theta)$ be given by
\begin{align}\label{e-gamma-gen}
 \mb T (2(\sqrt{p_\theta/p_0}-1)) = \gamma(\theta)'\psi\,.
\end{align}

\medskip

\begin{proposition}\label{p:genscore}
	Suppose that $\mc P$ satisfies the following regularity conditions:\\
	(a) $\{\log p : p \in \mc P\}$ is $P_0$-Glivenko Cantelli; \\
	(b) $\mc P$ is DQM, $\mc T$ is closed and convex and $\mathrm{Span}(\mc T)$ has finite dimension $d^* \geq 1$;\\
	(c) there exists $\varepsilon > 0$ such that $\ol{\mc D}_\varepsilon$ is Donsker and has envelope $D \in L^2(P_0)$.\\
	Then: there exists a sequence $(r_n)_{n \in \mb N}$ with $r_n \to \infty$ and $r_n = o(n^{1/4})$, such that Assumption \ref{a:quad} holds for the average log-likelihood (\ref{e:ll}) over $\Theta_{osn}:= \{ \theta : h(p_\theta,p_0) \leq r_n/\sqrt n\}$
	with $\ell_n = n\p_n \log p_0$,  $\sqrt n \hat \gamma_n = \mb V_n = \mb G_n(\psi)$, $\Sigma = I_{d^*}$ and  $\gamma(\theta)$  defined in (\ref{e-gamma-gen}).
\end{proposition}

Proposition \ref{p:genscore} is a set of sufficient conditions for i.i.d. data; see Lemma \ref{l:genscore} in Appendix \ref{a:proofs} for a more general result. Assumption \ref{a:rate}(ii) can be verified under additional mild conditions (see, e.g., Theorem 5.1 of \cite{GGV2000}).

\subsection{GMM models}\label{s:gmm}

Consider the GMM model $\{ \rho_\theta : \theta \in \Theta\}$ with $\rho : \mcr X \times \Theta \to \mb R^{d_\rho}$. Let $g(\theta) = E[\rho_\theta(X_i)]$ and the identified set be $\Theta_I = \{ \theta \in \Theta :  g(\theta) = 0\}$ (we assume throughout this subsection that $\Theta_I$ is non-empty). When $\rho$ is of higher dimension than $\theta$, the set $\mc G = \{ g(\theta) : \theta \in \Theta\}$ will not contain a neighborhood of the origin. But, if the map $\theta \mapsto  g(\theta)$ is smooth (e.g. $\mc G$ is a smooth manifold) then $\mc G$ can typically be locally approximated at the origin by a closed convex cone $\mc T \subset \mb R^{d_\rho}$.

To simplify notation, we assume that for any $v \in \mr{Span}(\mc T)$ we may partition $\Omega^{-1/2} v$ so that its upper $d^*$ elements $[\Omega^{-1/2} v]_1$ are (possibly) non-zero and the remaining $d_\rho -d^*$ elements $[\Omega^{-1/2} v]_2 = 0$ (this can always be achieved by multiplying the moment functions by a suitable rotation matrix).\footnote{See our July 2016 working paper version for details.} If $\mc G$ contains a neighborhood of the origin then we simply take $\mc T = \mb R^{d_\rho}$ and $[\Omega^{-1/2} v]_1 = \Omega^{-1/2} v$. Let $\mb T g(\theta)$ denote the projection of $g(\theta)$ onto $\mc T \subset \mb R^{d_\rho}$ and note that $[\Omega^{-1/2} \mb T g(\theta)]_2=0$. Finally, define $\Theta_I^\varepsilon = \{ \theta \in \Theta : \|  g(\theta)\| \leq \varepsilon\}$.

\medskip

\begin{proposition}\label{p:gmm:cue}
	Suppose that $\{ \rho_\theta : \theta \in \Theta\}$ satisfies the following regularity conditions: \\
	(a) there exists $\varepsilon_0 > 0$ such that $\{\rho_\theta : \theta \in \Theta_I^{\varepsilon_0}\}$ is Donsker; \\
	(b) $E[\rho_\theta(X_i) \rho_\theta(X_i)'] = \Omega$ for each $\theta \in \Theta_I$ and $\Omega$ is positive definite; \\
	(c) there exists $\theta^* \in \Theta_I$ such that $\sup_{\theta \in \Theta_I^\varepsilon} E[\|\rho_\theta(X_i) - \rho_{\theta^*}(X_i)\|^2] = o(1)$ as $\varepsilon \to 0$;\\
	(d) there exists $\delta > 0$ such that $\sup_{\theta \in \Theta_I^\varepsilon} \| g(\theta) - \mb T g(\theta)\| = o(\varepsilon^{1+\delta})$ as $\varepsilon \to 0$.\\
	Then: there exists a sequence $(r_n)_{n \in \mb N}$ with $r_n \to \infty$ and $r_n = o(n^{1/4})$ such that Assumption \ref{a:quad} holds for the CU-GMM criterion (\ref{e:cue}) over $\Theta_{osn} = \{ \theta \in \Theta : \|g(\theta)\| \leq r_n/\sqrt n\}$, where
	 $\ell_n = -\frac{1}{2}Z_n'\Omega^{-1} Z_n$, $Z_n=\mb G_n (\rho_{\theta^*})$,  $\gamma(\theta) = [\Omega^{-1/2} \mb Tg(\theta)]_1$, and $\sqrt n \hat \gamma_n = \mb V_n = -[\Omega^{-1/2}Z_n]_1$ and $\Sigma = I_{d^*}$.\\
	If $\mc G$ contains a neighborhood of the origin then $\gamma(\theta) = \Omega^{-1/2} g(\theta)$ and $\sqrt n \hat \gamma_n = \mb V_n = - \Omega^{-1/2} Z_n$.
\end{proposition}

\medskip

\begin{proposition}\label{p:gmm:opt}
	Let all the conditions of Proposition \ref{p:gmm:cue} hold and let: (e) $\|\wh W - \Omega^{-1}\| = o_\p(1)$. \\
	Then: the conclusions of Proposition \ref{p:gmm:cue} hold for the optimally-weighted GMM criterion (\ref{e:ow}).
\end{proposition}

\subsubsection{Moment inequality models}\label{s:miq0}

Consider the moment inequality model $\{ \tilde \rho(X_i,\mu) : \mu \in M\}$ where $\tilde \rho$ is a $d_\rho$ vector of moments and the space is $M \subseteq \mb R^{d_\mu}$. The identified set for $\mu$ is $M_I = \{ \mu \in M : E[\tilde \rho(X_i,\mu)] \leq 0\}$ (the inequality is understood to hold element-wise). We may reformulate the moment inequality model as a moment equality model by augmenting the parameter vector with a vector of slackness parameters $\eta \in H = \mb R^{d_\rho}_+$. Thus we re-parameterize the model by $\theta = (\mu,\eta) \in \Theta = M \times H$ and write the inequality model as a GMM model with
\begin{equation}\label{aug-gmm}
E[\rho_\theta(X_i)] = 0~\mbox{for}~\theta \in \Theta_I,~~~\rho_\theta(X_i) = \tilde \rho(X_i,\mu) + \eta~,
\end{equation}
where the identified set for $\theta$ is $\Theta_I = \{ \theta \in \Theta : E[\rho_\theta(X_i)] = 0\}$ and $M_I$ is the projection of $\Theta_I$ onto $M$. Here the objective function would be as in display (\ref{e:cue}) or (\ref{e:ow}) using $\rho_\theta(X_i) = \tilde \rho(X_i,\mu) + \eta$. We may then apply Propositions \ref{p:gmm:cue} or \ref{p:gmm:opt} to the reparameterized GMM model (\ref{aug-gmm}).

As the parameter of interest is $\mu$, one could use our Procedures 2 or 3 for inference on $M_I$. These procedures involve the profile criterion $\sup_{\eta \in H} L_n(\mu,\eta)$ which is simple to compute because the GMM objective function is quadratic in $\eta$ for given $\mu$ (since the optimal weighting or continuous updating weighting matrix will typically not depend on $\eta$). See Example 3 in Subsection \ref{s:miq}.

\subsection{Examples}\label{s:ex}

\subsubsection{Example 1: missing data model in Subsection \ref{s:missing}}\label{s:md}

We revisit the missing data example in Subsection \ref{s:missing}, where the parameter space $\Theta$ for $\theta = (\mu, \eta_1, \eta_2)$ is given in (\ref{e:theta:md}), the identified set for $\theta$ is $\Theta_I$ given in (\ref{e:thetaI:md}), and the identified set for $\mu$ is $M_I = [\tilde \gamma_{11},\tilde \gamma_{11}+\tilde \gamma_{00}]$.

\paragraph{Inference under partial identification:}

Consider the case in which the model is partially identified (i.e. $0 < \eta_2 < 1$). The likelihood of the $i$-th observation $(D_i,Y_iD_i)=(d,yd)$ is
\begin{align*}
 p_\theta(d,yd) & = [\tilde \gamma_{11}(\theta)]^{yd} [1-\tilde \gamma_{11}(\theta)-\tilde \gamma_{00}(\theta)]^{d - y d} [\tilde \gamma_{00}(\theta)]^{1-d} = q_{\tilde \gamma(\theta)}(d,yd)
\end{align*}
where:
\[
 \tilde \gamma(\theta) = \left( \begin{array}{c}  \tilde \gamma_{11}(\theta) - \tilde \gamma_{11} \\ \tilde \gamma_{00}(\theta) -\tilde \gamma_{00} \end{array} \right)
\]
with $\wt \Gamma = \{\tilde \gamma(\theta) : \theta \in \Theta\} = \{(g_{11}-\tilde \gamma_{11},g_{00}-\tilde \gamma_{00}) : (g_{11},g_{00}) \in[0,1]^2 , 0 \leq g_{11} \leq 1-g_{00}\}$. Conditions (a)-(b) of Proposition \ref{p:rfr} hold and Assumption \ref{a:quad} is satisfied with $\gamma(\theta) =  \mb I_{0}^{1/2} \tilde \gamma(\theta)$,
\begin{align*}
 \mb I_{0} & = \left[ \begin{array}{cc}
\frac{1}{\tilde \gamma_{11}} + \frac{1}{1-\tilde \gamma_{11}-\tilde \gamma_{00}} & \frac{1}{1-\tilde \gamma_{11}-\tilde \gamma_{00}} \\
\frac{1}{1-\tilde \gamma_{11}-\tilde \gamma_{00}} & \frac{1}{\tilde \gamma_{00}} + \frac{1}{1-\tilde \gamma_{11}-\tilde \gamma_{00}} \end{array} \right] &
\sqrt n \hat \gamma_n = \mb V_n & = \mb I_{0}^{-1/2} \mb G_n \left( \begin{array}{c}
\frac{y d }{\tilde \gamma_{11}} - \frac{d -y d }{1-\tilde \gamma_{11}-\tilde \gamma_{00}} \\
\frac{1-d }{\tilde \gamma_{00}} - \frac{d -y d }{1-\tilde \gamma_{11}-\tilde \gamma_{00}}
\end{array} \right)
\end{align*}
$\Sigma = I_2$ and $T = \mb R^2$. A flat prior on $\Theta$ in (\ref{e:theta:md}) induces a flat prior on $\Gamma$, which verifies Condition (c) of Proposition \ref{p:rfr} and Assumption \ref{a:prior}. Therefore, Theorem \ref{t:main}(ii) implies that our CSs $\wh \Theta_\alpha$ for $\Theta_I$ has asymptotically exact coverage.

Now consider CSs for $M_I = [\tilde \gamma_{11},\tilde \gamma_{11}+\tilde \gamma_{00}]$. Here $H_\mu = \{ (\eta_1,\eta_2) \in [0,1]^2 : 0 \leq \mu - \eta_1(1-\eta_2) \leq \eta_2\}$. By concavity in $\mu$, the profile log-likelihood for $M_I$ is:
\begin{align*}
PL_n(M_I)
& =  \min_{\mu \in\{ \ul \mu, \ol \mu \}} \sup_{\eta \in H_\mu} \p_n \log p_{(\mu,\eta)}
\end{align*}
where $\ul \mu= \tilde \gamma_{11}$ and $\ol \mu = \tilde \gamma_{11}+\tilde \gamma_{00}$. The inner maximization problem is:
\[
  \sup_{\eta \in H_\mu} \p_n \log p_{(\mu,\eta)}
 =  \sup_{ \substack{ 0 \leq g_{11} \leq \mu \\
		\mu  \leq g_{11}+g_{00} \leq 1 } } \!\!\!\!\! \p_n \Big(  yd \log g_{11} + (d-yd) \log (1-g_{11}-g_{00}) + (1- d) \log g_{00} \Big) .\label{e:md:lower}
\]
Let $g = (g_{11},g_{00})'$ and $\tilde \gamma = (\tilde \gamma_{11},\tilde \gamma_{00})'$ and let:
\begin{align*}
T_{\mu} & = \bigcup_{n \geq 1} \Big\{ \sqrt n \mb I_{0}^{1/2} (g-\tilde\gamma) \,:\, 0 \leq g_{11} \leq \mu , \;
\mu \leq g_{11}+g_{00} \leq 1 ,  \;
\|g-\tilde\gamma\|^2 \leq r_n^2/n \Big\}
\end{align*}
where $r_n$ is from Proposition \ref{p:rfr}.
It follows that:
\begin{align*}
nPL_n(M_I)
& = \ell_n + \frac{1}{2} \|\mb V_n\|^2 - \max_{\mu \in \{\ul \mu,\ol \mu\}} \frac{1}{2} \inf_{t \in T_\mu} \|\mb V_n - t\|^2 + o_\p(1)  \\
PQ_n(M_I) & = \max_{\mu \in \{\ul \mu,\ol \mu\}}  \inf_{t \in T_\mu} \|\mb V_n - t\|^2+ o_\p(1) \,.
\end{align*}
Equation (\ref{e:qlr:subvec}) and Assumption \ref{a:qlr:chi} therefore hold with $f(v) = \max_{\mu \in \{\ul \mu,\ol \mu\}}  \inf_{t \in T_\mu} \|v - t\|^2$ where $T_{\ul \mu}$ and $T_{\ol \mu}$ are regular halfspaces in $\mb R^2$. Theorem \ref{t:profile:chi} implies that the CS $\wh M_\alpha^\chi$ is asymptotically valid (but conservative) for $M_I$.

To verify Assumption \ref{a:qlr:profile}, take $n$ sufficiently large that $\gamma(\theta) \in \mr{int}(\Gamma)$ for all $\theta \in \Theta_{osn}$. Then:
\begin{align} \label{e:md:theta}
 PL_n (M ( \theta ))
& =  \min_{\mu \in \{ \tilde\gamma_{11}(\theta), \tilde\gamma_{11}(\theta)+ \tilde\gamma_{00}(\theta)\}} \sup_{ \eta \in H_\mu }  \p_n \log p_{(\mu,\eta)}\,.
\end{align}
This is geometrically the same as the profile QLR for $M_I$ up to a translation of the local parameter space from $(\tilde\gamma_{11},\tilde\gamma_{00})'$ to $(\tilde\gamma_{11}(\theta),\tilde\gamma_{00}(\theta))'$. The local parameter spaces are approximated by $T_{\ul \mu}(\theta) = T_{\ul \mu} + \sqrt n \gamma(\theta)$ and $T_{\ol \mu} (\theta) = T_{\ol \mu} + \sqrt n \gamma(\theta)$. It follows that uniformly in $\theta \in \Theta_{osn}$,
\begin{align*}
n PL_n ( M ( \theta ))
&  = \ell_n + \frac{1}{2} \|\mb V_n\|^2 - \frac{1}{2} f \big( \mb V_n - \sqrt n \gamma(\theta) \big) + o_\p(1)
\end{align*}
verifying Assumption \ref{a:qlr:profile}. Theorem \ref{t:main:profile}(ii) implies that $\wh M_\alpha$ has asymptotically exact coverage.

\paragraph{Inference under identification:}

Now consider the case in which the model is identified (i.e. $\eta_2 = 1$ and $\tilde \gamma_{00} = 0$) and $M_I = \{\mu_0\}$. Here each $D_i = 1$ so the likelihood of the $i$-th observation $(D_i,Y_iD_i)=(1,y)$ is
\begin{align*}
 p_\theta(1,y) & = [\tilde\gamma_{11}(\theta)]^{y} [1-\tilde\gamma_{11}(\theta)-\tilde\gamma_{00}(\theta)]^{1 - y} = q_{\tilde\gamma(\theta)}(1,y)
\end{align*}
Lemma \ref{l:md:singular} in Appendix \ref{a:proofs} shows that with $\Theta$ as in (\ref{e:theta:md}) and a flat prior, the posterior $\Pi_n$ concentrates on the local neighborhood $\Theta_{osn} = \{ \theta : |\tilde\gamma_{11}(\theta)-\tilde\gamma_{11}| \leq r_n/\sqrt n, \tilde\gamma_{00}(\theta) \leq r_n/n\}$ for any positive sequence $(r_n)_{n \in \mb N}$ with $r_n \to \infty$, $r_n/\sqrt n = o(1)$.

In this case, the reduced-form parameter is $\tilde\gamma_{11}(\theta)$ and the singular part is $\gamma_\perp(\theta) = \tilde \gamma_{00}(\theta) \geq 0$. Uniformly over $\Theta_{osn}$ we obtain:
\[
n L_n(\theta) = \ell_n - \frac{1}{2} \frac{(\sqrt n (\tilde\gamma_{11}(\theta)-\tilde\gamma_{11}))^2}{\tilde\gamma_{11}(1-\tilde\gamma_{11})} + \frac{\sqrt n( \tilde \gamma_{11}(\theta)-\tilde\gamma_{11})}{{\tilde\gamma_{11}(1-\tilde\gamma_{11})}} \mb G_n(y) - n \tilde\gamma_{00}(\theta)+ o_\p(1)
\]
which verifies Assumption \ref{a:quad:prime}'(i) with
\begin{align*}
 \gamma(\theta) & = \frac{\tilde \gamma_{11}(\theta)-\tilde\gamma_{11}}{\sqrt{\tilde\gamma_{11}(1-\tilde\gamma_{11})}} &
 \sqrt n \hat \gamma_n =\mb V_n & = \frac{\mb G_n(y)}{\sqrt{\tilde\gamma_{11}(1-\tilde\gamma_{11})}} &
 f_{n,\bot}(\gamma_\bot(\theta)) & = n\gamma_\bot(\theta)
\end{align*}
and $T = \mb R$. The remaining parts of Assumption \ref{a:quad:prime}' are easily shown to be satisfied. Therefore, Theorem \ref{t:main:prime} implies that $\wh \Theta_\alpha$ for $\Theta_I$ will be asymptotically valid but conservative.

For inference on $M_I = \{\mu_0\}$, the profile LR statistic is asymptotically $\chi^2_1$ and equation (\ref{e:qlr:subvec}) holds with $f(v) = v^2$ and $T = \mb R$. To verify Assumption \ref{a:qlr:profile}, for each $\theta \in \Theta_{osn}$ we need to solve
\begin{align*}
 \sup_{\eta \in H_\mu} \p_n \log p_{(\mu,\eta)}
 & =  \sup_{ \substack{ 0 \leq g_{11} \leq \mu \\
		\mu  \leq g_{11}+g_{00} \leq 1 } } \!\!\!\!\! \p_n \Big(  y \log g_{11} + (1-y) \log (1-g_{11}-g_{00})  \Big)
\end{align*}
at $\mu = \tilde\gamma_{11}(\theta)$ and $\mu = \tilde\gamma_{11}(\theta)+\tilde\gamma_{00}(\theta)$.
The maximum is achieved when $g_{00}$ is as small as possible, i.e., when $g_{00} = \mu - g_{11}$. Substituting in and maximizing with respect to $g_{11}$:
\[
 \sup_{\eta \in H_\mu} \p_n \log p_{(\mu,\eta)} = \p_n \big(y \log \mu + (1-y) \log(1-\mu) \big)\,.
\]
Therefore, we obtain the following expansion uniformly for $\theta \in \Theta_{osn}$:
\begin{align*}
n PL_n (M (\theta))
&  = \ell_n + \frac{1}{2} \mb V_n^2 - \frac{1}{2} \Big( \big( \mb V_n - \sqrt n \gamma(\theta) \big)^2 \vee \big( \mb V_n - \sqrt n ( \gamma(\theta) + \tilde \gamma_{00}(\theta)) \big)^2 \Big) + o_\p(1) \\
&  = \ell_n  + \frac{1}{2} \mb V_n^2 - \frac{1}{2} \big( \mb V_n - \sqrt n \gamma(\theta) \big)^2  + o_\p(1)
\end{align*}
where the last equality holds because $\sup_{\theta \in \Theta_{osn}} \tilde\gamma_{00}(\theta) \leq r_n/n = o(n^{-1/2})$. This verifies that Assumption \ref{a:qlr:profile} holds with  $f(v) = v^2$. Thus Theorem \ref{t:main:profile}(ii) implies that $\wh M_\alpha$ has asymptotically exact coverage for $M_I$, even though $\wh \Theta_\alpha$ is conservative for $\Theta_I$ in this case.

\subsubsection{Example 2: entry game with correlated shocks in Subsection \ref{s:game}}\label{s:game1}

Consider the bivariate discrete game with payoffs described in Subsection \ref{s:game}. Here we consider a slightly more general setting, in which $Q_\rho$ denotes a general joint distribution (not just bivariate Gaussian) for $(\epsilon_1,\epsilon_2)$ indexed by a parameter $\rho$. This model falls into the class of models dealt with in Proposition \ref{p:rfr}. Conditions (a)-(b) and (d) of Proposition \ref{p:rfr} hold with $\tilde \gamma(\theta) = (\tilde \gamma_{00}(\theta),\tilde \gamma_{10}(\theta),\tilde \gamma_{11}(\theta))'$ and $\wt \Gamma = \{\tilde \gamma(\theta) : \theta \in \Theta\}$ under very mild conditions on the parameterization $\theta \mapsto \tilde \gamma(\theta)$ (which, in turn, is determined by the specification of $Q_\rho$). Assumption \ref{a:quad} is therefore satisfied with:
\[
\mb I_{0} = \left[ \begin{array}{ccc}
\frac{1}{\tilde \gamma_{00}} & 0 & 0 \\
0 & \frac{1}{\tilde \gamma_{10}} & 0 \\
0 & 0 & \frac{1}{\tilde \gamma_{11}} \end{array} \right]  + \frac{1}{1-\tilde \gamma_{00}-\tilde \gamma_{10}-\tilde \gamma_{11}} \mf 1_{3 \times 3}
\]
where $\mf 1_{3 \times 3}$ denotes a $3 \times 3$ matrix of ones,
\[
 \sqrt n \hat \gamma_n = \mb V_n = \mb I_{0}^{-1/2} \mb G_n \left( \begin{array}{c}
\frac{d_{00}}{\tilde \gamma_{00}} - \frac{1-d_{00}-d_{10}-d_{11}}{1-\tilde \gamma_{00}-\tilde \gamma_{10}-\tilde \gamma_{11}} \\[4pt]
\frac{d_{01}}{\tilde \gamma_{10}} - \frac{1-d_{00}-d_{10}-d_{11}}{1-\tilde \gamma_{00}-\tilde \gamma_{10}-\tilde \gamma_{11}} \\[4pt]
\frac{d_{11}}{\tilde \gamma_{11}} - \frac{1-d_{00}-d_{10}-d_{11}}{1-\tilde \gamma_{00}-\tilde \gamma_{10}-\tilde \gamma_{11}}
\end{array} \right) \rightsquigarrow N(0,I_3)
\]
and $T = \mb R^3$. Condition (c) of Proposition \ref{p:rfr} and Assumption \ref{a:prior} can be verified under mild conditions on the map $\theta \mapsto \tilde \gamma(\theta)$ and the prior $\Pi$. For instance, consider the parameterization $\theta = (\Delta_1,\Delta_2,\beta_1,\beta_2,\rho,s)$ where the joint distribution of $(\epsilon_{1},\epsilon_{2})$ is a bivariate Normal with mean zero, standard deviations one and positive correlation $\rho \in [0,1]$. The parameter space is
\[
\Theta = \{ (\Delta_1, \Delta_2, \beta_1, \beta_2, \rho,s) \in \mb R^6 : \ul \Delta \leq \Delta_1, \Delta_2 \leq \ol \Delta , \ul \beta \leq \beta_1, \beta_2 \leq \ol \beta,  0 \leq \rho,s \leq 1\}\,.
\]
where $-\infty < \ul \Delta <  \ol \Delta < 0$ and $-\infty < \ul \beta <  \ol \beta < \infty$. The image measure  $\Pi_\Gamma$ of a flat prior on $\Theta$ is positive and continuous on a neighborhood of the origin, which verifies Condition (c) of Proposition \ref{p:rfr} and Assumption \ref{a:prior}. Therefore, Theorem \ref{t:main}(ii) implies that our MC CSs for $\Theta_I$ will have asymptotically exact coverage.

\subsubsection{Example 3: a moment inequality model}\label{s:miq}

As a simple illustration, suppose that $\mu \in M = \mb R_+$ is identified by the inequality $\mb E[\mu - X_i] \leq 0$ where $X_1,\ldots,X_n$ are i.i.d. with unknown mean $\mu^* \in \mb R_+$ and unit variance. The identified set for $\mu$ is $M_I = [0,\mu^*]$, which is the argmax of the population criterion function $L(\mu) = -\frac{1}{2} (( \mu - \mu^*) \vee 0)^2$ (see Figure \ref{f:Q}). The sample criterion $-\frac{1}{2 }((\mu - \bar X_n ) \vee 0)^2$ is typically used in the moment inequality literature but violates our Assumption \ref{a:quad}. However, we can rewrite the  model as the moment equality model:
$\mb E[ \mu + \eta - X_i] = 0$ where $\eta \in H = \mb R_+$ is a slackness parameter. The parameter space for $\theta = (\mu,\eta)$ is $\Theta = \mb R_+^2$. The identified set for $\theta$ is $\Theta_I = \{(\mu,\eta) \in \Theta : \mu + \eta = \mu^*\}$ and the identified set for $\mu$ is $M_I$ (see Figure \ref{f:Q}). The GMM objective function is then:
\[
 L_n(\mu,\eta) = -\frac{1}{2 } (\mu + \eta - \bar X_n)^2\,.
\]
It is straightforward to show that $2nL_n(\hat \mu,\hat \eta) = -((\mb V_n + \sqrt n \mu^* ) \wedge 0)^2$ where $\mb V_n = \sqrt n (\bar X_n - \mu^*)$. Moreover, $\sup_{\eta \in H_\mu} 2nL_n (\mu,\eta) = -((\mb V_n + \sqrt n (\mu^* - \mu)) \wedge 0)^2$ and so the profile QLR for $M_I$ is $PQ_n(M_I) = (\mb V_n \wedge 0)^2 - ((\mb V_n + \sqrt n \mu^* ) \wedge 0)^2$.

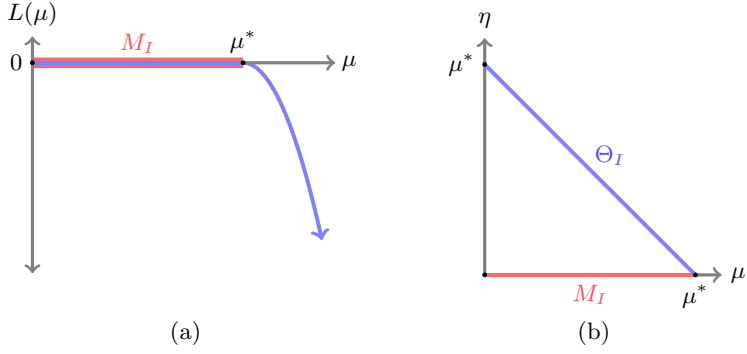
\begin{figure}[t]
	\begin{center}
	$\underset{\footnotesize \mbox{(a)}}{
		\begin{tikzpicture}[scale=0.7]
		\draw[red!60!white, line width=4pt] (0,4) -- (4,4);
		\draw[blue!50!white, line width=1.5pt] (0,4) -- (4,4);
		\draw [blue!50!white, line width=1.5pt,samples=100,domain=0:1.5, ->] plot(\x+4,{4-1.5*(\x)^2});
		\draw[gray, very thick, <->] (0,0) -- (0,4.5);
		\draw[gray, very thick, ->] (0,4) -- (5.75,4);
		\draw[blue!50!white, line width=1.5pt] (0,4) -- (4,4);
		\filldraw[black] (0,4) circle (1pt) node[anchor=east] {\footnotesize $0$};
		\filldraw[black] (6.35,4) circle (0pt) node[anchor=east] {\footnotesize $\mu$};
		\filldraw[black] (4,4) circle (1pt) node[anchor=south] {\footnotesize $\mu^*$};
		\filldraw[black] (6,0) circle (0pt) node[anchor=north] {\footnotesize $\phantom{(1)}$};
		\filldraw[black] (0,4.5) circle (0pt) node[anchor=south] {\footnotesize $L(\mu)$};
		\filldraw[red!70!white] (2,4) circle (0pt) node[anchor=south] {\footnotesize $M_I$};
		\filldraw[gray] (0,-0.76) circle (0pt)  ;
		\end{tikzpicture}}$%
		\quad \quad
		$\underset{\footnotesize \mbox{(b)}}{\begin{tikzpicture}[scale=0.7]
		\draw[gray, very thick, ->] (0,0) -- (0,4.5);
		\draw[gray, very thick, ->] (0,0) -- (4.5,0);
		\draw[blue!50!white, line width=1.5pt] (0,4) -- (4,0);
		\draw[red!60!white, line width=1.5pt] (0,0) -- (4,0);
		\filldraw[black] (0,0) circle (1pt) ;
		\filldraw[black] (0,4) circle (1pt) node[anchor=east] {\footnotesize $\mu^*$};
		\filldraw[black] (4,0) circle (1pt) node[anchor=north] {\footnotesize $\mu^*$};
		\filldraw[black] (5.15,0) circle (0pt) node[anchor=east] {\footnotesize $\mu$};
		\filldraw[black] (0,4.5) circle (0pt) node[anchor=south] {\footnotesize $\eta$};
		\filldraw[red!70!white] (2,0) circle (0pt) node[anchor=north] {\footnotesize $M_I$};
		\filldraw[blue!70!white] (1.9,1.9) circle (0pt) node[anchor=south west] {\footnotesize $\Theta_I$};
		\filldraw[gray] (0,-0.75) circle (0pt) ;
		\filldraw[gray] (-1,-0.70) circle (0pt) ;
		\end{tikzpicture}}$
		\vskip 4pt
		\parbox{12cm}{\caption{\small\label{f:Q}
		Panel (a): identified set $M_I$ for $\mu$ as the argmax of the population (moment inequality) criterion $L(\mu) = -\frac{1}{2}( (\mu - \mu^*) \vee 0)^2$. Panel (b): identified set $\Theta_I$ for $\theta =(\mu,\eta)$ for the moment equality model $  E[ \mu + \eta - X] = 0$.
		}}
	\end{center}
\end{figure}

For the posterior of the profile QLR, we also have $\Delta(\theta^b) = \{ \theta \in \Theta : \mu + \eta = \mu^b + \eta^b\}$ and $M(\theta^b) = [0,\mu^b + \eta^b]$. The profile QLR for $M(\theta^b)$ is
\[
 PQ_n(M(\theta^b)) = ((\mb V_n - \sqrt n (\mu^b + \eta^b - \mu^*)) \wedge 0)^2 - ((\mb V_n + \sqrt n \mu^* ) \wedge 0)^2
\]
This maps into our framework with the local reduced-form parameter $\gamma(\theta) = \mu + \eta - \mu^*$. Consider the case $\mu^* \in (cn^{\alpha-1/2},\infty)$ where $c > 0$ and $\alpha \in (0,\frac{1}{2}]$ are positive constants (we consider this case for the moment just to illustrate verification of our conditions). Here $T = \mb R$ and a positive continuous prior on $\mu$ and $\eta$ induces a prior on $\gamma$ that is positive and continuous at the origin. Moreover, Assumption \ref{a:qlr:profile} holds with $f(\kappa) = (\kappa \wedge 0)^2$. The regularity conditions of Theorem \ref{t:main:profile} hold, and hence $\wh M_\alpha$ has asymptotically exact coverage for $M_I$.

More generally, Appendix \ref{s:drift} shows that under very mild conditions our CS $\wh M_\alpha$ is uniformly valid over a class of DGPs $\mf P$, i.e.:
\[
 \liminf_{n \to \infty} \inf_{\p \in \mf P} \p( \mb M_I(\p) \subseteq \wh M_\alpha ) \geq \alpha
\]
where $M_I(\p) = [0,\mu^*(\p)]$ and the set $\mf P$ allows for any mean $\mu^*(\p) \in \mb R_+$ (encompassing, in particular, point-identified, partially identified, and drifting-to-point identified cases). In contrast, we construct sequences of DGPs $(\mr P_n)_{n \in \mb N} \subset \mf P$ along which bootstrap-based CSs $\wh M_\alpha^{boot}$ fail to cover with the prescribed coverage probability, i.e.:
\[
 \limsup_{n \to \infty} \mr P_n( \mb M_I(\mr P_n) \subseteq \wh M_\alpha^{boot} ) < \alpha\,.
\]
This reinforces the fact that our MC CSs for $M_I$ have very different asymptotic properties from bootstrap-based CSs for $M_I$.

\section{Conclusion}\label{sec-conclusion}

We propose new methods for constructing CSs for identified sets in partially-identified econometric models. Our CSs are relatively simple to compute and have asymptotically valid frequentist coverage uniformly over a class of DGPs, including partially- and point- identified parametric likelihood and moment based models. We show that under a set of sufficient conditions, and in broad classes of models, our set coverage is asymptotically exact. We also show that in models with singularities (such as the missing data example), our MC CSs for  $\Theta_I$ may be slightly conservative, but our MC CSs for identified sets of subvectors could still be asymptotically exact.
Simulation experiments demonstrate the good finite-sample coverage properties of our proposed CS constructions in standard difficult situations. We also illustrate our proposed CSs in two realistic empirical examples.

There are numerous extensions we plan to address in the future. The first natural extension is to allow for semiparametric likelihood or moment based models involving unknown and possibly partially-identified nuisance functions. We think this paper's MC approach could be extended to the partially-identified sieve MLE based inference in \cite{CTT}. A related, important extension is to allow for nonlinear structural models with latent state variables. Finally, we plan to study possibly misspecified and partially identified models.

\newpage

{ \small \singlespacing
\bibliographystyle{chicago}
\bibliography{mcmc-2}
}

\newpage

\appendix

\section{Additional details for the simulations and applications}\label{a:mc}

\subsection{An adaptive Sequential Monte Carlo algorithm}\label{s:smc}

We use an adaptive Sequential Monte Carlo (SMC) algorithm to sample from the quasi-posterior in (\ref{e:posterior}). Conventional MCMC algorithms such as the Metropolis-Hastings algorithm may fail to generate representative samples from the quasi-posterior in partially identified models or, more generally, models with multi-modal quasi-posteriors. For instance, the MCMC chain may get stuck exploring a single mode and fail to explore other modes if there is insufficient mass bridging the modes. In contrast, the SMC algorithm we use propagates large clouds of draws, in parallel, over a sequence of tempered distributions which begins with the prior, slowly incorporates information from the criterion, and ends with the quasi-posterior. The algorithm sequentially discards draws with relatively low mass as information is added, duplicates those with relatively high mass, then mutates the draws via a MCMC step to generate new draws (preventing particle impoverishment). Moreover, the algorithm is adaptive, i.e., the tuning parameters for the sequence of proposal distributions in the MCMC step are determined in a data-driven way.

The algorithm we use and its exposition below closely follows \cite{HS2014} who adapt a generic adaptive SMC algorithm to deal with large-scale DSGE models.\footnote{See \cite{Chopin2002,Chopin2004} and \cite{DDJ2006} for the generic SMC algorithm for estimating static model parameters, \cite{DDJ2012} and references therein for adaptive selection of tuning parameters with a SMC framework.} A similar algorithm is proposed by \cite{DG2014}, who emphasize its parallelizability. Let $J$ and $K$ be positive integers and let $\phi_1,\ldots,\phi_J$ be an increasing sequence with $\phi_1 = 0$ and $\phi_J = 1$. Set $w_1^b =  1$ for $b = 1,\ldots,B$ and draw $\theta^1_1,\ldots,\theta^B_1$ from the prior $\Pi(\theta)$. Then for $j = 2,\ldots,J$:
\begin{enumerate}
\item Correction: Let $v^b_j = e^{(\phi_j-\phi_{j-1})nL_n(\theta^b_{j-1})}$ and $w^b_j = (v^b_j w^b_{j-1})/(\frac{1}{B}\sum_{b=1}^B v^b_j w^b_{j-1})$.

\item Selection: Compute the effective sample size $ESS_j = B/(\frac{1}{B} \sum_{b=1}^B (w^b_j)^2)$. Then:
\begin{enumerate}
\item If $ESS_j > \frac{B}{2}$: set $\vartheta^b_j = \theta^b_{j-1}$ for $b = 1,\ldots,B$; or
\item If $ESS_j \leq \frac{B}{2}$: draw an i.i.d. sample $\vartheta^1_j,\ldots,\vartheta^B_j$ from the multinomial distribution with support $\theta^1_{j-1},\ldots,\theta^B_{j-1}$ and weights $w^1_{j},\ldots,w^B_{j}$, then set $w^b_j = 1$ for $b = 1,\ldots,B$.
\end{enumerate}

\item Mutation: Run $B$ separate and independent MCMC chains of length $K$ using the random-walk Metropolis-Hastings algorithm initialized at each $\vartheta_j^b$ for the tempered quasi-posterior $\Pi_j(\theta | \mf X_n) \propto e^{\phi_j n L_n(\theta)} \Pi(\theta)$ and let $\theta_j^b$ be the final draw of the $b$th chain.
\end{enumerate}

The resulting sample is $\theta^b = \theta^b_J$ for $b = 1,\ldots,B$. Multinomial resampling (step 2) and the $B$ independent MCMC chains (step 3) can both be computed in parallel, so the additional computational time relative to conventional MCMC methods is modest.

In practice, we take $J = 200$, $J = 1$, $4$ or $8$ (see below for the specific $K$ used in the simulations and empirical applications),  and $\phi_j = ( \frac{j-1}{J-1} )^{\lambda}$ with $\lambda = 2$. When the dimension of $\theta$ is low, in step 3 we use a $N(0,\sigma_j^2I)$ proposal density (all parameters are transformed to have full support) where $\sigma_j$ is chosen adaptively to target an acceptance ratio $\approx 0.35$ by setting $\sigma_2 = 1$ and
\[
 \sigma_j = \sigma_{j-1} \Big( 0.95 + 0.10 \frac{e^{16(A_{j-1}-0.35)}}{1+e^{16(A_{j-1}-0.35)}} \Big)
\]
for $j > 2$, where $A_{j-1}$ is the acceptance ratio from the previous iteration. If the dimension of $\theta$ is large, we partition $\vartheta^b_j$ into $L$ random blocks (we assign each element of $\vartheta^b_j$ to a block by drawing from the uniform distribution on $\{1,\ldots,L\}$) then apply a blockwise random-walk Metropolis-Hastings (i.e. Metropolis-within-Gibbs) algorithm. Here the proposal density we use for block $l \in \{1,\ldots,L\}$ is $N(0,\sigma_j^2\Sigma_{j-1}^l)$ where $\sigma_j$ is chosen as before, $\Sigma_{j-1}$ is the covariance of the draws from iteration $j-1$, and $\Sigma_j^l$ is the sub-matrix of $\Sigma_j$ corresponding to block $l$.

As the SMC procedure uses a particle approximation to the posterior, in practice compute quantiles for procedure 1 using:
\begin{equation} \label{e:quantile-weighted}
 \Pi( \{\theta : Q_n(\theta) \leq z\} | \mf X_n) = \frac{1}{B} \sum_{b=1}^B w^b_J \ind \{ Q_n(\theta^b) \leq z\}
\end{equation}
and similarly for the profile QLR for procedure 2.

\subsection{Example 1: missing data}\label{ax:smc:ex1}

{\bf SMC algorithm:} We implement the SMC algorithm with $K = 1$ and a $N(0,\sigma_j^2I)$ proposal in the mutation step for all simulations for this example.

{\bf Additional simulation results:} Here we present additional simulation results for the missing data example using (i) a likelihood criterion and curved prior and (ii) a continuously-updated GMM criterion and flat prior. For the ``curved'' prior, we take $\pi(\mu,\eta_1,\eta_2) = \pi(\mu|\eta_1,\eta_2)\pi(\eta_1)\pi(\eta_2)$ with $\pi(\eta_1) = \mr{Beta}(3,8)$, $\pi(\eta_2) = \mr{Beta}(8,1)$, and $\pi(\mu|\eta_1,\eta_2) = U[\eta_1(1-\eta_2),\eta_2+\eta_1(1-\eta_2)]$. Figure \ref{f:prior} plots the marginal curved priors for $\eta_1$ and $\eta_2$.

Results for the likelihood criterion with curved prior are presented in Table \ref{t:ex1-curved}, and are very similar to those presented in Table \ref{t:ex1}, though the coverage of percentile-based CSs is worse here for the partially identified cases ($c = 1,2$). Results for the CU-GMM criterion and flat prior are presented in Table \ref{t:ex1-cugmm}. Results for Procedures 2 and 3 are very similar to the results with a likelihood criterion and show coverage very close to nominal coverage in point point- and partially-identified cases. Here procedure 1 does not over-cover in the point-identified case because the weighting matrix is singular when the model is identified, which forces the draws to concentrate on the region in which $\eta_2 = 1$. This, in turn, means projection is no longer conservative in the point-identified case, though it is still very conservative in the partially-identified cases. Percentile CSs again under-cover badly in the partially-identified case.

\begin{sidewaystable}[p]{\begin{center} \footnotesize
\begin{tabular}{|c|cccccc|cccccc|cccccc|} \hline
	& \multicolumn{6}{c|}{$\eta_2 = 1-\frac{2}{\sqrt n}$} & \multicolumn{6}{c|}{$\eta_2 = 1-\frac{1}{\sqrt n}$}& \multicolumn{6}{c|}{$\eta_2 = 1$ (Point ID)}   \\
	& \multicolumn{2}{c}{0.90} & \multicolumn{2}{c}{0.95} & \multicolumn{2}{c|}{0.99} & \multicolumn{2}{c}{0.90} & \multicolumn{2}{c}{0.95} & \multicolumn{2}{c|}{0.99}  & \multicolumn{2}{c}{0.90} & \multicolumn{2}{c}{0.95} & \multicolumn{2}{c|}{0.99}  \\ \hline
 	& & \multicolumn{16}{c}{$\widehat{\Theta}_{\alpha}$ (Procedure 1)} & \\
100     &  .911  & ---   &  .957  & ---   &  .990  & ---   &  .898  & ---   &  .950  & ---   &  .989  & ---   &  .985  & ---   &  .994  & ---   &  .999  & ---  \\
250     &  .906  & ---   &  .954  & ---   &  .993  & ---   &  .899  & ---   &  .951  & ---   &  .991  & ---   &  .992  & ---   &  .997  & ---   &  1.000  & ---  \\
500     &  .908  & ---   &  .956  & ---   &  .992  & ---   &  .910  & ---   &  .957  & ---   &  .991  & ---   &  .994  & ---   &  .998  & ---   &  1.000  & ---  \\
1000    &  .896  & ---   &  .948  & ---   &  .989  & ---   &  .905  & ---   &  .954  & ---   &  .989  & ---   &  .996  & ---   &  .999  & ---   &  1.000  & ---  \\
	& & \multicolumn{16}{c}{$\widehat{M}_{\alpha}$ (Procedure 2)} & \\
100     &  .890  & [$ .33 ,\! .67 $]  &  .948  & [$ .31 ,\! .69 $]  &  .990  & [$ .29 ,\! .71 $]  &  .911  & [$ .37 ,\! .63 $]  &  .952  & [$ .36 ,\! .64 $]  &  .992  & [$ .33 ,\! .67 $]  &  .912  & [$ .41 ,\! .58 $]  &  .961  & [$ .40 ,\! .60 $]  &  .990  & [$ .37 ,\! .63 $] \\
250     &  .905  & [$ .39 ,\! .61 $]  &  .953  & [$ .38 ,\! .62 $]  &  .991  & [$ .36 ,\! .64 $]  &  .913  & [$ .42 ,\! .58 $]  &  .957  & [$ .41 ,\! .59 $]  &  .992  & [$ .39 ,\! .61 $]  &  .917  & [$ .45 ,\! .55 $]  &  .962  & [$ .44 ,\! .56 $]  &  .993  & [$ .42 ,\! .58 $] \\
500     &  .913  & [$ .42 ,\! .58 $]  &  .957  & [$ .41 ,\! .59 $]  &  .992  & [$ .40 ,\! .60 $]  &  .915  & [$ .44 ,\! .56 $]  &  .955  & [$ .43 ,\! .57 $]  &  .993  & [$ .42 ,\! .58 $]  &  .919  & [$ .46 ,\! .54 $]  &  .959  & [$ .46 ,\! .54 $]  &  .993  & [$ .44 ,\! .56 $] \\
1000    &  .898  & [$ .44 ,\! .56 $]  &  .948  & [$ .44 ,\! .56 $]  &  .988  & [$ .43 ,\! .57 $]  &  .898  & [$ .46 ,\! .54 $]  &  .948  & [$ .45 ,\! .55 $]  &  .990  & [$ .44 ,\! .56 $]  &  .913  & [$ .47 ,\! .53 $]  &  .954  & [$ .47 ,\! .53 $]  &  .992  & [$ .46 ,\! .54 $] \\
	& & \multicolumn{16}{c}{$\widehat{M}^{\chi}_{\alpha}$ (Procedure 3)} & \\
100     &  .922  & [$ .32 ,\! .68 $]  &  .950  & [$ .31 ,\! .69 $]  &  .989  & [$ .28 ,\! .72 $]  &  .912  & [$ .37 ,\! .63 $]  &  .942  & [$ .36 ,\! .64 $]  &  .989  & [$ .33 ,\! .67 $]  &  .909  & [$ .42 ,\! .58 $]  &  .941  & [$ .40 ,\! .59 $]  &  .985  & [$ .38 ,\! .62 $] \\
250     &  .910  & [$ .39 ,\! .61 $]  &  .949  & [$ .38 ,\! .62 $]  &  .989  & [$ .36 ,\! .64 $]  &  .913  & [$ .42 ,\! .58 $]  &  .951  & [$ .41 ,\! .59 $]  &  .990  & [$ .39 ,\! .61 $]  &  .891  & [$ .45 ,\! .55 $]  &  .952  & [$ .44 ,\! .56 $]  &  .992  & [$ .42 ,\! .58 $] \\
500     &  .898  & [$ .42 ,\! .58 $]  &  .952  & [$ .41 ,\! .59 $]  &  .992  & [$ .40 ,\! .60 $]  &  .910  & [$ .44 ,\! .56 $]  &  .947  & [$ .44 ,\! .56 $]  &  .990  & [$ .42 ,\! .58 $]  &  .911  & [$ .46 ,\! .54 $]  &  .948  & [$ .46 ,\! .54 $]  &  .992  & [$ .44 ,\! .56 $] \\
1000    &  .892  & [$ .44 ,\! .56 $]  &  .944  & [$ .44 ,\! .56 $]  &  .989  & [$ .43 ,\! .57 $]  &  .892  & [$ .46 ,\! .54 $]  &  .944  & [$ .45 ,\! .55 $]  &  .986  & [$ .45 ,\! .55 $]  &  .902  & [$ .48 ,\! .52 $]  &  .941  & [$ .47 ,\! .53 $]  &  .987  & [$ .46 ,\! .54 $] \\
	& & \multicolumn{16}{c}{$\widehat{M}^{proj}_{\alpha}$ (Projection)} & \\
100     &  .971  & [$ .30 ,\! .70 $]  &  .988  & [$ .29 ,\! .71 $]  &  .997  & [$ .26 ,\! .74 $]  &  .966  & [$ .35 ,\! .65 $]  &  .987  & [$ .33 ,\! .67 $]  &  .997  & [$ .31 ,\! .69 $]  &  .985  & [$ .38 ,\! .62 $]  &  .994  & [$ .36 ,\! .64 $]  &  .999  & [$ .34 ,\! .66 $] \\
250     &  .971  & [$ .37 ,\! .63 $]  &  .987  & [$ .36 ,\! .64 $]  &  .999  & [$ .34 ,\! .66 $]  &  .969  & [$ .40 ,\! .60 $]  &  .986  & [$ .39 ,\! .61 $]  &  .998  & [$ .37 ,\! .62 $]  &  .992  & [$ .42 ,\! .58 $]  &  .997  & [$ .41 ,\! .59 $]  &  1.000  & [$ .39 ,\! .61 $] \\
500     &  .975  & [$ .41 ,\! .59 $]  &  .988  & [$ .40 ,\! .60 $]  &  .998  & [$ .39 ,\! .61 $]  &  .972  & [$ .43 ,\! .57 $]  &  .989  & [$ .42 ,\! .58 $]  &  .998  & [$ .41 ,\! .59 $]  &  .994  & [$ .44 ,\! .56 $]  &  .998  & [$ .44 ,\! .57 $]  &  1.000  & [$ .42 ,\! .58 $] \\
1000    &  .965  & [$ .44 ,\! .56 $]  &  .983  & [$ .43 ,\! .57 $]  &  .997  & [$ .42 ,\! .58 $]  &  .966  & [$ .45 ,\! .55 $]  &  .985  & [$ .45 ,\! .55 $]  &  .998  & [$ .44 ,\! .56 $]  &  .996  & [$ .45 ,\! .55 $]  &  .999  & [$ .45 ,\! .55 $]  &  1.000  & [$ .44 ,\! .56 $] \\
	& & \multicolumn{16}{c}{$\widehat{M}^{perc}_{\alpha}$ (Percentile)} & \\
100     &  .000  & [$ .37 ,\! .54 $]  &  .037  & [$ .36 ,\! .56 $]  &  .398  & [$ .33 ,\! .59 $]  &  .458  & [$ .40 ,\! .56 $]  &  .646  & [$ .38 ,\! .57 $]  &  .866  & [$ .35 ,\! .60 $]  &  .902  & [$ .42 ,\! .58 $]  &  .951  & [$ .40 ,\! .59 $]  &  .989  & [$ .37 ,\! .62 $] \\
250     &  .000  & [$ .41 ,\! .53 $]  &  .075  & [$ .40 ,\! .54 $]  &  .438  & [$ .38 ,\! .56 $]  &  .480  & [$ .43 ,\! .54 $]  &  .653  & [$ .42 ,\! .55 $]  &  .867  & [$ .40 ,\! .57 $]  &  .909  & [$ .45 ,\! .55 $]  &  .954  & [$ .44 ,\! .56 $]  &  .992  & [$ .42 ,\! .58 $] \\
500     &  .000  & [$ .44 ,\! .52 $]  &  .098  & [$ .43 ,\! .53 $]  &  .468  & [$ .42 ,\! .54 $]  &  .488  & [$ .45 ,\! .53 $]  &  .660  & [$ .44 ,\! .54 $]  &  .878  & [$ .43 ,\! .55 $]  &  .910  & [$ .46 ,\! .54 $]  &  .955  & [$ .46 ,\! .54 $]  &  .991  & [$ .44 ,\! .56 $] \\
1000    &  .000  & [$ .46 ,\! .52 $]  &  .107  & [$ .45 ,\! .52 $]  &  .472  & [$ .44 ,\! .53 $]  &  .483  & [$ .47 ,\! .52 $]  &  .655  & [$ .46 ,\! .53 $]  &  .866  & [$ .45 ,\! .54 $]  &  .901  & [$ .47 ,\! .53 $]  &  .948  & [$ .47 ,\! .53 $]  &  .989  & [$ .46 ,\! .54 $] \\
	& & \multicolumn{16}{c}{Comparison with GMS CSs for $\mu$ via moment inequalities} & \\
100     &  .810  & [$ .34 ,\! .66 $]  &  .909  & [$ .32 ,\! .68 $]  &  .981  & [$ .29 ,\! .71 $]  &  .806  & [$ .39 ,\! .61 $]  &  .896  & [$ .37 ,\! .63 $]  &  .979  & [$ .34 ,\! .66 $]  &  .894  & [$ .42 ,\! .58 $]  &  .943  & [$ .40 ,\! .60 $]  &  .974  & [$ .39 ,\! .61 $] \\
250     &  .800  & [$ .40 ,\! .60 $]  &  .897  & [$ .39 ,\! .61 $]  &  .978  & [$ .36 ,\! .63 $]  &  .798  & [$ .43 ,\! .57 $]  &  .897  & [$ .42 ,\! .58 $]  &  .976  & [$ .40 ,\! .60 $]  &  .903  & [$ .45 ,\! .55 $]  &  .951  & [$ .44 ,\! .56 $]  &  .984  & [$ .43 ,\! .57 $] \\
500     &  .795  & [$ .43 ,\! .57 $]  &  .902  & [$ .42 ,\! .58 $]  &  .980  & [$ .40 ,\! .60 $]  &  .788  & [$ .45 ,\! .55 $]  &  .894  & [$ .44 ,\! .56 $]  &  .977  & [$ .43 ,\! .57 $]  &  .905  & [$ .46 ,\! .54 $]  &  .950  & [$ .46 ,\! .54 $]  &  .987  & [$ .45 ,\! .55 $] \\
1000    &  .785  & [$ .45 ,\! .55 $]  &  .885  & [$ .44 ,\! .56 $]  &  .973  & [$ .43 ,\! .57 $]  &  .785  & [$ .46 ,\! .54 $]  &  .884  & [$ .46 ,\! .54 $]  &  .973  & [$ .45 ,\! .55 $]  &  .895  & [$ .47 ,\! .53 $]  &  .943  & [$ .47 ,\! .53 $]  &  .986  & [$ .46 ,\! .54 $] \\ \hline
			\end{tabular}
			\parbox{14cm}{\caption{\small\label{t:ex1-curved} Missing data example: average coverage probabilities for $\Theta_I$ and $M_I$ and average lower and upper bounds of CSs for $M_I$ across 5000 MC replications. Procedures 1--3, Projection and Percentile are implemented using a likelihood criterion and \textbf{curved prior}.}}
	\end{center}	}	
\end{sidewaystable}

\begin{sidewaystable}[p]{\begin{center} \footnotesize
\begin{tabular}{|c|cccccc|cccccc|cccccc|} \hline
	& \multicolumn{6}{c|}{$\eta_2 = 1-\frac{2}{\sqrt n}$} & \multicolumn{6}{c|}{$\eta_2 = 1-\frac{1}{\sqrt n}$}& \multicolumn{6}{c|}{$\eta_2 = 1$ (Point ID)}   \\
	& \multicolumn{2}{c}{0.90} & \multicolumn{2}{c}{0.95} & \multicolumn{2}{c|}{0.99} & \multicolumn{2}{c}{0.90} & \multicolumn{2}{c}{0.95} & \multicolumn{2}{c|}{0.99}  & \multicolumn{2}{c}{0.90} & \multicolumn{2}{c}{0.95} & \multicolumn{2}{c|}{0.99}  \\ \hline
	& & \multicolumn{16}{c}{$\widehat{\Theta}_{\alpha}$ (Procedure 1)} & \\
100     &  .910  & ---   &  .951  & ---   &  .989  & ---   &  .888  & ---   &  .938  & ---   &  .972  & ---   &  .914  & ---   &  .945  & ---   &  .990  & ---  \\
250     &  .908  & ---   &  .958  & ---   &  .994  & ---   &  .901  & ---   &  .952  & ---   &  .988  & ---   &  .904  & ---   &  .950  & ---   &  .990  & ---  \\
500     &  .912  & ---   &  .957  & ---   &  .992  & ---   &  .923  & ---   &  .965  & ---   &  .993  & ---   &  .904  & ---   &  .954  & ---   &  .990  & ---  \\
1000    &  .909  & ---   &  .955  & ---   &  .989  & ---   &  .912  & ---   &  .963  & ---   &  .995  & ---   &  .911  & ---   &  .953  & ---   &  .989  & ---  \\
	& & \multicolumn{16}{c}{$\widehat{M}_{\alpha}$ (Procedure 2)} & \\
100     &  .921  & [$ .32 ,\! .68 $]  &  .965  & [$ .30 ,\! .70 $]  &  .993  & [$ .27 ,\! .73 $]  &  .913  & [$ .37 ,\! .63 $]  &  .949  & [$ .35 ,\! .65 $]  &  .991  & [$ .32 ,\! .68 $]  &  .914  & [$ .42 ,\! .58 $]  &  .945  & [$ .40 ,\! .60 $]  &  .990  & [$ .37 ,\! .63 $] \\
250     &  .911  & [$ .39 ,\! .61 $]  &  .956  & [$ .38 ,\! .62 $]  &  .993  & [$ .36 ,\! .64 $]  &  .913  & [$ .42 ,\! .58 $]  &  .953  & [$ .41 ,\! .59 $]  &  .993  & [$ .39 ,\! .61 $]  &  .904  & [$ .45 ,\! .55 $]  &  .950  & [$ .44 ,\! .56 $]  &  .990  & [$ .42 ,\! .58 $] \\
500     &  .908  & [$ .42 ,\! .58 $]  &  .954  & [$ .41 ,\! .59 $]  &  .991  & [$ .40 ,\! .60 $]  &  .912  & [$ .44 ,\! .56 $]  &  .953  & [$ .44 ,\! .57 $]  &  .991  & [$ .42 ,\! .58 $]  &  .904  & [$ .46 ,\! .54 $]  &  .954  & [$ .46 ,\! .54 $]  &  .990  & [$ .44 ,\! .56 $] \\
1000    &  .904  & [$ .44 ,\! .56 $]  &  .953  & [$ .44 ,\! .56 $]  &  .989  & [$ .43 ,\! .57 $]  &  .902  & [$ .46 ,\! .54 $]  &  .950  & [$ .45 ,\! .55 $]  &  .990  & [$ .44 ,\! .56 $]  &  .911  & [$ .48 ,\! .52 $]  &  .953  & [$ .47 ,\! .53 $]  &  .989  & [$ .46 ,\! .54 $] \\
	& & \multicolumn{16}{c}{$\widehat{M}^{\chi}_{\alpha}$ (Procedure 3)} & \\
100     &  .921  & [$ .32 ,\! .68 $]  &  .949  & [$ .31 ,\! .69 $]  &  .989  & [$ .28 ,\! .72 $]  &  .913  & [$ .37 ,\! .63 $]  &  .946  & [$ .35 ,\! .64 $]  &  .989  & [$ .32 ,\! .68 $]  &  .914  & [$ .42 ,\! .58 $]  &  .945  & [$ .40 ,\! .60 $]  &  .990  & [$ .37 ,\! .63 $] \\
250     &  .911  & [$ .39 ,\! .61 $]  &  .951  & [$ .38 ,\! .62 $]  &  .992  & [$ .36 ,\! .64 $]  &  .914  & [$ .42 ,\! .58 $]  &  .952  & [$ .41 ,\! .59 $]  &  .993  & [$ .39 ,\! .61 $]  &  .886  & [$ .45 ,\! .55 $]  &  .950  & [$ .44 ,\! .56 $]  &  .991  & [$ .42 ,\! .58 $] \\
500     &  .909  & [$ .42 ,\! .58 $]  &  .950  & [$ .41 ,\! .59 $]  &  .990  & [$ .40 ,\! .60 $]  &  .913  & [$ .44 ,\! .56 $]  &  .948  & [$ .44 ,\! .57 $]  &  .990  & [$ .42 ,\! .58 $]  &  .904  & [$ .46 ,\! .54 $]  &  .947  & [$ .46 ,\! .54 $]  &  .989  & [$ .44 ,\! .56 $] \\
1000    &  .902  & [$ .44 ,\! .56 $]  &  .950  & [$ .44 ,\! .56 $]  &  .989  & [$ .43 ,\! .57 $]  &  .901  & [$ .46 ,\! .54 $]  &  .950  & [$ .45 ,\! .55 $]  &  .988  & [$ .45 ,\! .55 $]  &  .902  & [$ .48 ,\! .52 $]  &  .948  & [$ .47 ,\! .53 $]  &  .988  & [$ .46 ,\! .54 $] \\
	& & \multicolumn{16}{c}{$\widehat{M}^{proj}_{\alpha}$ (Projection)} & \\
100     &  .977  & [$ .29 ,\! .71 $]  &  .990  & [$ .28 ,\! .72 $]  &  .999  & [$ .24 ,\! .76 $]  &  .970  & [$ .34 ,\! .66 $]  &  .988  & [$ .33 ,\! .67 $]  &  .998  & [$ .30 ,\! .70 $]  &  .914  & [$ .42 ,\! .58 $]  &  .945  & [$ .40 ,\! .60 $]  &  .990  & [$ .37 ,\! .63 $] \\
250     &  .972  & [$ .37 ,\! .63 $]  &  .992  & [$ .36 ,\! .64 $]  &  .999  & [$ .34 ,\! .66 $]  &  .976  & [$ .40 ,\! .60 $]  &  .992  & [$ .39 ,\! .61 $]  &  .999  & [$ .37 ,\! .64 $]  &  .904  & [$ .45 ,\! .55 $]  &  .950  & [$ .44 ,\! .56 $]  &  .990  & [$ .42 ,\! .58 $] \\
500     &  .975  & [$ .41 ,\! .59 $]  &  .988  & [$ .40 ,\! .60 $]  &  .998  & [$ .39 ,\! .61 $]  &  .979  & [$ .43 ,\! .57 $]  &  .992  & [$ .42 ,\! .58 $]  &  1.000  & [$ .40 ,\! .60 $]  &  .904  & [$ .46 ,\! .54 $]  &  .954  & [$ .46 ,\! .54 $]  &  .990  & [$ .44 ,\! .56 $] \\
1000    &  .972  & [$ .44 ,\! .57 $]  &  .986  & [$ .43 ,\! .57 $]  &  .998  & [$ .42 ,\! .58 $]  &  .975  & [$ .45 ,\! .55 $]  &  .991  & [$ .44 ,\! .56 $]  &  .999  & [$ .43 ,\! .57 $]  &  .911  & [$ .48 ,\! .52 $]  &  .953  & [$ .47 ,\! .53 $]  &  .989  & [$ .46 ,\! .54 $] \\
	& & \multicolumn{16}{c}{$\widehat{M}^{perc}_{\alpha}$ (Percentile)} & \\
100     &  .399  & [$ .38 ,\! .62 $]  &  .665  & [$ .36 ,\! .64 $]  &  .939  & [$ .31 ,\! .68 $]  &  .642  & [$ .41 ,\! .59 $]  &  .800  & [$ .39 ,\! .61 $]  &  .954  & [$ .35 ,\! .65 $]  &  .912  & [$ .42 ,\! .58 $]  &  .945  & [$ .40 ,\! .60 $]  &  .990  & [$ .37 ,\! .63 $] \\
250     &  .386  & [$ .42 ,\! .58 $]  &  .642  & [$ .41 ,\! .59 $]  &  .914  & [$ .38 ,\! .62 $]  &  .641  & [$ .44 ,\! .56 $]  &  .804  & [$ .43 ,\! .57 $]  &  .954  & [$ .41 ,\! .59 $]  &  .903  & [$ .45 ,\! .55 $]  &  .950  & [$ .44 ,\! .56 $]  &  .990  & [$ .42 ,\! .58 $] \\
500     &  .384  & [$ .44 ,\! .56 $]  &  .639  & [$ .44 ,\! .57 $]  &  .911  & [$ .42 ,\! .58 $]  &  .638  & [$ .46 ,\! .54 $]  &  .803  & [$ .45 ,\! .55 $]  &  .952  & [$ .43 ,\! .57 $]  &  .905  & [$ .46 ,\! .54 $]  &  .953  & [$ .46 ,\! .54 $]  &  .989  & [$ .44 ,\! .56 $] \\
1000    &  .392  & [$ .46 ,\! .54 $]  &  .647  & [$ .45 ,\! .55 $]  &  .908  & [$ .44 ,\! .56 $]  &  .651  & [$ .47 ,\! .53 $]  &  .803  & [$ .46 ,\! .54 $]  &  .950  & [$ .45 ,\! .55 $]  &  .909  & [$ .47 ,\! .53 $]  &  .954  & [$ .47 ,\! .53 $]  &  .989  & [$ .46 ,\! .54 $] \\
	& & \multicolumn{16}{c}{Comparison with GMS CSs for $\mu$ via moment inequalities} & \\
100     &  .806  & [$ .34 ,\! .66 $]  &  .908  & [$ .32 ,\! .68 $]  &  .978  & [$ .29 ,\! .71 $]  &  .803  & [$ .39 ,\! .61 $]  &  .900  & [$ .37 ,\! .63 $]  &  .976  & [$ .34 ,\! .66 $]  &  .899  & [$ .42 ,\! .58 $]  &  .946  & [$ .40 ,\! .60 $]  &  .975  & [$ .39 ,\! .62 $] \\
250     &  .787  & [$ .40 ,\! .60 $]  &  .898  & [$ .39 ,\! .61 $]  &  .979  & [$ .37 ,\! .64 $]  &  .798  & [$ .43 ,\! .57 $]  &  .898  & [$ .42 ,\! .58 $]  &  .978  & [$ .40 ,\! .60 $]  &  .896  & [$ .45 ,\! .55 $]  &  .946  & [$ .44 ,\! .56 $]  &  .981  & [$ .43 ,\! .57 $] \\
500     &  .788  & [$ .43 ,\! .57 $]  &  .896  & [$ .42 ,\! .58 $]  &  .978  & [$ .40 ,\! .60 $]  &  .786  & [$ .45 ,\! .55 $]  &  .898  & [$ .44 ,\! .56 $]  &  .976  & [$ .43 ,\! .57 $]  &  .898  & [$ .46 ,\! .54 $]  &  .950  & [$ .46 ,\! .54 $]  &  .985  & [$ .45 ,\! .56 $] \\
1000    &  .789  & [$ .45 ,\! .55 $]  &  .892  & [$ .44 ,\! .56 $]  &  .979  & [$ .43 ,\! .57 $]  &  .800  & [$ .46 ,\! .54 $]  &  .893  & [$ .46 ,\! .54 $]  &  .977  & [$ .45 ,\! .55 $]  &  .906  & [$ .47 ,\! .53 $]  &  .950  & [$ .47 ,\! .53 $]  &  .986  & [$ .46 ,\! .54 $] \\  \hline
			\end{tabular}
			\parbox{14cm}{\caption{\small\label{t:ex1-cugmm} Missing data example: average coverage probabilities for $\Theta_I$ and $M_I$ and average lower and upper bounds of CSs for $M_I$ across 5000 MC replications. Procedures 1--3, Projection and Percentile are implemented using a \textbf{CU-GMM criterion} and flat prior.}}
	\end{center}	}	
\end{sidewaystable}

\begin{figure}[t]
	\begin{center}
		\includegraphics[trim = 0cm 0cm 0cm 0cm, clip, width = 0.7\textwidth]{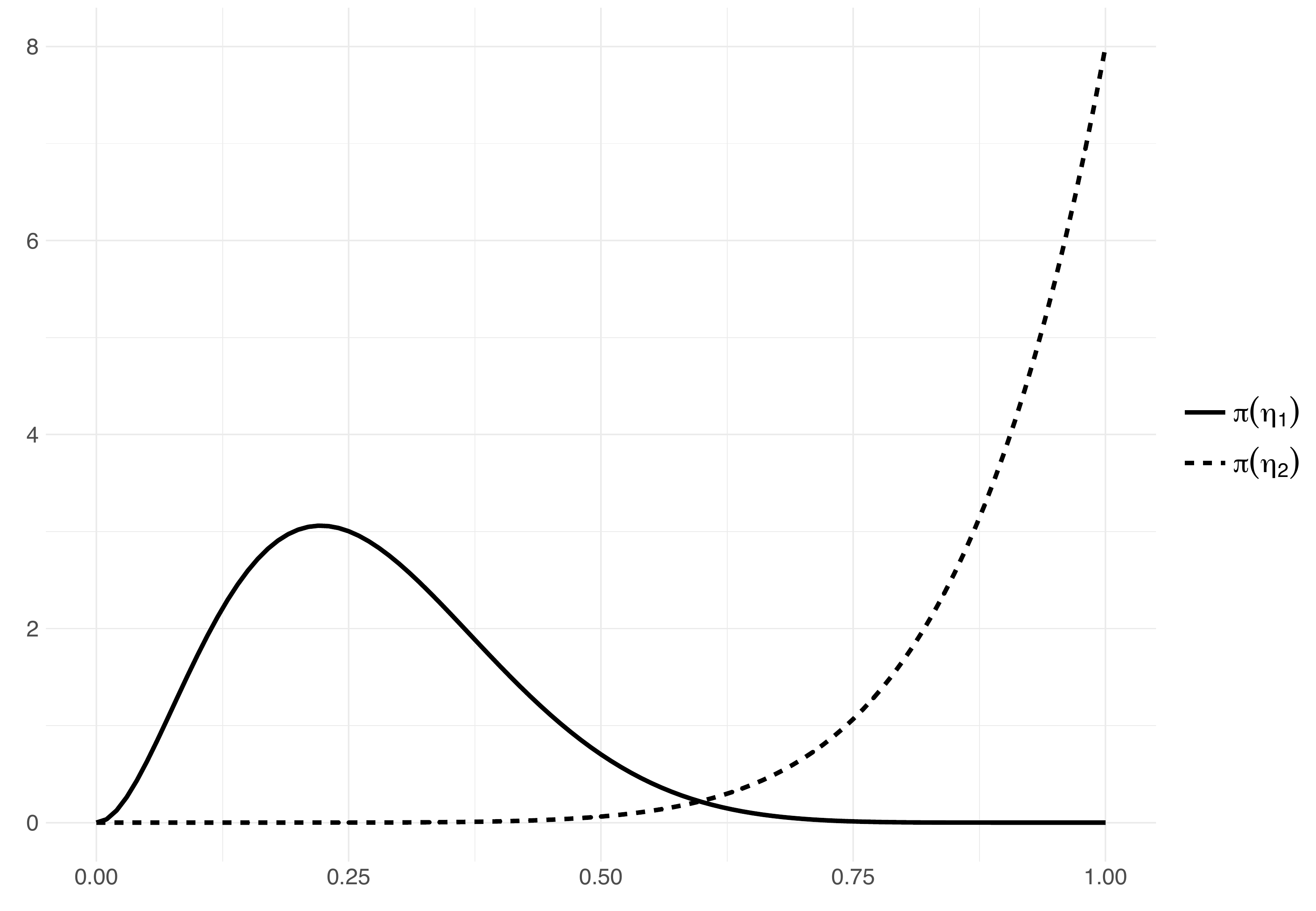}
		\parbox{12cm}{\caption{\small\label{f:prior} Missing data example: Marginal ``curved'' priors for $\eta_1$ (solid line) and $\eta_2$ (dashed line).}}
	\end{center}
\end{figure}

\subsection{Example 2: entry game with correlated shocks}\label{ax:smc:ex2}

{\bf SMC Algorithm:} As there are 6 partially-identified parameters here instead of 2 in the previous example, we initially increased $J$ to reduce the distance between the successive tempered distributions. Like \cite{HS2014}, whose DSGE examples use $(J,K) = (500,1)$, we also found the effect of increasing $K$ similar to the effect of increasing $J$. We therefore settled on $(J,K) = (200,4)$ which was more computationally efficient than using larger $J$. We again use a $N(0,\sigma_j^2I)$ proposal in the mutation step for all simulations for this example.

{\bf Procedure 2:} Unlike the missing data example, where $M(\theta)$ is known in closed form, here the set $M(\theta)$ is no longer known in closed form if $\rho \neq 0$. We therefore calculate $M(\theta^b)$ for $b=1,\ldots,B$ numerically in order to implement procedure 2 for $\mu = \Delta_1$ (in which case $\eta = (\Delta_2,\beta_1,\beta_2,\rho,s)$) and $\mu = \beta_1$ (in which case $\eta = (\Delta_1,\Delta_2,\beta_2,\rho,s)$) . Let $D_{KL}(p_\theta\|p_\vartheta)$ denote the KL distance between $p_\theta$ and $p_\vartheta$ or any $\theta,\vartheta \in \Theta$, which is given by
\[
 D_{KL}(p_\theta\|p_\vartheta) = \sum_{\{i,j\} \in \{0,1\}^2} p_\theta(a_1 = i,a_2 = j) \log \Big( \frac{p_\theta(a_1 = i,a_2 = j)}{p_\vartheta(a_1 = i,a_2 = j)} \Big)
\]
where $p_\theta(a_1 = i,a_2 = j)$ denotes the probability that player $1$ takes action $i$ and player $2$ takes action $j$ when the true structural parameter is $\theta$.
Clearly $\vartheta \in \Delta(\theta)$ if and only if $D_{KL}(p_\theta\|p_\vartheta) = 0$. We compute the endpoints of the interval $M(\theta^b)$ by solving
\begin{align} \label{e-endpoint}
 \min/\max \mu \quad \mbox{such that} \quad \inf_{\eta \in H_\mu} D_{KL}(p_{\theta^b}\|p_{(\mu,\eta)}) = 0
\end{align}
where $H_\mu = [-2,0] \times [-1,2]^2 \times [0,1]^2$ for $\mu = \Delta_1$ and $H_\mu = [-2,0]^2 \times [-1,2] \times [0,1]^2$ for $\mu = \beta_1$. The profiled distance $\inf_{\eta \in H_\mu} D_{KL}(p_\theta\|p_{(\mu,\eta)})$ is independent of the data and is very fast to compute.
Note that we do {\it not} make explicit use of the separable reparameterization in terms of reduced-form choice probabilities when computing $M(\theta^b)$. Moreover, computation of $M(\theta^b)$ can be run in parallel for $b = 1,\ldots,B$ once the draws $\theta^1,\ldots,\theta^B$ have been generated.

To accommodate a small amount of optimization error, in practice we replace the equality in (\ref{e-endpoint}) by a small tolerance $D_{KL}(p_{\theta^b}\|p_{(\mu,\eta)}) < 10^{-7}$. The effect of this slight relaxation is to make our CSs computed via procedure 2 slightly more conservative than if the interval $M(\theta^b)$ were known in closed form.

\subsection{Airline entry game application}\label{ax:smc:app1}

{\bf SMC algorithm:} We implement the adaptive SMC algorithm with $J=200$ iterations, $K=4$ blocked random-walk Metropolis-Hastings steps per iteration with $L =4$ blocks for the full model and $2$ blocks for the fixed-$s$ model

{\bf Procedure 2:} To implement procedure 2 here with any scalar subvector $\mu$ we calculate $M(\theta^b)$ numerically (in parallel), analogously to the entry game simulation example. We again compute the endpoints of $M(\theta^b)$ by solving (\ref{e-endpoint}) for the subvector of interest.

As the log-likelihood is conditional upon regressors, we replace $D_{KL}(p_{\theta^b}\|p_{(\mu,\eta)})$ by the sum of the KL distances between the conditional distributions of outcomes given regressors, namely:
\[
 \sum_{\{MS, MP_{OA}, MP_{LC}\} \in \{0,1\}^3} D_{KL}(p_{\theta^b}(\,\cdot\,|MS, MP_{OA}, MP_{LC})\|p_{(\mu,\eta)}(\,\cdot\,|MS, MP_{OA}, MP_{LC}))
\]
where $p_{\theta}(\,\cdot\,|MS, MP_{OA}, MP_{LC})$ denotes the probabilities of market outcomes conditional upon regressors when the structural parameter is $\theta$.

\begin{figure}[p]
\begin{subfigure}{.45\textwidth}
  \centering
  \includegraphics[width=.8\linewidth]{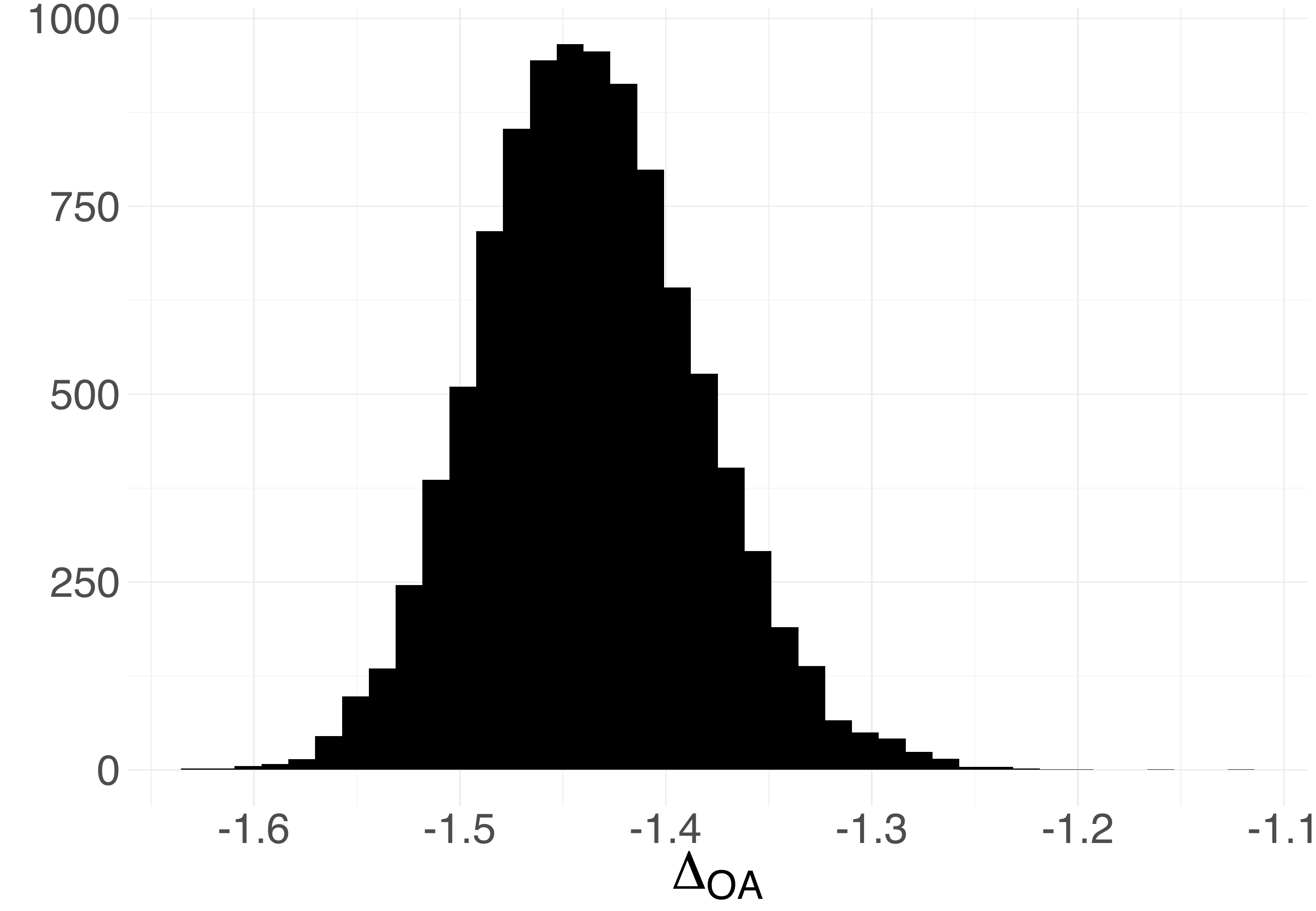}
  \end{subfigure}%
\begin{subfigure}{.45\textwidth}
  \centering
  \includegraphics[width=.8\linewidth]{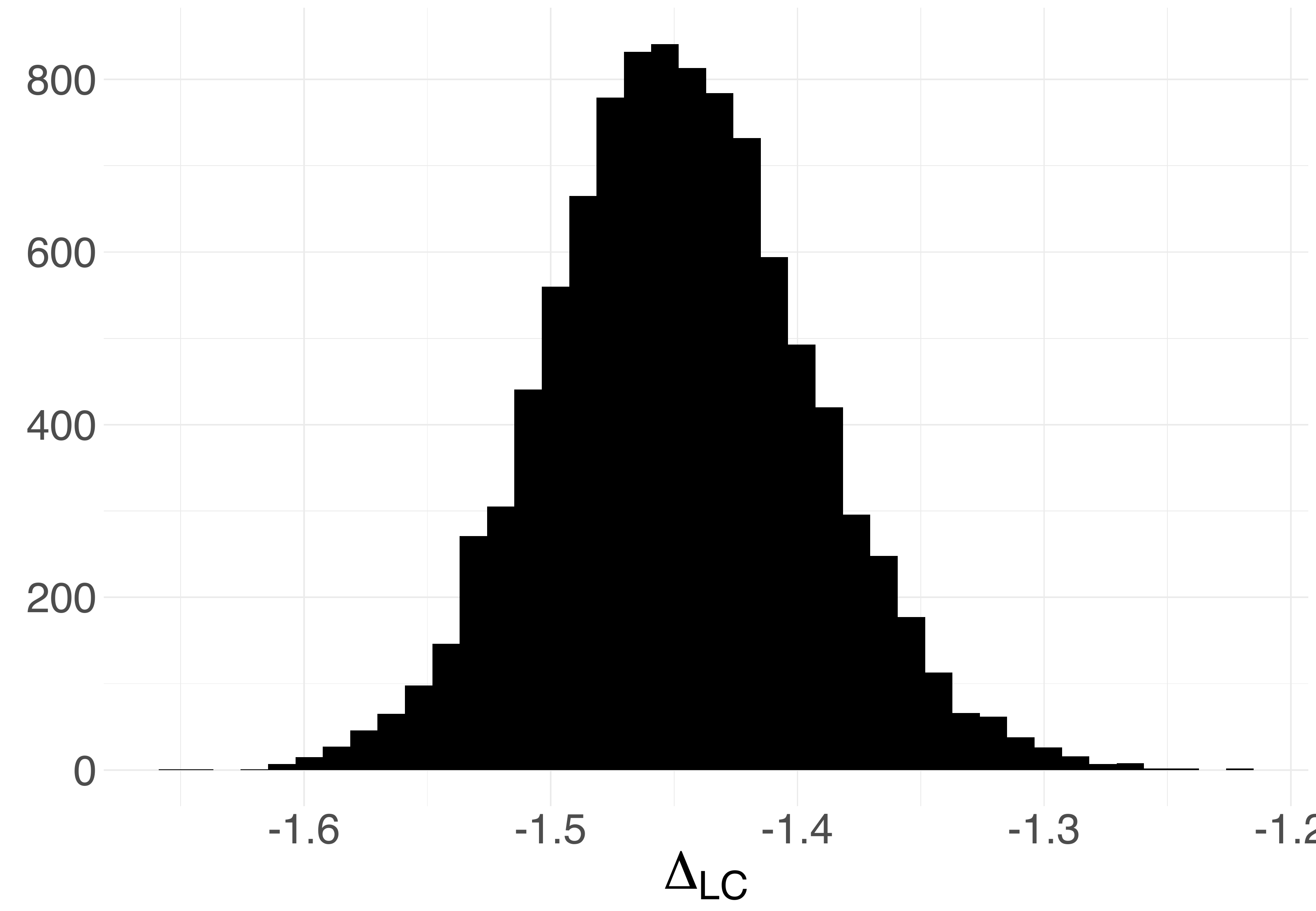}
\end{subfigure}
\begin{subfigure}{.45\textwidth}
  \centering
  \includegraphics[width=.8\linewidth]{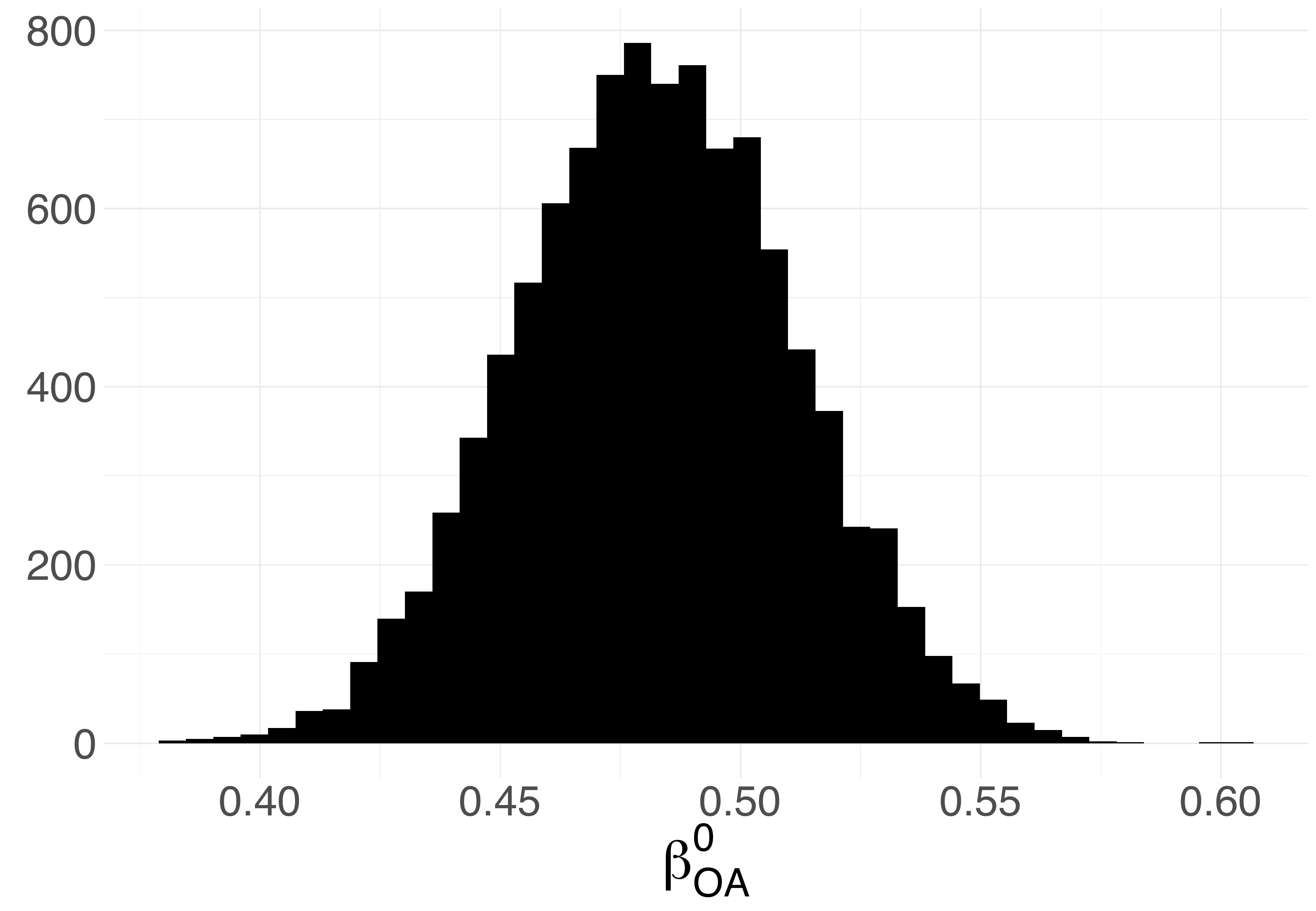}
  \end{subfigure}%
\begin{subfigure}{.45\textwidth}
  \centering
  \includegraphics[width=.8\linewidth]{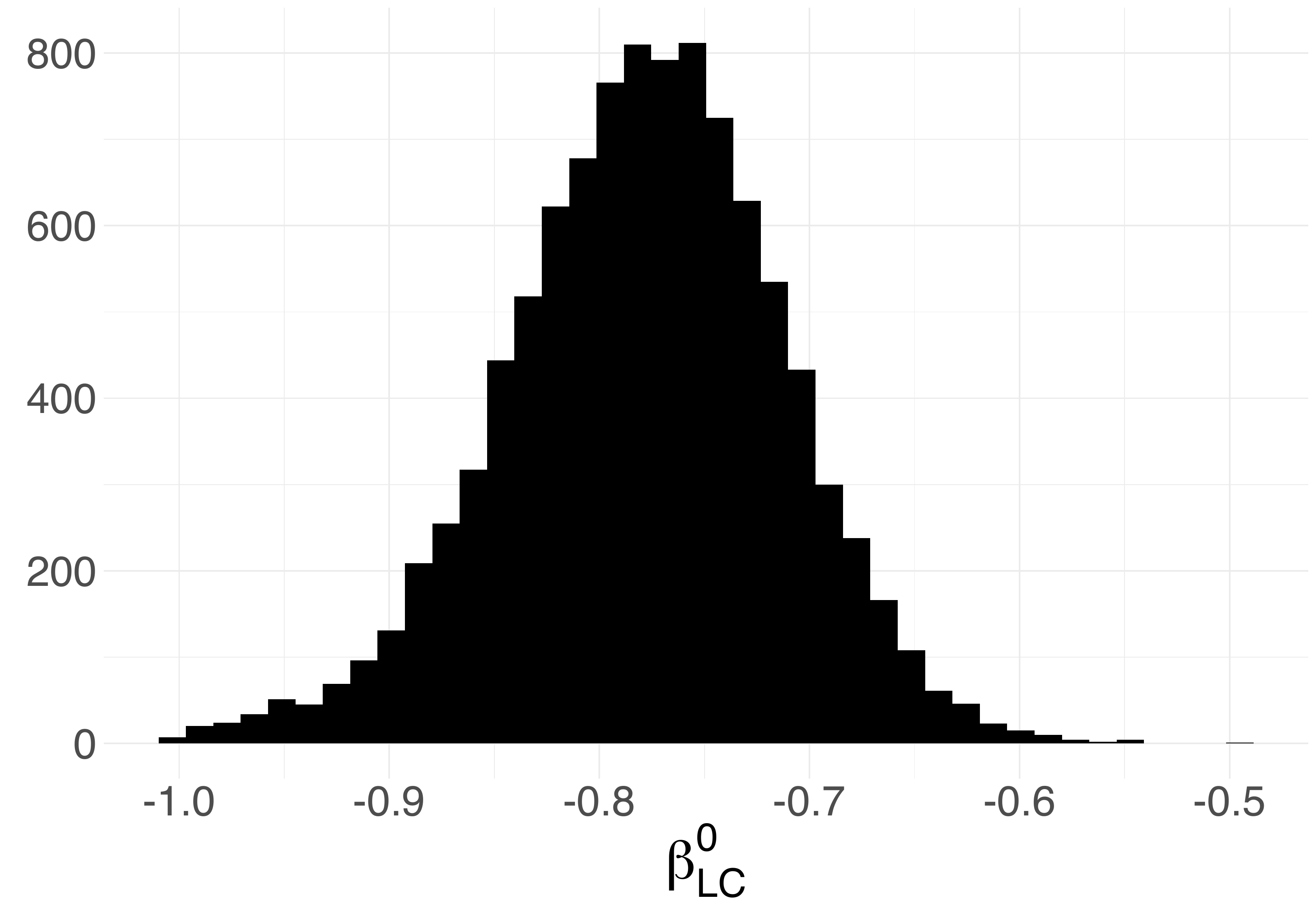}
\end{subfigure}
\begin{subfigure}{.45\textwidth}
  \centering
  \includegraphics[width=.8\linewidth]{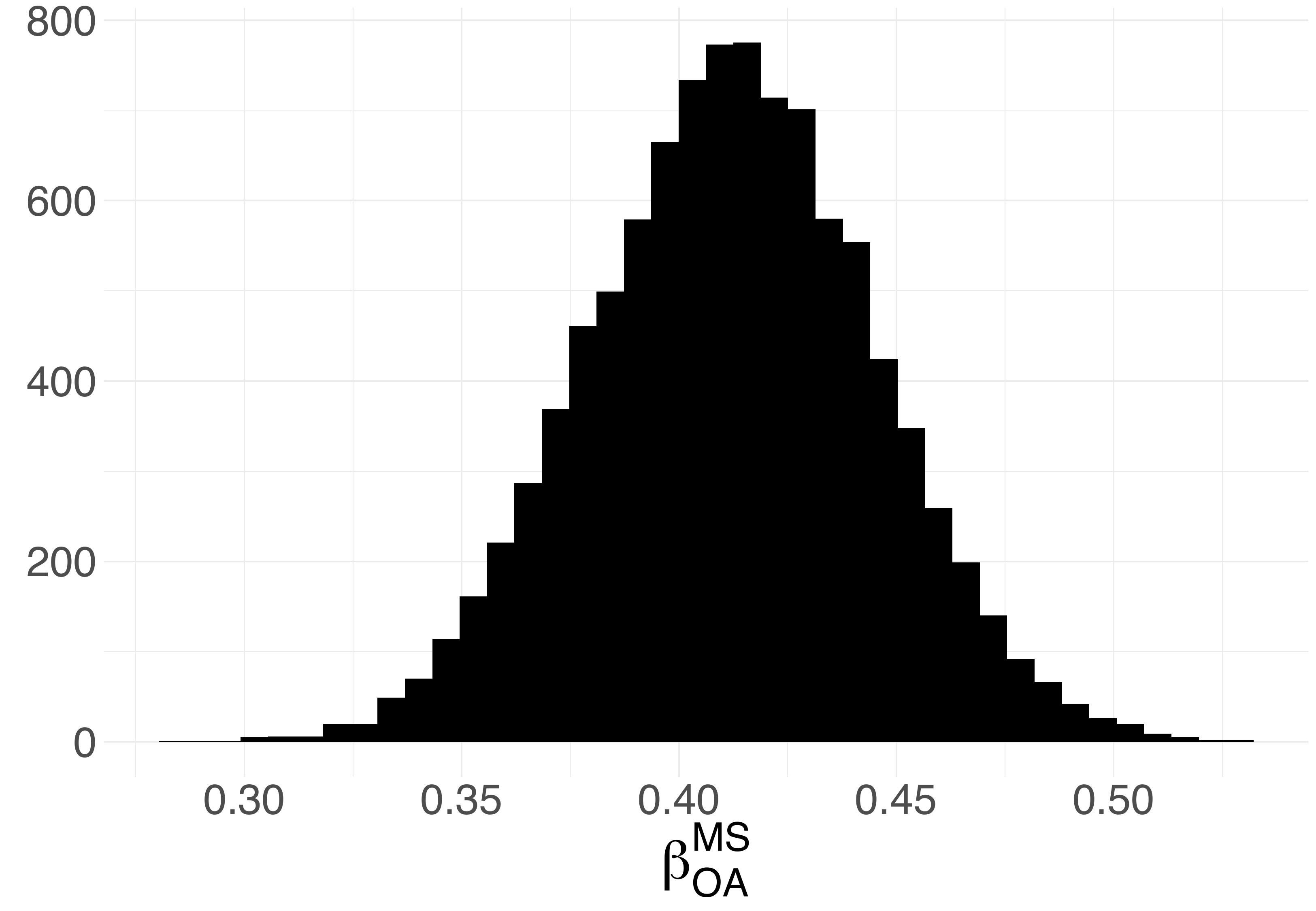}
  \end{subfigure}%
\begin{subfigure}{.45\textwidth}
  \centering
  \includegraphics[width=.8\linewidth]{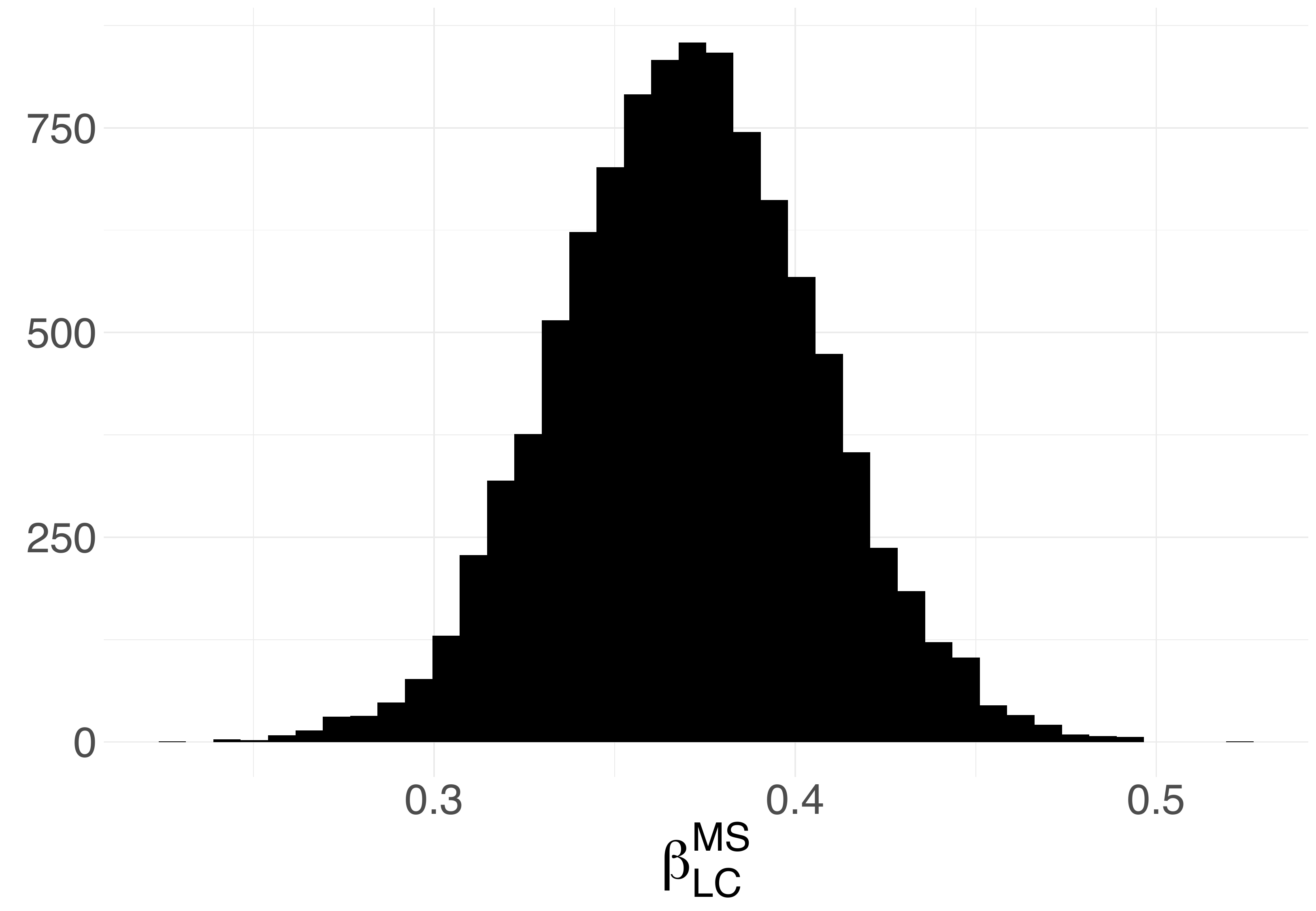}
\end{subfigure}
\begin{subfigure}{.45\textwidth}
  \centering
  \includegraphics[width=.8\linewidth]{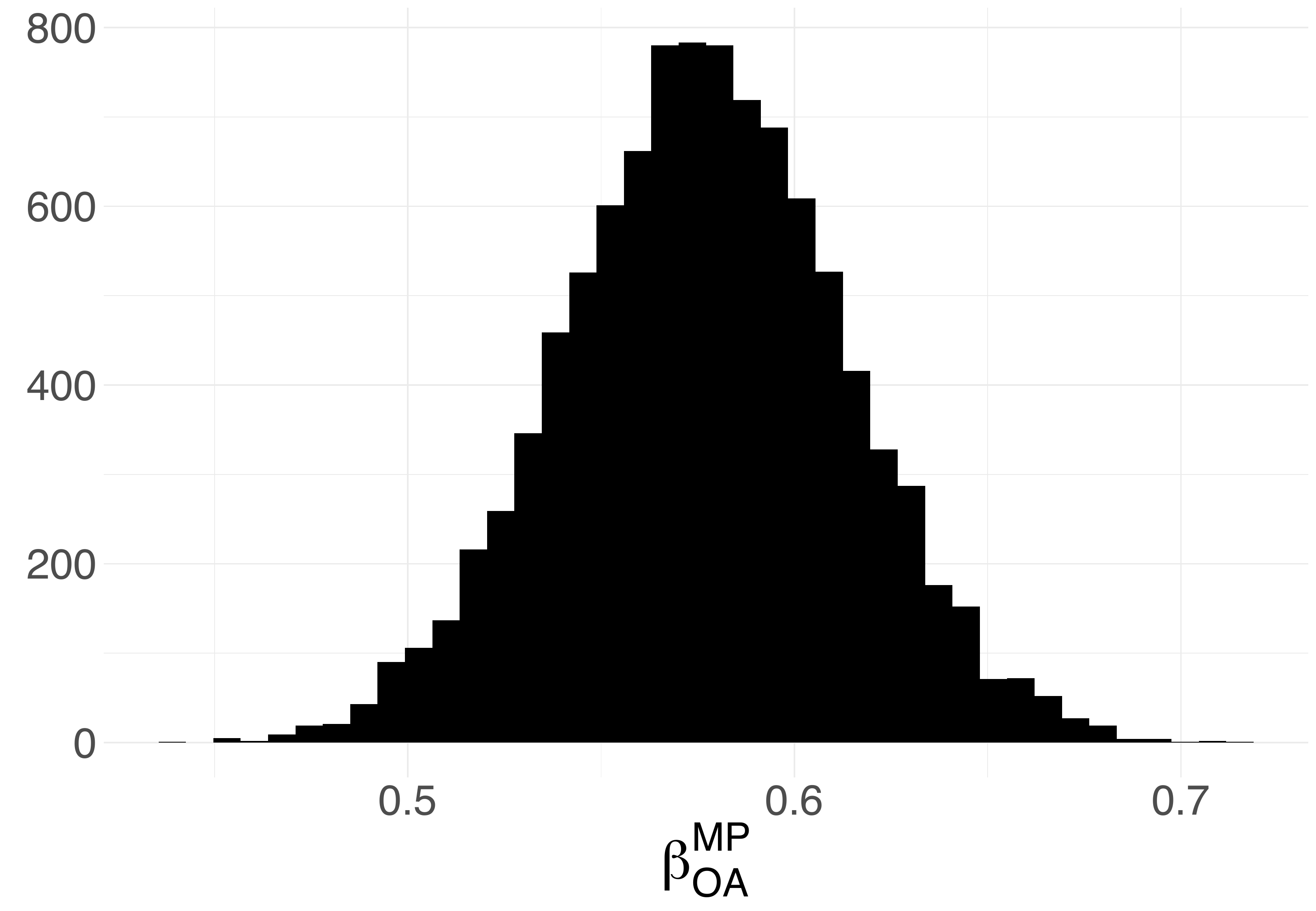}
  \end{subfigure}%
\begin{subfigure}{.45\textwidth}
  \centering
  \includegraphics[width=.8\linewidth]{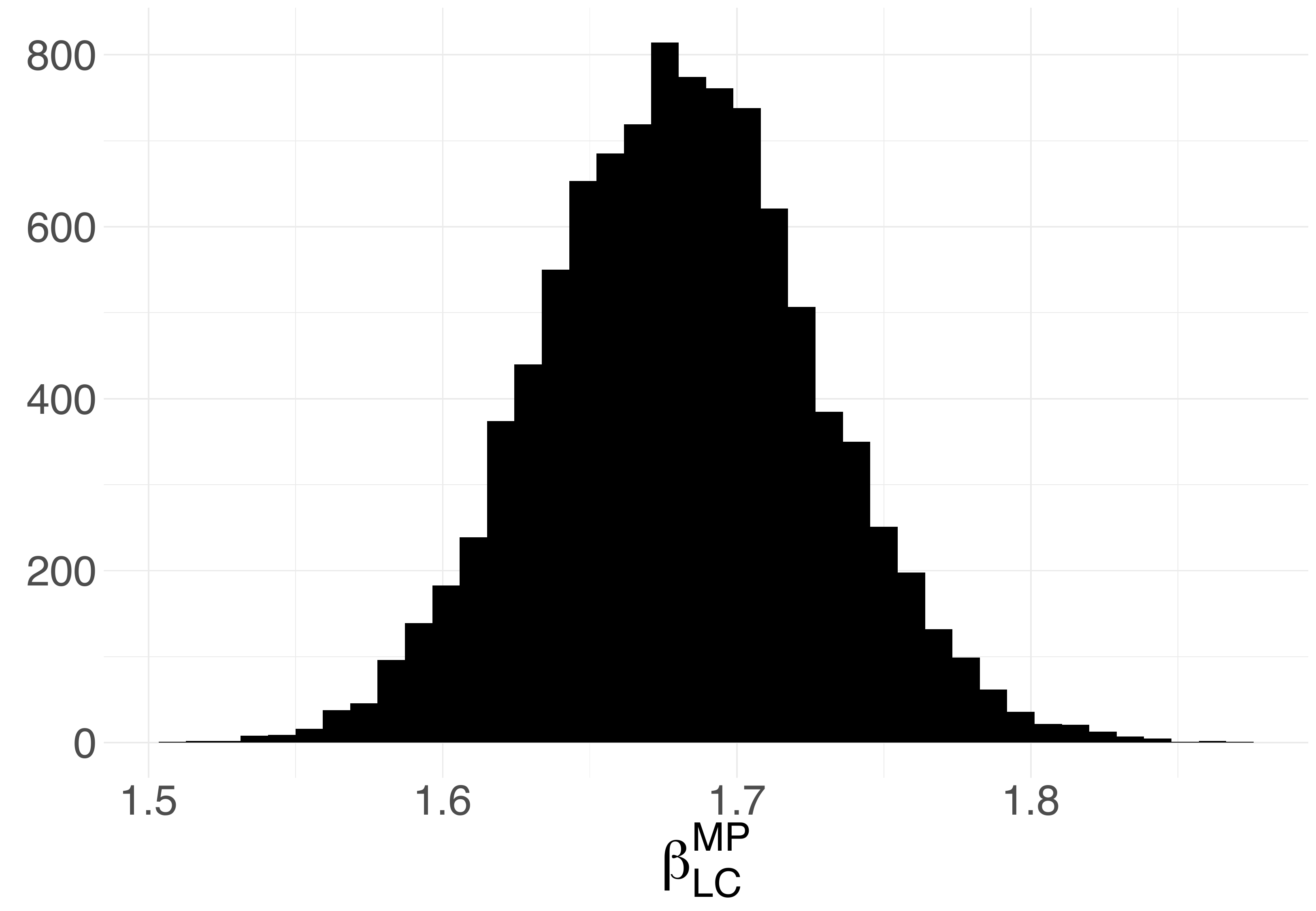}
\end{subfigure}
\begin{subfigure}{.45\textwidth}
  \centering
  \includegraphics[width=.8\linewidth]{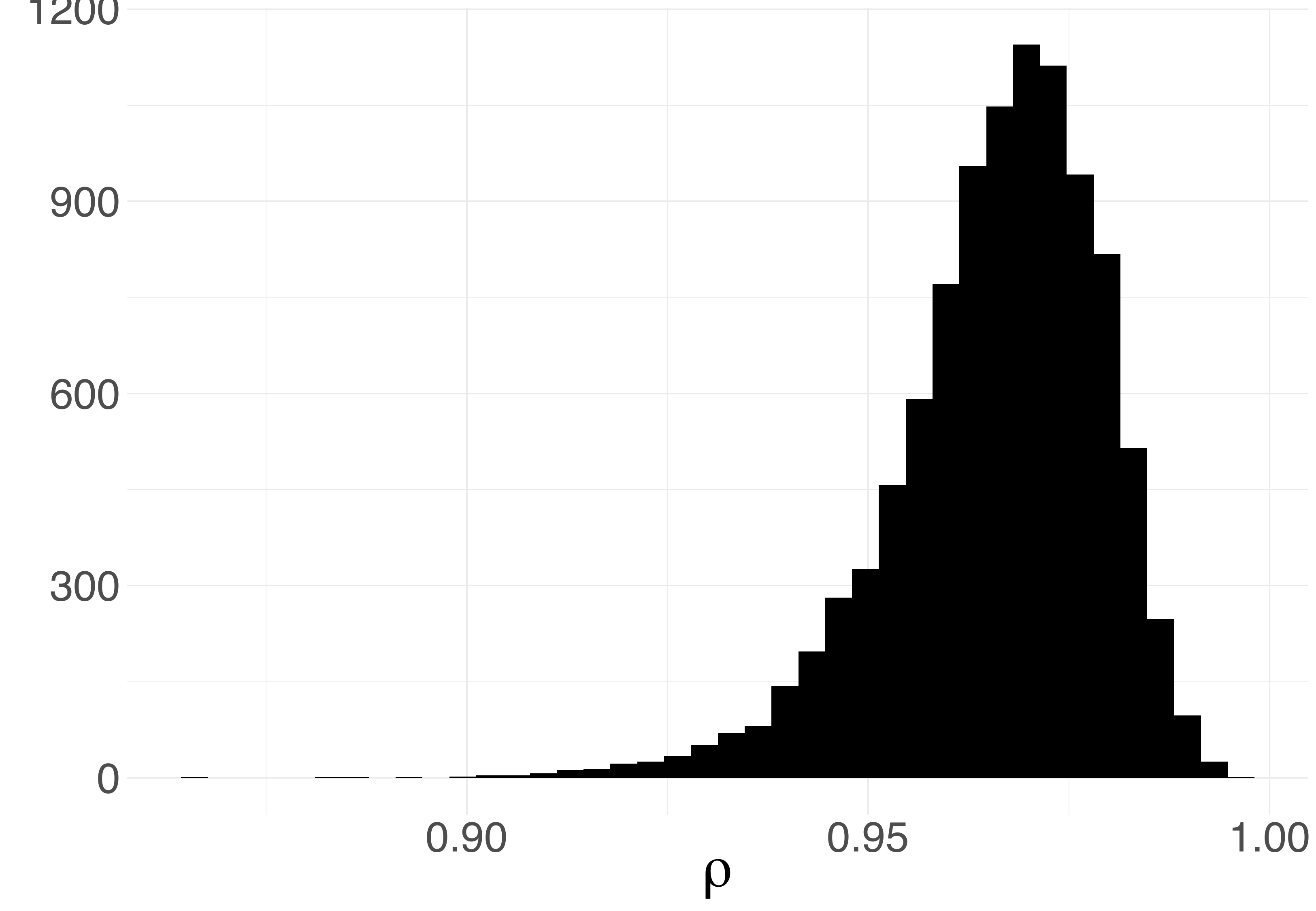}
  \end{subfigure}%
\begin{subfigure}{.45\textwidth}
  \centering
  \includegraphics[width=.8\linewidth]{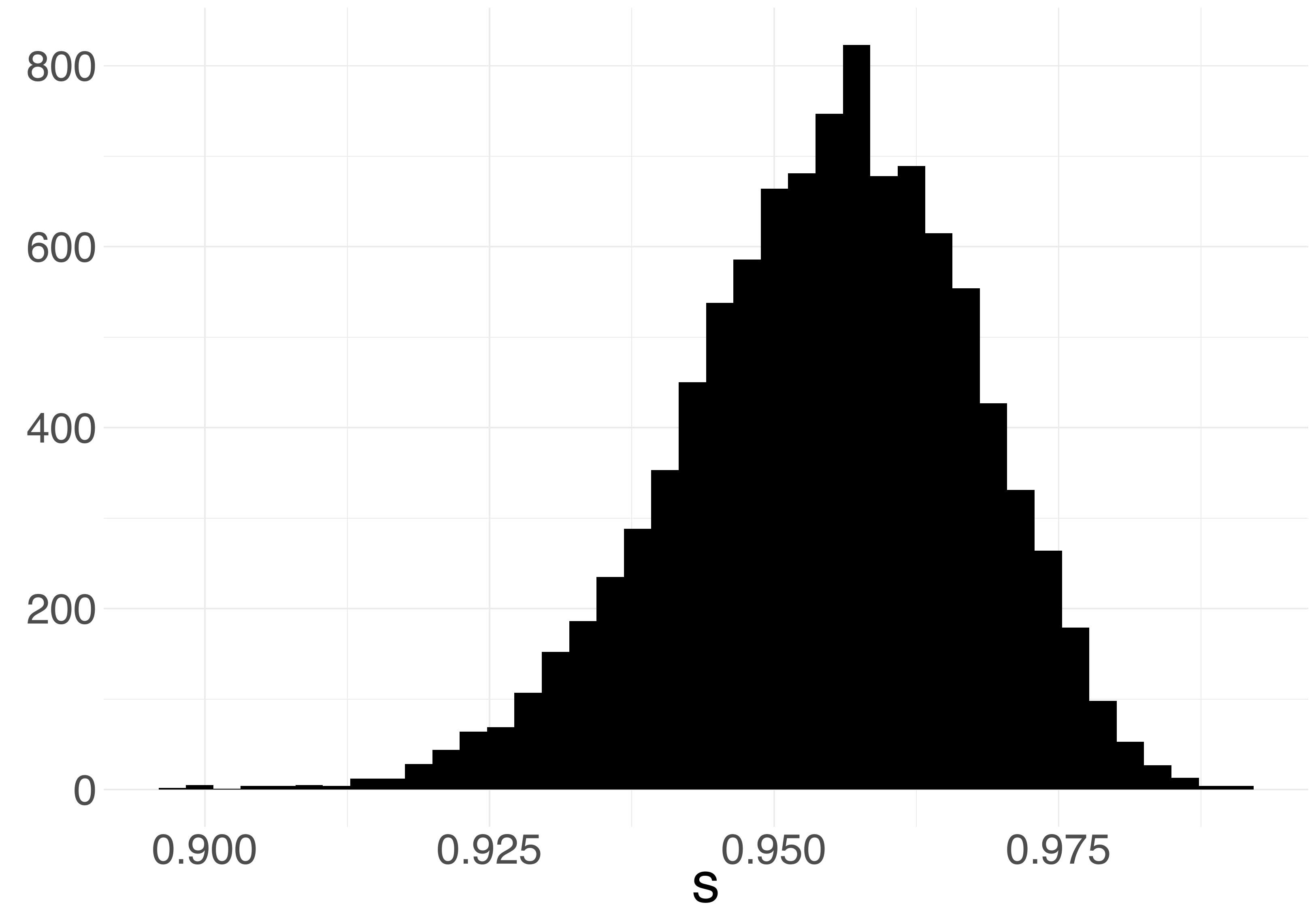}
\end{subfigure}
\centering
\parbox{12cm}{\caption{\small\label{f:game_0} Entry game application: histograms of the SMC draws for $\Delta_{OA}$, $\Delta_{LC}$, $\beta_{OA}^0$, $\beta_{LC}^0$, $\beta_{OA}^{MS}$, $\beta_{LC}^{MS}$, $\beta_{OA}^{MP}$, $\beta_{LC}^{MP}$, $\rho$ and $s$ for the \textbf{fixed-$\boldsymbol s$ model}.} }
\end{figure}

\begin{figure}[p]
\begin{subfigure}{.45\textwidth}
  \centering
  \includegraphics[width=.8\linewidth]{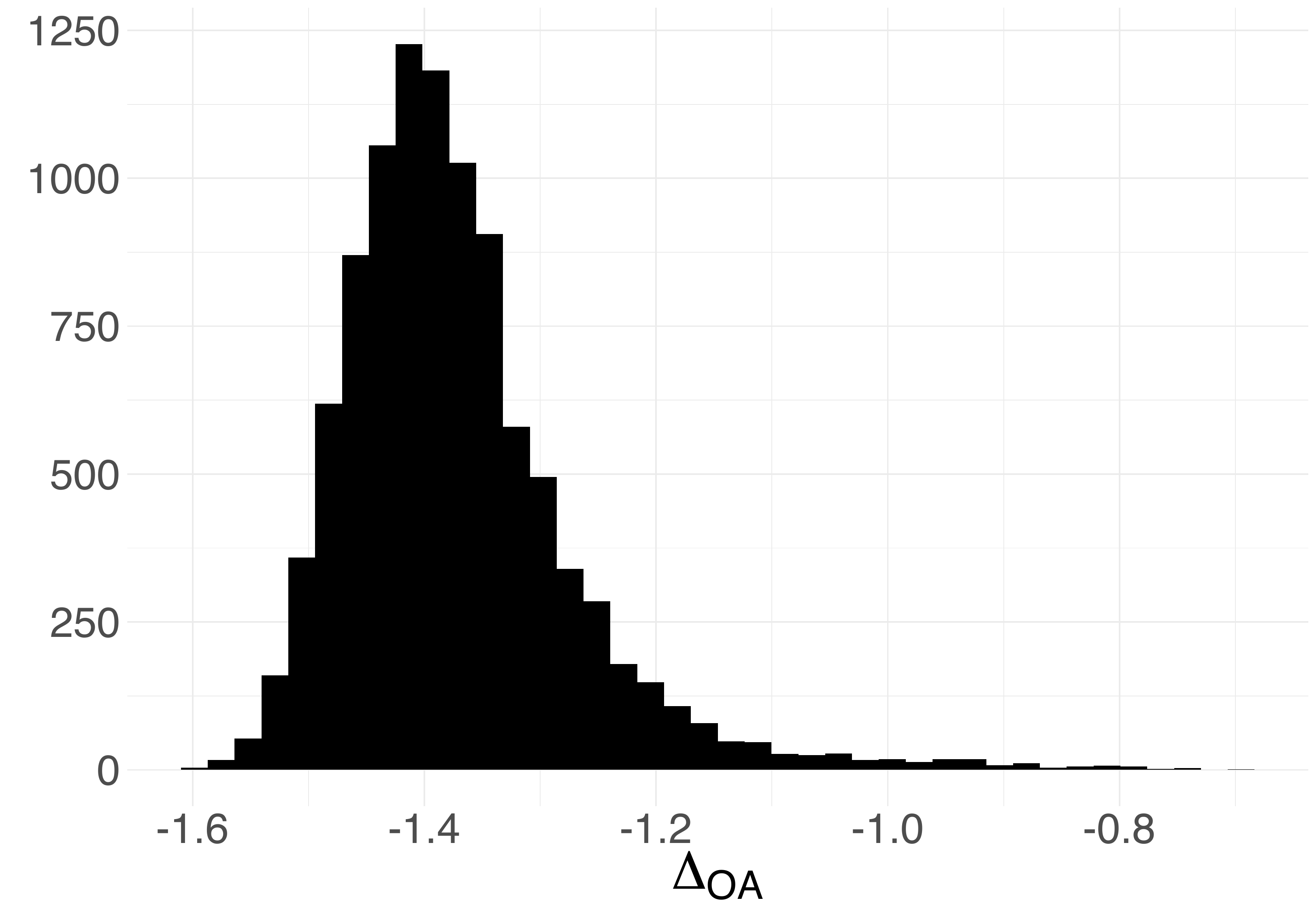}
  \end{subfigure}%
\begin{subfigure}{.45\textwidth}
  \centering
  \includegraphics[width=.8\linewidth]{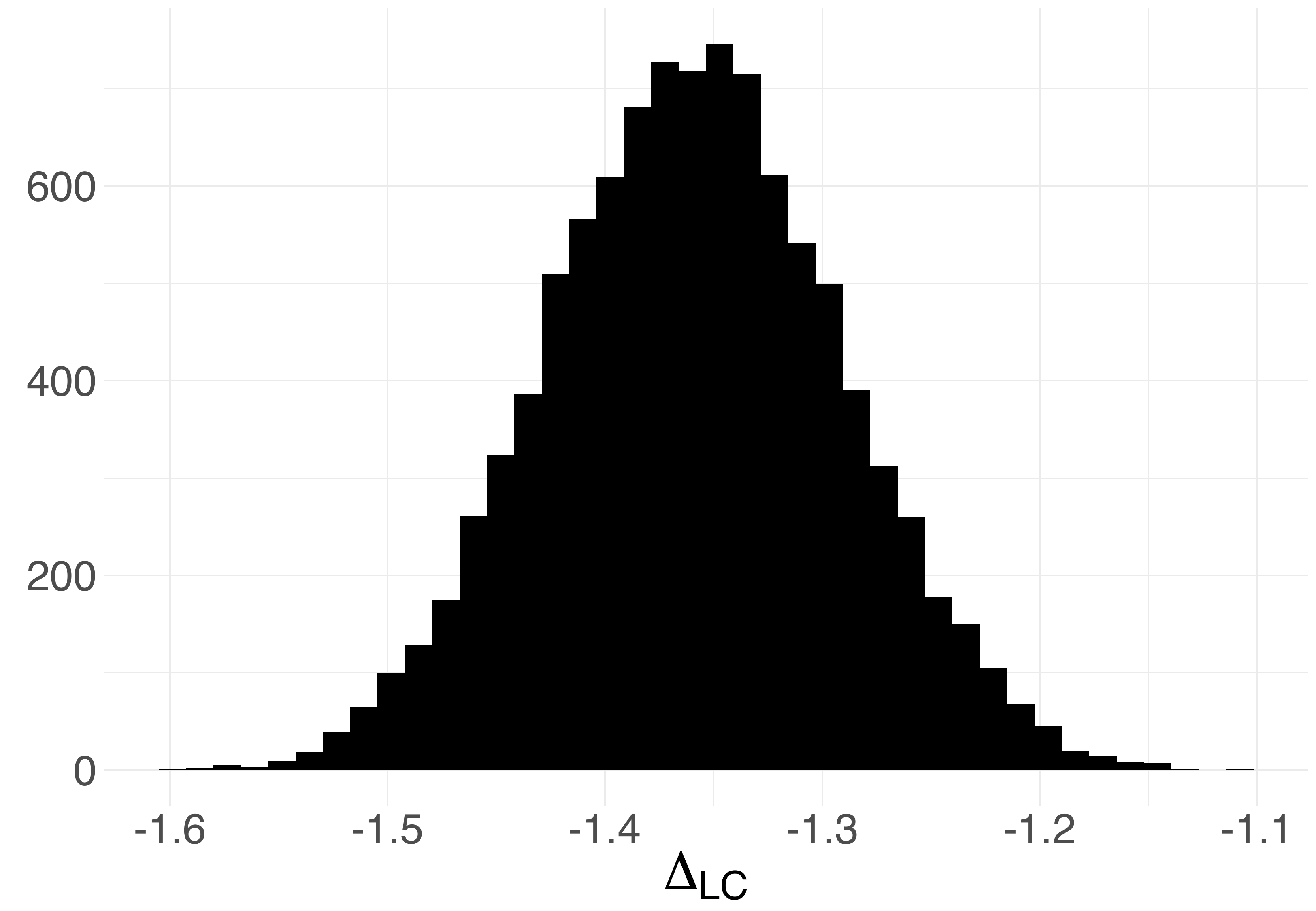}
\end{subfigure}
\begin{subfigure}{.45\textwidth}
  \centering
  \includegraphics[width=.8\linewidth]{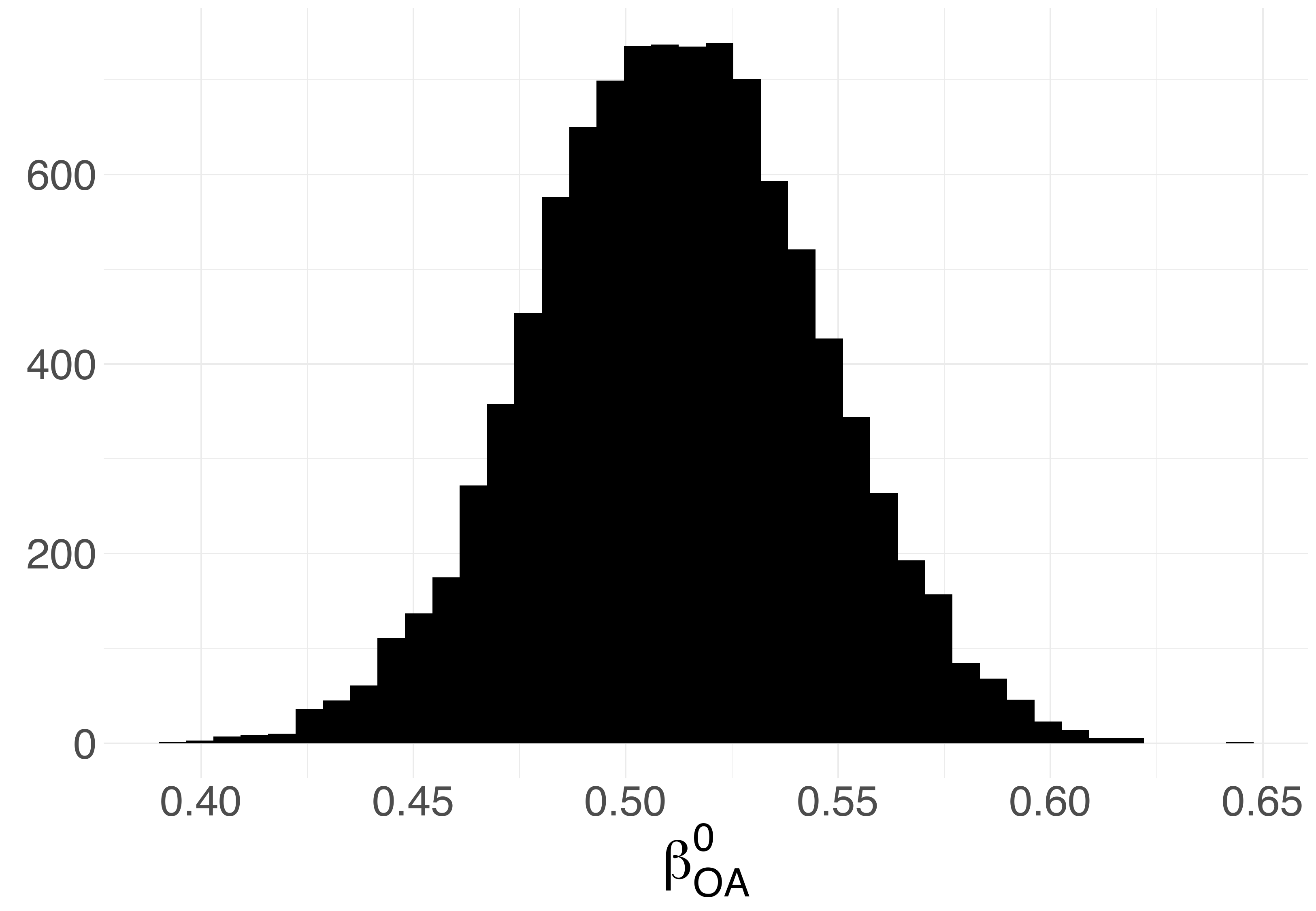}
  \end{subfigure}%
\begin{subfigure}{.45\textwidth}
  \centering
  \includegraphics[width=.8\linewidth]{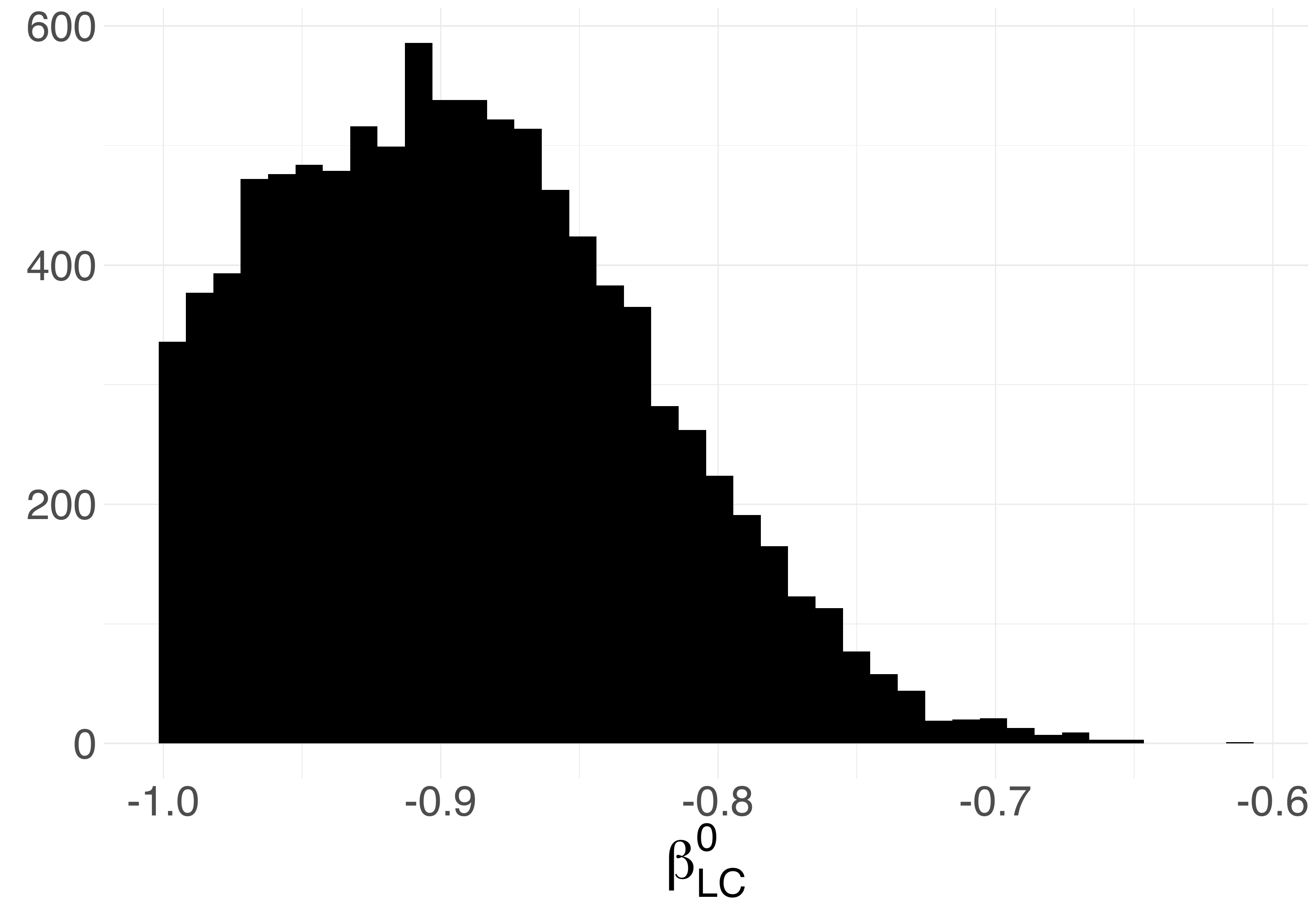}
\end{subfigure}
\begin{subfigure}{.45\textwidth}
  \centering
  \includegraphics[width=.8\linewidth]{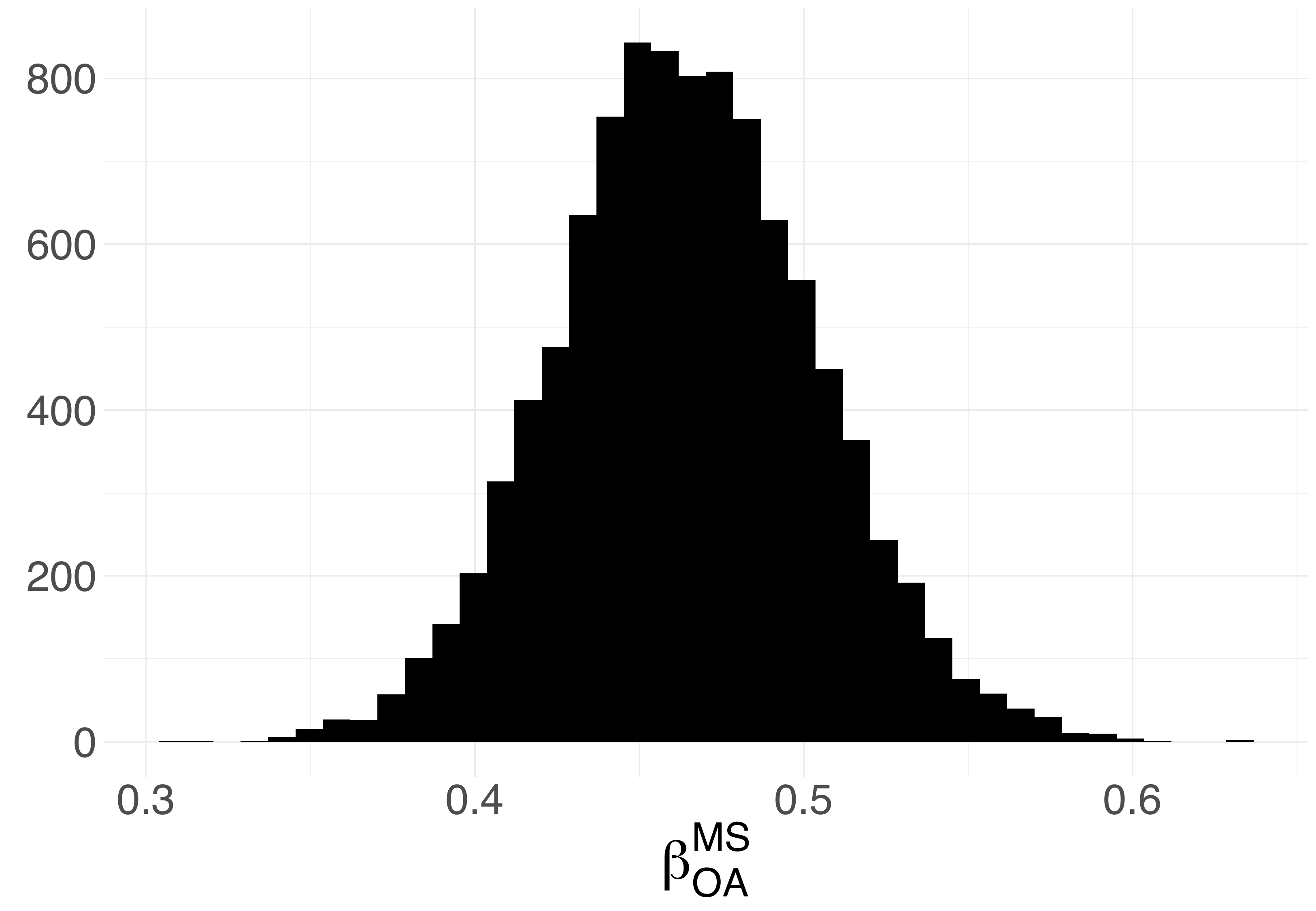}
  \end{subfigure}%
\begin{subfigure}{.45\textwidth}
  \centering
  \includegraphics[width=.8\linewidth]{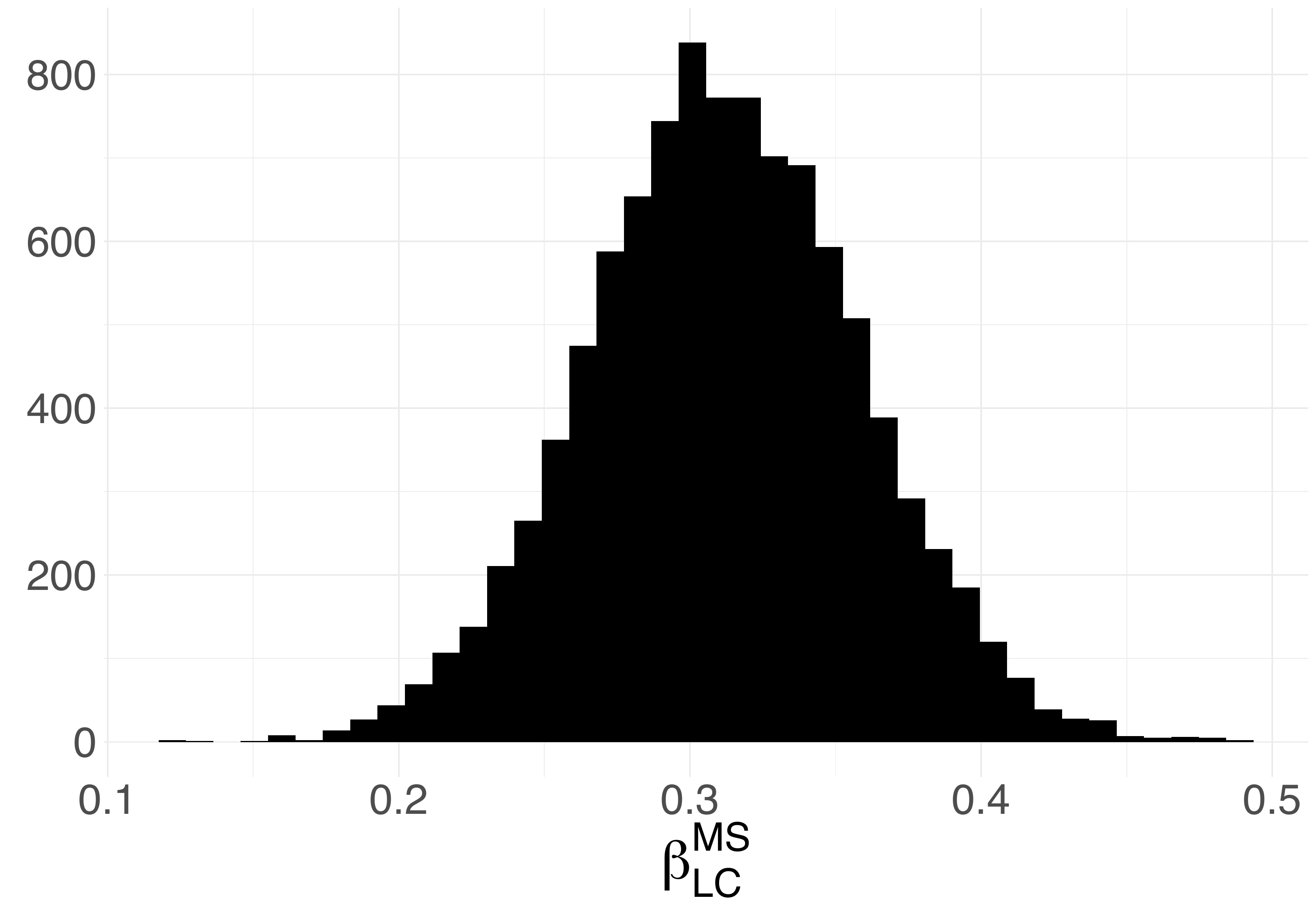}
\end{subfigure}
\begin{subfigure}{.45\textwidth}
  \centering
  \includegraphics[width=.8\linewidth]{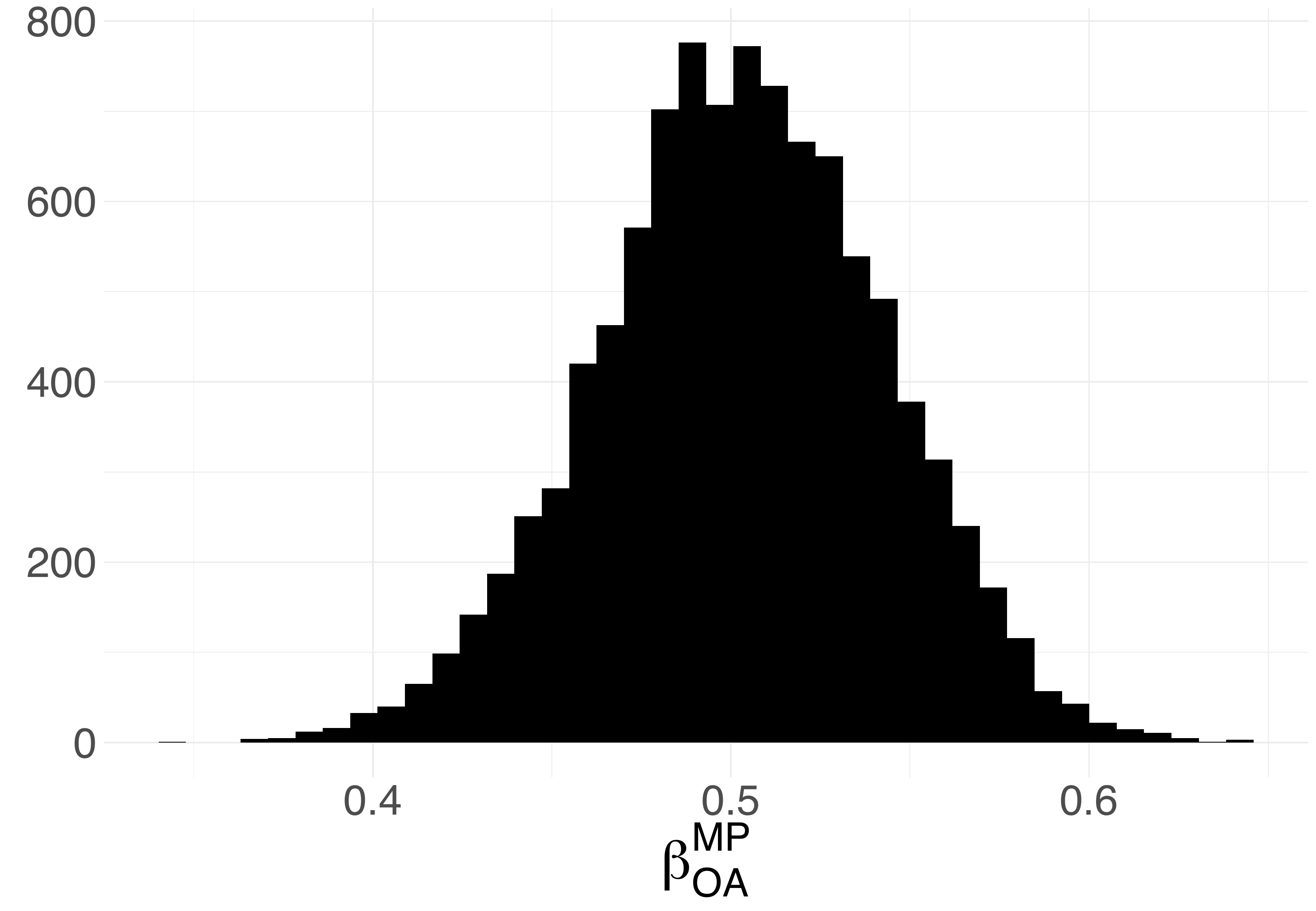}
  \end{subfigure}%
\begin{subfigure}{.45\textwidth}
  \centering
  \includegraphics[width=.8\linewidth]{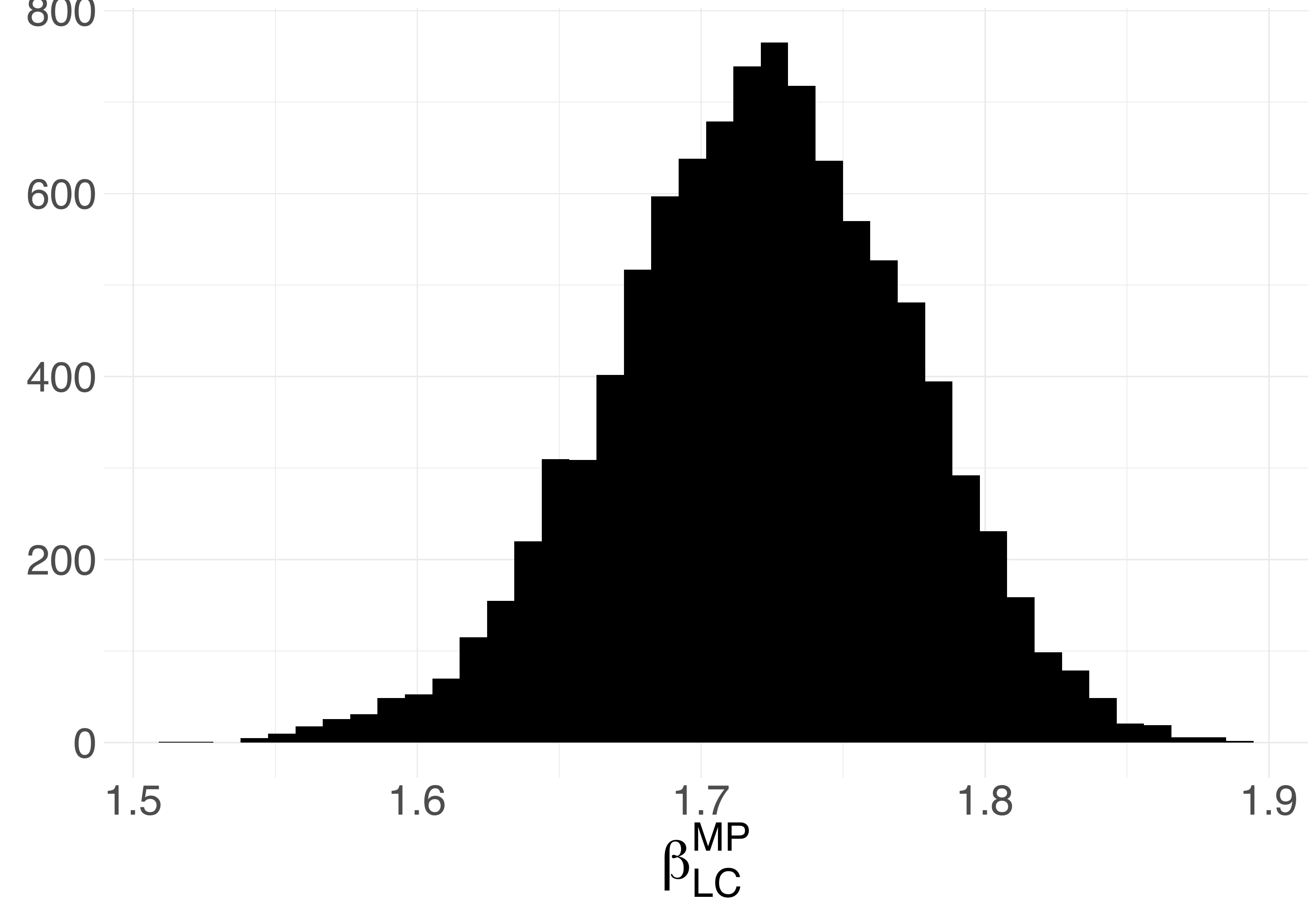}
\end{subfigure}
\begin{subfigure}{.45\textwidth}
  \centering
  \includegraphics[width=.8\linewidth]{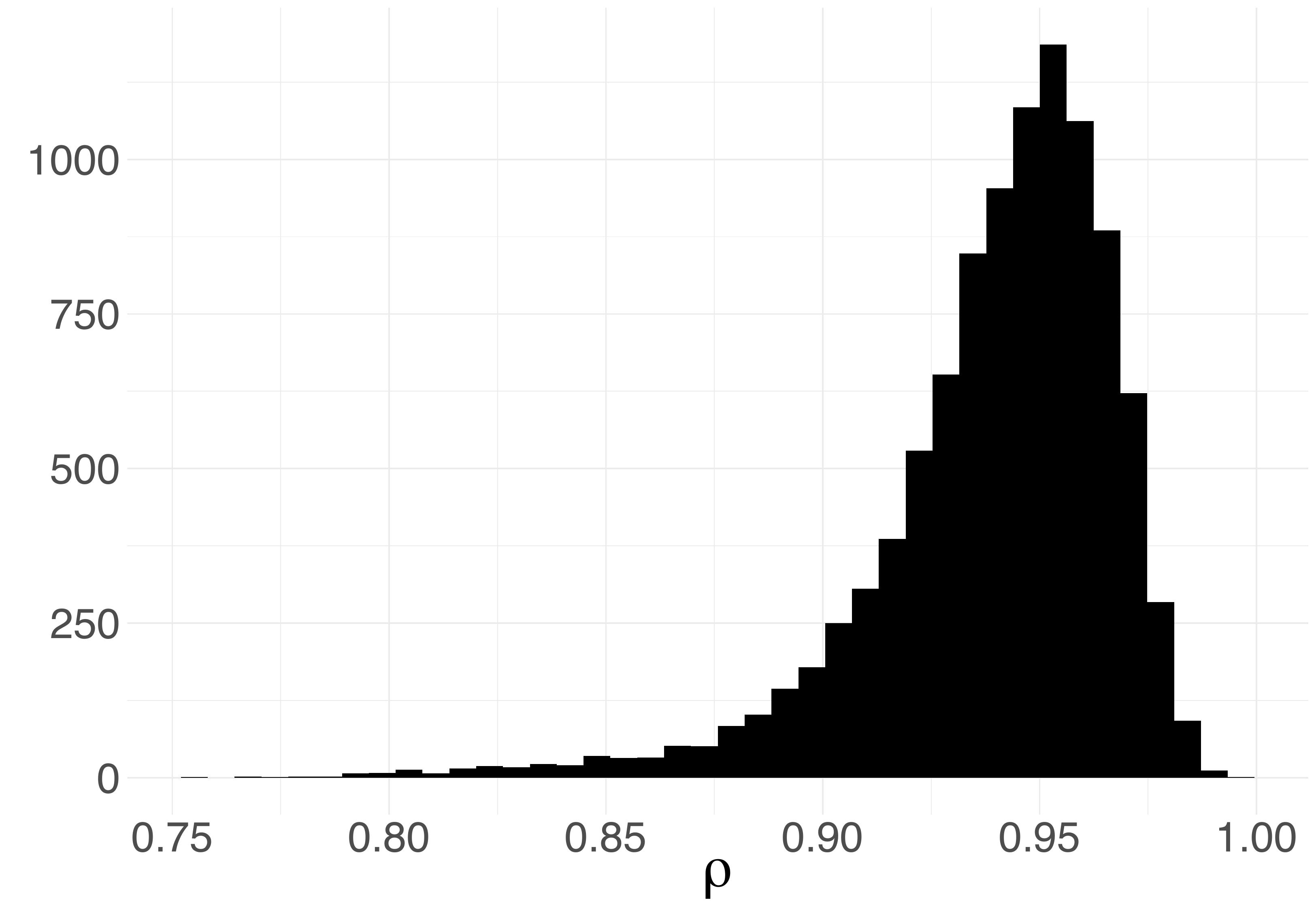}
  \end{subfigure}
\centering
\parbox{12cm}{\caption{\small\label{f:game_1} Entry game application: histograms of the SMC draws for $\Delta_{OA}$, $\Delta_{LC}$, $\beta_{OA}^0$, $\beta_{LC}^0$, $\beta_{OA}^{MS}$, $\beta_{LC}^{MS}$, $\beta_{OA}^{MP}$, $\beta_{LC}^{MP}$, $\rho$ for the \textbf{full model}.} }
\end{figure}

\subsection{Trade flow application}\label{ax:smc:app2}

{\bf Priors:} We use the change of variables $2\mr{arctanh}(\rho)$ and $\log \sigma_m^2$ and assume that the transformed correlation and variance all have full support. We specify and independent $N(0,100^2)$ priors on each of these 46 parameters.

{\bf SMC Algorithm:} Given the high dimensionality of the parameter vector and the lack of a natural restriction of the parameter space for many of the parameters, we use a slight modification of the SMC algorithm described as follows.

We initialize the procedure from drawing from the $N(\hat \theta,-\hat I(\hat \theta)^{-1})$ distribution, where $\hat \theta$ is the MLE and $-I(\hat \theta)^{-1}$ is the inverse negative hessian at the MLE.

There are two more minor modifications which need to be made to correct the particle weights from initializing the algorithm in this manner. First, in the correction step, we replace $v^b_j$ by $v^b_j = e(^{nL_n(\theta^b_{j-1})}\Pi(\theta^b_{j-1})/\mc N_n(\theta^b))^{\phi_j-\phi_{j-1}}$ where $\mc N_n(\theta^b)$ denotes the $N(\hat \theta,-I(\hat \theta)^{-1})$ density evaluated at $\theta^b$. Second, we use the tempered quasi-posterior $\Pi_j(\theta | \mf X_n) \propto (e^{n L_n(\theta)} \Pi(\theta))^{\phi_j} \mc N_n(\theta)^{1-\phi_j}$ in the updating step.

With these modifications, the algorithm is implemented  with $K = 8$ block random-walk Metropolis-Hastings steps per iteration and $L = 6$ blocks.

{\bf Procedure 2:}
To implement procedure 2 here with any scalar subvector $\mu$ we calculate $M(\theta^b)$ numerically. We find the smallest and largest values of $\mu$ for which the average (across regressors) KL divergence, namely
\[
 \frac{1}{n} \sum_{ij}^n D_{KL}( p_{\theta^b}(\,\cdot\,|X_{ij})\| p_{(\mu,\eta)}(\,\cdot\, | X_{ij}) )
\]
is approximately zero. We then set $M(\theta^b) = [\ul \mu(\theta^b),\ol \mu(\theta^b)]$ where $\ul \mu(\theta^b)$ and $\ol \mu (\theta^b)$ denote the smallest and largest such values of $\mu$ for which the average KL divergence is minimized. If $M(\theta^b)$ is not an interval then the interval $[\ul \mu(\theta^b),\ol\mu(\theta^b)]$ will be a superset of $M(\theta^b)$ and the resulting CSs will be slightly conservative.

To compute $D_{KL}( p_{\theta^b}(\,\cdot\,|X_{ij})\| p_{(\mu,\eta)}(\,\cdot\, | X_{ij}) )$, let $d_{ij}$ be a dummy variable denoting exports from $j$ to $i$ and let $m_{ij} = \log M_{ij}$. We may write the model more compactly as:
\begin{align*}
 d_{ij} m_{ij} & = \left\{ \begin{array}{ll}
  X_{ij}'(\beta_m + \delta \beta_z) + (\delta \eta_{ij}^* + u_{ij}) & \mbox{if $d_{ij} = 1$} \\
  0  & \mbox{if $d_{ij} = 0$} \end{array} \right. \\
  d_{ij} & =  \ind\{ X_{ij}'\beta_z + \eta_{ij}^* > 0\}
\end{align*}
where $X_{ij}$ collects the trade friction variables $f_{ij}$ and dummy variables for importer and exporter's continent and $\beta_z$ and $\beta_m$ collect all coefficients in the selection and outcome equations, respectively. Therefore,
\[
 \Pr(d_{ij} = 1 | X_{ij}) = \Phi \Big( \frac{X_{ij}'\beta_z}{\sigma_z(X_{ij})} \Big) \,.
\]
The likelihood is
\begin{align*}
 p_\theta(d_{ij},d_{ij}m_{ij}|X_{ij}) & = \bigg( 1- \Phi \bigg( \frac{X_{ij}'\beta_z}{\sigma_z(X_{ij})} \bigg) \bigg)^{1-d_{ij}}  \Bigg( \Phi \Bigg( \frac{\frac{X_{ij}'\beta_z}{\sigma_z(X_{ij})} + r(X_{ij})\frac{d_{ij} m_{ij}-X_{ij}'(\beta_m + \delta \beta_z)}{\sigma_v(X_{ij})}}{\sqrt{1-r^2(X_{ij})}} \Bigg) \\
 & \quad \quad \times \frac{1}{\sigma_v(X_{ij})} \phi \bigg( \frac{d_{ij}m_{ij}-X_{ij}'(\beta_m + \delta \beta_z)}{\sigma_v(X_{ij})} \bigg) \Bigg)^{d_{ij}}
\end{align*}
where
\begin{align*}
 \sigma_v^2(X_{ij}) & = \sigma_m^2 + 2 \delta \rho \sigma_m \sigma_z(X_{ij}) + \delta^2 \sigma_z^2 (X_{ij}) &
 r(X_{ij}) & = \frac{\rho \sigma_m \sigma_z (X_{ij}) + \delta \sigma_z^2(X_{ij})}{ \sigma_v(X_{ij}) \sigma_z(X_{ij}) } \,.
\end{align*}
The conditional KL divergence between $p_{\theta^b}$ and $p_{(\mu,\eta)}$ is then straightforward to compute numerically (e.g. via Gaussian quadrature). Note also that the sets $M(\theta^b)$ for $b=1,\ldots,B$ and for each subvector of interest can be computed in parallel once the draws $\theta^1,\ldots,\theta^B$ have been generated.

\section{Uniformity}\label{s:uniformity}

Here we present conditions under which our CSs $\wh \Theta_\alpha$ (Procedure 1) and $\wh M_\alpha$ (Procedure 2) are uniformly valid over a class of DGPs  $\mf P$. For each $\p \in \mf P$, let $L(\theta;\p)$ denote the population objective function under $\p$. We assume that for each $\p \in \mf P$, $L(\cdot;\p)$ and $L_n$ are upper semicontinuous and $\sup_{\theta \in \Theta} L(\theta;\p) < \infty$. The identified set is $\Theta_I(\p) = \{ \theta \in \Theta : L(\theta;\p) = \sup_{\vartheta \in \Theta} L(\vartheta;\p)\}$ and the identified set for a subvector $\mu$ is $M_I(\p) = \{\mu : (\mu,\eta) \in \Theta_I(\p)$ for some $\eta \}$.

We now show that, under a natural extension of the assumptions in Section \ref{sec-property}, the CSs $\wh \Theta_{\alpha}$ and $\wh M_{\alpha}$ are uniformly valid i.e.:
\begin{align}
\liminf_{n \to \infty} \inf_{\p \in \mf P} \p (\Theta_I (\p) \subseteq \wh \Theta_{\alpha} ) & \geq \alpha \label{e:unif:full} \\
\liminf_{n \to \infty} \inf_{\p \in \mf P} \p (M_I (\p) \subseteq \wh M_{\alpha} ) & \geq  \alpha  \label{e:unif:subvec}
\end{align}
both hold.
The following Lemmas are straightforward extensions of Lemmas \ref{l:basic} and \ref{l:basic:profile}, but are helpful to organize ideas. Let $(\upsilon_n)_{n \in \mb N}$ be a sequence of random variables. We say that $\upsilon_n = o_\p(1)$ uniformly in $\p$ if $\lim_{n \to \infty} \sup_{\p \in \mf P} \p (|\upsilon_n| > \epsilon) = 0$ for each $\epsilon > 0$, and that $\upsilon_n \leq o_\p(1)$ uniformly in $\p$ if $\lim_{n \to \infty} \sup_{\p \in \mf P} \p(\upsilon_n > \epsilon) = 0$ for each $\epsilon > 0$. Uniform $O_\p(1)$ statements are defined analogously.

\medskip

\begin{lemma}\label{l:basic:unif}
	Let there exist sequences of random variables $(W_n,v_{\alpha,n})_{n \in \mb N}$ such that: \\
	(i) $\sup_{\theta \in \Theta_I(\p)} Q_n(\theta)  - W_n \leq o_\p(1)$ uniformly in $\p$; and \\
	(ii) $\liminf_{n \to \infty} \inf_{\p \in \mf P}  \p ( W_{n} \leq v_{\alpha,n} - \varepsilon_n) \geq \alpha$ for any positive sequence $(\varepsilon_n)_{n \in \mb N}$ with $\varepsilon_n  = o(1)$. \\
	Then: (\ref{e:unif:full}) holds for $\wh \Theta_\alpha = \{ \theta \in \Theta : Q_n(\theta) \leq v_{\alpha,n}\}$.
\end{lemma}

\medskip

\begin{lemma}\label{l:basic:profile:unif}
Let there exist sequences of random variables $(W_n,v_{\alpha,n})_{n \in \mb N}$ such that: \\
	(i) $PQ_n(M_I(\p)) - W_n \leq o_\p(1)$ uniformly in $\p$; and \\
	(ii) $\liminf_{n \to \infty} \inf_{\p \in \mf P}  \p ( W_{n} \leq v_{\alpha,n} - \varepsilon_n) \geq \alpha$ for any positive sequence $(\varepsilon_n)_{n \in \mb N}$ with $\varepsilon_n  = o(1)$. \\
	Then: (\ref{e:unif:subvec}) holds for $\wh M_\alpha =\{ \mu \in M : \inf_{\eta \in H_\mu} Q_n(\mu,\eta) \leq v_{\alpha,n}\}$.
\end{lemma}

The following regularity conditions ensure that $\wh \Theta_\alpha$ and $\wh M_\alpha$ are uniformly valid over $\mf P$.
Let $(\Theta_{osn}(\p))_{n \in \mb N}$ denote  a sequence of local neighborhoods of $\Theta_I(\p)$ such that $\Theta_I(\p) \subseteq \Theta_{osn}(\p)$ for each $n$ and for each $\p \in \mf P$. In what follows we omit the dependence of $\Theta_{osn}(\p)$ on $\p$ to simplify notation.

\begin{assumption}\label{a:rate:unif} (Consistency, posterior contraction) \\
(i)	$L_n(\hat \theta) = \sup_{\theta \in \Theta_{osn}} L_n(\theta) + o_\p(n^{-1})$ uniformly in $\p$.\\
(ii) $\Pi_n(\Theta_{osn}^c |\,\mf X_n) = o_\p(1)$ uniformly in $\p$.
\end{assumption}

We restate our conditions on local quadratic approximation of the criterion allowing for singularity. Recall that a local reduced-form reparameterization is defined on a neighborhood $\Theta_I^N$ of $\Theta_I^{\phantom *}$. We require that $\Theta_{osn}(\p) \subseteq \Theta_I^N(\p)$ for all $\p \in \mf P$, for all $n$ sufficiently large. For nonsingular $\p \in \mf P$ the reparameterization is of the form $\theta \mapsto \gamma(\theta;\p)$ from $\Theta_I^N(\p)$ into $\Gamma(\p)$ where $\gamma(\theta) = 0$ if and only if $\theta \in \Theta_I(\p)$. For singular $\p \in \mf P$  the reparameterization is of the form $\theta \mapsto (\gamma(\theta;\p),\gamma_\bot(\theta;\p))$ from $\Theta_I^N(\p)$ into $\Gamma(\p)\times \Gamma_\bot(\p)$ where $(\gamma(\theta;\p),\gamma_\bot(\theta;\p)) = 0$ if and only if $\theta \in \Theta_I(\p)$. We require the dimension of $\gamma(\cdot;\p)$ to be between $1$ and $\ol d$ for each $\p \in \mf P$, with $\ol d < \infty$ independent of $\p$. Let $B_\delta$ denote a ball of radius $\delta$ centered at the origin (the dimension will be obvious depending on the context) and let $\nu_{d^*}$ denote Gaussian measure on $\mb R^{d^*}$.

To simply notation, in what follows we omit dependence of $d^*$,  $\gamma$, $\gamma_\bot$, $\Gamma$, $\Gamma_\bot$, $k_n$, $\ell_n$, $T$, $\mf T$, $T_{osn}$, $\tau$, $\Theta_I^N$, $\mb V_n$, $\Sigma$, and $f_{n,\bot}$ on $\p$.

\begin{assumption}\label{a:quad:unif} (Local quadratic approximation) \\
	(i) For each $\p \in \mf P$, there exist vectors $\tau \in T$, sequences of random variables $\ell_n$ and $\mb R^{d^*}$-valued random vectors $\hat \gamma_n$, and a sequence of non-negative measurable functions $f_{n,\bot}: \Gamma_\bot \to \mb R$ with $f_{n,\perp}(0)=0$ (we take $\gamma_\bot \equiv 0$ and $f_{n,\bot} \equiv 0$ for nonsingular $\p$), such that as $n \to \infty$:
	\begin{equation} \label{e:quad:unif}
	\sup_{\theta \in \Theta_{osn}} \left| n L_n(\theta) - \left( \ell_n + \frac{1}{2}\|\sqrt n(\hat \gamma_n - \tau)\|^2 - \frac{1}{2} \| \sqrt n (\hat \gamma_n - \tau - \gamma(\theta))\|^2 - f_{n,\perp}(\gamma_\perp(\theta)) \right)  \right| = o_\p(1)
	\end{equation}
	uniformly in $\p$, with $\sup_{\p \in \mf P} \sup_{\theta \in \Theta_{osn}} \|(\gamma(\theta),\gamma_\bot(\theta))\| \to 0$, $\sqrt n \hat \gamma_n = \mf T (\mb V_n+ \sqrt n \tau)$ and $\|\mb V_n \| = O_\p(1)$ (uniformly in $\p$); \\
	(ii) $\{\sqrt n \gamma(\theta) : \theta \in \Theta_{osn}\} \cap B_{k_n} = (T - \sqrt n \tau) \cap B_{k_n}$ where $\inf_{\p \in \mf P} k_n \to \infty$ and $\inf_{\p \in \mf P} \nu_{d^*}(T) > 0$; \\
	(iii) for each singular $\p \in \mf P$: $\{(\gamma(\theta),\gamma_{\bot}(\theta)) : \theta \in \Theta_{osn}\} = \{\gamma(\theta) : \theta \in \Theta_{osn}\} \times \{\gamma_{\bot}(\theta) : \theta \in \Theta_{osn}\}$.
\end{assumption}

Let $\Pi_{\Gamma^*}$ denote the image measure of $\Pi$ under the map $\theta \mapsto \gamma(\theta)$ if $\p$ is nonsingular and $\theta \mapsto ( \gamma(\theta),\gamma_\bot(\theta))$ if $\p $ is singular. We omit dependence of $\delta$, $\Pi_{\Gamma^*}$ and $\pi_{\Gamma^*}$ on $\p$ in what follows.

\begin{assumption}\label{a:prior:unif} (Prior) \\
	(i) $\int_\theta e^{nL_n(\theta)} \,\mr d \Pi(\theta) < \infty$ $\p$-almost surely for each $\p \in \mf P$; \\
	(ii) Each $\Pi_{\Gamma^*}$ has a density $\pi_{\Gamma^*}$  on $B_\delta \cap (\Gamma \times \Gamma_\bot)$ (or  $B_\delta \cap \Gamma $ if $\p$ is nonsingular) for some $\delta > 0$ which are uniformly (in $\mb P$) positive and continuous at the origin.
\end{assumption}

The next lemma is a uniform-in-$\p$ extension of Lemmas \ref{l:post} and \ref{l:post:prime}. Recall that $\mb P_{Z|\mf X_n}$ is the distribution of a $N(0,I_{d^*})$ random vector $Z$ (conditional on data).

\begin{lemma}\label{l:post:unif}
	Let Assumptions \ref{a:rate:unif}, \ref{a:quad:unif} and \ref{a:prior:unif} hold. Then:
 	\[
	\sup_z \biggl(   \Pi_n \big(\{ \theta : Q_n(\theta) \leq z\} \,\big|\,\mf X_n\big)  - \p_{Z|\mf X_n} \left( \|Z\|^2 \leq z | Z \in T - \sqrt n \hat \gamma_n \right) \biggr) \leq o_\p(1)
	\]
	uniformly in $\p$. If no $\p \in \mf P$ is singular, then:
	\[
	\sup_z \biggl|  \Pi_n \big(\{ \theta : Q_n(\theta) \leq z\} \,\big|\,\mf X_n\big)  - \p_{Z|\mf X_n} \left( \|Z\|^2 \leq z | Z \in T - \sqrt n \hat \gamma_n \right) \biggr| = o_\p(1)\,.
	\]
	uniformly in $\p$.
\end{lemma}

As in Section \ref{sec-property}, we let $\xi_{n,\alpha}^{post}$ denote the $\alpha$ quantile of $Q_n(\theta)$ under the posterior distribution $\Pi_n$.

\begin{assumption} \label{a:mcmc:unif}  (MC convergence)\\
	$\xi_{n,\alpha}^{mc} = \xi_{n,\alpha}^{post} + o_\p(1)$ uniformly in $\p$.
\end{assumption}

The following result is a uniform-in-$\p$ extension of Theorems \ref{t:main} and \ref{t:main:prime}. Recall that $F_{T}(z) = \p_Z (\| \mf T Z \|^2 \leq z)$ where $\p_Z$ denotes the distribution of a $N(0,I_{d^*})$ random vector. We say that the distributions $\{ F_T : \p \in \mf P\}$ are equicontinuous at their $\alpha$ quantiles (denoted $\xi_{\alpha,\p}$) if for each $\epsilon > 0$ there is $\delta > 0$ such that $F_T(\xi_{\alpha,\p} - \epsilon ) < \alpha - \delta$ for each $\p \in \mf P$ and $\inf_{\p \in \mf P} F_T(\xi_{\alpha,\p} - \epsilon ) \to \alpha$ as $\epsilon \to 0$. This is trivially true if $T = \mb R^{d^*}$ for each $\p \in \mf P$ and $\sup_{\p \in \mf P} d^* < \infty$.

\begin{theorem}\label{t:main:unif}
Let Assumptions \ref{a:rate:unif}, \ref{a:quad:unif}, \ref{a:prior:unif} and \ref{a:mcmc:unif} hold, and let
\[
 \sup_{\p \in \mf P} \sup_z | \p ( \|\mf T \mb V_n \|^2 \leq z ) - \p_Z (\|\mf T Z\|^2 \leq z) | = o(1)\,.
\]
(i) If $\| \mf T (\mb V_n + \sqrt n \tau) - \sqrt n \tau\|^2 \leq \| \mf T \mb V_n\|^2$ (almost surely) for each $\p \in \mf P$, then: (\ref{e:unif:full}) holds.\\
(ii) If no $\p \in \mf P$ is singular and $T = \mb R^{d^*}$ for each $\p$, then: (\ref{e:unif:full}) holds with equality.
\end{theorem}

To establish (\ref{e:unif:subvec}) we require a uniform version of Assumptions \ref{a:qlr:profile} and \ref{a:mcmc:profile}. In what follows, we omit dependence of $f$ on $\p$ to simplify notation.

\begin{assumption}\label{a:qlr:unif} (Profile QL) \\
 (i) For each $\p \in \mf P$, there exists a measurable function $f : \mb R^{d^*} \to \mb R$ such that:
\begin{align*}
 \sup_{\theta \in \Theta_{osn}} \left| nPL_n(M ( \theta)) - \left( \ell_n  + \frac{1}{2}\|\sqrt n(\hat \gamma_n - \tau)\|^2 - \frac{1}{2} f \left( \sqrt n (\hat \gamma_n - \tau - \gamma(\theta))  \right) \right) \right| = o_\p(1)
\end{align*}
uniformly in $\p$, with $\hat \gamma_n$, $\ell_n$, $\tau$ and $\gamma(\cdot)$ from Assumption \ref{a:quad:unif}; \\
(ii) $f(\mf T(\mb V_n + \sqrt n \tau) - \sqrt n \tau) \leq f(\mb V_n)$ (almost surely) for each $\p \in \mf P$; \\
(iii) $\sup_z (\p_Z (f(Z) \leq z | Z \in v - T) - \p_Z (f(Z) \leq z)) \leq 0$ for all $v \in T$.
\end{assumption}

Note that parts (ii) and (iii) of Assumption \ref{a:qlr:unif} automatically hold with equality if $T = \mb R^{d^*}$. These conditions are not needed in the following result that is a uniform-in-$\p$ extension of Lemma \ref{l:post:profile}.

\begin{lemma}\label{l:post:profile:unif}
	Let Assumptions \ref{a:rate:unif}, \ref{a:quad:unif}, \ref{a:prior:unif} and \ref{a:qlr:unif}(i) hold. Then for any interval $I = I(\p) \subseteq \mb R$ such that $\p_{Z } ( f(Z) \leq z   )$ is uniformly continuous on $I$ (in both $z$ and $\p$):
	\[
	\sup_{z \in I} \left|   \Pi_n \big(\{ \theta : PQ_n ( M (\theta)) \leq z\} \,\big|\,\mf X_n\big)  -  \p_{Z|\mf X_n} (f(Z) \leq z | Z \in \sqrt n \hat \gamma_n - T) \right| = o_\p(1)\,.
	\]
	uniformly in $\p$.
\end{lemma}

Let $\xi_{n,\alpha}^{post,p}$ denote the $\alpha$ quantile of $PQ_n ( M (\theta))$ under the posterior distribution $\Pi_n$.

\begin{assumption}\label{a:mcmc:profile:unif} (MC convergence) \\
	$\xi_{n,\alpha}^{mc,p} = \xi_{n,\alpha}^{post,p} + o_\p(1)$ uniformly in $\p$.
\end{assumption}

The following result is a uniform-in-$\p$ extension of Theorem \ref{t:main:profile}.

\begin{theorem}\label{t:profile:unif}
Let Assumptions \ref{a:rate:unif}, \ref{a:quad:unif}, \ref{a:prior:unif}, \ref{a:qlr:unif} and \ref{a:mcmc:profile:unif} hold, and let
\[
 \sup_{\p \in \mf P} \sup_z | \p (f (\mb V_n) \leq z) - \p_Z (f(Z) \leq z) | = o(1)
\]
where the distributions $\{ \p_Z(f(Z) \leq z) : \p \in \mf P\}$ are equicontinuous at their $\alpha$ quantiles.\\
(i) Then: (\ref{e:unif:subvec}) holds. \\
(ii) If Assumption \ref{a:qlr:unif}(ii)(iii) holds with equality for all $\p \in \mf P$, then:  (\ref{e:unif:subvec}) holds with equality.
\end{theorem}

\subsection{A uniform quadratic expansion for discrete distributions}\label{s:uniformity-qe}

In this subsection we present low-level conditions that show the uniform quadratic expansion assumption is satisfied over a large class of DGPs in discrete models. Let $\mf P$ (possibly depending on $n$) be a class of distributions such that for each $\p_\theta \in \mf P$, $X_1,\ldots,X_n$ are i.i.d. discretely distributed on sample space $\{1,\ldots,k\}$ where $k \geq 2$. Let the $k$-vector $p_\theta$ denote the probabilities $p_\theta(j) = \p_\theta(X_i=j)$ for $j=1,\ldots,k$ and write $p_\theta > 0$ if $p_\theta(j) > 0$ for all $1 \leq j \leq k$. We identify a vector $\p_\theta$ with its probability vector $p_\theta$ and a generic distribution $\p \in \mf P$ with the $k$-vector $p$.

Our uniform quadratic approximation result  encompasses a large variety of drifting sequence asymptotics, allowing $p(j)$ to drift towards $0$ at rate up to (but not including) $n^{-1}$. That is, the first set of results concern any class of distributions $\mf P$ for which
\begin{equation} \label{e:p:unif}
 \sup_{\p \in \mf P} \max_{1 \leq j \leq k} \frac{1}{p(j)} = o(n)\,.
\end{equation}

For any $\p \in \mf P$ with $p > 0$ and any $\theta$, define the (squared) chi-square distance of $\p_\theta$ from $\p$ as
\[
 \chi^2(p_\theta;p) = \sum_{j=1}^k \frac{(p_\theta(j)-p(j))^2}{p(j)} \,.
\]
For each $\p$, let $\Theta_{osn}(\p) = \{ \theta : p_\theta > 0, \chi^2(p_\theta;p) \leq r_n^2 n^{-1}\}$ where $(r_n)_{n \in \mb N}$ is a positive sequence to be defined below. Also let $e_x$ denote a $k$-vector with $1$ in its $x$th entry and $0$ elsewhere, let $\mb J_p = \mr{diag}(p(1)^{-1/2},\ldots,p(k)^{-1/2})$, and let $\sqrt p = (\sqrt{p(1)},\ldots,\sqrt{p(k)})'$.

\begin{lemma} \label{lem:quad:unif:disc}
Let (\ref{e:p:unif}) hold. Then: there exists a positive sequence $(r_n)_{n \in \mb N}$ with $r_n \to \infty$ as $n \to \infty$ such that:
\[
 \sup_{\theta \in \Theta_{osn}(\p)} \left| n L_n (p_\theta) - \left( \ell_n - \frac{1}{2} \| \sqrt n \tilde \gamma_{\theta;p}\|^2 + (\sqrt n \tilde \gamma_{\theta;p})'\tilde{\mb V}_{n;p} \right) \right| = o_\p(1)
\]
uniformly in $\p$, where for each $\p \in \mf P$:
\begin{align*}
 \ell_n = \ell_n(\p) & = n L_n(p) &
 \tilde \gamma_{\theta;p} & = \left[ \begin{array}{c}
 \frac{p_\theta(1)-p(1)}{\sqrt{p(1)}}  \\
 \vdots \\
 \frac{p_\theta(k)-p(k)}{\sqrt{p(k)}}  \end{array} \right] &
 \tilde{\mb V}_{n;p} & = \mb G_n( \mb J_p e_x) \overset{\p}{\rightsquigarrow} N(0,I-\sqrt p \sqrt p')\,.
\end{align*}
\end{lemma}

We are not quite done, as the covariance matrix is a rank $k-1$ orthogonal projection matrix.
Let $v_{1,p},\ldots,v_{k-1,p}$ denote an orthonormal basis for $\{ v \in \mb R^k : v'\sqrt p = 0\}$ and define the matrix $V_p$ by $V_p' = [ v_{1,p} \; \cdots \; v_{k-1,p} \; \sqrt p ]$. Notice that $V_p$ is orthogonal (i.e.  $V_p^{\phantom \prime} V_p' = V_p' V_p^{\phantom \prime} = I$) and
\begin{align} \label{e:quad:unif:disc:1}
 V_p \tilde \gamma_{\theta;p} & = \left[ \begin{array}{c}
 v_{1,p}'\tilde \gamma_{\theta;p} \\
 \vdots \\
 v_{k-1,p}'\tilde \gamma_{\theta;p} \\
 0 \end{array} \right] &
 V_p \mb G_n (\mb J_p e_x)  & = \left[ \begin{array}{c}
 v_{1,p}'\mb G_n (\mb J_p e_x) \\
 \vdots \\
 v_{k-1,p}'\mb G_n (\mb J_p e_x) \\
 0 \end{array} \right] \,.
\end{align}
Let $\gamma(\theta) = \gamma(\theta;\p)$ and $\mb V_n = \mb V_n(\p)$ denote the upper $k-1$ entries of  $V_p \tilde \gamma_{\theta;p}$ and $V_p \mb G_n (\mb J_p e_x)$:
\begin{align} \label{e:quad:unif:disc:2}
 \gamma(\theta) & = \left[ \begin{array}{c}
 v_{1,p}'\tilde \gamma_{\theta;p} \\
 \vdots \\
 v_{k-1,p}'\tilde \gamma_{\theta;p} \end{array} \right] &
 \mb V_n & =  \left[ \begin{array}{c}
 v_{1,p}'\mb G_n (\mb J_p e_x) \\
 \vdots \\
 v_{k-1,p}'\mb G_n (\mb J_p e_x) \end{array} \right] \,.
\end{align}
We say that $\mb V_n \overset{\p}{\rightsquigarrow} N(0,I_{k-1})$ uniformly in $\p$ if $\sup_{\p \in \mf P} d_{\pi}(\mb V_n,N(0,I_{k-1})) \to 0$ where $d_{\pi}$ denotes the distance (in the Prokhorov metric) between the sampling distribution of $\mb V_n$ and the $N(0,I_{k-1})$ distribution.

\begin{proposition}\label{p:quad:unif:disc}
Let (\ref{e:p:unif}) hold and $\Theta_{osn}(\p)$ be as described in Lemma \ref{lem:quad:unif:disc}. Then:
\[
 \sup_{\theta \in \Theta_{osn}(\p)} \left| n L_n (p_\theta) - \left( \ell_n - \frac{1}{2} \| \sqrt n \gamma(\theta)\|^2 + (\sqrt n \gamma(\theta))'\mb V_n \right) \right| = o_\p(1)
\]
uniformly in $\p$, where $\mb V_n \overset{\p}{\rightsquigarrow} N(0,I_{k-1})$ uniformly in $\p$.
\end{proposition}

We may generalize Proposition \ref{p:quad:unif:disc} to allow for the support $k = k(n) \to \infty$ as $n \to \infty$ under a very mild condition on the growth rate of $k$. This result would be very useful in extending our procedures to semi/nonparametric models via discrete approximations of growing dimension. As before, let $\Theta_{osn}(\p) = \{ \theta : p_\theta > 0, \chi^2(p_\theta;p) \leq r_n^2 n^{-1}\}$ where $(r_n)_{n \in \mb N}$ is a positive sequence to be defined below.

\begin{proposition}\label{p:quad:unif:disc:sieve}
Let $\sup_{\p \in \mf P} \max_{1 \leq j \leq k} (1/p(j)) = o(n/\log k)$. Then: there exists a positive sequence $(r_n)_{n \in \mb N}$ with $r_n \to \infty$ as $n \to \infty$ such that:
\[
 \sup_{\theta \in \Theta_{osn}(\p)} \left| n L_n (p_\theta) - \left( \ell_n - \frac{1}{2} \| \sqrt n \gamma(\theta)\|^2 + (\sqrt n \gamma(\theta))'\mb V_n \right) \right| = o_\p(1)
\]
uniformly in $\p$.
\end{proposition}

We now present two lemmas which are helpful in verifying the other conditions of Assumptions \ref{a:quad:unif} and \ref{a:qlr:unif}, respectively. Often, models may be parametrized such that $\{p_\theta : \theta \in \Theta, p_\theta > 0\} = \mr{int}(\Delta^{k-1})$ where $\Delta^{k-1}$ denotes the unit simplex in $\mb R^k$. The following  result shows that the sets $\{\sqrt n \gamma(\theta) : \theta \in \Theta_{osn}(\p)\}$ each cover a ball of radius $\rho_n$ (not depending on $\p$) with $\rho_n \to \infty$.

\begin{lemma}\label{lem:disc:cov}
Let (\ref{e:p:unif}) hold, $\{p_\theta : \theta \in \Theta, p_\theta > 0\} = \mr{int}(\Delta^{k-1})$ and $\Theta_{osn}(\p)$ be as described in Lemma \ref{lem:quad:unif:disc}. Then: for each $\p \in \mf P$, $\{\sqrt n \gamma(\theta) : \theta \in \Theta_{osn}(\p)\}$ covers a ball of radius $\rho_n \to \infty$ (with $\rho_n$ not depending on $\p$) as $n \to \infty$.
\end{lemma}

For the next result, let $\Theta_{osn}'(\p) = \{ \theta : p_\theta > 0, \chi^2(p_\theta;p) \leq (r_n')^2 n^{-1}\}$ where $(r_n')_{n \in \mb N}$ is a positive sequence to be defined below.

\begin{lemma}\label{lem:profile:disc}
Let (\ref{e:p:unif}) hold. Then: there exists a positive sequence $(r_n')_{n \in \mb N}$ with $r_n' \to \infty$ as $n \to \infty$ such that:
\[
 \sup_{\theta \in \Theta_{osn}'(\p)} \sup_{\mu \in M(\theta)} \left| \sup_{\eta \in H_\mu} n L_n(p_{\mu,\eta}) - \sup_{\eta \in H_\mu : (\mu,\eta) \in \Theta_{osn}'(\p)} nL_n(p_{\mu,\eta}) \right| = o_\p(1)
\]
uniformly in $\p$.
\end{lemma}

\section{Verification of main conditions for uniformity in examples}\label{s:uniform-ex}

\subsection{Example 1: uniform validity for missing data}

Here we apply Proposition \ref{p:quad:unif:disc} to establish uniform validity of our procedures. To make the missing data example fit the preceding notation, let $p_\theta = (\tilde \gamma_{11}(\theta),\tilde \gamma_{00}(\theta),1-\tilde\gamma_{00}(\theta)-\tilde\gamma_{11}(\theta))'$ and let $p = (\tilde \gamma_{11},\tilde \gamma_{00},1-\tilde\gamma_{00}-\tilde\gamma_{11})'$ denote the true probabilities under $\p$. The only requirement on $\mf P$ is that (\ref{e:p:unif}) holds. Therefore, Proposition \ref{p:quad:unif:disc} holds uniformly over a set of DGPs under which the probability of missing data can drift to zero at rate up to $n^{-1}$. As $\{p_\theta : \theta \in \Theta, p_\theta > 0\} = \mr{int}(\Delta^2)$, Lemma \ref{lem:disc:cov} implies that $\{\sqrt n \gamma(\theta) : \theta \in \Theta_{osn}(\p)\}$ covers a ball of radius $\rho_n$ (independently of $\p$) with $\rho_n \to \infty$ as $n \to \infty$. This verifies parts (i)--(iv) of Assumption \ref{a:quad:unif}.

By concavity, the infimum in the definition of the profile likelihood $PL_n(M(\theta))$ is attained at either the lower or upper bound of $M_I(\theta) = [\tilde \gamma_{11}(\theta),\tilde \gamma_{11}(\theta) + \tilde \gamma_{00}(\theta)]$. Moreover, at both $\mu = \tilde \gamma_{11}(\theta)$ and $\mu = \tilde \gamma_{11}(\theta) + \tilde \gamma_{00}(\theta)$, the profile likelihood is
\[
\sup_{ \substack{ 0 \leq g_{11} \leq \mu \\
		\mu  \leq g_{11}+g_{00} \leq 1 } } \!\!\!\!\! \Big( n \p_n \ind\{ yd = 1\} \log g_{11} + n \p_n \ind\{ 1-d = 1\} \log g_{00} + n \p_n \ind\{ d-yd = 1\} \log (1-g_{11}-g_{00}) \Big) \,.
\]
The constraint $g_{11} \leq \mu$ will be the binding constraint at the lower bound and the constraint $\mu  \leq g_{11}+g_{00} $ will be the binding constraint at the upper bound (wpa1, uniformly in $\p$). These constraints are equivalent to $a_1'(\gamma - \gamma(\theta)) \leq 0$ and $a_2'(\gamma - \gamma(\theta)) \leq 0$ for some $a_1 = a_1(\p) \in \mb R^2$ and $a_2 = a_2(\p) \in \mb R^2$. It now follows from Proposition \ref{p:quad:unif:disc} and Lemmas \ref{lem:disc:cov} and \ref{lem:profile:disc} that
\[
 \left| n PL_n(M_I) - \min_{j \in \{1,2\}} \sup_{ \gamma : a_j'\gamma  \leq 0 } \left( \ell_n - \frac{1}{2} \| \sqrt n \gamma\|^2 + (\sqrt n \gamma)'\mb V_n \right) \right| = o_\p(1)
\]
and
\[
 \sup_{\theta \in \Theta_{osn}'(\p)} \left| n PL_n(M(\theta)) - \min_{j \in \{1,2\}} \sup_{  \gamma : a_j'(\gamma - \gamma(\theta)) \leq 0 } \left( \ell_n - \frac{1}{2} \| \sqrt n \gamma\|^2 + (\sqrt n \gamma)'\mb V_n \right) \right| = o_\p(1)
\]
uniformly in $\p$. Let $T_j$ denote the closed convex cone in $\mb R^2$ defined by the inequality $a_j'\gamma \leq 0$ for $j = 1,2$. We may write the above as
\begin{align*}
 \left| n PL_n(M_I) - \left( \ell_n + \frac{1}{2} \| \mb V_n\|^2 - \max_{j \in \{1,2\}}  \inf_{t \in T_j} \| \mb V_n - t\|^2 \right) \right| & = o_\p(1) \\
 \sup_{\theta \in \Theta_{osn}'(\p)} \left| n PL_n(M(\theta)) - \left( \ell_n + \frac{1}{2} \| \mb V_n\|^2 - \max_{j \in \{1,2\}}  \inf_{t \in T_j} \| (\mb V_n - \sqrt n \gamma(\theta)) - t\|^2  \right) \right| & = o_\p(1)
\end{align*}
uniformly in $\p$. This verifies the uniform expansion of the profile criterion.

\subsection{Example 3: uniform validity of Procedure 2 vs the bootstrap}\label{s:drift}

We return to Example 3 considered in Subsection \ref{s:miq} and show that our MC CSs (based on the posterior distribution of the profile QLR) are uniformly valid under very mild conditions while bootstrap-based CSs (based on the bootstrap distribution of the profile QLR) can undercover along certain sequences of DGPs. This reinforces the fact that our MC CSs and bootstrap-based CSs have different asymptotic properties.

Recall that $X_1,\ldots,X_n$ are i.i.d. with unknown mean $\mu^* \in \mb R_+$ and $\mu \in \mb R_+$ is identified by the moment inequality $\mb E[\mu - X_i] \leq 0$.  The identified set for $\mu$ is $M_I = [0,\mu^*]$. We consider coverage of the CS for $M_I = [0,\mu^*]$
We introduce a slackness parameter $\eta \in \mb R_+$ to write this model as a moment equality model $\mb E[ \mu + \eta - X_i] = 0$. The parameter space for $\theta = (\mu,\eta)$ is
$\Theta = \mb R_+^2$.
The GMM objective function and profile QLR are
\begin{align}
 L_n(\mu,\eta) & = -\frac{1}{2} (\mu + \eta - \bar X_n)^2 \notag \\
 PQ_n(M_I) & = (\mb V_n \wedge 0)^2 - ((\mb V_n + \sqrt n \mu^* ) \wedge 0)^2 \label{e:pq:appunif} \\
 PQ_n(M(\theta)) & = ((\mb V_n - \sqrt n \gamma(\theta)) \wedge 0)^2 - ((\mb V_n + \sqrt n \mu^* ) \wedge 0)^2 \notag
\end{align}
where $\sqrt n \gamma(\theta) = \sqrt n( \mu + \eta - \mu^*) \in [-\sqrt n \mu^*,\infty)$.

\subsubsection{Uniform validity of Procedures 2 and 3}

Let $\mf P$ be the family of distributions under which the $X_i$ are i.i.d. with mean $\mu^* = \mu^*(\p) \in \mb R_+$ and unit variance and for which
\begin{equation} \label{e:unif-clt}
 \lim_{n \to \infty} \sup_{\p \in \mf P} \sup_{z \in \mb R} |\p( \mb V_n \leq z) - \Phi(z)| = 0
\end{equation}
holds, where $\mb V_n = \mb V_n(\p) = \sqrt n (\bar X_n - \mu^*)$. We first consider uniform coverage of our MC CSs $\wh M_\alpha$ for the identified set $M_I = M_I(\p) = [0,\mu^*(\p)]$.

 To focus solely on the essential ideas, assume the prior on $\theta$ induces a uniform prior on $\gamma$ (the posterior is still proper); this could be relaxed at the cost of more cumbersome notation without changing the results that follow. Letting $z \geq 0$, $\kappa = \sqrt n \gamma$ and $v_n = v_n(\p) = \mb V_n + \sqrt n \mu^*$, we have:
\[
 \Pi_n (\{ \theta : PQ_n(M(\theta)) \leq z\} | \mf X_n)
 = \frac{\int_{-\sqrt n\mu^*}^{\infty} \ind\{ ((\mb V_n - \kappa) \wedge 0)^2 - (v_n \wedge 0)^2 \leq z \}e^{-\frac{1}{2}(\mb V_n -\kappa) ^2} \mr d\kappa }{\int_{-\sqrt n\mu^*}^{\infty} e^{-\frac{1}{2}(\mb V_n -\kappa) ^2} \mr d\kappa } \,.
\]
A change of variables with $x = \mb V_n - \kappa$ yields:
\begin{align*}
\Pi_n (\{ \theta : PQ_n(M(\theta)) \leq z\} | \mf X_n)
 & = \frac{\int_{-\infty}^{v_n} \ind\{ ( x \wedge 0)^2 \leq z + (v_n \wedge 0)^2 \}e^{-\frac{1}{2}x^2} \mr d x }{\int_{-\infty}^{v_n} e^{-\frac{1}{2}x^2} \mr d x } \\
 & =  \p_{ Z | \mf X_n} (  - \sqrt{z + (v_n \wedge 0)^2} \leq Z | Z \leq v_n) = G(v_n;z) \,.
\end{align*}
As we have an explicit form for the posterior distribution of the profile QLR, we can compute the posterior critical value directly rather than resorting to MC sampling. Therefore, Assumption \ref{a:mcmc:profile:unif} is not required here (since we can trivially set $\xi_{n,\alpha}^{post,p} = \xi_{n,\alpha}^{mc,p}$). If MC sampling were to be used, we would require that Assumption \ref{a:mcmc:profile:unif}  holds.

For $v_n \geq 0$, we have
\[
 G(v_n;z)  = \p_{ Z | \mf X_n} (  - \sqrt{z} \leq Z | Z \leq v_n)
\]
and so the posterior $\alpha$-critical value $\xi_{n,\alpha}^{post,p} = \Phi^{-1}((1-\alpha) \Phi(v_n))^2$. Therefore,
\begin{align}
\p( PQ_n(M_I) \leq \xi_{n,\alpha}^{post,p} | v_n \geq 0) & = \p( (\mb V_n \wedge 0)^2 \leq \Phi^{-1}((1-\alpha) \Phi(v_n))^2 | v_n \geq 0) \notag \\
& = \p( \Phi^{-1}((1-\alpha) \Phi(v_n)) \leq \mb V_n | v_n \geq 0) \label{e:ex3:unif1} \,.
\end{align}
Now suppose that $v_n < 0$. Here we have
\[
 G(v_n;z) = \p_{ Z | \mf X_n} (  - \sqrt{z + v_n^2} \leq Z | Z \leq v_n) = \frac{\Phi(v_n)-\Phi( - \sqrt{z + v_n^2})}{\Phi(v_n)}
\]
from which it follows that $\xi_{n,\alpha}^{post,p} = \Phi^{-1}((1-\alpha) \Phi(v_n))^2 - v_n^2$ and hence:
\begin{align}
 \p( PQ_n(M_I) \leq \xi_{n,\alpha}^{post,p} | v_n < 0) & = \p( (\mb V_n \wedge 0)^2  \leq \Phi^{-1}((1-\alpha) \Phi(v_n))^2 | v_n < 0) \notag \\
 & = \p(  \Phi^{-1}((1-\alpha) \Phi(v_n)) \leq \mb V_n | v_n < 0) \label{e:ex3:unif2} \,.
\end{align}
Combining (\ref{e:ex3:unif1}) and (\ref{e:ex3:unif2}), we obtain:
\[
 \p( PQ_n(M_I) \leq \xi_{n,\alpha}^{post,p})
  = \p((1-\alpha) \Phi(v_n) \leq \Phi(\mb V_n) )
  \geq \p( (1-\alpha) \leq \Phi(\mb V_n))
\]
which, together with (\ref{e:unif-clt}), delivers the uniform coverage result for Procedure 2:
\[
 \liminf_{n \to \infty} \inf_{\p \in \mf P} \p( \mb M_I(\p) \subseteq \wh M_\alpha ) \geq \alpha\,.
\]

For uniform validity of Procedure 3, first note that (\ref{e:pq:appunif}) implies that the inequality
\[
 \p( PQ_n(M_I) \leq \chi^2_{1,\alpha}) \geq \p ( (\mb V_n \wedge 0)^2 \leq \chi^2_{1,\alpha})
\]
holds uniformly in $\p$. It follows by (\ref{e:unif-clt}) that:
\[
 \liminf_{n \to \infty} \inf_{\p \in \mf P} \p( \mb M_I(\p) \subseteq \wh M_\alpha^\chi ) > \alpha\,.
\]

\subsubsection{Lack of uniformity of the bootstrap}

We now show that bootstrap-based CSs for $M_I$ are not uniformly valid when the standard (i.e. nonparametric) bootstrap is used. The bootstrap criterion function $L_n^\star(\mu,\eta)$ is
\[
  L_n^\star(\mu,\eta) = -\frac{1}{2}(\mu + \eta - \bar X_n^\star)^2
\]
where $\bar X_n^\star$ is the bootstrap sample mean. Let $\wh M_I = [0,(\bar X_n \vee 0)]$. Consider a subsequence $(\mr P_n)_{n \in \mb N} \subset \mf P$ with $\mu^*(\mr P_n) = c/\sqrt n$ for some $c > 0$ (chosen below). By similar calculations to Subsection \ref{s:miq}, along this sequence of DGPs,
the bootstrapped profile QLR statistic for $M_I$ is:
\begin{align*}
 PQ_n^\star(M_I) & = 2n L_n^\star(\hat \mu^\star,\hat \eta^\star)  - \inf_{\mu \in \wh M_I} \sup_{\eta \in H_\mu} 2n L_n^\star(\mu , \eta) \\
 & = ((\mb V_n^\star + ((\mb V_n + c) \wedge 0) ) \wedge 0)^2  - ((\mb V_n^\star + \mb V_n + c ) \wedge 0)^2\,.
\end{align*}
Let $\xi_{n,\alpha}^{boot,p}$ denote the $\alpha$-quantile of the distribution of $PQ_n^\star(M_I)$. Consider
\[
 \wh M_\alpha^{boot} = \{ \mu : {\textstyle \sup_{\eta \in H_\mu} Q_n(\mu,\eta)} \leq \xi_{n,\alpha}^{boot,p}\}
\]

We now show that for any $\alpha \in (\frac{1}{2},1)$ we may choose $c>0$ in the definition of $(\mr P_n)_{n \in \mb N}$ such that the asymptotic coverage of $\wh M_\alpha^{boot}$ is strictly less than $\alpha$ along this sequence of DGPs. Since
\[
 PQ_n^\star(M_I) = ( ( \mb V_n^\star \wedge 0 )^2 - ((\mb V_n^\star + \mb V_n + c) \wedge 0)^2 ) \ind \{ \mb V_n + c \geq 0\}
\]
it follows that whenever $\mb V_n + c < 0$ the bootstrap distribution of the profile QLR for $M_I$ is point mass at the origin, and the $\alpha$-quantile of the bootstrap distribution is $\xi_{n,\alpha}^{boot,p} = 0$. However, the QLR statistic for $M_I$ is $PQ_n(M_I) =(\mb V_n \wedge 0)^2 - ((\mb V_n + c ) \wedge 0)^2$. So whenever $\mb V_n + c < 0$ we also have that $PQ_n(M_I) = \mb V_n^2 - (\mb V_n + c)^2 > 0$. Therefore,
\[
 \mr P_n( M_I(\mr P_n) \subseteq \wh M_\alpha^{boot} | \mb V_n + c < 0) = 0\,.
\]
It follows by (\ref{e:unif-clt}) that for any $c$ for which $\Phi(c) < \alpha$, we have:
\[
 \limsup_{n \to \infty} \mr P_n( M_I(\mr P_n) \subseteq \wh M_\alpha^{boot} ) \leq \lim_{n \to \infty} \mr P_n( \mb V_n + c \geq 0) < \alpha\,.
\]

\subsubsection{An alternative recentering}

An alternative is to recenter the criterion function at $(\bar X_n \vee 0)$, that is, one could use instead
\[
 L_n(\mu,\eta) = - \frac{1}{2} (\mu + \eta - (\bar X_n \vee 0))^2
\]
similar to the idea of a sandwich (quasi-)likelihood with $(\bar X_n \vee 0) = \hat \gamma_n$. This maps into the setup described in Appendix \ref{s:uniformity}, where
\[
 nL_n(\theta) = \ell_n - \frac{1}{2} ( \sqrt n \gamma(\theta))^2 + \sqrt n (\gamma(\theta))( \sqrt n (\hat \gamma_n - \tau))
\]
where $\ell_n = - \frac{1}{2}(\sqrt n (\hat \gamma_n - \tau))^2$, $\theta = (\mu,\eta)$ and
\begin{align*}
 \gamma(\theta) & = \mu + \eta - \mu^*  &
 \tau & = \mu^* &
 \hat \gamma_n & = (\bar X_n \vee 0) &
 \sqrt n (\hat \gamma_n - \tau) & = (\mb V_n \wedge -\sqrt n \mu^*)
\end{align*}
where $\mb V_n = \sqrt n (\bar X_n - \mu^*)$, $\gamma(\theta) \in [- \mu^*,\infty)$, and $\mu^* \in \mb R_+$.

Assumption \ref{a:rate:unif} and \ref{a:quad:unif}(i)--(iii) hold with $\Theta_{osn} = \Theta$, $k_n = +\infty$, $ T = \mb R_+$, and $\mf T v = (v \vee 0)$ (none of the models are singular).  We again take a prior on $\theta$ that induces a flat prior on $\gamma$ to concentrate on the essential ideas, verifying Assumption \ref{a:prior:unif}.

For inference on $M_I = [0,\mu^*(\p)]$, observe that
\begin{align*}
 PQ_n(M(\theta)) & = f(\sqrt n (\hat \gamma_n - \tau) - \sqrt n \gamma(\theta)) &
 PQ_n(M_I) & = f(\sqrt n (\hat \gamma_n - \tau))
\end{align*}
where $f(v) = ( v \wedge 0)^2$ for each $\p$, verifying Assumption \ref{a:qlr:unif}(i). Assumption \ref{a:qlr:unif}(ii) also holds for this $f$. Finally, for Assumption \ref{a:qlr:unif}(iii), for any $z,v \geq 0$ we have
\begin{align*}
 \p_Z (f(Z) \leq z | Z  \in  v - T ) = \frac{\Phi(v) - \Phi(-\sqrt z)}{\Phi(v)} \leq 1- \Phi(-\sqrt z) = \p_Z (f(Z) \leq z)\,.
\end{align*}
Theorem \ref{t:profile:unif}, together with (\ref{e:unif-clt}), delivers uniform coverage for Procedure 2.

Similarly, for uniform validity of Procedure 3 we have:
\[
 \p( PQ_n(M_I) \leq \chi^2_{1,\alpha}) \geq \p ( (\mb V_n \wedge 0)^2 \leq \chi^2_{1,\alpha})
\]
which, together with (\ref{e:unif-clt}), delivers uniform coverage for Procedure 3.

Now consider bootstrap-based inference. As before, let $\wh M_I = [0,(\bar X_n \vee 0)]$ and consider a subsequence $(\mr P_n)_{n \in \mb N} \subset \mf P$ with $\mu^*(\mr P_n) = c/\sqrt n$ for some $c > 0$. Under $\mr P_n$, we then have:
\begin{align*}
  L_n^\star(\mu,\eta) & = -\frac{1}{2}(\mu + \eta - (\bar X_n^\star \vee 0) )^2 \\
 PQ_n^\star(M_I) & = ([((\mb V_n^\star + \mb V_n) \vee -c)  -  (\mb V_n \vee -c)] \wedge 0)^2
\end{align*}
 and the true QLR statistic is $PQ_n(M_I) = ((\mb V_n \vee -c)   \wedge 0)^2$.  We again show that for any $\alpha \in (\frac{1}{2},1)$ we may choose $c>0$ in the definition of $(\mr P_n)_{n \in \mb N}$ such that the asymptotic coverage of $\wh M_\alpha^{boot}$ is strictly less than $\alpha$ along this sequence of DGPs. Observe that when $\mb V_n < -c$ we have $PQ_n(M_I) = c^2 > 0$ and $PQ_n^\star(M_I) = 0$. Therefore,
\[
 \mr P_n( M_I(\mr P_n) \subseteq \wh M_\alpha^{boot} | \mb V_n + c < 0) = 0\,.
\]
It follows by (\ref{e:unif-clt}) that for any $c$ for which $\Phi(c) < \alpha$, we again have:
\[
 \limsup_{n \to \infty} \mr P_n( M_I(\mr P_n) \subseteq \wh M_\alpha^{boot} ) \leq \lim_{n \to \infty} \mr P_n( \mb V_n + c \geq 0) < \alpha\,.
\]

\section{Local power}\label{s:lp}

In this appendix we study the behavior of the CSs $\wh \Theta_\alpha$ and $\wh M_\alpha$ under $n^{-1/2}$-local (contiguous) alternatives. We maintain the same setup as in Section \ref{sec-property}. Fix $a \in \mb R^{d^*}$.

\begin{assumption} \label{a:lp}
There exist sequences of distributions $(\mr P_{n,a})_{n \in \mb N}$ such that as $n \to \infty$: \\
(i) $L_n(\hat \theta) = \sup_{\theta \in \Theta_{osn}} L_n(\theta) + o_{\mr P_{n,a}}(n^{-1})$; \\
(ii) $\Pi_n(\Theta_{osn}^c | \mf X_n) = o_{\mr P_{n,a}}(1)$; \\
(iii) There exist sequences of random variables $\ell_n$ and $\mb R^{d^*}$-valued random vectors $\hat \gamma_n$ (both measurable in $\mf X_n$) such that:
	\begin{equation} \label{e:quad:lp}
	\sup_{\theta \in \Theta_{osn}} \left| n L_n(\theta) - \left(\ell_n +\frac{1}{2} \| \sqrt n \hat \gamma_n \|^2 - \frac{1}{2} \|\sqrt n (\hat \gamma_n - \gamma(\theta))\|^2 \right) \right| = o_{\mr P_{n,a}}(1)
	\end{equation}
	with $\sup_{\theta \in \Theta_{osn}} \|\gamma(\theta)\| \to 0$, $\sqrt n \hat \gamma_n = \mb V_n$ where $\mb V_n \overset{\mr P_{n,a}}{\rightsquigarrow} N(a,I_{d^*})$ and $T=\mb R^{d^*}$; \\
(iv)  $\int_\Theta e^{nL_n(\theta)} \,\mr d \Pi(\theta) < \infty $ holds $\mr P_{n,a}$-almost surely; \\
(v) $\Pi_\Gamma$ has a continuous, strictly positive density $\pi_\Gamma$ on $B_\delta \cap \Gamma$ for some $\delta > 0$; \\
(vi) $\xi_{n,\alpha}^{mc} = \xi_{n,\alpha}^{post} + o_{\mr P_{n,a}}(1)$.
\end{assumption}

Assumption \ref{a:lp} is essentially a restatement of Assumptions \ref{a:rate} to \ref{a:mcmc} with a modified quadratic expansion. Notice that with $a = 0$ we obtain $\mr P_{n,a} = \mb P$ and Assumption \ref{a:lp} corresponds to Assumptions \ref{a:rate} to \ref{a:mcmc} with generalized information equality $\Sigma = I_{d^*}$ and $T = \mb R^{d^*}$.

Let  ${\chi^2_{d^*}(a'a)}$ denote the noncentral $\chi^2$ distribution with $d^*$ degrees of freedom and noncentrality parameter $a'a$ and let $F_{\chi^2_{d^*}(a'a)}$ denote its cdf. Let $\chi^2_{d^*,\alpha}$ denote the $\alpha$ quantile of the (standard) $\chi^2_{d^*}$ distribution $F_{\chi^2_{d^*}}$.

\begin{theorem}\label{t:lp}
Let Assumption \ref{a:lp}(i)(iii) hold. Then:
\[
 \sup_{\theta \in \Theta_I} Q_n(\theta) \overset{\mr P_{n,a}}{\rightsquigarrow} \chi^2_{d^*}(a'a);
\]
if further Assumption \ref{a:lp}(ii)(iv)(v) holds, then:
\[
\sup_z \left| \Pi_n \big(\{ \theta : Q_n(\theta) \leq z \}\big|\,\mf X_n\big) - F_{\chi^2_{d^*}}(z) \right| = o_{\mr P_{n,a}}(1);
\]
and if further Assumption \ref{a:lp}(vi) holds, then:
\[
 \lim_{n \to \infty} \mr P_{n,a} (\Theta_I \subseteq \wh \Theta_\alpha) = F_{\chi^2_{d^*}(a'a)}( \chi^2_{d^*,\alpha}) < \alpha~~\mbox{whenever}~~ a \neq 0.
\]
\end{theorem}

We now present a similar result for $\wh M_\alpha$. To do so, we extend the conditions in Assumption \ref{a:lp}.

\setcounter{assumption}{0}

\begin{assumption}
Let the following also hold under the local alternatives: \\
(vii) 	There exists a measurable  $f : \mb R^{d^*} \to \mb R_{+}$ such that:
	\begin{align*}
	& \sup_{\theta \in \Theta_{osn}} \left|  nP L_n(M(\theta)) - \left( \ell_n + \frac{1}{2} \|\mb V_n\|^2 - \frac{1}{2} f \left( \mb V_n - \sqrt n \gamma(\theta) \right)  \right) \right| = o_{\mr P_{n,a}}(1) \,
	\end{align*} with $\mb V_n$ from Assumption \ref{a:lp}(iii).\\
(vi$\,^\prime\!$) $\xi_{n,\alpha}^{mc,p} = \xi_{n,\alpha}^{post,p} + o_{\mr P_{n,a}}(1)$.
\end{assumption}

Assumption \ref{a:lp}(vii) and (vi$^\prime$) are essentially Assumptions \ref{a:qlr:profile} and \ref{a:mcmc:profile}.

Let $Z \sim N(0,I_{d^*})$ and $\mb P_Z$ denote the distribution of $Z$. Let the distribution of $f(Z)$ be continuous at its $\alpha$-quantile, which we denote by $z_\alpha$.
\begin{theorem}\label{t:lp:profile}
Let Assumption \ref{a:lp}(i)(iii)(vii) hold. Then:
\[
 PQ_n(M_I) \overset{\mr P_{n,a}}{\rightsquigarrow} f(Z + a)~;
\]
if further Assumption \ref{a:lp}(ii)(iv)(v) holds, then for a neighborhood $I$ of $z_\alpha$:
\[
 \sup_{z  \in I} \left| { \Pi_n \big( \{\theta:PQ_n( M(\theta)) \leq  z \} \,\big|\, \mf X_n \big) } - \p_{Z | \mf X_n} \big( f(Z) \leq z  \big)  \right| = o_{\mr P_{n,a}}(1)
\]
and if further Assumption \ref{a:lp}(vi$\,^\prime\!$) holds, then:
\[
 \lim_{n \to \infty} \mr P_{n,a} (M_I \subseteq \wh M_\alpha) = \p_Z (f(Z + a) \leq z_\alpha)~.
\]
\end{theorem}

It follows from Anderson's lemma \cite[Lemma 8.5]{vdV} that$$\lim_{n \to \infty} \mr P_{n,a} (M_I \subseteq \wh M_\alpha) \leq \alpha$$ whenever $f$ is subconvex. In particular, this includes the case in which $M_I$ is a singleton.

\section{Parameter-dependent support}\label{ax:pds}

In this appendix we briefly describe how our procedure may be applied to models with parameter dependent support under loss of identifiability. Parameter-dependent support is a feature of certain auction models (e.g., \cite{HP}, \cite{CH04}) and some structural models in labor economics (e.g., \cite{FlinnHeckman}). For simplicity we just deal with inference on the full vector, though the following results could be extended to subvector inference in this context.

We again presume the existence of a local reduced-form parameter $\gamma$ such that $\gamma(\theta) = 0$ if and only if $\theta \in \Theta_I$. In what follows we assume without loss of generality that $L_n(\hat \theta) = \sup_{\theta \in \Theta_{osn}} L_n(\theta)$ since $\hat \theta$ is not required in order to compute the confidence set. We replace Assumption \ref{a:quad} (local quadratic approximation) with the following assumption, which permits the support of the data to depend on certain components of the local reduced-form parameter $\gamma$.

\setcounter{assumption}{1}

\begin{assumption}\label{a:quad:gamma}
(i) There exist functions $\gamma : \Theta^N_I \to \Gamma \subseteq \mb R^{d^*}$ and $h : \Gamma \to \mb R_+$, a sequence of $\mb R^{d^*}$-valued random vectors $\hat \gamma_n$, and a positive sequence $(a_n)_{n \in \mb N}$ with $a_n \to 0$ such that:
\[
 \sup_{\theta \in \Theta_{osn}} \left| \frac{ \frac{a_n}{2} Q_n(\theta) - h(\gamma(\theta) - \hat \gamma_n)}{h(\gamma(\theta) - \hat \gamma_n)} \right| = o_\p(1)
\]
with $\sup_{\theta \in \Theta_{osn}} \|\gamma(\theta)\| \to 0$ and $\inf\{h(\gamma) : \|\gamma\| = 1\} > 0$;\\
(ii) there exist $r_1,\ldots,r_{d^*} > 0$ such that $th(\gamma) = h(t^{r_1} \gamma_1,t^{r_2} \gamma_2,\ldots,t^{r_{d^*}} \gamma_{d^*})$ for each $t >0$;\\
(iii) the sets $K_{osn} = \{ (b_n^{-r_1} (\gamma_1(\theta) - \hat \gamma_{n,1}),\ldots,b_n^{-r_{d^*}} (\gamma_{d^*}(\theta) - \hat \gamma_{n,d^*}))' : \theta \in \Theta_{osn}\}$ cover $\mb R^{d^*}_+$ for any positive sequence $(b_n)_{n \in \mb N}$ with $b_n \to 0$ and $a_n/b_n \to 1$.
\end{assumption}

This assumption is similar to Assumptions 2-3 in \cite{FHW} but has been modified to allow for non-identifiable parameters $\theta$. Let $F_\Gamma$ denote a Gamma distribution with shape parameter $r^* = \sum_{i=1}^{d^*} r_i$ and scale parameter $2$. The following lemma shows that the posterior distribution of the QLR converges to $F_\Gamma$.

\begin{lemma}\label{l:post:gamma}
Let Assumptions \ref{a:rate}, \ref{a:quad:gamma}, and \ref{a:prior} hold. Then:
\[
 \sup_z \left|  \Pi_n( \{\theta : Q_n(\theta ) \leq z \} | \mf X_n) - F_\Gamma(z) \right| = o_\p(1)\,.
\]
\end{lemma}

By modifying appropriately the arguments in \cite{FHW} one can show that, under Assumption \ref{a:quad:gamma}, $\sup_{\theta \in \Theta_I} Q_n(\theta) \rightsquigarrow F_\Gamma$. The following theorem states that one still obtains asymptotically correct frequentist coverage of $\wh \Theta_\alpha$.

\begin{theorem} \label{t:main:gamma}
Let Assumptions \ref{a:rate}, \ref{a:quad:gamma}, \ref{a:prior}, and \ref{a:mcmc} hold and $\sup_{\theta \in \Theta_I} Q_n(\theta) \rightsquigarrow F_\Gamma$. Then:
\[
 \lim_{n \to \infty} \p(\Theta_I \subseteq \wh \Theta_{\alpha}) = \alpha\,.
\]
\end{theorem}

We finish this section with a simple example. Consider a model in which $X_1,\ldots,X_n$ are i.i.d. $U[0,(\theta_1 \vee \theta_2)]$ where $(\theta_1,\theta_2) \in \Theta = \mb R_+^2$. Let the true distribution of the data be $U[0,\tilde \gamma]$. The identified set is $\Theta_I = \{ \theta \in \Theta : \theta_1 \vee \theta_2 = \tilde \gamma\}$.

Then we use the reduced-form parameter $\gamma(\theta) = (\theta_1 \vee \theta_2) - \tilde \gamma$. Let $\hat \gamma_n = \max_{1 \leq i \leq n} X_i - \tilde \gamma$. Here we take $\Theta_{osn} = \{ \theta : (1+\varepsilon_n) \hat \gamma_n \geq \gamma(\theta) \geq \hat \gamma_n\}$ where $\varepsilon_n \to 0$ slower than $n^{-1}$ (e.g. $\varepsilon_n = (\log n)/n$). It is straightforward to show that:
\[
 \sup_{\theta \in \Theta_I} Q_n(\theta) = 2 n \log \left( \frac{\tilde \gamma}{\hat \gamma_n + \tilde \gamma} \right) \rightsquigarrow F_\Gamma
\]
where $F_\Gamma$ denotes the Gamma distribution with shape parameter $r^*=1$ and scale parameter $2$. Furthermore, taking $a_n= n^{-1}$ and $h(\gamma(\theta)-\hat \gamma_n) = \tilde \gamma^{-1} (\gamma(\theta)-\hat \gamma_n)$ we may deduce that:
\[
 \sup_{\theta \in \Theta_{osn}} \left| \frac{ \frac{1}{2n} Q_n(\theta) - h(\gamma(\theta) - \hat \gamma_n)}{h(\gamma(\theta) - \hat \gamma_n)} \right| = o_\p(1)\,.
\]
Notice that $r^*=1$ and that the sets $K_{osn} = \{ n (\gamma(\theta) - \hat \gamma_n) : \theta \in \Theta_{osn}\} = \{ n (\gamma - \hat \gamma_n) : (1+\varepsilon_n) \hat \gamma \geq \gamma \geq \hat \gamma_n\}$ cover $\mb R^+$. A smooth prior on $\Theta$ will induce a smooth prior on $\gamma(\theta)$, and the result follows from Theorem \ref{t:main:gamma}.

\newpage

\section{Proofs and Additional Results}\label{a:proofs}

\subsection{Proofs and Additional Lemmas for Sections \ref{sec-procedure} and \ref{sec-property}}

\begin{proof}[\textbf{Proof of Lemma \ref{l:basic}}]
	By (ii), there is a positive sequence $(\varepsilon_n)_{n \in \mb N}$ with $\varepsilon_n = o(1)$ such that $w_{n,\alpha} \geq w_\alpha-\varepsilon_n$ holds wpa1. Therefore:
	\[
	\begin{array}{rcl}
	\p (\Theta_I \subseteq \wh \Theta_\alpha) & = & \p (\sup_{\theta \in \Theta_I} Q_n(\theta) \leq w_{n,\alpha}) \\
	& \geq & \p (\sup_{\theta \in \Theta_I} Q_n(\theta) \leq w_{\alpha} - \varepsilon_n) + o(1)
	\end{array}
	\]
	and the result follows by part (i). If $w_{n,\alpha} = w_\alpha + o_\p(1)$ then we may replace the preceding inequality by an equality.
\end{proof}

\begin{proof}[\textbf{Proof of Lemma \ref{l:basic:profile}}]
	Follows by similar arguments to the proof of Lemma \ref{l:basic}.
\end{proof}

\begin{lemma}\label{l:quad}
	Let Assumptions \ref{a:rate}(i) and \ref{a:quad} hold. Then:
	\begin{align}
	\sup_{\theta \in \Theta_{osn}} \left | Q_n(\theta) - \| \sqrt n \gamma(\theta) - \mf T\mb V_n \|^2 \right| &= o_\p(1) \label{e:quad:1}  \,.
	\end{align}
And hence $\sup_{\theta \in \Theta_I} Q_n(\theta) = \|\mf T \mb V_n\|^2 + o_\p(1)$.
\end{lemma}

\begin{proof}[\textbf{Proof of Lemma \ref{l:quad}}]
	By Assumptions \ref{a:rate}(i) and \ref{a:quad}, we obtain:
	\begin{align}
	2nL_n(\hat \theta) & = \sup_{\theta \in \Theta_{osn}} 2n L_n(\theta) + o_\p(1) \notag \\
	& =  2\ell_n + \|\sqrt n \hat \gamma_n\|^2 - \inf_{\theta \in \Theta_{osn}} \|\sqrt n \gamma(\theta) - \mf T \mb V_n\|^2 + o_\p(1) \notag \\
	& =  2\ell_n + \|\mf T \mb V_n\|^2 - \inf_{t \in T}  \|t -\mf T \mb V_n\|^2 + o_\p(1)  \label{e:lnhattheta}
	\end{align}
	where $\inf_{t \in T}  \|t -\mf T \mb V_n\|^2 = 0$ because $\mf T \mb V_n \in T$.
	Now by Assumption \ref{a:quad},
	\begin{align*}
	Q_n(\theta)  & = \left( 2\ell_n + \|\mf T \mb V_n\|^2 + o_\p(1)\right) - \left( 2 \ell_n +\|\mf T \mb V_n\|^2 - \| \sqrt n \gamma(\theta) - \mf T \mb V_n \|^2 + o_\p(1) \right) \\
	& =  \| \sqrt n \gamma(\theta) - \mf T\mb V_n \|^2  + o_\p(1)
	\end{align*}
	where the $o_\p(1)$ term holds uniformly over $\Theta_{osn}$. This proves expression (\ref{e:quad:1}). Finally, since $\gamma(\theta) = 0$ for $\theta \in \Theta_I$, we have $\sup_{\theta \in \Theta_I} Q_n(\theta) = \|\mf T \mb V_n \|^2 + o_\p(1)$.
\end{proof}

\begin{proof}[\textbf{Proof of Lemma \ref{l:post}}]
	We first prove equation (\ref{e:c:post:1}). Since $|\Pr(A) - \Pr (A \cap B)| \leq \Pr(B^c)$, we have:
	\begin{align}
	\sup_{z} \big|\Pi_n (\{ \theta : Q_n(\theta) \leq z \} |\mf X_n) - \Pi_n ( \{ \theta : Q_n(\theta) \leq z \} \cap \Theta_{osn}|\mf X_n) \big| \leq \Pi_n (\Theta_{osn}^c|\mf X_n) = o_\p(1) \label{e-post-target}
	\end{align}
	by Assumption \ref{a:rate}(ii). Moreover by Assumptions \ref{a:rate}(ii) and \ref{a:prior}(i),
	\begin{align*}
	\left| \frac{ \int_{\Theta_{osn}}  e^{nL_n(\theta)} \mr d\Pi(\theta) }{\int_{\Theta} e^{nL_n(\theta)} \mr d\Pi(\theta)} -1 \right| = \Pi_n (\Theta_{osn}^c|\mf X_n) = o_\p(1)
	\end{align*}
	and hence:
	\begin{equation} \label{e-denombd}
	\sup_{z} \left| \Pi_n ( \{ \theta : Q_n(\theta) \leq z \} \cap \Theta_{osn}\,|\,\mf X_n)  - \frac{ \int_{\{ \theta : Q_n(\theta) \leq z \} \cap \Theta_{osn}}  e^{nL_n(\theta)} \mr d\Pi(\theta) }{\int_{\Theta_{osn}} e^{nL_n(\theta)} \mr d\Pi(\theta)}  \right|  = o_\p(1) \,.
	\end{equation}
	In view of (\ref{e-post-target}) and (\ref{e-denombd}), it suffices to characterize the large-sample behavior of:
	\begin{equation} \label{e-Rn}
	R_{n}(z) :=  \frac{\int_{\{\theta:Q_n(\theta) \leq  z \} \cap \Theta_{osn}}\! e^{nL_n(\theta) - \ell_n - \frac{1}{2}\|\mf T \mb V_n\|^2} \mr d\Pi(\theta)}{\int_{\Theta_{osn}} \!e^{n L_n(\theta) - \ell_n - \frac{1}{2}\|\mf T \mb V_n\|^2} \mr d\Pi(\theta)} \,.
	\end{equation}
	Lemma \ref{l:quad} and Assumption \ref{a:quad} imply that there exists a positive sequence $(\varepsilon_n)_{n \in \mb N}$ independent of $z$ with $\varepsilon_n = o(1)$ such that the inequalities:
	\begin{align*}
	\sup_{\theta \in \Theta_{osn}} \left | Q_n(\theta) - \| \sqrt n \gamma(\theta) - \mf T\mb V_n \|^2 \right| & \leq \varepsilon_n  \\
	\sup_{\theta \in \Theta_{osn}} \left| nL_n(\theta) - \ell_n - \frac{1}{2}\|\mf T \mb V_n\|^2 + \frac{1}{2} \|\sqrt n \gamma(\theta) - \mf T \mb V_n\|^2  \right| & \leq \varepsilon_n
	\end{align*}
	both hold wpa1. Therefore, wpa1 we have:
	\begin{align}
	&  e^{-2\varepsilon_n} \frac{\int_{\{\theta: \|\sqrt n \gamma(\theta) - \mf T \mb V_n\|^2 \leq  z -  \varepsilon_n\} \cap \Theta_{osn}}\! e^{ - \frac{1}{2} \|\sqrt n \gamma(\theta) - \mf T \mb V_n\|^2} \mr d\Pi(\theta)}{\int_{\Theta_{osn}} \!e^{ - \frac{1}{2} \|\sqrt n \gamma(\theta) - \mf T \mb V_n\|^2} \mr d\Pi(\theta)} \notag \\
	& \leq R_n(z)  \leq  e^{2\varepsilon_n} \frac{\int_{\{\theta:\|\sqrt n \gamma(\theta) - \mf T \mb V_n\|^2 \leq  z +  \varepsilon_n\} \cap \Theta_{osn}}\! e^{ - \frac{1}{2}\|\sqrt n \gamma(\theta) - \mf T \mb V_n\|^2} \mr d\Pi(\theta)}{\int_{\Theta_{osn}} \!e^{ - \frac{1}{2} \|\sqrt n \gamma(\theta) -\mf T \mb V_n\|^2} \mr d\Pi(\theta)}  \notag
	\end{align}
	uniformly in $z$. Let $\Gamma_{osn} = \{ \gamma(\theta) : \theta \in \Theta_{osn}\}$. A change of variables yields:
	\begin{align}
	&  e^{-2 \varepsilon_n} \frac{\int_{\{\gamma: \|\sqrt n \gamma - \mf T \mb V_n\|^2 \leq  z  -  \varepsilon_n\} \cap \Gamma_{osn}}\! e^{ - \frac{1}{2} \|\sqrt n \gamma - \mf T \mb V_n\|^2} \mr d\Pi_\Gamma(\gamma)}{\int_{\Gamma_{osn}} \!e^{ - \frac{1}{2} \|\sqrt n \gamma - \mf T \mb V_n\|^2} \mr d\Pi_\Gamma(\gamma)} \notag \\
	& \leq R_n(z)  \leq  e^{2 \varepsilon_n} \frac{\int_{\{\gamma:\|\sqrt n \gamma - \mf T \mb V_n\|^2 \leq  z +  \varepsilon_n\} \cap \Gamma_{osn}}\! e^{ - \frac{1}{2}\|\sqrt n \gamma - \mf T \mb V_n\|^2} \mr d\Pi_\Gamma(\gamma)}{\int_{\Gamma_{osn}} \!e^{ - \frac{1}{2} \|\sqrt n \gamma - \mf T \mb V_n\|^2} \mr d\Pi_\Gamma(\gamma)}   \,.\label{e:rn:bound}
	\end{align}
	
	Recall $B_\delta$ from Assumption \ref{a:prior}(ii). The inclusion $\Gamma_{osn} \subset B_\delta \cap \Gamma$ holds for all $n$ sufficiently large by Assumption \ref{a:quad}. Taking $n$ sufficiently large and using Assumption \ref{a:prior}(ii), we may deduce that there exists a positive sequence $(\bar \varepsilon_n)_{n \in \mb N}$ with $\bar \varepsilon_n= o(1)$ such that:
	\[
	\left| \frac{\sup_{\gamma \in \Gamma_{osn}} \pi_{\Gamma}(\gamma)}{\inf_{\gamma \in \Gamma_{osn}}  \pi_{\Gamma}( \gamma) } - 1 \right| \leq \bar \varepsilon_n
	\]
	for each $n$.
	Substituting into (\ref{e:rn:bound}):
	\begin{align*}
	&  (1-\bar \varepsilon_n)e^{-2\varepsilon_n} \frac{\int_{\{\gamma: \|\sqrt n \gamma - \mf T \mb V_n\|^2 \leq  z -  \varepsilon_n\} \cap \Gamma_{osn}}\! e^{ - \frac{1}{2} \|\sqrt n \gamma - \mf T \mb V_n\|^2} \mr d\gamma}{\int_{\Gamma_{osn}} \!e^{ - \frac{1}{2} \|\sqrt n \gamma - \mf T \mb V_n\|^2} \mr d\gamma} \notag \\
	& \leq R_n(z)  \leq  (1+\bar \varepsilon_n)e^{2 \varepsilon_n} \frac{\int_{\{\gamma:\|\sqrt n \gamma - \mf T \mb V_n\|^2 \leq  z + \varepsilon_n\} \cap \Gamma_{osn}}\! e^{ - \frac{1}{2}\|\sqrt n \gamma - \mf T \mb V_n\|^2} \mr d\gamma}{\int_{\Gamma_{osn}} \!e^{ - \frac{1}{2} \|\sqrt n \gamma -\mf T \mb V_n\|^2} \mr d\gamma}  \notag
	\end{align*}
	uniformly in $z$, where ``$\mr d \gamma$'' denotes integration with respect to Lebesgue measure on $\mb R^{d^*}$.

	Let $T_{osn} = \{ \sqrt n \gamma : \gamma \in \Gamma_{osn}\}$ and let $B_z$ denote a ball of radius $z$ in $\mb R^{d^*}$ centered at the origin. Using the change of variables $\sqrt n \gamma - \mf T \mb V_n \mapsto \kappa$, we can rewrite the preceding inequalities as:
	\begin{align*}
	& (1-\bar \varepsilon_n) e^{-2 \varepsilon_n} \frac{\int_{B_{\sqrt {z - \varepsilon_n}} \cap (T_{osn}- \mf T \mb V_n) }\! e^{ - \frac{1}{2} \| \kappa \|^2} \mr d\kappa }{\int_{(T_{osn}- \mf T \mb V_n)} \!e^{ - \frac{1}{2} \| \kappa \|^2} \mr d\kappa}
	 \leq R_n(z)
	 \leq (1+\bar \varepsilon_n) e^{2 \varepsilon_n} \frac{\int_{B_{\sqrt{z  + \varepsilon_n}} \cap (T_{osn}- \mf T \mb V_n) }\! e^{ - \frac{1}{2} \| \kappa \|^2} \mr d\kappa }{\int_{(T_{osn}-\mf T \mb V_n)} \!e^{ - \frac{1}{2} \| \kappa \|^2} \mr d\kappa}
	\end{align*}
	with the understanding that $B_{\sqrt{z - \varepsilon_n}}$ is empty if $\varepsilon_n > z$.

	Let $\nu_{d^*}(A) = (2\pi)^{-d^*/2} \int_A e^{-\frac{1}{2} \|\kappa\|^2} \, \mr d \kappa$ denote Gaussian measure. We now show that:
	\begin{align}
	\sup_z \left| \frac{ \nu_{d^*}( B_{\sqrt {z \pm \varepsilon_n}} \cap (T_{osn}- \mf T \mb V_n) ) }{ \nu_{d^*}( T_{osn} - \mf T \mb V_n ) } - \frac{ \nu_{d^*}( B_{\sqrt {z \pm \varepsilon_n}} \cap (T- \mf T \mb V_n)  ) }{ \nu_{d^*}( T - \mf T \mb V_n ) } \right| & =  o_\p(1) \label{e:denom:1} \\
	\sup_z \left| \frac{ \nu_{d^*}( B_{\sqrt {z \pm \varepsilon_n}} \cap (T- \mf T \mb V_n)  ) }{ \nu_{d^*}( T - \mf T \mb V_n ) } - \frac{ \nu_{d^*}( B_{\sqrt z} \cap (T- \mf T \mb V_n) ) }{ \nu_{d^*}( T -\mf T \mb V_n ) } \right| & =  o_\p(1)\,. \label{e:denom:2}
	\end{align}
	Consider (\ref{e:denom:1}). To simplify presentation, we assume wlog that $T_{osn} \subseteq T$. As
	\begin{align}
	 \left| \frac{\Pr(A \cap B)}{\Pr(B)} - \frac{\Pr(A \cap C)}{\Pr(C)} \right| \leq 2 \frac{\Pr(C \setminus B)}{\Pr(C)} \label{e:abc:ineq}
	\end{align}
	holds for events $A,B,C$ with $B \subseteq C$, we have:
	\begin{align*}
	& \sup_z  \left|\frac{ \nu_{d^*}( B_{\sqrt {z \pm \varepsilon_n}} \cap (T_{osn}- \mf T \mb V_n) ) }{ \nu_{d^*}( T_{osn} - \mf T \mb V_n ) } - \frac{ \nu_{d^*}( B_{\sqrt {z \pm \varepsilon_n}} \cap (T- \mf T \mb V_n)  ) }{ \nu_{d^*}( T - \mf T \mb V_n ) } \right|
	\leq 2 \frac{ \nu_{d^*}( (T  \setminus T_{osn}) -\mf T \mb V_n ) }{ \nu_{d^*}( T -\mf T\mb V_n ) }
	\end{align*}
	As $\mb V_n$ is tight and $T\subseteq \mb R^{d^*}$ has positive volume, we may deduce that
	\begin{align} \label{e:gmeasbd:2}
	 1/{\nu_{d^*}( T -\mf T\mb V_n )}= O_\p(1)\,.
	\end{align}
	It also follows by tightness of $\mb V_n$ and Assumption \ref{a:quad} that $\nu_{d^*}( (T  \setminus T_{osn}) -\mf T \mb V_n ) = o_\p(1)$, which proves (\ref{e:denom:1}).
 Result (\ref{e:denom:2}) now follows by (\ref{e:gmeasbd:2})  and the fact that:
	\begin{align*}
	 \sup_z | \nu_{d^*}( B_{\sqrt {z \pm \varepsilon_n}} \cap (T_{osn}- \mf T \mb V_n)  ) - \nu_{d^*}( B_{\sqrt {z }} \cap (T_{osn}- \mf T \mb V_n)  ) |
	& \leq \sup_z  |F_{\chi^2_{d^*}}(z\pm \varepsilon_n) - F_{\chi^2_{d^*}}(z) | = o(1)
	\end{align*}
	since $\nu_{d^*}( B_{\sqrt {z }}  ) = F_{\chi^2_{d^*}}(z)$. This completes the proof of result (\ref{e:c:post:1}).

Part (i) follows by combining (\ref{e:c:post:1}) and the inequality:
\begin{equation} \label{e:c:post:2}
 \sup_z \left( \p_Z \Big( \|Z\|^2 \leq z  \Big| Z \in T - \mf T v\Big) - \p_Z (\| \mf T Z \|^2 \leq z) \right) \leq 0 \quad \mbox{for all $v \in \mb R^{d^*}$}
\end{equation}
(see Theorem 2 in \cite{ChenGao}). Part (ii) also follows from (\ref{e:c:post:1}) by observing that if $T = \mb R^{d^*}$ then $T - \mb V_n = \mb R^{d^*}$.
\end{proof}

\begin{proof}[\textbf{Proof of Theorem \ref{t:main}}]
	We verify the conditions of Lemma \ref{l:basic}. We may assume without loss of generality that $L_n(\hat \theta) = \sup_{\theta \in \Theta_{osn}} L_n(\theta) + o_\p(n^{-1})$ because $\wh \Theta_\alpha$ does not depend on the precise $\hat \theta$ used (cf. Remark \ref{rmk:full}). By Lemma \ref{l:quad} we have:
\[
\sup_{\theta \in \Theta_I} Q_n(\theta)  =  \|\mf T \mb V_n\|^2 + o_\p(1) \rightsquigarrow  \| \mf T Z\|^2
\]
with $Z \sim N(0,I_{d^*})$  when $\Sigma=I_{d^*}$. Let $z_{\alpha}$ denote the $\alpha$ quantile of the distribution of $ \| \mf T Z\|^2$.

For part (i), Lemma \ref{l:post}(i) shows that the posterior distribution of the QLR asymptotically (first-order) stochastically dominates the distribution of $ \| \mf T Z\|^2$ which implies that $\xi_{n,\alpha}^{post} \geq z_\alpha + o_\p(1)$. Therefore:
\begin{align*}
 \xi_{n,\alpha}^{mc} & = z_\alpha + (\xi_{n,\alpha}^{post} - z_\alpha) + (\xi_{n,\alpha}^{mc} - \xi_{n,\alpha}^{post})
 \geq z_\alpha + (\xi_{n,\alpha}^{mc} - \xi_{n,\alpha}^{post}) + o_\p(1)  = z_\alpha + o_\p(1)
\end{align*}
where the final equality is by Assumption \ref{a:mcmc}.

For part (ii), when $T = \mb R^{d^*}$ and $\Sigma=I_{d^*}$, we have:
	\begin{align*}
	\sup_{\theta \in \Theta_I} Q_n(\theta)  = \|\mb V_n \|^2 + o_\p(1)  \rightsquigarrow \chi^2_{d^*}\,,~~\mbox{and hence}~~z_{\alpha} = \chi^2_{d^*,\alpha}\,.
	\end{align*}
	Further:
	\begin{align*}
	\xi_{n,\alpha}^{mc} =  \chi^2_{d^*,\alpha} + (\xi_{n,\alpha}^{post} - \chi^2_{d^*,\alpha}) + (\xi_{n,\alpha}^{mc} - \xi_{n,\alpha}^{post})  = \chi^2_{d^*,\alpha} + o_\p(1)
	\end{align*}
	by Lemma \ref{l:post}(ii) and Assumption \ref{a:mcmc}.
\end{proof}

\begin{lemma}\label{l:quad:prime}
	Let Assumptions \ref{a:rate}(i) and \ref{a:quad:prime}' hold. Then:
	\begin{align}
	\sup_{\theta \in \Theta_{osn}} \left| Q_n(\theta) - \left( \|\sqrt n \gamma(\theta) - \mf T \mb V_n\|^2 +2 f_{n,\bot}(\gamma_\bot(\theta)) \right) \right|  & = o_\p(1) \label{e:quad:1:prime} \,.
	\end{align}
And hence $ \sup_{\theta \in \Theta_I} Q_n(\theta) = \|\mf T \mb V_n\|^2 + o_\p(1)$.
\end{lemma}

\begin{proof}[\textbf{Proof of Lemma \ref{l:quad:prime}}]
	Using Assumptions \ref{a:rate}(i) and \ref{a:quad:prime}', we obtain:
	\begin{align}
	2nL_n(\hat \theta)
	& = \sup_{\theta \in \Theta_{osn}}\left( 2\ell_n + \|\mf T \mb V_n\|^2 - \| \sqrt n \gamma(\theta) - \mf T \mb V_n \|^2 - 2 f_{n,\bot}(\gamma_\bot(\theta))\right)  + o_\p(1)  \notag \\
	& =    2 \ell_n +  \|\mf T \mb V_n\|^2 - \inf_{t \in T_{osn}} \| t - \mf T \mb V_n \|^2   -\inf_{\theta \in \Theta_{osn}}2 f_{n,\bot}(\gamma_\bot(\theta)) + o_\p(1) \notag \\
	& =  2 \ell_n +   \|\mf T \mb V_n\|^2  + o_\p(1) \label{e:lnhattheta:prime} \,,
	\end{align}
because $\mf T \mb V_n \in T$ and $ f_{n,\bot}(\cdot ) \geq 0$ with $f_{n,\bot}(0 )=0$, $\gamma_\bot(\theta)=0$ for all $\theta \in \Theta_I$ thus:
\[
0\leq \inf_{\theta \in \Theta_{osn}}f_{n,\bot}(\gamma_\bot(\theta)) \leq f_{n,\bot}(\gamma_\bot(\bar \theta))=0~~\mbox{for any}~\bar \theta \in \Theta_I\,.
\]
 Then by Assumption \ref{a:quad:prime}'(i) and definition of $Q_n$, we obtain:
	\begin{align*}
	Q_n(\theta)  & =  2 \ell_n +   \|\mf T \mb V_n\|^2  + o_\p(1)  - \left(  2\ell_n +  \|\mf T \mb V_n\|^2 -   \| \sqrt n \gamma(\theta) - \mf T \mb V_n \|^2 - 2f_{n,\bot}(\gamma_\bot(\theta)) + o_\p(1) \right) \\
	& =    \| \sqrt n \gamma(\theta) - \mf T \mb V_n \|^2  +2 f_{n,\bot}(\gamma_\bot(\theta)) + o_\p(1)
	\end{align*}
	where the $o_\p(1)$ term holds uniformly over $\Theta_{osn}$. This proves expression (\ref{e:quad:1:prime}).

Since $\gamma(\theta) = 0$ and $\gamma_\bot (\theta) = 0$ for $\theta \in \Theta_I$, and $f_{n,\bot}(0)=0$ (almost surely), we therefore have $\sup_{\theta \in \Theta_I} Q_n(\theta) =\|\mf T \mb V_n\|^2 + o_\p(1)$.
\end{proof}

\begin{proof}[\textbf{Proof of Lemma \ref{l:post:prime}}]
We first show that relation (\ref{e:post:qlr:prime}) holds.	By identical arguments to the proof of Lemma \ref{l:post}, it is enough to characterize the large-sample behavior of $R_n (z)$ defined in (\ref{e-Rn}).
	By Lemma \ref{l:quad:prime} and Assumption \ref{a:quad:prime}', there exists a positive sequence $(\varepsilon_n)_{n \in \mb N}$ independent of $z$ with $\varepsilon_n = o(1)$ such that:
	\begin{align*}
	\sup_{\theta \in \Theta_{osn}} \left| Q_n(\theta) - \left( \|\sqrt n \gamma(\theta) - \mf T \mb V_n\|^2 +2 f_{n,\bot}(\gamma_\bot(\theta)) \right) \right| & \leq \varepsilon_n  \\
	\sup_{\theta \in \Theta_{osn}} \left| nL_n(\theta) - \ell_n - \frac{1}{2}\|\mf T \mb V_n\|^2  + \frac{1}{2} \|\sqrt n \gamma(\theta) - \mf T \mb V_n\|^2 + f_{n,\bot}(\gamma_\bot(\theta)) \right| & \leq \varepsilon_n
	\end{align*}
	both hold wpa1. Also note that for any $z$, we have
	\begin{align*}
	\left\{\theta \in \Theta_{osn} :   \|\sqrt n \gamma(\theta) - \mf T \mb V_n\|^2 +2 f_{n,\bot}(\gamma_\bot(\theta)) \pm \varepsilon_n \leq z \right\}
	& \subseteq \left\{\theta \in \Theta_{osn} :  \|\sqrt n \gamma(\theta) - \mf T \mb V_n\|^2 \pm \varepsilon_n \leq z \right\}
	\end{align*}
	because $f_{n,\bot}(\cdot ) \geq 0$. Therefore, wpa1 we have:
	\begin{align*}
	R_n(z)
	& \leq  e^{2\varepsilon_n} \frac{\int_{\{\theta:\|\sqrt n \gamma(\theta) - \mf T \mb V_n\|^2 \leq  z+  \varepsilon_n\} \cap \Theta_{osn}}\! e^{ - \frac{1}{2}\|\sqrt n \gamma(\theta) - \mf T \mb V_n\|^2 - f_{n,\bot}(\gamma_\bot(\theta))} \mr d\Pi(\theta)}{\int_{\Theta_{osn}} \!e^{ - \frac{1}{2} \|\sqrt n \gamma(\theta) -\mf T \mb V_n\|^2 - f_{n,\bot}(\gamma_\bot(\theta)) } \mr d\Pi(\theta)}
	\end{align*}
	uniformly in  $z$.
	Define $\Gamma_{osn} = \{\gamma(\theta) : \theta \in \Theta_{osn}\}$ and $\Gamma_{\bot,osn} = \{\gamma_\bot(\theta) : \theta \in \Theta_{osn}\}$. By similar arguments to the proof of Lemma \ref{l:post}, Assumption \ref{a:prior:prime}'(ii) and a change of variables yield:
	\[
	R_n(z)
	\leq e^{2 \varepsilon_n}(1+\bar \varepsilon_n) \frac{\int_{(\{ \gamma : \| \sqrt n \gamma - \mf T \mb V_n \|^2 \leq z  + \varepsilon_n\} \cap \Gamma_{osn} ) \times  \Gamma_{\bot,osn} }\! e^{ - \frac{1}{2} \|\sqrt n \gamma - \mf T \mb V_n\|^2 - f_{n,\bot}(\gamma_\bot)} \mr d( \gamma,\gamma_\bot)}{\int_{ \Gamma_{osn} \times \Gamma_{\bot,osn}} \!e^{ - \frac{1}{2} \|\sqrt n \gamma - \mf T \mb V_n\|^2 - f_{n,\bot}(\gamma_\bot)}  \mr d( \gamma,\gamma_\bot)}
	\]
	which holds uniformly in $z$ (wpa1) for some $\bar \varepsilon_n = o(1)$. By Tonelli's theorem and Assumption \ref{a:quad:prime}'(ii), the preceding inequality becomes:
	\[
	R_n(z) \leq e^{2 \varepsilon_n}(1+\bar \varepsilon_n)   \frac{\int_{(\{ \gamma : \| \sqrt n \gamma - \mf T \mb V_n \|^2 \leq z + \varepsilon_n) \cap \Gamma_{osn}  }\! e^{ - \frac{1}{2} \|\sqrt n \gamma - \mf T \mb V_n\|^2 } \mr d \gamma }{\int_{ \Gamma_{osn} } \!e^{ - \frac{1}{2} \|\sqrt n \gamma - \mf T \mb V_n\|^2 }  \mr d \gamma}  \,.\label{e:rbd:prime:1}
	\]
	The rest of the proof of inequality (\ref{e:post:qlr:prime}) follows  by similar arguments to the proof of Lemma \ref{l:post}. The conclusion now follows by combining inequalities (\ref{e:post:qlr:prime}) and (\ref{e:c:post:2}).
\end{proof}

\begin{proof}[\textbf{Proof of Theorem \ref{t:main:prime}}]
	We verify the conditions of Lemma \ref{l:basic}. Again, we assume wlog that $L_n(\hat \theta) = \sup_{\theta \in \Theta_{osn}} L_n(\theta) + o_\p(n^{-1})$. By Lemma \ref{l:quad:prime}, when $\Sigma=I_{d^*}$, we have:
	\begin{align} \label{e:dist:prime}
	\sup_{\theta \in \Theta_I} Q_n(\theta)  = \|\mf T \mb V_n \|^2 + o_\p(1)  \rightsquigarrow \| \mf T Z\|^2
	\end{align}
where $Z \sim N(0,I_{d^*})$.
Lemma \ref{l:post:prime} shows that
the posterior distribution of the QLR asymptotically (first-order) stochastically dominates the $F_T$ distribution. The result follows by the same arguments as the proof of Theorem \ref{t:main}(i).
\end{proof}

\begin{lemma}\label{l:quad:profile}
	Let Assumptions \ref{a:rate}(i) and \ref{a:quad} or \ref{a:quad:prime}' and \ref{a:qlr:profile} hold. Then:
	\begin{align*}
	& \sup_{\theta \in \Theta_{osn}} \left| PQ_n ( M (\theta)) - f \left( \mf T \mb V_n - \sqrt n \gamma(\theta) \right)  \right| = o_\p(1) \,.
	\end{align*}
\end{lemma}

\begin{proof}[\textbf{Proof of Lemma \ref{l:quad:profile}}]
	By display (\ref{e:lnhattheta}) in the proof of Lemma \ref{l:quad} or display (\ref{e:lnhattheta:prime}) in the proof of Lemma \ref{l:quad:prime} and Assumption \ref{a:qlr:profile}, we obtain:
	\begin{align*}
	PQ_n (  M (\theta))
	& = 2nL_n(\hat \theta) - 2n PL_n (  M (\theta)) \\
	& =  2\ell_n +  \|\mf T \mb V_n\|^2  - \bigg( 2 \ell_n + \|\mf T \mb V_n\|^2  -  f \left( \mf T \mb V_n -  \sqrt n \gamma(\theta)  \right) \bigg) + o_\p(1)
	\end{align*}
	where the $o_\p(1)$ term holds uniformly over $\Theta_{osn}$.
\end{proof}

\begin{proof}[\textbf{Proof of Lemma \ref{l:post:profile}}]
	We prove the result under Assumptions  \ref{a:rate}, \ref{a:quad:prime}', \ref{a:prior:prime}', and \ref{a:qlr:profile}. The proof under Assumptions  \ref{a:rate}, \ref{a:quad}, \ref{a:prior}, and \ref{a:qlr:profile} follows similarly. By the same arguments as the proof of Lemma \ref{l:post}, it suffices to characterize the large-sample behavior of:
	\begin{equation} \label{e-Rn-profile}
	R_{n}(z)  := \frac{\int_{\{\theta:PQ_n (  M (\theta)) \leq  z \} \cap \Theta_{osn}}\! e^{nL_n(\theta)  }\,\mr d \Pi(\theta)}{\int_{\Theta_{osn}} \!e^{nL_n(\theta)  }\,\mr d \Pi(\theta)} \,.
	\end{equation}
	By Lemma \ref{l:quad:profile} and Assumption \ref{a:quad:prime}', there exists a positive sequence $(\varepsilon_n)_{n \in \mb N}$ independent of $z$ with $\varepsilon_n = o(1)$ such that the inequalities:
	\begin{align*}
	\sup_{\theta \in \Theta_{osn}} \left| PQ_n (  M (\theta))  - f(\mf T \mb V_n - \sqrt n \gamma(\theta)) \right| & \leq \varepsilon_n  \\
	\sup_{\theta \in \Theta_{osn}} \left| n L_n(\theta) - \ell_n - \frac{1}{2}\|\mf T \mb V_n\|^2 - \left( - \frac{1}{2} \|\sqrt n \gamma(\theta) - \mf T \mb V_n\|^2  - f_{n,\bot}(\gamma_\bot(\theta))  \right) \right| & \leq \varepsilon_n
	\end{align*}
	both hold wpa1. 	Therefore, wpa1 we have:
	\begin{align*}
	&  e^{-2 \varepsilon_n} \frac{\int_{\{\theta: f( \mf T \mb V_n - \sqrt n \gamma(\theta) ) \leq  z - \varepsilon_n\} \cap \Theta_{osn}}\! e^{ - \frac{1}{2} \|\sqrt n \gamma(\theta) - \mf T \mb V_n\|^2  - f_{n,\bot}(\gamma_\bot(\theta)) } \mr d\Pi(\theta)}{\int_{\Theta_{osn}} \!e^{ - \frac{1}{2} \|\sqrt n \gamma(\theta) - \mf T \mb V_n\|^2  - f_{n,\bot}(\gamma_\bot(\theta)) } \mr d\Pi(\theta)} \notag \\
	& \leq R_n(z)
	\leq  e^{2\varepsilon_n} \frac{\int_{\{\theta: f( \mf T \mb V_n - \sqrt n \gamma(\theta)  ) \leq  z +  \varepsilon_n\} \cap \Theta_{osn}}\! e^{ - \frac{1}{2}\|\sqrt n \gamma(\theta) - \mf T \mb V_n\|^2 - f_{n,\bot}(\gamma_\bot(\theta)) } \mr d\Pi(\theta)}{\int_{\Theta_{osn}} \!e^{ - \frac{1}{2} \|\sqrt n \gamma(\theta) - \mf T \mb V_n\|^2 - f_{n,\bot}(\gamma_\bot(\theta)) } \mr d\Pi(\theta)}
	\end{align*}
	uniformly in $z$. By similar arguments to the proof of Lemma \ref{l:post:prime}, we may use the change of variables $\theta \mapsto ( \gamma(\theta),\gamma_\bot(\theta)) $, continuity of $\pi_{\Gamma^*}$ (Assumption \ref{a:prior:prime}'(ii)), and Tonelli's theorem to restate the preceding inequalities as:
	\begin{align*}
	&  (1-\bar \varepsilon_n) e^{-2 \varepsilon_n} \frac{\int_{\{  \gamma :  f(\mf T \mb V_n - \sqrt n \gamma) \leq z - \varepsilon_n  \} \cap \Gamma_{osn} } e^{ - \frac{1}{2} \|\sqrt n \gamma - \mf T {\mb V_n}\|^2} \mr d \gamma}{\int_{\Gamma_{osn}} \!e^{ - \frac{1}{2} \|\sqrt n \gamma - \mf T \mb V_n \|^2 } \mr d \gamma}   \notag \\
	& \leq R_n(z)
	\leq (1+\bar \varepsilon_n) e^{2 \varepsilon_n} \frac{\int_{\{  \gamma :  f(\mf T \mb V_n - \sqrt n  \gamma) \leq z + \varepsilon_n  \} \cap \Gamma_{osn} } e^{ - \frac{1}{2} \|\sqrt n \gamma - \mf T {\mb V_n}\|^2} \mr d \gamma}{\int_{\Gamma_{osn}} \!e^{ - \frac{1}{2} \|\sqrt n \gamma - \mf T \mb V_n \|^2 } \mr d \gamma}
	\end{align*}
	which holds (wpa1) for some $\bar \varepsilon_n = o(1)$. Let $f^{-1}(z) = \{ \kappa \in \mb R^{d^*} : f(\kappa) \leq z\}$. A second change of variables $\mf T \mb V_n - \sqrt n  \gamma \mapsto \kappa$ yields:
	\begin{align*}
	&  (1-\bar \varepsilon_n) e^{-2 \varepsilon_n} \frac{\nu_{d^*} ( (f^{-1}( z - \varepsilon_n )) \cap (\mf T \mb V_n - T_{osn}))}{ \nu_{d^*} (\mf T \mb V_n - T_{osn})}  \\
	& \leq R_n(z)
	\leq  (1+\bar \varepsilon_n) e^{2 \varepsilon_n} \frac{\nu_{d^*} ( (f^{-1}( z + \varepsilon_n ) )\cap (\mf T \mb V_n - T_{osn}))}{ \nu_{d^*} (\mf T \mb V_n - T_{osn})}
	\end{align*}
	uniformly in $z$, where it should be understood that $\mf T \mb V_n - T_{osn}$ is the Minkowski sum $\mf T \mb V_n + (-T_{osn})$ with $-T_{osn} = \{-\kappa : \kappa \in T_{osn}\}$.
	
	The remainder of the proof follows by similar arguments to the proof of Lemma \ref{l:post}, noting that
	\begin{align*}
	& \sup_{z \in I} \left|\frac{\nu_{d^*}( (f^{-1}(z \pm \varepsilon_n)) \cap (\mf T \mb V_n - T))}{\nu_{d^*}( \mf T \mb V_n - T)} - \p_{Z | \mf X_n} \Big( f(Z) \leq z  \Big| Z \in \mf T \mb V_n - T \Big) \right| \\
	& \leq \sup_{z \in I} |\nu_{d^*}(f^{-1}(z \pm \varepsilon_n)) - \nu_{d^*}(f^{-1}(z))| = o(1)
	\end{align*}
	where the final equality is by uniform continuity of bounded, monotone continuous functions.
\end{proof}

\begin{proof}[\textbf{Proof of Theorem \ref{t:main:profile}}]
	We verify the conditions of Lemma \ref{l:basic:profile}.  Again, we assume wlog that $L_n(\hat \theta) = \sup_{\theta \in \Theta_{osn}} L_n(\theta) + o_\p(n^{-1})$.
	
	To prove Theorem \ref{t:main:profile}(i), let $\xi_{\alpha}$ denote the $\alpha$ quantile of $f( \mf T Z )$. By Lemma  \ref{l:post:profile} we need to show that
	\[
	\p_Z ( f(Z) \leq z  | Z \in \mf T v - T ) \leq \p_Z (f( \mf T Z ) \leq z)
	\]
	holds for all $v \in \mb R^{d^*}$ and all $z$ in a neighborhood of $\xi_\alpha$. To prove this, i is sufficient to show that
	\[
	 \nu_{d^*}( f^{-1}(z) \cap (\mf T v - T))  \leq \nu_{d^*}( \mf T v - T ) \times \nu_{d^*} ( \{ \kappa \in \mb R^{d^*} : f(\mf T \kappa) \leq z\})
	\]
	holds for each $z$ and each $v \in \mb R^{d^*}$. But notice that
	\[
	 \nu_{d^*}( f^{-1}(z) \cap (\mf T v - T)) \leq \nu_{d^*}( (f^{-1}(z)- T^o) \cap (\mf T v - T))  \leq \nu_{d^*}( f^{-1}(z)- T^o) \times \nu_{d^*} (\mf T v - T)
	\]
	where the first inequality is because $f^{-1}(z) \subseteq f^{-1}(z) - T^o = \{\kappa_1 + \kappa_2 : \kappa_1 \in f^{-1}(z) , -\kappa_2 \in T^o\}$ since $0 \in T^o$ and the second inequality is by Theorem 1 of \cite{ChenGao} (taking  $A = \{ \mf T v\}$, $B = f^{-1}(z)$, $C = -T$ and $D = -T^o$ in their notation). 
	
	Whenever $\nu_{d^*}( f^{-1}(z)- T^o) \leq \nu_{d^*} ( \{ \kappa \in \mb R^{d^*} : f(\mf T \kappa) \leq z\})$ holds (which it does, in particular, when $f$ is subconvex), we therefore have:
	\begin{align*}
 \xi_{n,\alpha}^{mc,p} & = \xi_\alpha + (\xi_{n,\alpha}^{post,p} - \xi_\alpha) + (\xi_{n,\alpha}^{mc,p} - \xi_{n,\alpha}^{post,p})
 \geq \xi_\alpha + (\xi_{n,\alpha}^{mc,p} - \xi_{n,\alpha}^{post,p}) + o_\p(1)  = \xi_\alpha + o_\p(1)
\end{align*}
where the final equality is by Assumption \ref{a:mcmc:profile}.
	
	To prove Theorem \ref{t:main:profile}(ii), when $T = \mb R^{d^*}$ we have $PQ_n (  M_I)  \rightsquigarrow f ( Z )$. Let $\xi_{\alpha}$ denote the $\alpha$ quantile of $f( Z )$. Then:
	\[
	\xi_{n,\alpha}^{mc} = \xi_\alpha + ( \xi_{n,\alpha}^{post} - \xi_\alpha) + (\xi_{n,\alpha}^{mc} - \xi_{n,\alpha}^{post}) = \xi_\alpha + o_\p(1)
	\]
	by Lemma \ref{l:post:profile} and Assumption \ref{a:mcmc:profile}.
\end{proof}

\begin{proof}[\textbf{Proof of Theorem \ref{t:profile:chi}}]
	By Lemma \ref{l:basic:profile}, it is enough to show that $\Pr(W^* \leq w) \geq F_{\chi^2_1}(w)$ holds for $w \geq 0$, where $W^* = \max_{i \in \{1,2\}} \inf_{t \in T_i} \| Z - t\|^2$.

\textbf{Case 1: $d^* = 1$.} Wlog let $T_1 = [0,\infty)$ and $T_1^o = (-\infty,0]$. If $T_2 = T_1$ then $\mf T_1^o Z = \mf T_2^o Z = (Z \wedge 0)$ so $W^* = (Z \wedge 0)^2 \leq Z^2 \sim \chi^2_1$. If $T_2 = T_1^o$ then $\mf T_1^o Z = (Z \wedge 0)$ and $\mf T_2^o Z = (Z \vee 0)$, so $W^* = Z^2 \sim \chi^2_1$. In either case, we have: $\Pr(W^* \leq w) \geq F_{\chi^2_1}(w)$ for any $w \geq 0$.

\textbf{Case 2: $d^* = 2$.} Wlog let $T_1 = \{(x,y) : y \leq 0\}$ then $T_1^o$ is the positive $y$-axis. Let $Z = (X,Y)'$. If $T_1 = T_2$ then $\mf T_1^o Z = \mf T_2^o Z = (Y \vee 0)$, so $W^* = (Y \vee 0)^2 \leq Y^2 \sim \chi^2_1$. If $T_2 = \{(x,y) : y \geq 0\}$ then $T_2^o$ is the negative $y$-axis. So, in this case, $\mf T_1^o Z = (Y \vee 0)$, $\mf T_2^o = (Y \wedge 0)$ and so $W^* = Y^2 \sim \chi^2_1$.

Now let $T_2$ be the rotation of $T_1$ by $\varphi \in (0,\pi)$ radians. This is plotted in Figure \ref{f:cones} for $\varphi \in (0,\pi/2)$ (left panel) and $\varphi \in (\frac{\pi}{2},\pi)$ (right panel). The axis of symmetry is the line $y = -x \cot (\frac{\varphi}{2})$, which bisects the angle between $T_1^o$ and $T_2^o$.

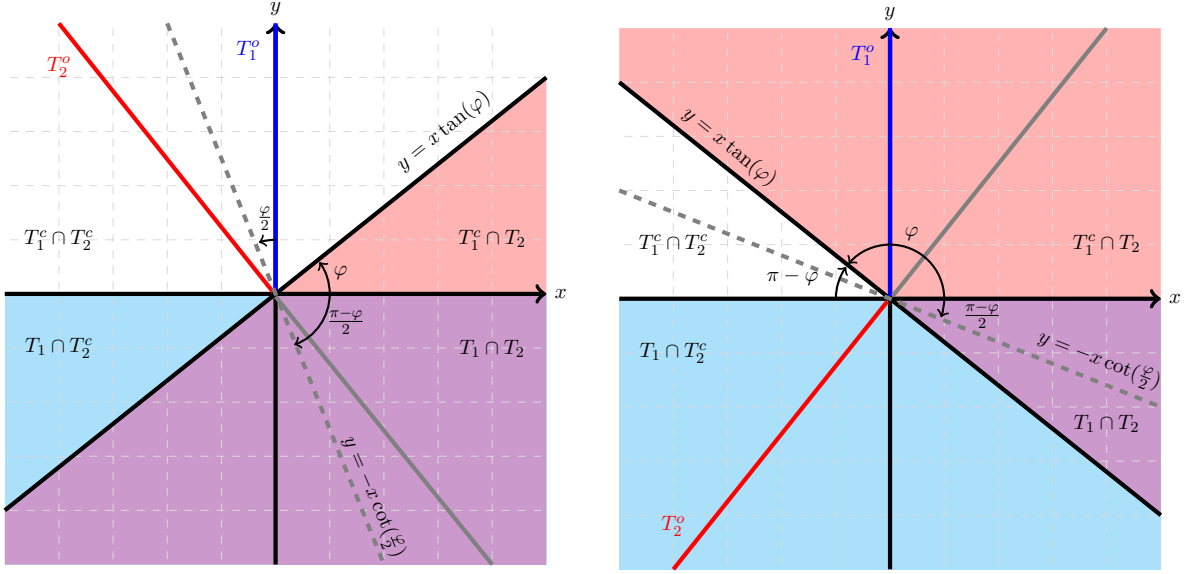
\begin{figure}[t]
\begin{center}
\begin{tikzpicture}[thick,scale=.72, every node/.style={transform shape}]
\fill[cyan!30!white] (-5,0) rectangle (5,-5);
   \node (r0) at (0,0) {}; 
   \node (u1) at (5,4) {}; 
   \node (u2) at (5,0) {}; 
\fill[red!30] (r0.center)--(u1.center)--(u2.center);
   \node (l1) at (-5,-4) {}; 
   \node (l2) at (-5,-5) {}; 
   \node (l3) at (0,-5) {}; 
\fill[violet!40] (r0.center)--(l1.center)--(l2.center)--(l3.center);
\fill[violet!40!white] (0,0) rectangle (5,-5);
\draw[help lines, color=gray!30, dashed] (-4.9,-4.9) grid (4.9,4.9);
\draw[->,ultra thick] (-5,0)--(5,0) node[right]{$x$};
\draw[->,ultra thick] (0,-5)--(0,5) node[above]{$y$};
\draw[ultra thick,domain=-5:5,smooth,variable=\x,black] plot ({\x},{0.8*\x});
\draw[-,red,ultra thick] (-4,5)--(-0.05,0.04);
\draw[-,blue,ultra thick] (0,0.05)--(0,5);
\draw[-,gray,ultra thick] (0,-0.05)--(4,-5);
\draw[dashed,gray,ultra thick] (-2,5)--(2,-5);
\node[red] at (-4,4.2){$T_2^o$};
\node[blue] at (-0.5,4.5){$T_1^o$};
\node[black] at (-4,-1){$T_1 \cap T_2^c$};
\node[black] at (4,1){$T_1^c \cap T_2$};
\node[black] at (4,-1){$T_1 \cap T_2$};
\node[black] at (-4,1){$T_1^c \cap T_2^c$};
\draw[->] (1,0) arc (0:36:1);
\node[black] at (1.2,0.4){$\varphi$};
\draw[->] (0,1) arc (90:108:1);
\node[black] at (-0.2,1.4){$\frac{\varphi}{2}$};
\draw[->] (1,0) arc (0:-70:1);
\node[black] at (1.3,-0.5){$\frac{\pi-\varphi}{2}$};
\node[label={[label distance=0cm,text depth=0,rotate=39]right:$y = x\tan (\varphi)$}] at (2,2.3) {};
\node[label={[label distance=0cm,text depth=0,rotate=-67]right:$y = -x \cot(\frac{\varphi}{2})$}] at (1.2,-2.5) {};
\end{tikzpicture}
\quad
\begin{tikzpicture}[thick,scale=0.72, every node/.style={transform shape}]
\fill[cyan!30!white] (-5,0) rectangle (5,-5);
   \node (r0) at (0,0) {}; 
   \node (l1) at (5,-4) {}; 
   \node (l2) at (5,0) {}; 
\fill[violet!40] (r0.center)--(l1.center)--(l2.center);
   \node (u1) at (-5,4) {}; 
   \node (u2) at (-5,5) {}; 
   \node (u3) at (5,5) {}; 
   \node (u4) at (5,0) {}; 
\fill[red!30] (r0.center)--(u1.center)--(u2.center)--(u3.center)--(u4.center);
\draw[help lines, color=gray!30, dashed] (-4.9,-4.9) grid (4.9,4.9);
\draw[->,ultra thick] (-5,0)--(5,0) node[right]{$x$};
\draw[->,ultra thick] (0,-5)--(0,5) node[above]{$y$};
\draw[ultra thick,domain=-5:5,smooth,variable=\x,black] plot ({\x},{-0.8*\x});
\draw[-,red,ultra thick] (-4,-5)--(-0.05,-0.05);
\draw[-,blue,ultra thick] (0,0.05)--(0,5);
\draw[-,gray,ultra thick] (0.05,0.05)--(4,5);
\draw[dashed,gray,ultra thick] (-5,2)--(5,-2);
\node[red] at (-4,-4.2){$T_2^o$};
\node[blue] at (-0.5,4.5){$T_1^o$};
\draw[->] (1,0) arc (0:140:1);
\node[black] at (0.4,1.2){$\varphi$};
\draw[->] (-1,0) arc (180:144:1);
\node[black] at (-1.8,0.4){$\pi-\varphi$};
\draw[->] (1,0) arc (0:-20:1);
\node[black] at (1.7,-0.3){$\frac{\pi-\varphi}{2}$};
\node[black] at (4,-2.3){$T_1 \cap T_2$};
\node[black] at (-4,1){$T_1^c \cap T_2^c$};
\node[black] at (-4,-1){$T_1 \cap T_2^c$};
\node[black] at (4,1){$T_1^c \cap T_2$};
\node[label={[label distance=0cm,text depth=0,rotate=-39]right:$y = x\tan (\varphi)$}] at (-4,3.6) {};
\node[label={[label distance=0cm,text depth=0,rotate=-21]right:$y = -x \cot(\frac{\varphi}{2})$}] at (2.5,-0.6) {};
\end{tikzpicture}
\end{center}
\centering

\parbox{12cm}{\caption{\small\label{f:cones} Cones and polar cones for the proof of Theorem \ref{t:profile:chi}.} }
\end{figure}

Suppose $Z = (X,Y)'$ lies in the half-space $Y \geq -X \cot(\frac{\varphi}{2})$. There are three options:
\begin{itemize}
\item $Z \in (T_1^{\phantom c} \cap T_2^{\phantom c})$ (purple region): $\mf T_1^o Z = 0$, $\mf T_2^o Z = 0$, so $W^* = 0$
\item $Z \in (T_1^c \cap T_2^{\phantom c})$ (red region): $\mf T_1^o Z = (0,Y)'$, $\mf T_2^o Z = 0$, so $W^* = Y^2$
\item $Z \in (T_1^c \cap T_2^c)$ (white region): $\mf T_1^o Z = (0,Y)'$. To calculate $\mf T_2^o Z$, observe that if we rotate about the origin by $-\varphi$ then the polar cone $\mf T_2^o$ becomes the positive $y$ axis.
Under the rotation, $\mf T_2^o Z = (0,Y^*)$ where $Y^*$ is the $y$-value of the rotation of $(X,Y)$ by negative $\varphi$. The point $(X,Y)$ rotates to $(X \cos\varphi + Y \sin\varphi, Y \cos\varphi - X \sin\varphi)$, so we get $\| \mf T_2^o Z\|^2 = (Y \cos\varphi - X \sin\varphi)^2$. We assumed $Y \geq -X \cot (\frac{\varphi}{2})$. By the half-angle formula $\cot (\frac{\varphi}{2}) = \frac{\sin \varphi}{1-\cos \varphi}$, this means that $Y \geq Y \cos \varphi - X \sin \varphi$. But $Y \cos \varphi - X \sin \varphi \geq 0$ as $Y \geq X \tan \varphi$. Therefore, $(Y \cos\varphi - X \sin\varphi)^2 \leq Y^2$ and so $W^* = Y^2$.
\end{itemize}

We have shown that $W^* \leq Y^2$ whenever $Y \geq -X \cot (\frac{\varphi}{2})$. Now, for any $w \geq 0$:
\begin{align} \label{e-p3pf-1}
 \Pr(W^* \leq w | Y \geq -X\cot({\textstyle \frac{\varphi}{2}})) & \geq \Pr(Y^2 \leq w |Y \geq -X\cot({\textstyle \frac{\varphi}{2}}))  = \Pr(Y^2 \leq w | V \geq 0)
\end{align}
where $V = Y \sin (\frac{\varphi}{2}) + X \cos(\frac{\varphi}{2})$. Note that $Y$ and $V$ are jointly normal with mean $0$, unit variance, and correlation $\rho = \sin(\frac{\varphi}{2})$.
The pdf of $Y$ given $V \geq 0$ is:
\begin{align*}
 f(y|V \geq 0) & = \frac{\int_0^\infty f_{Y|V}(y|v) f_V(v) \mr dv}{\int_0^\infty f_V(v) \mr dv }
  = 2 f_Y(y)  (1 - F_{V|Y}(0|y)) \,.
\end{align*}
As $V|Y=y \sim N(\rho y , (1-\rho^2) )$, we have:
\[
 F_{V|Y}(0|y) = \Phi \Big( \frac{- \rho y}{\sqrt{1-\rho^2}} \Big) = 1- \Phi \Big( \frac{ \rho}{\sqrt{1-\rho^2}} y \Big)
\]
and so
\[
 f(y|V \geq 0) = 2 \phi(y) \Phi \Big( \frac{ \rho}{\sqrt{1-\rho^2}} y \Big)\,.
\]
Therefore:
\begin{align}
 \Pr(Y^2 \leq w | V \geq 0) & = \Pr(-\sqrt w \leq y \leq \sqrt w  | V \geq 0)
  = \int_{-\sqrt w}^{\sqrt w} 2 \phi(y) \Phi \Big( \frac{ \rho}{\sqrt{1-\rho^2}} y \Big) \, \mr dy \label{e-p3pf-2}  \,.
\end{align}
But differentiating the right-hand side of (\ref{e-p3pf-2}) with respect to $\rho$ gives:
\[
 \frac{\mr d}{\mr d \rho} \int_{-\sqrt w}^{\sqrt w} 2 \phi(y) \Phi \Big( \frac{ \rho}{\sqrt{1-\rho^2}} y \Big) \, \mr dy =  \frac{1}{(1-\rho^2)^{3/2}} \int_{-\sqrt w}^{\sqrt w} 2 y \phi(y) \phi \Big( \frac{ \rho}{\sqrt{1-\rho^2}} y \Big) \, \mr d y = 0
\]
for any $\rho \in (-1,1)$, because $y \phi(y) \phi \big( \rho y/\sqrt{1-\rho^2} \big)$ is an odd function. Therefore, the probability in display (\ref{e-p3pf-2}) doesn't depend on the value of $\rho$. Setting $\rho = 0$, we obtain:
\begin{align*}
 \Pr(Y^2 \leq w | V \geq 0) & = \int_{-\sqrt w}^{\sqrt w} 2 \phi(y) \Phi ( 0 ) \, \mr dy = \Phi(\sqrt w) - \Phi(-\sqrt w) = F_{\chi^2_1}(w) \,.
\end{align*}
Therefore, by inequality (\ref{e-p3pf-1}) we have:
\[
 \Pr(W^* \leq w | Y \geq -X\cot({\textstyle \frac{\varphi}{2}})) \geq F_{\chi^2_1}(w)\,.
\]
By symmetry, we also have $\Pr(W^* \leq w | Y < -X\cot({\textstyle \frac{\varphi}{2}})) \geq F_{\chi^2_1}(w)$. Therefore, we have shown that $\Pr(W^* \leq w) \geq F_{\chi^2_1}(w)$ holds for each $w \geq 0$.
A similar argument applies when $T_2$ is the rotation of $T_1$ by $\varphi \in (-\pi,0)$ radians. This completes the proof of the case $d^* = 2$.

\textbf{Case 3: $d^* \geq 3$.} As $T_1$ and $T_2$ are closed half-spaces we have $T_1 = \{z \in \mb R^{d^*} : a'z \leq 0\}$ and $T_2 = \{z \in \mb R^{d^*} : b'z \leq 0\}$ for some $a,b \in \mb R^{d^*}\setminus \{0\}$. The polar cones are the rays $T_1^o = \{ s  a : s \geq 0\}$ and $T_2^o = \{s b : s \geq 0\}$. There are three sub-cases to consider.

Case 3a: $a = s b$ for some $s > 0$. Let $u_a = \frac{a}{\|a\|}$. Here $T_1 = T_2$, $\mf T_1^o Z = \mf T_2^o Z = 0$ if $Z \in T_1$, and
\[
 \mf T_1^o Z = \mf T_2^o Z = u_a (Z'u_a) \quad \mbox{if $Z \not \in T_1$ (i.e. if $Z'u_a > 0$)}\,.
\]
Therefore, $W^* = (Z'u_a \vee 0)^2 \leq (Z'u_a)^2 \sim \chi^2_1$.

Case 3b: $a =s b$ for some $s < 0$. Here $T_1 = -T_2$ and $T_1^o = -T_2^o$, so $\mf T_1^o Z = 0$ and $\mf T_2^o Z = u_a(Z'u_a)$ if $Z \in T_1$ (i.e. if $Z'u_a \leq 0$) and $\mf T_1^o Z = u_a (Z'u_a)$ and $\mf T_2^o Z = 0$ if $Z \not \in T_1$ (i.e. if $Z'u_a > 0$). Therefore $W^* = (Z'u_a)^2\sim \chi^2_1$.

Case 3c: $a$ and $b$ are linearly independent. Without loss of generality,\footnote{By Gram-Schmidt, we can always define a new set of coordinate vectors $e_1,e_2,\ldots,e_{d^*}$ for $\mb R^{d^*}$ with $e_2=u_a$ and such that $b$ is in the span of $e_1$ and $e_2$.} we can take $T_1^o$ to be the positive $y$-axis (i.e. $a = (0, a_2, 0, \ldots, 0)'$ for some $a_2 > 0)$ and take $T_2^o$ to lie in the $(x,y)$-plane (i.e. $b = (b_1, b_2, 0 , \ldots, 0)'$ for some $b_1 \neq 0$).

Now write $Z = (X,Y,U)$ where $U \in \mb R^{d^*-2}$. Note that $a'Z = a_2 Y$ and $b'Z = b_1 X + b_2 Y$. So only the values of $X$ and $Y$ matter in determining whether or not $Z$ belongs to $T_1$ and $T_2$.

Without loss of generality we may assume that $(b_1,b_2)'$ is, up to scale, a rotation of $(0,a_2)'$ by $\varphi \in (0,\pi)$ (the case $(-\pi,0)$ can be handled by similar arguments, as in Case 2).

Suppose that $Y \geq -X \cot(\frac{\varphi}{2})$. As in Case 2, there are three options:
\begin{itemize}
\item $Z \in (T_1^{\phantom c} \cap T_2^{\phantom c})$: $\mf T_1^o Z = 0$, $\mf T_2^o Z = 0$, so $W^* = 0$
\item $Z \in (T_1^c \cap T_2^{\phantom c})$: $\mf T_1^o Z = (0,Y,0,\ldots,0)'$, $\mf T_2^o Z = 0$, so $W^* = Y^2$
\item $Z \in (T_1^c \cap T_2^c)$: $\|\mf T_1^o Z\|^2 = Y^2$ and $\| \mf T_2^o Z\|^2 = (Y \cos\varphi - X \sin\varphi)^2 \leq Y^2$, so $W^* = Y^2$.
\end{itemize}
Arguing as in Case 2, we obtain $\Pr(W^* \leq w | Y \geq -X\cot({\textstyle \frac{\varphi}{2}})) \geq F_{\chi^2_1}(w)$. By symmetry, we also have $\Pr(W^* \leq w | Y < -X\cot({\textstyle \frac{\varphi}{2}})) \geq F_{\chi^2_1}(w)$. Therefore, $\Pr(W^* \leq w) \geq F_{\chi^2_1}(w)$.
\end{proof}

\begin{proof}[\textbf{Proof of Proposition \ref{p:qlr:chi}}]
	It follows from condition (i) and display (\ref{e:lnhattheta}) or display (\ref{e:lnhattheta:prime}) that:
	\[
	2n L_n (\hat \theta) = 2\ell_n + \|\mb V_n\|^2  + o_\p(1) \,.
	\]
	Moreover, applying conditions (ii) and (iii), we obtain:
	\begin{align*}
	\inf_{\mu \in M_I} \sup_{\eta \in H_\mu} 2n L_n(\mu,\eta)
	& = \min_{\mu \in \{\ul \mu,\ol \mu\}} \sup_{\eta \in H_\mu} 2n L_n(\mu,\eta) + o_\p(1) \\
	& = \min_{\mu \in \{\ul \mu,\ol \mu\}} \left( 2\ell_n + \|\mb V_n\|^2 - \inf_{t \in T_\mu} \| \mb V_n - t \|^2 \right) + o_\p(1) \,.
	\end{align*}
	Therefore:
	\[
	\sup_{\mu \in M_I} \inf_{\eta \in H_\mu} Q_n(\mu,\eta) = \max_{\mu \in \{ \ul \mu,\ol \mu\}} \inf_{t \in T_{\mu}} \|\mb V_n - t\|^2  + o_\p(1) \,.
	\]
	The result now follows from $\Sigma = I_{d^*}$.
\end{proof}

\subsection{Proofs and Additional Lemmas for Section \ref{s:suff}}

\begin{proof}[\textbf{Proof of Proposition \ref{p:rfr}}]
	Wlog we can take $\tilde \gamma_0 = 0$. Also take $n$ large enough that $\{\tilde \gamma : \|\tilde \gamma\| \leq n^{-1/4}\} \subseteq U$. Then by condition (b), for any such $\tilde \gamma$ we have:
	\begin{align*}
	n L_n(\tilde \gamma) & =n L_n(\tilde \gamma_0) + (\sqrt n \tilde \gamma)'(\sqrt n \p_n \dot \ell_{\tilde \gamma_0}) + \frac{1}{2}(\sqrt n \tilde \gamma)' (\p_n \ddot \ell_{\tilde  \gamma^*}) (\sqrt n \tilde \gamma)
	\end{align*}
	where $\tilde \gamma^*$ is in the segment between $\tilde \gamma$ and $\tilde \gamma_0$ for each element of $\p_n \ddot \ell_{\tilde \gamma^*}$. We may deduce from Lemma 2.4 of \cite{NeweyMcFadden} that $\sup_{\tilde \gamma : \|\tilde \gamma \| \leq n^{-1/4}} \| (\p_n \ddot \ell_{\tilde \gamma^*})  - P_0(\ddot \ell_{\gamma_0})\| = o_\p(1)$ holds under conditions (a) and (b). Since this term is $o_\p(1)$, we can choose a positive sequence $(r_n)_{n \in \mb N}$ with $r_n \to \infty$, $r_n = o(n^{-1/4})$ such that $r_n^2 \sup_{\tilde \gamma : \|\tilde \gamma \| \leq n^{1/4}} \| (\p_n \ddot \ell_{\tilde \gamma^*})  - P_0(\ddot \ell_{\tilde \gamma_0})\| = o_\p(1)$. Assumption \ref{a:quad} then holds over $\Theta_{osn} = \{ \theta \in \Theta : \|\tilde \gamma(\theta)\| \leq r_n/\sqrt n\}$ with $\ell_n = n L_n(\tilde \gamma_0)$, $\gamma(\theta) = \mb I_{\tilde \gamma_0}^{1/2} \tilde \gamma(\theta)$, $\sqrt n \hat \gamma_n = \mb V_n =\mb I_{\tilde \gamma_0}^{-1/2} \mb P_n (\dot \ell_{\gamma_0})$, $\Sigma = I_{d^*}$.
	
	It remains to show that the posterior concentrates on $\Theta_{osn}$. Choose $\varepsilon$ sufficiently small that $U_\varepsilon = \{\tilde \gamma : \|\tilde \gamma\| < \varepsilon\} \subseteq U$. By a similar expansion to the above and condition (c), we have $D_{KL}(p_0\|q_{\tilde \gamma}) = -\frac{1}{2} \tilde \gamma' P_0(\ddot \ell_{\tilde \gamma^*}) \tilde \gamma$ where $\tilde \gamma^*$ is in the segment between $\tilde \gamma^*$ and $\tilde \gamma_0$. Since $\|P_0(\ddot \ell_{\tilde \gamma^*}) + \mb I_{\tilde \gamma_0}\| \to 0$ as $\|\tilde \gamma\| \to 0$, we may reduce $\varepsilon$ so that $\inf_{\tilde \gamma \in U_\varepsilon} \|P_0(\ddot \ell_{\tilde \gamma^*}) + \mb I_{\tilde \gamma_0}\| \leq \frac{1}{2} \lambda_{\min} (\mb I_{\tilde \gamma_0})$. On $U_\varepsilon$ we then have that there exist finite positive constants $\ul c$ and $\ol c$ such that $\ul c \|\tilde \gamma\|^2 \leq D_{KL}(p_0\|q_{\tilde \gamma}) \leq \ol c \|\tilde \gamma\|^2$. Also note that $\inf_{\tilde \gamma \in \wt \Gamma \setminus U_\varepsilon} D_{KL}(p_0 \| q_{\tilde \gamma})=: \delta$ with $\delta > 0$ by identifiability of $\tilde \gamma_0$, continuity of the map $\tilde \gamma \mapsto P_0 \ell_{\tilde \gamma}$, and compactness of $\wt \Gamma$. Standard consistency arguments (e.g. the Corollary to Theorem 6.1 in \cite{Schwartz1965})) then imply that $\Pi_n( U_{\varepsilon} | \mf X_n) \to_{a.s.} 1$. Since the posterior concentrates on $U_\varepsilon$ and $\Theta_{osn} \subset U_\varepsilon$ for all $n$ sufficiently large, it's enough to confine attention to $U_\varepsilon$. We have shown that $\ul c \|\tilde \gamma\|^2 \leq D_{KL}(p_0\|q_{\tilde \gamma}) \leq \ol c \|\tilde \gamma\|^2$ holds on $U_\varepsilon$. It now follows by the parametric Bernstein-von Mises theorem (e.g. Theorem 10.1 in \cite{vdV}) that the posterior contracts at a $\sqrt n$-rate, verifying Assumption \ref{a:rate}(ii).
\end{proof}

For the following lemma, let $(r_n)_{n \in \mb N}$ be a positive sequence with $r_n \to \infty$ and $r_n = o(n^{1/2})$, $\mc P_{osn} = \{ p \in \mc P : h(p,p_0) \leq r_n/\sqrt n\}$ and $\Theta_{osn} = \{ \theta \in \Theta : h(p_\theta,p_0) \leq r_n/\sqrt n\}$. For each $p \in \mc P$ with $p \neq p_0$, define $S_p = \sqrt{p/p_0} - 1$ and $s_p = {S_p}/{h(p,p_0)}$. Recall the definitions of $\ol{\mc D}_{\varepsilon}$, the tagent cone $\mc T$ and the projection $\mb T$ from Section \ref{s:genscore}. We say $\mc P$ is {$r_n$-DQM} if each $p$ is absolutely continuous with respect to $p_0$ and for each $p \in \mc P$ there are $g_p \in \mc T$ and $R_p \in L^2(\lambda)$ such that:
\[
\sqrt{p_{\phantom{.}}} - \sqrt{p_0} = g_p  \sqrt{p_0} + h(p,p_0) R_p
\]
with $\sup\{ r_n\| R_p \|_{L^2(\lambda)} : h(p,p_0) \leq r_n/\sqrt n \} \to 0$ as $n \to \infty$. Let $\ol{\mc D}_\varepsilon^2 = \{d^2 : d \in \ol{\mc D}_\varepsilon\}$.

\begin{lemma}\label{l:genscore}
	Let the following conditions hold.\\
	(i) $\mc P$ is $r_n$-DQM\\
	(ii) there exists $\varepsilon > 0$ such that  $\ol{\mc D}_{\varepsilon}^2$ is $P_0$-Glivenko Cantelli and $\ol{\mc D}_\varepsilon$ has envelope $D \in L^2(P_0)$ with $\max_{i \leq i \leq n } D(X_i) = o_\p(\sqrt n/r_n^3)$ \\
	(iii) $\sup_{p \in \mc P_{osn}} |\mb G_n (S_p - \mb T  S_p)| = o_\p(n^{-1/2})$ \\
	(iv) $\sup_{p \in \mc P_{osn}} |(\p_n-P_0)S_p^2| = o_\p(n^{-1})$. \\
	Then:
	\[
	\sup_{\theta \in \Theta_{osn}} \left|n L_n(\theta) - \left( n \p_n \log p_0  - \frac{1}{2} n P_0( (2\mb T S_{p_\theta})^2) +  n \p_n (2\mb T S_{p_\theta}) \right) \right| = o_\p(1)\,.
	\]
	If, in addition, $\mathrm{Span}(\mc T)$ has finite dimension $d^* \geq 1$ then Assumption \ref{a:quad} holds over $\Theta_{osn}$ with $\ell_n = n\p_n \log p_0$, $\sqrt n \hat \gamma_n = \mb V_n = \mb G_n(\psi)$, $\Sigma = I_{d^*}$ and $\gamma(\theta)$ defined in (\ref{e-gamma-gen}).
\end{lemma}

\begin{proof}[\textbf{Proof of Lemma \ref{l:genscore}}]
	We first prove
	\begin{equation} \label{e:lin:0}
	\sup_{p \in \mc P_{osn}} \left| n \p_n \log (p/p_0) - 2n \p_n (S_p -P_0(S_p)) + n (\p_n S_p^2 + h^2(p,p_0)) \right| = o_\p(1)
	\end{equation}
	by adapting arguments used in Theorem 1 of \cite{AGM}, Theorem 3.1 in \cite{Gassiat2002}, and Theorem 2.1 in \cite{LiuShao}.
	
	Take $n$ large enough that $r_n/\sqrt n \leq \varepsilon$. Then for each $p \in \mc P_{osn}\setminus \{p_0\}$:
	\begin{equation} \label{e:lik:expand}
	n \p_n \log (p/p_0) = 2n \p_n S_p -  n \p_n S_p^2 + 2 n \p_n S_p^2 r(S_p)
	\end{equation}
	where $r(u) = (\log(1 + u) - u - \frac{1}{2}u^2)/u^2$ and $\lim_{u \to 0} |r(u)/( \frac{1}{3} u)-1| = 0$. By condition (ii), $\max_{1 \leq i \leq n} |S_p(X_i)| \leq r_n/\sqrt n \times \max_{1 \leq i \leq n} D(X_i)  = o_\p(r_n^{-2})$ uniformly for $p \in \mc P_{osn}$. This implies that $\sup_{p \in \mc P_{osn}} \max_{1 \leq i \leq n} |r(S_p(X_i))|  = o_\p(r_n^{-2})$. Therefore, by the Glivenko-Cantelli condition in (ii):
	\[
	\sup_{p \in \mc P_{osn}} | 2 n \p_n S_p^2 r(S_p) |  \leq 2 r_n^2 \times o_\p(r_n^{-2}) \times \sup_{p \in \mc P_{osn}} \p_n s_p^2 = o_\p(1) \times (1+o_\p(1)) =  o_\p(1)\,.
	\]
	Display  (\ref{e:lin:0}) now follows by adding and subtracting $2n P_0(S_p) = -n h^2(p,p_0)$ to (\ref{e:lik:expand}).

	Each element of $\mc T$ has mean zero and so $P_0(\mb T S_p) = 0$ for each $p$. By Condition (iii):
	\[
	\sup_{p \in \mc P_{osn}} \left| \p_n (S_p -P_0(S_p)- \mb T S_p) \right| = n^{-1/2} \times \sup_{p \in \mc P_{osn}} | \mb G_n (S_p - \mb T S_p) | = o_\p(1)\,.
	\]
	
	It remains to show:
	\begin{align}
	\sup_{p \in \mc P_{osn}} \left| \p_n (S_p^2) + h^2(p,p_0) - 2P_0((\mb T S_p)^2) \right| & = o_\p(n^{-1}) \,. \label{e:emp:2}
	\end{align}
	By condition (iv) and $P_0(S_p^2) = h^2(p,p_0)$, to establish (\ref{e:emp:2}) it is enough to show:
	\[
    \sup_{p \in \mc P_{osn}} |P_0 (S_p^2) - P_0((\mb T S_p)^2) | = o_\p(n^{-1}) \,.
	\]
	Observe by definition of $\mb T$ and condition (i), for each $p \in \mc P$ there is a $g_p \in \mc T$ and remainder $R^*_p = R_p/\sqrt{p_0}$ such that $S_p = g_p + h(p,p_0) R^*_p$, and so:
	\begin{equation} \label{e:contract:r}
	\| S_p - \mb T S_p\|_{L^2(P_0)} \leq \| S_p - g_p\|_{L^2(P_0)} = h(p,p_0) \| R^*_p\|_{L^2(P_0)} = h(p,p_0) \| R_p\|_{L^2(\lambda)}
	\end{equation}
	By Moreau's decomposition theorem and inequality (\ref{e:contract:r}), we may deduce:
	\begin{align*}
	\sup_{p \in \mc P_{osn}} |P_0 (S_p^2) - P_0((\mb T S_p)^2 )|
	& = \sup_{p \in \mc P_{osn}} \|S_p - \mb T S_p\|^2_{L^2(P_0)}  \leq \sup_{p \in \mc P_{osn}}  h(p,p_0)^2 \| R_p\|^2_{L^2(\lambda)}
	\end{align*}
	which is $o_\p(n^{-1})$ by condition (i) and definition of $\mc P_{osn}$. This proves the first result.

	The second result is immediate by defining $\mb V_n = \mb G_n(\psi)$ with $\psi = (\psi_1,\ldots,\psi_{d^*})'$ where $\psi_1,\ldots,\psi_{d^*}$ is an orthonormal basis for $\mathrm{Span}(\mc T)$, and $\gamma(\theta)$ as in (\ref{e-gamma-gen}), then noting that $P_0((\mb T (2 S_{p_\theta}) )^2) = \gamma(\theta)' P_0(\psi \psi') \gamma(\theta) = \|\gamma(\theta)\|^2$.
\end{proof}

\begin{proof}[\textbf{Proof of Proposition \ref{p:genscore}}]
	We verify the conditions of Lemma \ref{l:genscore}. By DQM (condition (b)) we have $\sup\{ \| R_p \|_{L^2(\lambda)} : h(p,p_0) \leq n^{-1/4} \} \to 0$ as $n \to \infty$. Therefore, we may choose a sequence $(a_n)_{n \in \mb N}$ with  $a_n \leq n^{1/4}$ but $a_n \to \infty$ slowly enough that
	\[
	\sup\{ a_n \| R_p \|_{L^2(\lambda)} : h(p,p_0) \leq a_n/\sqrt n \} \to 0 \quad \mbox{as } n \to \infty
	\]
	and hence $\sup\{ r_n \| R_p \|_{L^2(\lambda)} : h(p,p_0) \leq r_n/\sqrt n \} \to 0$ as $n \to \infty$ for any slowly diverging positive sequence $(r_n)_{n \in \mb N}$ with $r_n \leq a_n$. This verifies condition (i) of Lemma \ref{l:genscore}.

	For condition (ii), $\ol{\mc D}_{\varepsilon}^2$ is Glivenko-Cantelli by condition (c) and Lemma 2.10.14 of \cite{vdVW}. Moreover, it follows from the envelope condition (in condition (c)) that $\max_{1 \leq i \leq n} D(X_i) = o_\p(n^{1/2})$. We can therefore choose a positive sequence $(c_n)_{n \in \mb N}$ with $c_n \to \infty$ such that $c_n^3 \max_{1 \leq i \leq n} D(X_i) = o_\p(n^{1/2})$ and so $\max_{1 \leq i \leq n} D(X_i) = o_\p(n^{1/2}/r_n^3)$ for any $0 <r_n \leq c_n$.

	For condition (iv), since $\ol{\mc D}_{\varepsilon}^2$ is Glivenko-Cantelli we may choose a positive sequence $(b_n)_{n \in \mb N}$ with $b_n \to \infty$ such that $b_n^2 \sup_{s_p \in \ol{\mc D}_{\varepsilon}} |(\p_n-P_0)s_p^2| = o_\p(1)$. Therefore, for any $0 < r_n \leq b_n$ we have:
	\[
	\sup_{p : h(p,p_0) \leq r_n/\sqrt n} | (\p_n - P_0)S_p^2 | \leq  \sup_{p : h(p,p_0) \leq r_n/\sqrt n} r_n^2| (\p_n - P_0)s_p^2 |/n = o_\p(n^{-1}) \,.
	\]

	Finally, for condition (iii), note that condition (c) implies that $\ol {\mc D}_\varepsilon^o := \{ s_p - \mb T s_p : s_p \in \ol{\mc D}_\varepsilon\}$ is Donsker. Also note that the singleton $\{0\}$ is the only limit point of $\ol {\mc D}_{\varepsilon}^o$ as $\varepsilon \searrow 0$ because:
	\[
	\sup\{ \| s_p - \mb T s_p \|_{L^2(P_0)} : h(p,p_0) \leq \varepsilon \}
 \leq \sup\{ \| R_p \|_{L^2(\lambda)} : h(p,p_0) \leq \varepsilon \}  \to 0 \quad (\mbox{as } \varepsilon \to 0)
	\]
	by DQM (condition (b)). Asymptotic equicontinuity of $\mb G_n$ on $\ol {\mc D}_\varepsilon^o$ then implies that
	\[
	 \sup_{p : h(p,p_0) \leq n^{-1/4}} | \mb G_n (s_p - \mb T s_p) | = o_\p(1) \,.
	\]
	We can therefore choose a positive sequence $(d_n)_{n \in \mb N}$ with $d_n \leq n^{1/4}$ but $d_n \to \infty$ slowly enough that $d_n \sup_{p : h(p,p_0) \leq n^{-1/4}}  |\mb G_n (s_p - \mb T s_p)| = o_\p(1)$ and so for any $0 < r_n \leq d_n$:
	\begin{align*}
	\sup_{p : h(p,p_0) \leq r_n/\sqrt n} |\mb G_n (S_p - \mb T S_p)| & \leq \frac{r_n}{\sqrt n} \sup_{p  : h(p,p_0) \leq n^{-1/4}} \mb G_n (s_p - \mb T s_p) = o_\p(n^{-1/2})\,.
	\end{align*}
	The result follows by taking $r_n = (a_n \wedge b_n \wedge c_n \wedge d_n)$.
\end{proof}

\begin{proof}[\textbf{Proof of Proposition \ref{p:gmm:cue}}]
	We first show that:
	\begin{equation} \label{e:gmm:obj:1}
	\sup_{\theta: \|g(\theta)\| \leq r_n/\sqrt n} \left| n L_n(\theta) - \left(-\frac{1}{2} ( \mb T(\sqrt n g(\theta)) + Z_n )' \Omega^{-1} (\mb T (\sqrt n g(\theta)) + Z_n ) \right) \right| = o_\p(1)
	\end{equation}
	holds  for a positive sequence $(r_n)_{n \in \mb N}$ with $r_n \to \infty$ with $Z_n = \mb G_n(\rho_{\theta^*})$. Take $n$ large enough that $n^{-1/4} \leq \varepsilon_0$. By conditions (a)--(c) and Lemma 2.10.14 of \cite{vdVW}, we have that $\sup_{\theta : \|g(\theta)\| \leq n^{-1/4}} \| \p_n (\rho_\theta^{\phantom \prime} \rho_\theta') - \Omega\| = o_\p(1)$. Therefore, we may choose a positive sequence $(a_n)_{n \in \mb N}$ with $a_n \to \infty$, $a_n =o( n^{1/4})$ such that $\sup_{\theta : \|g(\theta)\| \leq n^{-1/4}} a_n^2 \| \p_n (\rho_\theta^{\phantom \prime} \rho_\theta') - \Omega\| = o_\p(1)$
	and hence:
	\begin{align} \label{e:emp:omg}
	\sup_{\theta : \|g(\theta)\| \leq r_n/\sqrt n}\| \p_n (\rho_\theta^{\phantom \prime} \rho_\theta') - \Omega\| = o_\p(r_n^{-2})
	\end{align}
	for any $0 < r_n \leq a_n$.

	Notice that $Z_n \rightsquigarrow N(0,\Omega)$ by condition (a) and that the covariance of each element of $\rho_\theta(X_i)-\rho_{\theta^*}(X_i)$ vanishes uniformly over $\Theta_I^\varepsilon$ as $\varepsilon \to 0$ by condition (c). Asymptotic equicontinuity of $\mb G_n$ (which holds under (a))  then implies that $\sup_{\theta : \|g(\theta)\| \leq n^{-1/4}} \| \mb G_n(\rho_\theta) - Z_n\| = o_\p(1)$. We can therefore choose a positive sequence $(b_n)_{n \in \mb N}$ with $b_n \to \infty$, $b_n = o(n^{1/4})$ as $n \to \infty$ such that $b_n \sup_{\theta : \|g(\theta)\| \leq b_n/\sqrt n} \| \mb G_n(\rho_\theta) - Z_n\| = o_\p(1)$ and hence:
	\begin{align} \label{e:emp:rho:unif}
	\sup_{\theta : \|g(\theta)\| \leq r_n/\sqrt n} | \sqrt n \p_n \rho_\theta - ( \sqrt n g(\theta) + Z_n) |= o_\p(r_n^{-1}) \,.
	\end{align}
	for any $0 < r_n \leq b_n$.

	Condition (d) implies that we may choose a sequence $(c_n)_{n \in \mb N}$ with $c_n \to \infty$, $c_n = o(n^{1/4})$ such that $\sup_{\theta : \|g(\theta)\| \leq c_n/\sqrt n} \sqrt n \| g(\theta) - \mb T  g(\theta)\| = o(c_n^{-1})$ and so:
	\begin{align} \label{e:proj:control}
	\sup_{\theta : \|g(\theta)\| \leq r_n/\sqrt n}  \|\sqrt n g(\theta) -  \mb T (\sqrt n g(\theta))\| = o(r_n^{-1})
	\end{align}
	for any $0 < r_n \leq c_n$.

	Result  (\ref{e:gmm:obj:1}) now follows by taking $r_n = (a_n \wedge b_n \wedge c_n)$ and using (\ref{e:emp:omg}), (\ref{e:emp:rho:unif}) and (\ref{e:proj:control}). To complete the proof, expanding the quadratic in (\ref{e:gmm:obj:1}) we obtain:
	\begin{align*}
	-\frac{1}{2}( \mb T(\sqrt n g(\theta)) + Z_n )' \Omega^{-1} (\mb T (\sqrt n g(\theta)) + Z_n )
	& = -\frac{1}{2} Z_n'\Omega^{-1} Z_n -\frac{1}{2} \| [\Omega^{-1/2} \mb T(\sqrt n g(\theta)) ]_1\|^2 \\
	& \quad   -  [ \Omega^{-1/2}Z_n]'_1 [\Omega^{-1/2} \mb T(\sqrt n g(\theta)) ]_1
	\end{align*}
	and the result follows with $\ell_n = Z_n'\Omega^{-1} Z_n$,  $\gamma(\theta) = [\Omega^{-1/2} \mb Tg(\theta)]_1$, and $\mb V_n = -[\Omega^{-1/2}Z_n]_1$.
\end{proof}

\begin{proof}[\textbf{Proof of Proposition \ref{p:gmm:opt}}]
	Follows by similar arguments to the proof of Proposition \ref{p:gmm:cue}, noting that by condition (e) we may choose a positive sequence $(a_n)_{n \in \mb N}$ with $a_n \to \infty$ slowly such that $a_n^2 \|\wh W - \Omega^{-1}\| = o_\p(1)$. Therefore $ \|\wh W - \Omega^{-1}\| = o_\p(r_n^{-2})$ holds for any $0 < r_n \leq a_n$.
\end{proof}

\begin{lemma} \label{l:md:singular}
	Consider the missing data model with a flat prior on $\Theta$. Suppose that the model is point identified (i.e. the true $\eta_2 = 1$). Then Assumption \ref{a:rate}(ii) holds for $$\Theta_{osn} = \{ \theta : |\tilde\gamma_{11}(\theta)-\tilde\gamma_{11}| \leq r_n/\sqrt n, \tilde\gamma_{00}(\theta) \leq r_n/n\}$$ for any positive sequence $(r_n)_{n \in \mb N}$ with $r_n \to \infty$, $r_n/\sqrt n = o(1)$
\end{lemma}

\begin{proof}[\textbf{Proof of Lemma \ref{l:md:singular}}]
	 The flat prior on $\Theta$ induces a flat prior on $\{ (a,b) \in [0,1] : 0 \leq a \leq 1-b\}$  under the map $\theta \mapsto (\tilde \gamma_{11}(\theta),\tilde \gamma_{00}(\theta))$. Take $n$ large enough that $[\tilde \gamma_{11} - r_n/\sqrt n,\tilde \gamma_{11} + r_n/\sqrt n] \subseteq [0,1]$ and $r_n/n <1$. Then with $S_n := \sum_{i=1}^n Y_i$, we have:
	\begin{align*}
	\Pi_n ( \Theta_{osn}^c | \mf X_n )
	& = \frac{\int_{[0,\tilde \gamma_{11} - r_n/\sqrt n] \cup [\tilde \gamma_{11} + r_n/\sqrt n]} \int_{0}^{1-a} (a)^{S_n}( 1-a-b )^{n-S_n}   \, \mr d b \, \mr d a}{\int_{0}^{1} \int_0^{1-a} (a )^{S_n}( 1-a-b \big)^{n-S_n}   \, \mr d b \, \mr d a} \\
	& \quad \quad  +  \frac{\int_{\kappa_{11} - k_n/\sqrt n}^{\kappa_{11} + k_n/\sqrt n} \int_{k_n/n}^{1-a} ( a)^{S_n} ( 1-a-b )^{n-S_n}   \, \mr d b \, \mr d a}{\int_{0}^{1} \int_0^{1-a} (a)^{S_n}(1-a-b)^{n-S_n}   \, \mr d b \, \mr d a}  =:  I_1 + I_2  \,.
	\end{align*}
	Integrating $I_1$ first with respect to $b$ yields:
	\begin{align*}
	I_1 & = \frac{\int_{[0,\tilde \gamma_{11} - r_n/\sqrt n] \cup [\tilde \gamma_{11} + r_n/\sqrt n]} \int_{0}^{1-a} (a)^{S_n}( 1-a )^{n-S_n+1} \, \mr d a}{\int_{0}^{1}  (a )^{S_n}( 1-a \big)^{n-S_n+1}  \, \mr d a}
	= \mb P_{U|S_n} ( | U - \tilde \gamma_{11} | > r_n/\sqrt n)
	\end{align*}
	where $U|S_n \sim \mr{Beta}(S_n+1,n-S_n+2)$. Note that this implies:
	\begin{align*}
	\mb E[ U | S_n ]  & = \frac{S_n+1}{n+3}  &
	\mr{Var}[ U|S_n ] & = \frac{(S_n+1)(n-S_n+2)}{(n+3)^2(n+4)}\,.
	\end{align*}
	By the triangle inequality, the fact that $\mb E[ U | S_n ] = \tilde \gamma_{11} + O_\p(n^{-1/2})$, and Chebyshev's inequality:
	\begin{align*}
	I_1 & \leq \mb P_{U|S_n} \Big( | U - \mb E[ U | S_n ] | > r_n/(2\sqrt n)\Big) + \ind\Big\{ |\mb E[ U | S_n ] - \tilde \gamma_{11}| > r_n/(2\sqrt n)\Big\} \\
	& = \mb P_{U|S_n} \Big( | U - \mb E[ U | S_n ] | > r_n/(2\sqrt n)\Big) + o_\p(1) \\
	& \leq \frac{4}{r_n^2} \frac{(\frac{S_n}{n}+\frac{1}{n})(1-\frac{S_n}{n}+\frac{2}{n})}{(1+\frac{3}{n})^2(1+\frac{4}{n})}  + o_\p(1) = o_\p(1) \,.
	\end{align*}

	Similarly:
	\[
	I_2  = \frac{\int_{\tilde \gamma_{11} - r_n/\sqrt n}^{\tilde \gamma_{11} + r_n/\sqrt n} ( a)^{S_n} ( 1-a-(r_n/n) )^{n-S_n+1}  \, \mr d a}{\int_{0}^{1}  (a)^{S_n}(1-a)^{n-S_n+1}   \, \mr d a}
	 \leq \frac{\int_{0}^{1-r_n/n} ( a)^{S_n} ( 1-a-(r_n/n) )^{n-S_n+1}  \, \mr d a}{\int_{0}^{1}  (a)^{S_n}(1-a)^{n-S_n+1}   \, \mr d a} \,.
	\]
	Using the change of variables $a \mapsto c(a) := \frac{1-a-r_n/n}{1-r_n/n}$ in the numerator yields:
	\begin{align*}
	I_2 & \leq (1-(r_n/n))^{n+2} \frac{\int_{0}^{1} ( 1-c)^{S_n} ( c )^{n-S_n+1}  \, \mr d c}{\int_{0}^{1}  (a)^{S_n}(1-a)^{n-S_n+1}   \, \mr d a} = (1-(r_n/n))^{n+2} \to 0\,.
	\end{align*}
Therefore, $\Pi_n ( \Theta_{osn}^c | \mf X_n ) = o_\p(1)$, as required.
\end{proof}

\subsection{Proofs and Additional Lemmas for Appendix \ref{s:uniformity}}

\begin{proof}[\textbf{Proof of Lemma \ref{l:basic:unif}}]
	By condition (i), there exists a positive sequence $(\varepsilon_n)_{n \in \mb N}$, $\varepsilon_n = o(1)$ such that $\sup_{\p \in \mf P} \p( \sup_{\theta \in \Theta_I(\p)} Q_n(\theta) - W_{n} > \varepsilon_n) = o(1)$. Let $\mc A_{n,\p}$ denote the event on which $\sup_{\theta \in \Theta_I(\p)} Q_n(\theta) - W_n \leq \varepsilon_n$. Then:
	\begin{align*}
	\inf_{\p \in \mf P} \p (\Theta_I (\p) \subseteq \wh \Theta_{\alpha} )
	& \geq \inf_{\p \in \mf P}  \p (\{ \Theta_I (\p) \subseteq \wh \Theta_{\alpha}\} \cap \mc A_{n,\p} ) \\
	& = \inf_{\p \in \mf P}  \p (\{ \textstyle \sup_{\theta \in \Theta_I(\p)} Q_n(\theta) \leq v_{\alpha,n} \} \cap \mc A_{n,\p} )  \\
	& \geq \inf_{\p \in \mf P}  \p (\{ W_n \leq v_{\alpha,n} - \varepsilon_n \} \cap \mc A_{n,\p} )  \,,
	\end{align*}
	where the second equality is by the definition of $\wh \Theta_{\alpha}$. Since $\p( A \cap B) \geq 1- \p(A^c) - \p(B^c)$, we therefore have:
	\begin{align*}
	\inf_{\p \in \mf P} \p (\Theta_I (\p) \subseteq \wh \Theta_{\alpha} )
	& \geq 1 - \sup_{\p \in \mf P} \p ( W_n > v_{\alpha,n} - \varepsilon_n ) - \sup_{\p \in \mf P}  \p (\mc A_{n,\p}^c ) \\
	& = 1 - (1-\inf_{\p \in \mf P} \p ( W_n \leq v_{\alpha,n} - \varepsilon_n )) - o(1) \\
	& = \inf_{\p \in \mf P} \p ( W_n \leq v_{\alpha,n} - \varepsilon_n ) - o(1) \\
	& \geq \alpha - o(1)\,,
	\end{align*}
	where the final line is by condition (ii) and definition of $\mc A_{n,\p}$.
\end{proof}

\begin{proof}[\textbf{Proof of Lemma \ref{l:basic:profile:unif}}]
	Follows by similar arguments to the proof of Lemma \ref{l:basic:unif}.
\end{proof}

We use the next Lemma several times in the following proofs.

\begin{lemma}\label{lem:proj:contract}
Let $T \subseteq \mb R^{d}$ be a closed convex cone and let $\mf T$ denote the projection onto $T$. Then:
\[
 \| \mf T (x + t) - t\| \leq \|x\|
\]
 for any $x \in \mb R^d$ and $t \in T$.
\end{lemma}

\begin{proof}[Proof of Lemma \ref{lem:proj:contract}]
Let $\mf T^o$ denote the projection onto the polar cone $T^o$ of $T$. Since $u't \leq 0$ holds for any $u \in T^o$ and $\| \mf T v \| \leq \|v\|$ holds for any $v \in \mb R^d$, we obtain:
\[
 \| \mf T (x + t)\|^2 + 2 (\mf T^o (x + t))'t \leq \| \mf T (x + t)\|^2  \leq \|x + t\|^2\,.
\]
Subtracting $2 (x+t)'t$ from both sides and using the fact that $v = \mf T v + \mf T^o v$ yields:
\[
 \| \mf T (x + t)\|^2 - 2 (\mf T (x + t))'t \leq \|x + t\|^2 - 2 (x+t)'t \,.
\]
Adding $\|t\|^2$ to both sides and completing the square gives $\| \mf T (x + t) - t\|^2 \leq \|x + t - t\|^2 = \|x\|^2$.
\end{proof}

In view of Lemma \ref{lem:proj:contract} and Assumption \ref{a:quad:unif}(i), for each $\p \in \mf P$ we have:
\begin{equation} \label{e:proj:contract}
 \| \sqrt n (\hat \gamma_n - \tau)\| \leq \| \mb V_n\|\,.
\end{equation}

\begin{lemma}\label{l:quad:unif}
	Let Assumptions \ref{a:rate:unif}(i) and \ref{a:quad:unif} hold. Then:
	\begin{align}
	\sup_{\theta \in \Theta_{osn}} \left| Q_n(\theta) - \|\sqrt n \gamma(\theta) - \sqrt n (\hat \gamma_n - \tau)\|^2 - 2 f_{n,\bot}(\gamma_\bot(\theta)) \right|  & = o_\p(1) \label{e:quad:1:unif}
	\end{align}
	 uniformly in $\p$.
	
	If, in addition, Assumption \ref{a:qlr:unif}(i) holds, then:
	\begin{align} \label{e:quad:qlr:unif}
	& \sup_{\theta \in \Theta_{osn}} \left| PQ_n (M( \theta) ) - f \left( \sqrt n (\hat \gamma_n - \tau) - \sqrt n \gamma(\theta)  \right)    \right| = o_\p(1)
	\end{align}
	 uniformly in $\p$.
\end{lemma}

\begin{proof}[\textbf{Proof of Lemma \ref{l:quad:unif}}]
	To show (\ref{e:quad:1:unif}),  by Assumptions \ref{a:rate:unif}(i) and \ref{a:quad:unif}(i)(iii):
	\begin{align}
	nL_n(\hat \theta)
	& = \sup_{\theta \in \Theta_{osn}} \left( \ell_n + \frac{n}{2}\|\hat \gamma_n - \tau\|^2 - \frac{1}{2} \|\sqrt n \gamma(\theta) - \sqrt n (\hat \gamma_n - \tau)\|^2 - f_{n,\bot} (\gamma_\perp(\theta))\right)  + o_\p(1)  \notag \\
	& = \ell_n + \frac{n}{2}\|\hat \gamma_n - \tau\|^2 - \inf_{\theta \in \Theta_{osn}} \frac{1}{2} \|\sqrt n \gamma(\theta) - \sqrt n (\hat \gamma_n - \tau)\|^2 + o_\p(1) \label{e:lnhattheta:prime:unif}
	\end{align}
	uniformly in $\p$. But observe that by Assumption \ref{a:quad:unif}(i)(ii), for any $\epsilon > 0$:
	\begin{align*}
	& \sup_{\p \in \mf P} \p \left( \inf_{\theta \in \Theta_{osn}} \|\sqrt n \gamma(\theta) - \sqrt n (\hat \gamma_n - \tau)\|^2 > \epsilon \right) \\
	& \leq \sup_{\p \in \mf P} \p \left(  \left\{ \inf_{t \in (T - \sqrt n \tau) \cap B_{k_n}}  \|t- \sqrt n (\hat \gamma_n - \tau)\|^2 > \epsilon \right\} \cap \left\{ \|\hat \gamma_n - \tau\| < \frac{k_n}{\sqrt n} \right\} \right) \\
	& \quad \quad +  \sup_{\p \in \mf P} \p \left(  \| \hat \gamma_n - \tau\| \geq \frac{k_n}{\sqrt n}  \right)
	\end{align*}
	where $\inf_{t \in (T - \sqrt n \tau) \cap B_{k_n}}  \|t- \sqrt n (\hat \gamma_n - \tau)\|^2  = 0$ whenever $\|\sqrt n (\hat \gamma_n - \tau)\| < k_n$ (because $\sqrt n \hat \gamma_n \in T$). Notice $\|\sqrt n (\hat \gamma_n - \tau)\| = o_\p(k_n)$ uniformly in $\p $ by (\ref{e:proj:contract}) and the condition $\|\mb V_n\| = O_\p(1)$ (uniformly in $\p$). This proves (\ref{e:quad:1:unif}). Result (\ref{e:quad:qlr:unif}) follows by Assumption \ref{a:qlr:unif}(i).
\end{proof}

\begin{proof}[\textbf{Proof of Lemma \ref{l:post:unif}}]
	We only prove the case with singularity; the case without singularity follows similarly. By identical arguments to the proof of Lemma \ref{l:post:prime}, it is enough to characterize the large-sample behavior of $R_n (z)$ defined in equation (\ref{e-Rn}) uniformly in $\p$.
	By Lemma \ref{l:quad:unif} and Assumption \ref{a:quad:unif}(i)--(iii), there exist a positive sequence $(\varepsilon_n)_{n \in \mb N}$ independent of $z$ with $\varepsilon_n = o(1)$ and a sequence of events $(\mc A_n)_{n \in \mb N} \subset \mc F$ with $\inf_{\p \in \mf P} \p(\mc A_n) = 1-o(1)$ such that:
	\begin{align*}
	\sup_{\theta \in \Theta_{osn}} \left| Q_n(\theta) - \left( \|\sqrt n \gamma(\theta) - \sqrt n (\hat \gamma_n - \tau)\|^2 + 2 f_{n,\bot}(\gamma_\bot(\theta)) \right) \right| & \leq \varepsilon_n  \\
	\sup_{\theta \in \Theta_{osn}} \left| nL_n(\theta) - \ell_n - \frac{n}{2}\|\hat \gamma_n - \tau\|^2 + \frac{1}{2} \|\sqrt n \gamma(\theta) - \sqrt n (\hat \gamma_n - \tau)\|^2 + f_{n,\bot} (\gamma_\perp(\theta))  \right| & \leq \varepsilon_n
	\end{align*}
	both hold on $\mc A_n$ for all $\p \in \mf P$. Also note that for any $z \in \mb R$ and any singular $\p \in \mf P$, we have
	\begin{align*}
	& \left\{\theta \in \Theta_{osn} : \|\sqrt n \gamma(\theta) - \sqrt n (\hat \gamma_n - \tau)\|^2 + 2 f_{n,\bot}(\gamma_\bot(\theta)) \leq z  + \varepsilon_n  \right\}  \\
	& \subseteq \left\{\theta \in \Theta_{osn} :  \|\sqrt n \gamma(\theta) - \sqrt n (\hat \gamma_n - \tau)\|^2 \leq z + \varepsilon_n  \right\}
	\end{align*}
	because $f_{n,\bot} \geq 0$. Therefore, on $\mc A_n$ we have:
	\begin{align*}
	R_n(z)
	& \leq  e^{2 \varepsilon_n} \frac{\int_{\{\theta: \|\sqrt n \gamma(\theta) - \sqrt n (\hat \gamma_n - \tau)\|^2 \leq  z + \varepsilon_n \} \cap \Theta_{osn}}\! e^{ - \frac{1}{2}\|\sqrt n \gamma(\theta) - \sqrt n (\hat \gamma_n - \tau)\|^2- f_{n,\bot}(\gamma_\bot(\theta))} \mr d\Pi(\theta)}{\int_{\Theta_{osn}} \!e^{ - \frac{1}{2} \|\sqrt n \gamma(\theta) - \sqrt n (\hat \gamma_n - \tau)\|^2 - f_{n,\bot}(\gamma_\bot(\theta)) } \mr d\Pi(\theta)}
	\end{align*}
	uniformly in $z$, for all $\p \in \mf P$.

	Define $\Gamma_{osn} = \{\gamma(\theta) : \theta \in \Theta_{osn}\}$ and $\Gamma_{\bot,osn} = \{\gamma_\bot(\theta) : \theta \in \Theta_{osn}\}$ (if $\p$ is singular). The condition $\sup_{\p \in \mf P} \sup_{\theta \in \Theta_{osn}} \|(\gamma(\theta),\gamma_\bot(\theta))\| \to 0$ in Assumption \ref{a:quad:unif}(i) implies that for all $n$ sufficiently large we have $\Gamma_{osn} \times \Gamma_{\bot,osn} \subset B^*_\delta$ for all $\p \in \mf P$. By similar arguments to the proof of Lemma \ref{l:post:prime}, we use Assumption \ref{a:prior:unif}(ii), a change of variables and Tonelli's theorem to obtain:
	\begin{align*}
	R_n(z)
	& \leq e^{2 \varepsilon_n}(1+\bar \varepsilon_n) \frac{\int_{(\{ \gamma :\|\sqrt n \gamma - \sqrt n (\hat \gamma_n - \tau)\|^2 \leq  z + \varepsilon_n) \cap \Gamma_{osn}  }\! e^{ - \frac{1}{2} \|\sqrt n \gamma - \sqrt n (\hat \gamma_n - \tau)\|^2 } \mr d \gamma }{\int_{ \Gamma_{osn} } \!e^{ - \frac{1}{2} \|\sqrt n \gamma - \sqrt n (\hat \gamma_n - \tau)\|^2 }  \mr d \gamma}
	\end{align*}
	which holds uniformly in $z$ for all $\p \in \mf P$ (on $\mc A_n$ with $n$ sufficiently large) for some sequence $(\bar \varepsilon_n)_{n \in \mb N}$ with $\bar \varepsilon_n = o(1)$. A second change of variables with $\sqrt n \gamma - \sqrt n (\hat \gamma_n - \tau) \mapsto \kappa$ yields:
	\begin{align*}
	R_n(z)
	& \leq e^{2 \varepsilon_n}(1+\bar \varepsilon_n) \frac{\nu_{d^*}(\{ \kappa : \| \kappa \|^2 \leq z  + \varepsilon_n\} \cap (T_{osn}- \sqrt n (\hat \gamma_n - \tau)) ) }{\nu_{d^*}( T_{osn} -  \sqrt n (\hat \gamma_n - \tau)) }
	\end{align*}
	where $T_{osn} = \{ \sqrt n \gamma : \gamma \in \Gamma_{osn}\}= \{ \sqrt n \gamma(\theta) : \theta \in \Theta_{osn}\}$.
	
	Recall that $B_\delta \subset \mb R^{d^*}$ denotes a ball of radius $\delta$ centered at zero. To complete the proof, it is enough to show that:
	\begin{equation} \label{e:prob:unif:z}
	 \sup_z \bigg| \frac{ \nu_{d^*}( B_{\sqrt{z  + \varepsilon_n }} \cap (T_{osn}- \sqrt n (\hat \gamma_n - \tau)) ) }{ \nu_{d^*}( T_{osn}- \sqrt n (\hat \gamma_n - \tau) ) }  -  \frac{\nu_{d^*}( B_{\sqrt{z}}  \cap (T - \sqrt n \hat \gamma_n) )}{\nu_{d^*}(T - \sqrt n \hat \gamma_n)}  \bigg| =  o_\p(1)
	 \end{equation}
	 uniformly in $\p$. We split this into three parts. First note that
	 \begin{align}
	 & \sup_z \bigg| \frac{ \nu_{d^*}( B_{\sqrt{z  + \varepsilon_n }} \cap (T_{osn}- \sqrt n (\hat \gamma_n - \tau)) ) }{ \nu_{d^*}( T_{osn}- \sqrt n (\hat \gamma_n - \tau) ) }  -   \frac{ \nu_{d^*}( B_{\sqrt{z  + \varepsilon_n }} \cap (T_{osn}\cap B_{k_n}- \sqrt n (\hat \gamma_n - \tau)) ) }{ \nu_{d^*}( T_{osn}\cap B_{k_n}- \sqrt n (\hat \gamma_n - \tau) ) }  \bigg| \notag \\
	 & \leq 2 \frac{ \nu_{d^*}(  ((T_{osn}\setminus B_{k_n})- \sqrt n (\hat \gamma_n - \tau)) ) }{ \nu_{d^*}( T_{osn}- \sqrt n (\hat \gamma_n - \tau) ) } \notag \\
	 & \leq 2 \frac{ \nu_{d^*}(  B_{k_n}^c- \sqrt n (\hat \gamma_n - \tau)) }{ \nu_{d^*}( T_{osn}- \sqrt n (\hat \gamma_n - \tau) ) } \label{e:prob:ratio:unif}
	 \end{align}
	 where the first inequality is by (\ref{e:abc:ineq}) and the second is by the inclusion $(T_{osn}\setminus B_{k_n}) \subseteq B_{k_n}^c$. Since $\|\sqrt n (\hat \gamma_n - \tau)\| \leq \| \mb V_n\|$ (by \ref{e:proj:contract}) where $\| \mb V_n\| = O_\p(1)$ uniformly in $\p$ and $\inf_{\p \in \mf P} k_n(\p) \to \infty$ and $d^* = d^*(\p) \leq \ol d < \infty$, we have
	 \[
	 \nu_{d^*}(B_{k_n}^c- \sqrt n (\hat \gamma_n - \tau)) = o_\p(1)
	\]
	uniformly in $\p$. Also notice that, by Assumption \ref{a:quad:unif}(ii),
	 \[
	   \frac{ \nu_{d^*}( B_{\sqrt{z  + \varepsilon_n }} \cap (T_{osn}\cap B_{k_n}- \sqrt n (\hat \gamma_n - \tau)) ) }{ \nu_{d^*}( T_{osn}\cap B_{k_n}- \sqrt n (\hat \gamma_n - \tau) ) }
	   = \frac{ \nu_{d^*}( B_{\sqrt{z  + \varepsilon_n }} \cap ((T - \sqrt n \tau)\cap B_{k_n}- \sqrt n (\hat \gamma_n - \tau)) ) }{ \nu_{d^*}( (T - \sqrt n \tau)\cap B_{k_n}- \sqrt n (\hat \gamma_n - \tau) ) }
	 \]
	 where, by similar arguments to (\ref{e:prob:ratio:unif}),
	 \begin{align}
	 & \sup_z \left| \frac{ \nu_{d^*}( B_{\sqrt{z  + \varepsilon_n }} \cap ((T - \sqrt n \tau)\cap B_{k_n}- \sqrt n (\hat \gamma_n - \tau)) ) }{ \nu_{d^*}( (T - \sqrt n \tau)\cap B_{k_n}- \sqrt n (\hat \gamma_n - \tau) ) } \notag
	 - \frac{ \nu_{d^*}( B_{\sqrt{z  + \varepsilon_n }} \cap (T - \sqrt n \hat \gamma_n  ) }{ \nu_{d^*}(T - \sqrt n \hat \gamma_n  ) }  \right| \\
	 & \leq 2 \frac{ \nu_{d^*}(  ((T - \sqrt n \tau)\setminus B_{k_n})- \sqrt n (\hat \gamma_n - \tau) ) }{ \nu_{d^*}( T - \sqrt n \hat \gamma_n ) } \label{e:prob:ratio:unif0} \\
	 & \leq 2 \frac{ \nu_{d^*}( B_{k_n}^c- \sqrt n (\hat \gamma_n - \tau) ) }{ \nu_{d^*}(T - \sqrt n \hat \gamma_n  ) } \label{e:prob:ratio:unif2} \,.
	 \end{align}
	 A sufficient condition for the right-hand side of display (\ref{e:prob:ratio:unif2}) to be $o_\p(1)$ (uniformly in $\p$) is that
	 \begin{align} \label{e:prob:ratio:unif3}
	 1/\nu_{d^*}(T - \sqrt n \hat \gamma_n  )  = O_\p(1) \quad \mbox{(uniformly in $\p$).}
	 \end{align}
	 But notice that $T - \sqrt n \hat \gamma_n  = (T - \sqrt n \tau) - \sqrt n (\hat \gamma_n - \tau)$ where the $\sqrt n (\hat \gamma_n - \tau)$ are uniformly tight (by (\ref{e:proj:contract}) and the condition $\| \mb V_n\| = O_\p(1)$ uniformly in $\p$) and $T - \sqrt n \tau \supseteq T$. We may therefore deduce by the condition $\inf_{\p \in \mf P} \nu_{d^*}(T) > 0$ in Assumption \ref{a:quad:unif}(ii) that (\ref{e:prob:ratio:unif3}) holds, and so:
	 \[
	  \sup_z \left| \frac{ \nu_{d^*}( B_{\sqrt{z  + \varepsilon_n }} \cap ((T - \sqrt n \tau)\cap B_{k_n}- \sqrt n (\hat \gamma_n - \tau)) ) }{ \nu_{d^*}( (T - \sqrt n \tau)\cap B_{k_n}- \sqrt n (\hat \gamma_n - \tau) ) } - \frac{ \nu_{d^*}( B_{\sqrt{z  + \varepsilon_n }} \cap (T - \sqrt n \hat \gamma_n  ) }{ \nu_{d^*}(T - \sqrt n \hat \gamma_n  ) }  \right| = o_\p(1)
	 \]
	 uniformly in $\p$.
	 To see that the right-hand side of (\ref{e:prob:ratio:unif}) is $o_\p(1)$ (uniformly in $\p$), first note that:
	 \begin{align*}
	 \frac{ 1 }{ \nu_{d^*}( T_{osn}- \sqrt n (\hat \gamma_n - \tau) ) } & = \frac{ \frac{ 1 }{ \nu_{d^*}( T_{osn} \cap B_{k_n}- \sqrt n (\hat \gamma_n - \tau) ) } }{1 + \frac{\nu_{d^*}( T_{osn} \setminus B_{k_n}- \sqrt n (\hat \gamma_n - \tau) )}{ \nu_{d^*}( T_{osn} \cap B_{k_n}- \sqrt n (\hat \gamma_n - \tau) ) } }
	 \end{align*}
	 where
	 \begin{align*}
	 \frac{1}{\nu_{d^*}( T_{osn} \cap B_{k_n}- \sqrt n (\hat \gamma_n - \tau) )} & = \frac{ 1 }{ \nu_{d^*}( (T-\sqrt n \tau) \cap B_{k_n}- \sqrt n (\hat \gamma_n - \tau) ) } \\
	 & = \frac{1}{(1-o_\p(1)) \times \nu_{d^*}( T-\sqrt n \tau) )} = O_\p(1)
	 \end{align*}
	 (uniformly in $\p$) by (\ref{e:prob:ratio:unif3}) and because the $o_\p(1)$ term holds uniformly in $\p$ by (\ref{e:prob:ratio:unif0}) and (\ref{e:prob:ratio:unif2}). It follows that $1/\nu_{d^*}( T_{osn}- \sqrt n (\hat \gamma_n - \tau) ) = O_\p(1)$ (uniformly in $\p$) and so, by (\ref{e:prob:ratio:unif}), we obtain:
	 \[
	  \sup_z \bigg| \frac{ \nu_{d^*}( B_{\sqrt{z  + \varepsilon_n }} \cap (T_{osn}- \sqrt n (\hat \gamma_n - \tau)) ) }{ \nu_{d^*}( T_{osn}- \sqrt n (\hat \gamma_n - \tau) ) }  -   \frac{ \nu_{d^*}( B_{\sqrt{z  + \varepsilon_n }} \cap (T_{osn}\cap B_{k_n}- \sqrt n (\hat \gamma_n - \tau)) ) }{ \nu_{d^*}( T_{osn}\cap B_{k_n}- \sqrt n (\hat \gamma_n - \tau) ) }  \bigg| = o_\p(1)
	 \]
	 (uniformly in $\p$). To complete the proof of (\ref{e:prob:unif:z}), it remains to show that
	 \[
	  \sup_z \left| \nu_{d^*}( B_{\sqrt{z+\varepsilon_n}}  \cap (T - \sqrt n \hat \gamma_n) ) - \nu_{d^*}( B_{\sqrt{z}}  \cap (T - \sqrt n \hat \gamma_n) ) \right| = o_\p(1)
	 \]
	holds uniformly in $\p$. But here we have:
	\begin{align*}
	 & \sup_z \left| \nu_{d^*}( B_{\sqrt{z+\varepsilon_n}}  \cap (T - \sqrt n \hat \gamma_n) ) - \nu_{d^*}( B_{\sqrt{z}}  \cap (T - \sqrt n \hat \gamma_n) ) \right| \\
	 & \leq \sup_z \left| \nu_{d^*}( B_{\sqrt{z+\varepsilon_n}}\setminus B_{\sqrt{z}}  ) \right| \\
	 & = \sup_z \left| F_{\chi^2_{d^*}}(z+\varepsilon_n) - F_{\chi^2_{d^*}}(z) \right| \to 0
	\end{align*}
	by uniform equicontinuity of $\{ F_{\chi^2_{d}} : d \leq \ol d\}$.
\end{proof}

\begin{proof}[\textbf{Proof of Theorem \ref{t:main:unif}}]
	We first prove part (i) by verifying the conditions of Lemma \ref{l:basic:unif}. We assume w.l.o.g. that $L_n(\hat \theta) = \sup_{\theta \in \Theta_{osn}} L_n(\theta) + o_\p(n^{-1})$ uniformly in $\p$. By display (\ref{e:quad:1:unif}) in Lemma \ref{l:quad:unif} we have $\sup_{\theta \in \Theta_I(\p)} Q_n(\theta) = \| \mf T(\mb V_n + \sqrt n \tau) - \sqrt n \tau\|^2  + o_\p(1)$ uniformly in $\p$. This verifies condition (i) with $W_n =\| \mf T(\mb V_n + \sqrt n \tau) - \sqrt n \tau\|^2$.
	
	For condition (ii) let $\xi_{\alpha,\p}$ denote the $\alpha$ quantile of $F_{T}$ under $\p$ and let $(\varepsilon_n)_{n \in \mb N}$ be a positive sequence with $\varepsilon_n = o(1)$. By the conditions $\| \mf T (\mb V_n + \sqrt n \tau) - \sqrt n \tau\|^2 \leq \| \mf T \mb V_n\|^2$ (almost surely) for each $\p \in \mf P$, $\sup_{\p \in \mf P} \sup_z | \p ( \| \mf T \mb V_n \|^2 \leq z ) - \p_Z (\|\mf T Z\|^2 \leq z) | = o(1)$ and the equicontinuity of $\{ F_T : \p \in \mf P\}$ at their $\alpha$ quantiles, we have:
	\begin{align*}
	\liminf_{n \to \infty} \inf_{\p \in \mf P} \p(W_n \leq \xi_{\alpha,\p} - \varepsilon_n)
	& \geq \liminf_{n \to \infty} \inf_{\p \in \mf P} \p( \| \mf T \mb V_n\|^2 \leq \xi_{\alpha,\p} - \varepsilon_n) \\
	& \geq \liminf_{n \to \infty} \inf_{\p \in \mf P} \p_Z( \| \mf T Z\|^2 \leq \xi_{\alpha,\p} - \varepsilon_n) \\
	& = \alpha.
	\end{align*}
	By Condition \ref{a:mcmc:unif} it suffices to show that for each $\epsilon > 0$
	\[
	\lim_{n \to \infty} \sup_{\p \in \mf P} \p( \xi_{\alpha,\p} - \xi_{n,\alpha}^{post} > \epsilon) = 0\,.
	\]
	A sufficient condition is that
	\[
	 \lim_{n \to \infty} \inf_{\p \in \mf P} \p ( \Pi_n (\{ \theta : Q_n(\theta) \leq \xi_{\alpha,\p} - \epsilon\} | \mf X_n) < \alpha )= 1\,.
	\]
	By Lemma \ref{l:post:unif} there exists a sequence of positive constants $(u_n)_{n \in \mb N}$ with $u_n = o(1)$ and a sequence of events $(\mc A_n)_{n \in \mb N}$ (possibly depending on $\p$) with $\inf_{\p \in \mf P} \p(\mc A_n) = 1-o(1)$ such that:
	\[
	\Pi_n \big(\{ \theta : Q_n(\theta) \leq \xi_{\alpha,\p} - \epsilon \} \,\big|\,\mf X_n\big) \leq \p_{Z|\mf X_n} \left( \|Z\|^2 \leq \xi_{\alpha,\p} - \epsilon  | Z \in T - \sqrt n \hat \gamma_n  \right)  + u_n
	\]
	holds on $\mc A_n$ for each $\p$. But by Theorem 2 of \cite{ChenGao} we also have:
	\[
	\p_{Z|\mf X_n} \left( \|Z\|^2 \leq \xi_{\alpha,\p} - \epsilon  | Z \in T - \sqrt n \hat \gamma_n  \right) \leq  F_T(\xi_{\alpha,\p} - \epsilon)
	\]
	and hence
	\[
	\Pi_n \big(\{ \theta : Q_n(\theta) \leq \xi_{\alpha,\p} - \epsilon \} \,\big|\,\mf X_n\big) \leq F_T(\xi_{\alpha,\p} - \epsilon) + u_n
	\]
	holds on $\mc A_n$ for each $\p$. Also note that by the equicontinuity of $\{F_T : \p \in \mf P\}$ at their $\alpha$ quantiles:
	\begin{equation} \label{e:unif:ineq:1}
	 \limsup_{n \to \infty} \sup_{\p \in \mf P} F_T(\xi_{\alpha,\p} - \epsilon )  + u_n < \alpha - \delta
	\end{equation}
	for some $\delta > 0$.
	
	We therefore have:
	\begin{align*}
	 & \lim_{n \to \infty} \inf_{\p \in \mf P} \p ( \Pi_n (\{ \theta : Q_n(\theta) \leq \xi_{\alpha,\p} - \epsilon\} | \mf X_n) < \alpha ) \\
	 & \geq \liminf_{n \to \infty} \inf_{\p \in \mf P} \p \Big( \Big\{ \Pi_n (\{ \theta : Q_n(\theta) \leq \xi_{\alpha,\p} - \epsilon\} | \mf X_n) < \alpha \Big\} \cap \mc A_n \Big) \\
	 & \geq \liminf_{n \to \infty} \inf_{\p \in \mf P} \p \Big( \Big\{ F_T(\xi_{\alpha,\p} - \epsilon )  + u_n < \alpha \Big\} \cap \mc A_n \Big) \\
	 & \geq 1 - \limsup_{n \to \infty} \sup_{\p \in \mf P} \ind\{ F_T(\xi_{\alpha,\p} - \epsilon )  + u_n \geq \alpha \}  - \limsup_{n \to \infty} \sup_{\p \in \mf P} \p( \mc A_n^c ) \\
	 & = 1
	\end{align*}
	where the final line is by (\ref{e:unif:ineq:1}) and definition of $\mc A_n$.
	
	The proof of part (ii) is similar.
\end{proof}

\begin{proof}[\textbf{Proof of Lemma \ref{l:post:profile:unif}}]
	It suffices to characterize the large-sample behavior of $R_{n}(z)$ defined in (\ref{e-Rn-profile}) uniformly in $\p$. By Lemma \ref{l:quad:unif} and Assumption \ref{a:quad:unif}(i)--(iii), there exist a positive sequence $(\varepsilon_n)_{n \in \mb N}$ independent of $z$ with $\varepsilon_n = o(1)$ and a sequence of events $(\mc A_n)_{n \in \mb N} \subset \mc F$ with $\inf_{\p \in \mf P} \p(\mc A_n) = 1-o(1)$ such that:
	\begin{align*}
	\sup_{\theta \in \Theta_{osn}} \left| PQ_n ( M (\theta))  - f(\sqrt n(\hat \gamma_n - \tau) - \sqrt n \gamma(\theta) ) \right| & \leq \varepsilon_n  \\
	\sup_{\theta \in \Theta_{osn}} \left| nL_n(\theta) - \ell_n - \frac{n}{2}\|\hat \gamma_n - \tau\|^2 +  \frac{1}{2} \|\sqrt n \gamma(\theta) - \sqrt n (\hat \gamma_n - \tau)\|^2 + f_{n,\bot} (\gamma_\perp(\theta)) \right| & \leq \varepsilon_n
	\end{align*}
	both hold on $\mc A_n$ for all $\p \in \mf P$. By similar arguments to the proof of Lemma \ref{l:post:profile}, wpa1 we obtain:
	\begin{align*}
	&  (1-\bar \varepsilon_n) e^{-2 \varepsilon_n} \frac{\nu_{d^*} ( (f^{-1}( z - \varepsilon_n )) \cap (\sqrt n(\hat \gamma_n - \tau) - T_{osn}))}{ \nu_{d^*} (\sqrt n(\hat \gamma_n - \tau) - T_{osn})}  \\
	& \leq R_n(z)
	\leq  (1+\bar \varepsilon_n) e^{2 \varepsilon_n} \frac{\nu_{d^*} ( (f^{-1}( z + \varepsilon_n ) )\cap (\sqrt n(\hat \gamma_n - \tau) - T_{osn}))}{ \nu_{d^*} (\sqrt n(\hat \gamma_n - \tau) - T_{osn})}
	\end{align*}
	uniformly in $z$ for all $\p \in \mf P$, for some positive sequence $(\bar \varepsilon_n)_{n \in \mb N}$ with $\bar \varepsilon_n = o(1)$.
		To complete the proof, it remains to show that:
	\[
	\sup_{z \in I} \left| \frac{\nu_{d^*} ( (f^{-1}( z + \varepsilon_n ) )\cap (\sqrt n(\hat \gamma_n - \tau) - T_{osn}))}{ \nu_{d^*} (\sqrt n(\hat \gamma_n - \tau) - T_{osn})}  -  \frac{\nu_{d^*} ( f^{-1}( z  )\cap (\sqrt n \hat \gamma_n -T))}{ \nu_{d^*} (\sqrt n \hat \gamma_n -T)} \right| = o_\p(1)
	\]
	uniformly in $\p$. This follows by the uniform continuity condition on $I$ in the statement of the lemma, using similar arguments to the proofs of Lemmas \ref{l:post:profile} and \ref{l:post:unif}.
\end{proof}

\begin{proof}[\textbf{Proof of Theorem \ref{t:profile:unif}}]
	We verify the conditions of Lemma \ref{l:basic:profile:unif}. We assume w.l.o.g. that $L_n(\hat \theta) = \sup_{\theta \in \Theta_{osn}} L_n(\theta) + o_\p(n^{-1})$ uniformly in $\p$. By display (\ref{e:quad:qlr:unif}) in Lemma \ref{l:quad:unif} we have $PQ_n(M_I) = f( \sqrt n (\hat \gamma_n - \tau) ) + o_\p(1)$ uniformly in $\p$. This verifies condition (i) with $W_n = f( \sqrt n (\hat \gamma_n - \tau) ) = f( \mf T (\mb V_n + \sqrt n \tau) - \sqrt n \tau)$.
	
	For condition (ii) let $\xi_{\alpha,\p}$ denote the $\alpha$ quantile of $f( Z)$ under $\p$ and let $(\varepsilon_n)_{n \in \mb N}$ be a positive sequence with $\varepsilon_n = o(1)$. By Assumption \ref{a:qlr:unif}(ii), the condition $\sup_{\p \in \mf P} \sup_z | \p (f (\mb V_n) \leq z) - \p_Z (f(Z) \leq z) | = o(1)$, and equicontinuity of $f(Z)$ at thier $\alpha$ quantiles, we have:
\begin{align*}
\liminf_{n \to \infty} \inf_{\p \in \mf P} \p( W_n \leq \xi_{\alpha,\p} - \varepsilon_n)
 & \geq \liminf_{n \to \infty} \inf_{\p \in \mf P} \p( f(\mb V_n) \leq \xi_{\alpha,\p} - \varepsilon_n) \\
 & \geq \liminf_{n \to \infty} \inf_{\p \in \mf P} \p_Z( f(Z) \leq \xi_{\alpha,\p} - \varepsilon_n) \\
 & = \alpha \,.
\end{align*}
By condition \ref{a:mcmc:profile:unif} it suffices to show that for each $\epsilon > 0$:
\[
 \lim_{n \to \infty} \sup_{\p \in \mf P} \p (\xi_{\alpha,\p} - \xi_{n,\alpha}^{post,p} > \epsilon) = 0\,.
\]
A sufficient condition is that
\[
 \lim_{n \to \infty} \inf_{\p \in \mf P} \p ( \Pi_n (\{ \theta : PQ_n(M(\theta)) \leq \xi_{\alpha,\p} - \epsilon\} | \mf X_n) < \alpha )= 1\,.
\]
By Lemma \ref{l:post:profile:unif} there exists a sequence of positive constants $(u_n)_{n \in \mb N}$ with $u_n = o(1)$ and a sequence of events $(\mc A_n)_{n \in \mb N}$ (possibly depending on $\p$) with $\inf_{\p \in \mf P} \p(\mc A_n) = 1-o(1)$ such that:
\[
 \Pi_n \big(\{ \theta : PQ_n(M(\theta)) \leq \xi_{\alpha,\p} - \epsilon\} \,\big|\,\mf X_n\big) \leq \p_{Z|\mf X_n} (f(Z) \leq \xi_{\alpha,\p} - \epsilon | Z \in \sqrt n \hat \gamma_n - T)  + u_n
\]
holds on $\mc A_n$ for each $\p$. But by Assumption \ref{a:qlr:unif}(iii) we may deduce that
\[
 \ \Pi_n \big(\{ \theta : PQ_n(M(\theta)) \leq \xi_{\alpha,\p} - \epsilon\} \,\big|\,\mf X_n\big) \leq \p_Z (f(Z) \leq \xi_{\alpha,\p} - \epsilon )  + u_n
\]
holds on $\mc A_n$ for each $\p$. By equicontinuity of the distribution of $\{f(Z) : \p \in \mf P\}$ we have:
	\[
	 \limsup_{n \to \infty} \sup_{\p \in \mf P} \p_Z (f(Z) \leq \xi_{\alpha,\p} - \epsilon )  + u_n  < \alpha - \delta
	\]
	for some $\delta > 0$. The result now follows by the same arguments as the proof of Theorem \ref{t:main:unif}.
\end{proof}

\begin{proof}[Proof of Lemma \ref{lem:quad:unif:disc}]
To simplify notation, let $D_{\theta;p} = \sqrt{\chi^2(p_\theta;p)}$. Define the {\it generalized score} of $\p_\theta$ with respect to $\p$ as $S_{\theta;p}(x) = g_{\theta;p}' e_{x}$ where
\[
 g_{\theta;p} = \frac{1}{D_{\theta;p}} \left[ \begin{array}{c}
 \frac{p_\theta(1)-p(1)}{p(1)}  \\
 \vdots \\
 \frac{p_\theta(k)-p(k)}{p(k)}  \end{array} \right]  \,.
\]
Note that $P S_{\theta;p} = 0$ and $P (S_{\theta;p}^2) = 1$. Also define $u_{\theta;p} = \mb J_p^{-1} g_{\theta;p}$ and notice that $u_{\theta;p}$ is a unit vector (i.e. $\|u_{\theta;p}\| = 1$). Therefore,
\begin{align}\label{e:disc:smax}
 |S_{\theta;p}(x)| \leq 1/(\min_{1 \leq j \leq k} \sqrt{p(j)})
\end{align} for each $\theta$ and $\p \in \mf P$.

 For any $p_\theta > 0$, a Taylor series expansion of $\log(u+1)$ about $u = 0$ yields
\begin{align}
 n L_n (p_\theta) - n L_n (p) & = n \p_n \log ( D_{\theta;p}S_{\theta;p}  + 1) \notag \\
 & = n D_{\theta;p} \p_n S_{\theta;p} - \frac{n D_{\theta;p}^2}{2} \p_n S_{\theta;p}^2 + nD_{\theta;p}^2 \p_n (S_{\theta;p}^2R(D_{\theta;p}S_{\theta;p})) \label{e:quad:disc:1}
\end{align}
where $R(u) \to 0$ as $u \to 0$.

By (\ref{e:disc:smax}), we may choose $(a_n)_{n \in \mb N}$ be a positive sequence with $a_n \to \infty$ as $n \to \infty$ such that $a_n \sup_{\theta: p_\theta > 0} \max_{1 \leq i \leq n} |S_{\theta;p}(X_i)| = o_\p(\sqrt n)$ (uniformly in $\p$). Then, for any $r_n \leq a_n$:
\begin{align} \label{e:quad:disc:2}
 \sup_{\theta \in \Theta_{osn}(\p)} \max_{1 \leq i \leq n} |D_{\theta;p}S_{\theta;p}(X_i)| = o_\p(1) \quad \mbox{(uniformly in $\p$).}
\end{align}
By the two-sided Chernoff bound, for any $\delta \in (0,1)$:
\begin{equation} \label{e:chernoff:unif}
 \sup_{\p \in \mf P} \p \Big( \max_{1 \leq j \leq k} \Big|\frac{\p_n \ind\{x=j\}}{p(j)}-1\Big| > \delta \Big) \leq 2ke^{-n (\inf_{\p \in P} \min_{1 \leq j \leq k} p(j)) \frac{ \delta^2}{3}} \to 0
\end{equation}
because $\sup_{\p \in \mf P} \max_{1 \leq j \leq k} (1/p(j)) = o(n)$. It follows that $\p_n (\mb J_p e_x^{\phantom \prime} e_x' \mb J_p) = I+o_\p(1)$ uniformly in $\p$. Also notice that $S_{\theta;p}^2(x) = u_{\theta;p}'\mb J_p e_x^{\phantom \prime} e_x' \mb J_p u_{\theta;p}$ where each $u_{\theta;p}$ is a unit vector. Therefore,
\begin{align} \label{e:quad:disc:3}
 \sup_{\theta : p_\theta > 0} |\p_n S_{\theta;p}^2 - 1| = o_\p(1) \quad \mbox{(uniformly in $\p$).}
\end{align}
Substituting (\ref{e:quad:disc:2}) and (\ref{e:quad:disc:3}) into (\ref{e:quad:disc:1}) yields:
\[
 nL_n(p_\theta) - nL_n(p) =  n D_{\theta;p} \p_n S_{\theta;p} - \frac{nD_{\theta;p}^2}{2} + nD_{\theta;p}^2 \times o_\p(1)
\]
where the $o_\p(1)$ term holds uniformly for all $\theta$ with $p_\theta > 0$, uniformly for all $\p \in \mf P$. We may therefore choose a positive sequence $(b_n)_{n \in \mb N}$ with $b_n \to \infty$ slowly such that $b_n^2$ times the $o_\p(1)$ term is still $o_\p(1)$ uniformly in $\p$. Letting $r_n = (a_n \wedge b_n)$, we obtain
\[
 \sup_{\theta \in \Theta_{osn}(\p)} \left| nL_n(p_\theta) - nL_n(p) -  n D_{\theta;p} \p_n S_{\theta;p} + \frac{nD_{\theta;p}^2}{2} \right| = o_\p(1) \quad \mbox{(uniformly in $\p$)}
\]
where $n D_{\theta;p} \p_n S_{\theta;p} = \sqrt n D_{\theta;p} \mb G_n (S_{\theta;p}) = \sqrt n \tilde \gamma_{\theta;p} \mb G_n (\mb J_p e_x)$ and $D_{\theta;p}^2 = \| \tilde \gamma_{\theta;p}\|^2$.
\end{proof}

\begin{proof}[Proof of Proposition \ref{p:quad:unif:disc}]
The quadratic expansion follows from Lemma \ref{lem:quad:unif:disc} and (\ref{e:quad:unif:disc:1}) and (\ref{e:quad:unif:disc:2}), which give $\| \tilde \gamma_{\theta;p} \|^2 = \tilde \gamma_{\theta;p}' \tilde \gamma_{\theta;p}^{\phantom \prime} = \tilde \gamma_{\theta;p}' V_p' V_p^{\phantom \prime} \tilde \gamma_{\theta;p}^{\phantom \prime} = \gamma(\theta)'\gamma(\theta)$ and  $\tilde \gamma_{\theta;p}'\tilde{\mb V}_{n,p} = \tilde \gamma_{\theta;p}'V_p' V_p^{\phantom \prime}\tilde{\mb V}_{n,p} = \gamma(\theta)' \mb V_n$.

Uniform convergence in distribution is by Proposition A.5.2 of \cite{vdVW}, since $\sup_{\p \in \mf P} \max_{1 \leq j \leq k} (1/p(j)) = o(n)$ implies $\sup_{\p \in \mf P} |v_{j,p}' \mb J_p e_x| \leq  1/(\min_{1 \leq j \leq k} \sqrt{p(j)}) = o(n^{1/2})$.
\end{proof}

\begin{proof}[Proof of Proposition \ref{p:quad:unif:disc:sieve}]
The condition $\sup_{\p \in \mf P} \max_{1 \leq j \leq k} (1/p(j)) = o(n/\log k)$ ensures that display (\ref{e:chernoff:unif}) holds with $k = k(n) \to \infty$. The rest of the proof follows that of Proposition \ref{p:quad:unif:disc}.
\end{proof}

\begin{proof}[Proof of Lemma \ref{lem:disc:cov}]
For any $\p \in \mf P$, the mapping $p_\theta \mapsto V_p \tilde \gamma_{\theta;p}$ is a homeomorphism because $p > 0$ and $V_p$ is an orthogonal matrix. Recall that the upper $k-1$ elements of $V_p\tilde \gamma_{\theta;p}$ is the vector $\gamma(\theta) = \gamma(\theta;\p)$ and the remaining $k$th element is zero. Therefore, for each $\p \in \mf P$ the mapping $p_\theta \mapsto \gamma(\theta)$ is a homeomorphism. Since  $\{p_\theta : \theta \in \Theta, p_\theta > 0\} = \mr{int}(\Delta^{k-1})$ and $p \in \mr{int}(\Delta^{k-1})$ for each $\p \in \mf P$, it follows that $\{\gamma(\theta) : \theta \in \Theta, p_\theta > 0\}$ contains a ball of radius $\epsilon = \epsilon(\p) > 0$ for each $\p \in \mf P$ (because homeomorphisms map interior points to interior points).

Recall that $\theta \in \Theta_{osn}(\p)$ if and only if $\|\gamma(\theta)\| \leq r_n/\sqrt n$ (because $\|\gamma(\theta)\|^2 = \|\tilde \gamma_{\theta;p}\|^2 = \chi^2(p_\theta;p)$). Let $\epsilon(\p) = \sup\{\epsilon > 0 : B_\epsilon \subseteq \{\gamma(\theta) : \theta \in \Theta, p_\theta > 0\}\}$. It suffices to show that $\inf_{\p \in \mf P} \sqrt n \epsilon(\p) \to \infty$ as $n \to \infty$. We can map back from any $\gamma \in \mb R^{k-1}$ by the inverse mapping $q_{\gamma;p}$ given by
\[
 q_{\gamma;p}(j) = p(j) + \sqrt{p(j)} [ V_p^{-1} ((\gamma' \; 0)')]_j
\]
for $1 \leq j \leq k$, where $[ V_p^{-1} ((\gamma' \; 0)')]_j$ denotes the $j$th element of $[ V_p^{-1} ((\gamma' \; 0)')]$. An equivalent definition of $\epsilon(\p)$ is $\inf\{ \epsilon > 0 : q_p(\gamma) \not \in \mr{int}(\Delta^{k-1}) \mbox{ for some } \gamma \in B_\epsilon\}$. As $p > 0$ and $\sum_{j=1}^k q_{\gamma;p}(j) = 1$ for each $\gamma$ by construction, we therefore need to find the smallest $\epsilon > 0$ for which $q_{\gamma;p}(j) \leq 0$ for some $j$, for some $\gamma \in B_\epsilon$. This is equivalent to finding the smallest $\epsilon > 0$ for which
\begin{equation}\label{e:disc:cov:1}
 \frac{1}{\sqrt{p(j)}} \geq \frac{1}{[ V_p^{-1} ((\gamma' \; 0)')]_j}
\end{equation}
for some $j$, for some $\gamma \in B_\epsilon$. The left-hand side is $o(\sqrt n)$ uniformly for $1 \leq j \leq k$ and uniformly in $\p$ under the condition $\sup_{\p \in \mf P} \max_{1 \leq j \leq k} (1/p(j)) = o(n)$. Also notice that, since the $\ell^2$ norm dominates the maximum norm and $V_p$ is an orthogonal matrix, we have
\begin{equation}\label{e:disc:cov:2}
 \frac{1}{[ V_p^{-1} ((\gamma' \; 0)')]_j}  \geq \frac{1}{\| V_p^{-1} ((\gamma' \; 0)')\|} = \frac{1}{\| \gamma\|} \geq \frac{1}{\epsilon}
\end{equation}
It follows from (\ref{e:disc:cov:1}) and (\ref{e:disc:cov:2}) that $\sqrt n \inf_{\p \in \mf P}\epsilon(\p) \geq \frac{\sqrt n}{o(\sqrt n)} \to \infty$ as $n \to \infty$, as required.
\end{proof}

\begin{proof}[Proof of Lemma \ref{lem:profile:disc}]
Condition (\ref{e:chernoff:unif}) implies that
\begin{align*}
 & \sup_{\theta \in \Theta_{osn}(\p)} \sup_{\mu \in M(\theta)} \left| \sup_{\eta \in H_\mu} n L_n(p_{\mu,\eta}) - \sup_{\eta \in H_\mu : (\mu,\eta) \in \Theta_{osn}(\p)} nL_n(p_{\mu,\eta}) \right| \\
 & = \sup_{\theta \in \Theta_{osn}(\p)} \sup_{\mu \in M(\theta)} \left| \inf_{\eta \in H_\mu} n D_{KL}(p\|p_{\mu,\eta}) - \inf_{\eta \in H_\mu : (\mu,\eta) \in \Theta_{osn}(\p)} n D_{KL}(p\|p_{\mu,\eta}) \right|(1+o_\p(1))
\end{align*}
where the $o_\p(1)$ term holds uniformly in $\p$ and $D_{KL}(p\|p_\theta) = \sum_{j=1}^k p(j) \log (p(j)/p_\theta(j))$. By a Taylor expansion of $-\log(u+1)$ about $u = 0$, it is straightforward to deduce that
\begin{align} \label{e:div:equiv}
 \lim_{\epsilon \to 0} \sup_{\p \in \mf P} \sup_{\theta \in \Theta : \chi^2(p_\theta;p) \leq \epsilon}\left| \frac{D_{KL}(p\|p_\theta)}{\frac{1}{2}\chi^2(p_\theta;p)}-1 \right| = o(1) \,.
\end{align}
In particular, for any $\theta \in \Theta_{osn}(\p)$ and any $\mu \in M(\theta)$, we have
\begin{align} \label{e:div:equiv:1}
 \inf_{\eta \in H_\mu} D_{KL}(p\|p_{\mu,\eta}) \leq \inf_{\eta \in H_\mu : (\mu,\eta) \in \Theta_{osn}(\p)} D_{KL}(p\|p_{\mu,\eta}) \leq \frac{\chi^2(p_\theta;p)}{2}(1+o(1))
\end{align}
uniformly in $\p$. We want to show that an equivalence (\ref{e:div:equiv}) holds uniformly over shrinking $KL$-divergence neighborhoods (rather $\chi^2$-divergence neighborhoods). By similar arguments to Lemma 3.1 in \cite{LiuShao}, we may deduce that
\begin{align*}
 \frac{1}{\chi^2(p_\theta;p)}|4h^2(p_\theta,p) - \chi^2(p_\theta;p)| \leq \frac{3}{D_{\theta;p}} \max_x |S_{\theta,p}(x)|  h^2(p_\theta,p)
\end{align*}
where again $D_{\theta;p} = \sqrt{\chi^2(p_\theta;p)}$. But, $h(p_\theta,p) \leq D_{\theta;p}$. Moreover, the proof of Proposition \ref{p:quad:unif:disc} also shows that $|S_{\theta;p}|\leq 1/(\min_{1 \leq j \leq k} \sqrt{p(j)})$ holds for each $\theta$ and each $\p \in \mf P$ so $\max_x |S_{\theta,p}(x)| = o(\sqrt n)$ uniformly in $\p$. This, together with the fact that $h(p_\theta,p) \leq \sqrt{D_{KL}(p\|p_\theta)}$, yields
\begin{align*}
 \frac{1}{\chi^2(p_\theta;p)}|4h^2(p_\theta,p) - \chi^2(p_\theta;p)| \leq o(\sqrt n) \times \sqrt{D_{KL}(p\|p_\theta)}
\end{align*}
where the $o(\sqrt n)$ term holds uniformly for $\theta \in \Theta$ and $\p \in \mf P$. Let $(a_n)_{n \in \mb N}$ be a positive sequence with $a_n \leq r_n$ and $a_n \to \infty$ sufficiently slowly that $a_n$ times the $o(\sqrt n)$ term in the above display is still $o(\sqrt n)$ (uniformly in $\theta$ and $\p$). We then have
\[
 \sup_{\p \in \mf P} \sup_{\theta : D_{KL}(p\|p_\theta) \leq \frac{a_n}{\sqrt n}} \frac{1}{\chi^2(p_\theta;p)}|4h^2(p_\theta,p) - \chi^2(p_\theta;p)| = o(1)\,.
\]
Since $h^2(p_\theta,p) \leq D_{KL}(p\|p_\theta)$, this implies that
\[
 \sup_{\p \in \mf P} \sup_{\theta : D_{KL}(p\|p_\theta) \leq \frac{a_n}{\sqrt n}} \chi^2(p_\theta;p) = o(1)
\]
and so, by (\ref{e:div:equiv}), we obtain
\[
 \sup_{\p \in \mf P} \sup_{\theta \in \Theta : D_{KL}(p\|p_\theta) \leq \frac{a_n}{\sqrt n}}\left| \frac{D_{KL}(p\|p_\theta)}{\frac{1}{2}\chi^2(p_\theta;p)}-1 \right| = o(1) \,.
\]
It now follows by (\ref{e:div:equiv}) that
\begin{align*}
 & \sup_{\theta : \chi^2(p_\theta;p) \leq \frac{a_n}{n}} \sup_{\mu \in M(\theta)} \left| \inf_{\eta \in H_\mu} n D_{KL}(p\|p_{\mu,\eta}) - \inf_{\eta \in H_\mu :\chi^2(p_{(\mu,\eta)};p) \leq \frac{a_n}{n} } n D_{KL}(p\|p_{\mu,\eta}) \right|(1+o_\p(1))  \\
 & = \frac{n}{2} \sup_{\theta : \chi^2(p_\theta;p) \leq \frac{a_n}{n}} \sup_{\mu \in M(\theta)} \left| \inf_{\eta \in H_\mu} \chi^2(p_{\mu,\eta};p) - \inf_{\eta \in H_\mu :\chi^2(p_{\mu,\eta};p) \leq \frac{a_n}{n} } \chi^2(p_{\mu,\eta};p) \right|(1+o_\p(1))
\end{align*}
where $ \inf_{\eta \in H_\mu} \chi^2(p_{\mu,\eta};p) - \inf_{\eta \in H_\mu :\chi^2(p_{\mu,\eta};p) \leq \frac{a_n}{n} } \chi^2(p_{\mu,\eta};p) = 0$ because $\chi^2(p_\theta;p) \leq \frac{a_n}{n}$ and $\mu \in M(\theta)$ implies that there exists an $\eta \in H(\mu)$ with $p_\theta = p_{\mu,\eta}$, so the constraint $\chi^2(p_{\mu,\eta};p) \leq \frac{a_n}{n}$ is never violated for any $\mu \in M(\theta)$, for any such $\theta$. The result follows by taking $r_n' = a_n$.
\end{proof}

\subsection{Proofs for Appendix \ref{s:lp}}

\begin{proof}[\textbf{Proof of Theorem \ref{t:lp}}]
We first derive the asymptotic distribution of $\sup_{\theta \in \Theta_I} Q_n(\theta)$ under $\mr P_{n,a}$. By similar arguments to the proof of Theorem \ref{t:main}, we have:
\begin{align*}
 \sup_{\theta \in \Theta_I} Q_n(\theta) & = \| \mb V_n \|^2 + o_{\mr P_{n,a}}(1) \; \overset{\mr P_{n,a}}{\rightsquigarrow} \chi^2_{d^*}(a'a)\,.
\end{align*}
Identical arguments to the proof of Lemma \ref{l:post} yield:
\[
 \sup_{z} | \Pi_n (\{ \theta : Q_n(\theta) \leq z\} | \mf X_n) - F_{\chi^2_{d^*}}(z)| = o_{\mr P_{n,a}}(1)\,.
\]
Therefore, $\xi_{n,\alpha}^{mc} = \chi^2_{d^*,\alpha} + o_{\mr P_{n,a}}(1)$ and we obtain:
\[
 \mr P_{n,a} (\Theta_I \subseteq \wh \Theta_\alpha) = \Pr ( \chi^2_{d^*}(a'a) \leq \chi^2_{d^*,\alpha} )+ o(1)
\]
as required.
\end{proof}

\begin{proof}[\textbf{Proof of Theorem \ref{t:lp:profile}}]
By similar arguments to the proof of Theorem \ref{t:main:profile}, we have:
\[
 PQ_n(M_I) = f(\mb V_n) + o_{\mr P_{n,a}}(1)  \overset{\mr P_{n,a}}{\rightsquigarrow} f(Z + a)
\]
where $Z \sim N(0,I_{d^*})$. Identical arguments to the proof of Lemma \ref{l:post:profile} yield:
\[
 \sup_{z  \in I} \left| {\textstyle \Pi_n \big( \{\theta:PQ_n( M(\theta)) \leq  z \} \,\big|\, \mf X_n \big) }  - \mb P_{Z| \mf X_n} \big( f ( Z )  \leq z\big) \right| = o_{\mr P_{n,a}}(1)
\]
for a neighborhood $I$ of $z_\alpha$. Therefore, $\xi_{n,\alpha}^{mc,p} = z_\alpha + o_{\mr P_{n,a}}(1)$ and we obtain:
\[
 \mr P_{n,a} (M_I \subseteq \wh M_\alpha) = \p_Z  ( f(Z+a) \leq z_\alpha )+ o(1)
\]
as required.
\end{proof}

\subsection{Proofs for Appendix \ref{ax:pds}}

\begin{proof}[\textbf{Proof of Lemma \ref{l:post:gamma}}]
By equations (\ref{e-post-target}) and (\ref{e-denombd}) in the proof of Lemma \ref{l:post}, it suffices to characterize the large-sample behavior of:
\[
	R_{n}(z) :=  \frac{\int_{\{\theta:Q_n(\theta) \leq  z \} \cap \Theta_{osn}}\! e^{-\frac{1}{2}Q_n(\theta)} \mr d\Pi(\theta)}{\int_{\Theta_{osn}} \!e^{-\frac{1}{2}Q_n(\theta)} \mr d\Pi(\theta)}\,.
\]
By Assumption \ref{a:quad:gamma}(i), there exists a positive sequence $(\varepsilon_n)_{n \in \mb N}$ with $\varepsilon_n = o(1)$ such that: $(1-\varepsilon_n) h(\gamma(\theta) - \hat \gamma_n) \leq \frac{a_n}{2}Q_n(\theta) \leq (1+\varepsilon_n) h(\gamma(\theta) - \hat \gamma_n)$ holds uniformly over $\Theta_{osn}$. Therefore:
\begin{align*}
 & \frac{\int_{\{\theta:2a_n^{-1}(1+\varepsilon_n)h(\gamma(\theta) -\hat \gamma_n) \leq  z \} \cap \Theta_{osn}}\! e^{-a_n^{-1}(1+\varepsilon_n)h(\gamma(\theta) -\hat \gamma_n)} \mr d\Pi(\theta)}{\int_{\Theta_{osn}} \!e^{-a_n^{-1}(1-\varepsilon_n)h(\gamma(\theta) -\hat \gamma_n)} \mr d\Pi(\theta)} \\
 & \leq R_n(z) \leq \frac{\int_{\{\theta:2a_n^{-1}(1-\varepsilon_n)h(\gamma(\theta) -\hat \gamma_n) \leq  z \} \cap \Theta_{osn}}\! e^{-a_n^{-1}(1-\varepsilon_n)h(\gamma(\theta) -\hat \gamma_n)} \mr d\Pi(\theta)}{\int_{\Theta_{osn}} \!e^{-a_n^{-1}(1+\varepsilon_n)h(\gamma(\theta) -\hat \gamma_n)} \mr d\Pi(\theta)}\,.
\end{align*}
By similar arguments to the proof of Lemma \ref{l:post}, under Assumption \ref{a:prior} there exists a positive sequence $(\bar \varepsilon_n)_{n \in \mb N}$ with $\bar \varepsilon_n = o(1)$ such that for all $n$ sufficiently large we have:
\begin{align*}
 & (1-\bar \varepsilon_n)\frac{\int_{\{\gamma:2a_n^{-1}(1+\varepsilon_n)h(\gamma -\hat \gamma_n) \leq  z \} \cap \Gamma_{osn}}\! e^{-a_n^{-1}(1+\varepsilon_n)h(\gamma -\hat \gamma_n)} \mr d\gamma}{\int_{\Gamma_{osn}} \!e^{-a_n^{-1}(1-\varepsilon_n)h(\gamma -\hat \gamma_n)} \mr d\gamma} \\
 & \leq R_n(z) \leq (1+\bar \varepsilon_n)\frac{\int_{\{\gamma:2a_n^{-1}(1-\varepsilon_n)h(\gamma -\hat \gamma_n) \leq  z \} \cap \Gamma_{osn}}\! e^{-a_n^{-1}(1-\varepsilon_n)h(\gamma -\hat \gamma_n)} \mr d\gamma}{\int_{\Gamma_{osn}} \!e^{-a_n^{-1}(1+\varepsilon_n)h(\gamma -\hat \gamma_n)} \mr d\gamma}\,.
\end{align*}
under the change of variables $\theta \mapsto \gamma(\theta)$, where $\Gamma_{osn} = \{\gamma(\theta) : \theta \in \Theta_{osn}\}$.

Assumption \ref{a:quad:gamma}(ii) implies that:
\[
 a_n^{-1}(1\pm\varepsilon_n)h(\gamma -\hat \gamma_n) = h\Big(a_n^{-r_1}(1\pm\varepsilon_n)^{r_1}(\gamma_1 -\hat \gamma_{n,1}),\ldots,a_n^{-r_{d^*}}(1\pm\varepsilon_n)^{r_{d^*}}(\gamma_{d^*} -\hat \gamma_{n,d^*})\Big)\,.
\]
Using a change of variables:
\[
 \gamma \mapsto \kappa_\pm(\gamma) = \big(a_n^{-r_1}(1\pm\varepsilon_n)^{r_1}(\gamma_1 -\hat \gamma_{n,1}),\ldots,a_n^{-r_{d^*}}(1\pm\varepsilon_n)^{r_{d^*}}(\gamma_{d^*} -\hat \gamma_{n,d^*})\big)
\]
(with choice of sign as appropriate) and setting $r^* = r_1 + \ldots + r_{d^*}$, we obtain:
\begin{align}
 & (1-\bar \varepsilon_n)\frac{(1-\varepsilon_n)^{r^*}}{(1+\varepsilon_n)^{r^*}}\frac{\int_{\{\kappa:2h(\kappa) \leq  z \} \cap K_{osn}^+}\! e^{-h(\kappa)} \mr d\kappa}{\int  \!e^{-h(\kappa)} \mr d\kappa} \notag \\
 & \leq R_n(z) \leq (1+\bar \varepsilon_n)\frac{(1+\varepsilon_n)^{r^*}}{(1-\varepsilon_n)^{r^*}}\frac{\int_{\{\kappa:2h(\kappa) \leq  z \}}\! e^{-h(\kappa)} \mr d\kappa}{\int_{K_{osn}^+} \!e^{-h(\kappa)} \mr d\kappa} \label{e:pf:gamma:0}
\end{align}
uniformly in $z$, where $K^+_{osn} = \{ \kappa_{+}(\gamma) : \gamma \in \Gamma_{osn}\}$.

We can use a change of variables for $\kappa \mapsto t = 2h(\kappa)$ to obtain:
\begin{align} \label{e:pf:gamma:1}
 \int_{\{\kappa : h(\kappa ) \leq z/2\}} e^{-h(\kappa)} \mr d \kappa & = 2^{-r^*} \mr{V}(S) \int_0^{z} e^{-t/2} t^{r^* -1} dt  & \int e^{-h(\kappa)} \mr d \kappa & = 2^{-r^*} \mr{V}(S) \int_0^\infty e^{-t/2} t^{r^* -1} dt
\end{align}
where $\mr{V}(S)$ denotes the volume of the set $S = \{ \kappa : h(\kappa) = 1\}$.

For the remaining integrals over $K_{osn}^+$ we first fix any $\omega \in \Omega$ so that $K_{osn}^+(\omega)$ becomes a deterministic sequence of sets. Let $C_n(\omega) = K_{osn}^+(\omega) \cap B_{k_n}$. Assumption \ref{a:quad:gamma}(iii) gives $\mb R^{d^*}_+ = \ol{\cup_{n \geq 1} C_n(\omega)}$ for almost every $\omega$. Now clearly:
\[
 \int e^{-h(\kappa)} \mr d \kappa \geq \int_{K_{osn}^+(\omega)} e^{-h(\kappa)} \mr d \kappa \geq \int \ind\{\kappa \in C_n(\omega)\} e^{-h(\kappa)} \mr d \kappa \to \int e^{-h(\kappa)} \, \mr d\kappa
\]
(by dominated convergence) for almost every $\omega$. Therefore:
\begin{align} \label{e:pf:gamma:2}
 \int_{K_{osn}^+} \!e^{-h(\kappa)} \mr d\kappa \to_{p} 2^{-r^*} \mr{V}(S) \int_0^\infty e^{-t/2} t^{r^* -1} dt\,.
\end{align}
We may similarly deduce that:
\begin{align} \label{e:pf:gamma:3}
 \sup_z \left| \int_{\{\kappa : h(\kappa ) \leq 2z\}\cap K_{osn}^+} e^{-h(\kappa)} \mr d \kappa  - 2^{-r^*} \mr{V}(S) \int_0^{z} e^{-t/2} t^{r^* -1} dt \right| \text{x}o_p 0\,.
\end{align}
The result follows by substituting (\ref{e:pf:gamma:1}), (\ref{e:pf:gamma:2}), and (\ref{e:pf:gamma:3}) into (\ref{e:pf:gamma:0}).
\end{proof}

\begin{proof}[\textbf{Proof of Theorem \ref{t:main:gamma}}]
We verify the conditions of Lemma \ref{l:basic}.
Lemma \ref{l:post:gamma} shows that the posterior distribution of the QLR is asymptotically $F_\Gamma = \Gamma(r^*,1/2)$, and hence $\xi_{n,\alpha}^{post} = z_\alpha + o_\p(1)$, where $z_\alpha$ denotes the $\alpha$ quantile of the $F_\Gamma$. By Assumption $\sup_{\theta \in \Theta_I} Q_n(\theta) \rightsquigarrow F_\Gamma$. Then:
\begin{align*}
 \xi_{n,\alpha}^{mc} & = z_\alpha + (\xi_{n,\alpha}^{post} - z_\alpha) + (\xi_{n,\alpha}^{mc} - \xi_{n,\alpha}^{post})
  = z_\alpha + o_\p(1)
\end{align*}
where the final equality is by Assumption \ref{a:mcmc}.
\end{proof}

\end{document}